%% file: main-ccs.tex
\documentclass[sigconf]{acmart}

\usepackage[english]{babel}
\usepackage[utf8]{inputenc}
\usepackage[title]{appendix}
\usepackage{amsfonts, amsmath, amsthm, amssymb}
\usepackage{framed}
\usepackage[scr=boondoxo]{mathalfa}
\usepackage{pifont}
\usepackage{mathpartir}
\usepackage{stmaryrd}
\usepackage{xifthen}
\usepackage{xspace}
\usepackage{enumerate}
\usepackage[shortlabels]{enumitem}
\usepackage{multirow}
\usepackage{xspace}
\usepackage{tikz}

\newboolean{techreport}
\setboolean{techreport}{true}

\newboolean{draft}
\setboolean{draft}{false}

\begin{document}

\ifthenelse{\boolean{techreport}}{%
\title{A Type System for Privacy Properties (Technical Report)}
}
{
\title{A Type System for Privacy Properties}
}

\author{Véronique Cortier}
\affiliation{%
  \institution{CNRS, LORIA}
  \city{Nancy} 
  \country{France} 
}
\email{veronique.cortier@loria.fr}

\author{Niklas Grimm}
\affiliation{%
  \institution{TU Wien}
  \city{Vienna}
  \country{Austria} 
}
\email{niklas.grimm@tuwien.ac.at}

\author{Joseph Lallemand}
\affiliation{%
  \institution{Inria, LORIA}
  \city{Nancy} 
  \country{France} 
}
\email{joseph.lallemand@loria.fr}

\author{Matteo Maffei}
\affiliation{%
  \institution{TU Wien}
  \city{Vienna}
  \country{Austria} 
}
\email{matteo.maffei@tuwien.ac.at}

\begin{abstract}
Mature push button tools have emerged for checking trace properties
(e.g. secrecy or authentication) of security protocols. The case of indistinguishability-based
privacy properties (e.g. ballot privacy or anonymity) is more complex
and constitutes an active research topic with several recent propositions of techniques
and tools.

We explore a novel approach based on type systems and provide a
(sound) type system for proving equivalence of protocols, for a
bounded or an unbounded number of sessions. The resulting prototype
implementation has been tested on various protocols of the
literature. It provides a significant speed-up (by orders of magnitude) compared to tools for a
bounded number of sessions and complements in terms of expressiveness
other state-of-the-art tools, such as ProVerif and Tamarin: e.g., we
show that our analysis technique is the first one to handle a faithful
encoding of the Helios e-voting protocol \veroniquebis{in the context of
  an untrusted ballot box}.
\end{abstract}

\input{macros-ccs}

\ifthenelse{\boolean{techreport}}{%
\renewcommand\footnotetextcopyrightpermission[1]{} 
}
{

\copyrightyear{2017}
\acmYear{2017}
\setcopyright{acmlicensed}
\acmConference{CCS '17}{October 30-November 3, 2017}{Dallas, TX, USA}\acmPrice{15.00}\acmDOI{10.1145/3133956.3133998}
\acmISBN{978-1-4503-4946-8/17/10}


}

\newif\ifallrules
\allrulesfalse

\maketitle

\input{introduction}
\input{overview}

\input{calculus}

\input{typing}
\input{consistency}
\input{results}
\input{experiments}
\input{conclusion}

\clearpage

\bibliographystyle{ACM-Reference-Format}
\bibliography{references}

%
%
\ifthenelse{\boolean{techreport}}{%
\begin{appendices}
\onecolumn
\input{proofs/typing2}

\input{proofs/newproofs}
\input{proofs/typerepl}
\input{proofs/consistency2}

\input{proofs/consistency-repl}

\end{appendices}}
{}
\end{document}

%% file: macros-ccs.tex

\newcommand{\mypar}[1]{\smallskip\noindent\textbf{#1.}}

\renewcommand{\infer}[3][]{\inferrule{#2}{#3}\ifthenelse{\isempty{#1}}{}{\;(\textsc{#1})}}

	\newcommand{\nonce}[3]{\ensuremath{\tau^{#1,#2}_{#3}}\xspace}

\newcommand{\TypeEq}{TypeEq}

\ifthenelse{\boolean{draft}}{
\newcommand{\fix}[3]{{#1\textbf{[#2:} #3\textbf{]}}}
}{
\newcommand{\fix}[3]{}
}

\ifthenelse{\boolean{draft}}{
\newcommand{\colordraft}[2]{{#1 #2}}
}{
\newcommand{\colordraft}[2]{{#2}}
}

\newcommand{\case}[1]{\underline{#1}}

\newcommand{\todo}[1]{\fix{\color{orange}}{TODO}{#1}}
\newcommand{\joseph}[1]{\fix{\color{blue}}{j}{#1}}
\newcommand{\josephbis}[1]{\colordraft{\color{blue}}{#1}}
\newcommand{\veronique}[1]{\fix{\color{violet}}{V}{#1}}
\newcommand{\veroniquebis}[1]{\colordraft{\color{violet}}{#1}}
\newcommand{\niklas}[1]{\fix{\color{red}}{N}{#1}}
\newcommand{\matteo}[1]{\fix{\color{red}}{M}{#1}}


\newcommand{\eg}{\emph{e.g.}\xspace}
\newcommand{\ie}{\emph{i.e.}\xspace}
\newcommand{\wrt}{\emph{w.r.t.}\xspace}
\newcommand{\cf}{\emph{cf.}\xspace}
\newcommand{\etal}{\emph{et al.}\xspace}

\newcommand{\eqdef}{\mathop{\overset{\mathrm{def}}{=}}}

\renewenvironment{description}
  {\begin{list}{}{
    \setlength{\labelwidth}{0pt}
    \setlength{\leftmargin}{1.5\parindent}
    \setlength{\itemindent}{-\leftmargin}
    \setlength{\itemsep}{0pt}
    \let\makelabel\descriptionlabel}}
  {\end{list}}

\newcommand{\proc}[1]{\mathtt{#1}}
\newcommand{\NEWNA}{\proc{new}}
\newcommand{\NEWN}[1]{\NEWNA\; #1}
\newcommand{\NEW}[2]{\NEWNA\; #1:#2}
\newcommand{\NEWnew}[2]{\NEWNA\; #1 : #2}
\newcommand{\OUTNA}{\proc{out}}
\newcommand{\OUT}[1]{\OUTNA(#1)}
\newcommand{\INNA}{\proc{in}}
\newcommand{\IN}[1]{\INNA(#1)}
\newcommand{\PAR}{~|~}
\newcommand{\REPL}{\proc{!}}
\newcommand{\ZERO}{\proc{0}}
\newcommand{\LETNA}{\proc{let}}
\newcommand{\LETINNA}{\proc{in}}
\newcommand{\LETELSENA}{\proc{else}}
\newcommand{\LETIN}[2]{\LETNA\; #1 = #2 \;\LETINNA}
\newcommand{\LET}[4]{\LETIN{#1}{#2}\; #3 \;\LETELSENA\; #4}
\newcommand{\IFNA}{\proc{if}}
\newcommand{\THENNA}{\proc{then}}
\newcommand{\ELSENA}{\proc{else}}
\newcommand{\IFTHEN}[2]{\IFNA\;#1 = #2 \;\THENNA}
\newcommand{\ITE}[4]{\IFTHEN{#1}{#2}\; #3 \;\ELSENA\; #4}
\newcommand{\INLNA}{\proc{inl}}
\newcommand{\INL}[1]{\INLNA(#1)}
\newcommand{\INRNA}{\proc{inr}}
\newcommand{\INR}[1]{\INRNA(#1)}
\newcommand{\MATCHNA}{\proc{match}}
\newcommand{\MATCH}[5]{\MATCHNA\; #1 \;\proc{with}\; \INL{#2} \Rightarrow  #3~|~ \INR{#4} \Rightarrow #5}
\newcommand{\CHOICENA}{\proc{choice}}
\newcommand{\CHOICE}[2]{\CHOICENA(#1,#2)}

\newcommand{\FN}{\mathcal{FN}}
\newcommand{\BN}{\mathcal{BN}}

\newcommand{\fsym}[1]{\mathtt{#1}}
\newcommand{\PAIRNA}{\langle \cdot, \cdot \rangle}
\newcommand{\PAIR}[2]{\langle #1, #2 \rangle}
\newcommand{\HASHNA}{\fsym{h}}
\newcommand{\HASH}[1]{\HASHNA(#1)}
\newcommand{\FSTNA}{\pi_1}
\newcommand{\FST}[1]{\FSTNA(#1)}
\newcommand{\SNDNA}{\pi_2}
\newcommand{\SND}[1]{\SNDNA(#1)}
\newcommand{\ENCNA}{\fsym{enc}}
\newcommand{\ENC}[2]{\ENCNA(#1,#2)}
\newcommand{\DECNA}{\fsym{dec}}
\newcommand{\DEC}[2]{\DECNA(#1,#2)}
\newcommand{\AENCNA}{\fsym{aenc}}
\newcommand{\AENC}[2]{\AENCNA(#1,#2)}
\newcommand{\ADECNA}{\fsym{adec}}
\newcommand{\ADEC}[2]{\ADECNA(#1,#2)}
\newcommand{\SIGNNA}{\fsym{sign}}
\newcommand{\SIGN}[2]{\SIGNNA(#1,#2)}
\newcommand{\CHECKNA}{\fsym{checksign}}
\newcommand{\CHECK}[2]{\CHECKNA(#1,#2)}
\newcommand{\PUBKNA}{\fsym{pk}}
\newcommand{\PUBK}[1]{\PUBKNA(#1)}
\newcommand{\VKNA}{\fsym{vk}}
\newcommand{\VK}[1]{\VKNA(#1)}
\newcommand{\MATCHLNA}{\fsym{matchl}}
\newcommand{\MATCHL}[1]{\MATCHLNA(#1)}
\newcommand{\MATCHRNA}{\fsym{matchr}}
\newcommand{\MATCHR}[1]{\MATCHRNA(#1)}

\newcommand{\F}{\mathcal{F}}
\newcommand{\N}{\mathcal{N}}
\newcommand{\Names}{\mathit{\mathcal{NK}}}
\newcommand{\X}{\mathcal{X}}
\newcommand{\K}{\mathcal{K}}
\newcommand{\AX}{\mathcal{AX}}
\newcommand{\V}{\mathcal{V}}
\newcommand{\W}{\mathcal{W}}
\newcommand{\T}{\mathcal{T}}
\newcommand{\CST}{\mathcal{C}}

\newcommand{\Lleft}{\mathscr{l}}
\newcommand{\Rright}{\mathscr{r}}
\newcommand{\typ}[1]{\mathit{#1}}
\renewcommand{\L}{\mathtt{LL}}
\renewcommand{\H}{\mathtt{HL}}
\renewcommand{\S}{\mathtt{HH}}
\newcommand{\skey}[2]{\mathrm{key}^{#1}(#2)}
\newcommand{\skeyB}[2]{\mathrm{key}^{#1}\left(#2\right)}
\newcommand{\prikey}[2]{\mathrm{prikey}^{#1}(#2)}
\newcommand{\encT}[2]{{(#1)}_{#2}}
\newcommand{\aencT}[2]{{\{#1\}}_{#2}}
\newcommand{\LRT}[2]{\llbracket ~#1\; ; \; ~#2\rrbracket}
\newcommand{\LRTnew}[4]{\llbracket \nonceentry{#1}{#2}\,;\, \nonceentry{#3}{#4}\rrbracket}
\newcommand{\LRTnewnew}[2]{\llbracket #1\,;\, #2\rrbracket}
\newcommand{\LRTn}[6]{\LRTnewnew{\noncetypelab{#1}{#2}{#3}}{\noncetypelab{#4}{#5}{#6}}}
\newcommand{\orT}{\,\vee\,}
\newcommand{\constantnoncetypes}{\T_{\CST,\N}}
\newcommand{\noncetypes}{\T_\N}
\newcommand{\constanttypes}{\T_\CST}
\newcommand{\noncetypessingle}{\noncetypes^1}
\newcommand{\noncetypesinf}{\noncetypes^{\infty}}
\newcommand{\noncetypebase}[1]{\tau_{#1}}
\newcommand{\noncetype}[1]{\noncetypebase{#1}^{\oneorinf}}
\newcommand{\noncetypelab}[3]{\ensuremath{\tau^{#1,#2}_{#3}}\xspace}
\newcommand{\noncetypesingle}[1]{\noncetypebase{#1}^1}
\newcommand{\noncetypeinf}[1]{\noncetypebase{#1}^{\infty}}
\newcommand{\noncenames}[1]{names(#1)}
\newcommand{\nonceentry}[2]{(#1\!:\!#2)}
\newcommand{\oneorinf}{{a}}
\newcommand{\unfold}[1]{\widetilde{#1}}

\newcommand{\branch}[1]{\mathrm{branches}(#1)}

\newcommand{\subtyp}{<:}



\newcommand{\eqT}{\sim}

\newcommand{\teqT}[7]{\color{red}#1;#2;(#3,#4) \vdash #5 \eqT #6 : #7}

\newcommand{\teqTcnew}[5]{#1 \vdash #2 \eqT #3 : #4 \rightarrow #5}
\newcommand{\teqTc}[8]{\teqTcnew{#1}{#5}{#6}{#7}{#8}}

\newcommand{\teqTGDE}[3]{\teqT{\Gamma}{\Delta}{\E}{\E'}{#1}{#2}{#3}}

\newcommand{\teqTcGDE}[4]{\teqTc{\Gamma}{\Delta}{\E}{\E'}{#1}{#2}{#3}{#4}}

\newcommand{\tewf}[1]{#1 \vdash \diamond}


\newcommand{\eqstr}{\approx}

\newcommand{\eqP}{\sim}

\newcommand{\teqPnew}[4]{#1 \vdash #2 \eqP #3 \rightarrow #4}
\newcommand{\teqP}[7]{\teqPnew{#1}{#5}{#6}{#7}}

\newcommand{\teqPGDE}[3]{\teqP{\Gamma}{\Delta}{\E}{\E'}{#1}{#2}{#3}}


\newcommand{\tDestnew}[3]{#1 \vdash #2 : #3}
\newcommand{\tDest}[4]{\tDestnew{#1}{#2\eqC #3}{#4}}


\newcommand{\varenv}[1]{#1_{\X}}
\newcommand{\keyenv}[1]{#1_{\K}}
\newcommand{\nonceenv}[1]{#1_{\N}}




\newcommand{\wtcnew}[4]{\novar{#1} \vdash #2 \eqT #3 : \onlyvar{#1} \rightarrow #4}
\newcommand{\wtc}[7]{\wtcnew{#6}{#4}{#5}{#7}}

\newcommand{\novar}[1]{{{#1}_{\N,\K}}}
\newcommand{\onlyvar}[1]{{{#1}_{\X}}}


\newcommand{\wtcDE}[4]{\wtc{\Delta}{\E}{\E'}{#1}{#2}{#3}{#4}}

\newcommand{\eqC}{\sim} 
\newcommand{\UnionAll}{{\cup_{\forall}}}
\newcommand{\UnionCart}{{\cup_{\times}}}
\newcommand{\BigUnionCart}{\mathop{\bigcup_{\times}}}

\newcommand{\constSubst}[6]{c^{#1}_{#5,#6}}
\newcommand{\constSubstGDE}[2]{\constSubst{\Gamma}{\Delta}{\E}{\E'}{#1}{#2}}

\newcommand{\phiL}[1]{\phi_\Lleft(#1)}
\newcommand{\phiR}[1]{\phi_\Rright(#1)}

\newcommand{\phiEnew}[1]{\phi_\L^{#1}}
\newcommand{\phiE}[3]{\phiEnew{#1}}
\newcommand{\phiEE}{\phiE{\Gamma}{\E}{\E'}}
\newcommand{\phiEEO}{\phiE{\Gammao}{\E}{\E'}}
\newcommand{\phiZ}{\phiE{\Gamma}{\EPO}{\EQO}}

\newcommand{\E}{\ensuremath{\mathcal{E}}}
\newcommand{\EPO}{\E_P^0}
\newcommand{\EQO}{\E_Q^0}
\newcommand{\R}[2]{\mathcal{R}(#1,#2)}
\newcommand{\PP}{\mathcal{P}}
\newcommand{\QQ}{\mathcal{Q}}
\newcommand{\PPP}{\mathscr{P}}
\newcommand{\redAction}[1]{\xrightarrow{\;#1\;}}
\newcommand{\silentAction}{\tau}
\newcommand{\redWord}[1]{\xrightarrow{\;#1\;}_*}
\newcommand{\redSplit}{\longrightarrow_s}
\newcommand{\processNorm}[1]{\overline{#1}}
\newcommand{\domNA}{\mathrm{dom}}
\newcommand{\dom}[1]{\domNA(#1)}
\newcommand{\varNA}{\mathrm{vars}}
\newcommand{\var}[1]{\varNA(#1)}
\newcommand{\bvarsNA}{\mathrm{bvars}}
\newcommand{\bvars}[1]{\bvarsNA(#1)}
\newcommand{\namesNA}{\mathrm{names}}
\newcommand{\names}[1]{\namesNA(#1)}
\newcommand{\keysNA}{\mathrm{keys}}
\newcommand{\keys}[1]{\keysNA(#1)}
\newcommand{\bnamesNA}{\mathrm{nnames}}
\newcommand{\bnames}[1]{\bnamesNA(#1)}
\newcommand{\nnamesNA}{\mathrm{nnames}}
\newcommand{\nnames}[1]{\nnamesNA(#1)}
\newcommand{\fvarNA}{\mathrm{fvars}}
\newcommand{\fvar}[1]{\fvarNA(#1)}
\newcommand{\fnamesNA}{\mathrm{fnames}}
\newcommand{\fnames}[1]{\fnamesNA(#1)}
\newcommand{\traceNA}{\mathrm{trace}}
\newcommand{\trace}[1]{\traceNA(#1)}

\renewcommand{\c}{\overline{c}}
\newcommand{\cc}{\tilde{c}}

\newcommand{\Gammao}{\overline{\Gamma}}

\newcommand{\EG}{\E_\Gamma}
\newcommand{\EGG}{\E_{\Gamma'}}
\newcommand{\EGGG}{\E_{\Gamma''}}
\newcommand{\EGO}{\E_{\Gammao}}


\newcommand{\incTrace}{\sqsubseteq_t}
\newcommand{\equivTrace}{\approx_t}

\newcommand{\evalNA}{\downarrow}
\newcommand{\eval}[1]{#1 \evalNA}
\newcommand{\eqEv}{=_\evalNA}
\newcommand{\eqCEv}{\eqC_\evalNA}
\newcommand{\eqL}[5]{#1, #2 \vdash #3 =_\L #4 \rightarrow #5}
\newcommand{\eqSilent}{=_\silentAction}

\newcommand{\inst}[3]{{\left\llbracket#1\right\rrbracket}_{#2, #3}}

\newcommand{\splitr}[1]{{\longrightarrow_{#1}}}
\newcommand{\splitrd}{\splitr{\Gamma}}


\newcommand{\Voter}{\mathsf{Voter}}
\newcommand{\Pass}{\mathsf{Pass}}
\newcommand{\Read}{\mathsf{Read}}


\newcommand{\refine}[2]{\{| x_\Lleft = #1 \wedge x_\Rright = #2 |\}}


\newcommand{\MATCHXORNA}[1]{\proc{match\_xor_{#1}}}
\newcommand{\MATCHXOR}[6]{\MATCHXORNA{#1}\; #2 \;\proc{with}\; \INL{#3} \Rightarrow  #4~|~ \INR{#5} \Rightarrow #6}

\newcommand{\nrep}[1]{{{#1}^*}}
\newcommand{\instTerm}[3]{{\left[\:{#1}\:\right]_{#2}^{#3}}}
\newcommand{\instTerma}[2]{\instTerm{#1}{#2}{\Gamma}}
\newcommand{\instProc}[3]{\instTerm{#1}{#2}{#3}}
\newcommand{\instProca}[2]{\instProc{#1}{#2}{\Gamma}}
\newcommand{\instConst}[3]{\instTerm{#1}{#2}{#3}}
\newcommand{\instConsta}[2]{\instConst{#1}{#2}{\Gamma}}
\newcommand{\instTyp}[2]{{\left[\:{#1}\:\right]^{#2}}}
\newcommand{\instG}[3]{{\left[\:{#1}\:\right]^{#3}_{#2}}}
\newcommand{\instGr}[2]{{\left[\:{#1}\:\right]_{#2}}}
\newcommand{\instGG}[2]{\instTyp{#1}{#2}}
\newcommand{\instD}[2]{\instTyp{#1}{#2}}
\newcommand{\instE}[2]{\instTyp{#1}{#2}}
\newcommand{\instCst}[3]{\instG{#1}{#2}{#3}}
\newcommand{\instCstr}[2]{\instGr{#1}{#2}}
\newcommand{\NN}{\N_0}
\newcommand{\NI}{\N_i}
\newcommand{\NR}{\N_*}
\newcommand{\XX}{\X_0}
\newcommand{\XR}{\X_*}
\newcommand{\XI}{\X_i}

\newcommand{\checkconst}{\mathtt{check\_const}}
\newcommand{\checkconststar}{\mathtt{check\_const}}
\newcommand{\stepOa}{\mathtt{step0a}}
\newcommand{\stepI}{\mathtt{step1}}
\newcommand{\stepII}{\mathtt{step2}}
\newcommand{\stepIII}{\mathtt{step3}}
\newcommand{\stepIV}{\mathtt{step4}}
\newcommand{\stepIVstar}{\mathtt{step4}}

\newcommand{\TNonce}{\textsc{TNonce}\xspace}
\newcommand{\TNonceL}{\textsc{TNonceL}\xspace}
\newcommand{\TCst}{\textsc{TCstFN}\xspace}
\newcommand{\TPubkey}{\textsc{TPubKey}\xspace}
\newcommand{\TVkey}{\textsc{TVKey}\xspace}
\newcommand{\TKey}{\textsc{TKey}\xspace}
\newcommand{\TVar}{\textsc{TVar}\xspace}
\newcommand{\TPair}{\textsc{TPair}\xspace}
\newcommand{\TEnc}{\textsc{TEnc}\xspace}
\newcommand{\TAenc}{\textsc{TAenc}\xspace}
\newcommand{\TEncH}{\textsc{TEncH}\xspace}
\newcommand{\TAencH}{\textsc{TAencH}\xspace}
\newcommand{\TEncL}{\textsc{TEncL}\xspace}
\newcommand{\TAencL}{\textsc{TAencL}\xspace}
\newcommand{\TSignL}{\textsc{TSignL}\xspace}
\newcommand{\TSignH}{\textsc{TSignH}\xspace}
\newcommand{\THash}{\textsc{THash}\xspace}
\newcommand{\THashL}{\textsc{THashL}\xspace}
\newcommand{\THigh}{\textsc{THigh}\xspace}
\newcommand{\TSub}{\textsc{TSub}\xspace}
\newcommand{\TOr}{\textsc{TOr}\xspace}
\newcommand{\TLRone}{\ensuremath{\textsc{TLR}^1}\xspace}
\newcommand{\TLRinf}{\ensuremath{\textsc{TLR}^\infty}\xspace}
\newcommand{\TLRp}{{\textsc{TLR}'}\xspace}
\newcommand{\TLRLp}{{\textsc{TLRL}'}\xspace}
\newcommand{\TLRVar}{\textsc{TLRVar}\xspace}

\newcommand{\PZero}{\textsc{PZero}\xspace}
\newcommand{\POut}{\textsc{POut}\xspace}
\newcommand{\PIn}{\textsc{PIn}\xspace}
\newcommand{\PNew}{\textsc{PNew}\xspace}
\newcommand{\PPar}{\textsc{PPar}\xspace}
\newcommand{\POr}{\textsc{POr}\xspace}
\newcommand{\PLet}{\textsc{PLet}\xspace}
\newcommand{\PLetLR}{\textsc{PLetLR}\xspace}
\newcommand{\PIfL}{\textsc{PIfL}\xspace}
\newcommand{\PIfLR}{\textsc{PIfLR}\xspace}
\newcommand{\PIfS}{\textsc{PIfS}\xspace}
\newcommand{\PIfLRinf}{\textsc{PIfLR*}\xspace}
\newcommand{\PIfP}{\textsc{PIfP}\xspace}
\newcommand{\PIfI}{\textsc{PIfI}\xspace}
\newcommand{\PIfLRp}{\textsc{PIfLR'*}\xspace}
\newcommand{\PLetDec}{\textsc{PLetDec}\xspace}
\newcommand{\PLetAdecSame}{\textsc{PLetAdecSame}\xspace}
\newcommand{\PLetAdecDiff}{\textsc{PLetAdecDiff}\xspace}
\newcommand{\PIfAll}{\textsc{PIfAll}\xspace}

\newcommand{\GNil}{\textsc{GNil}\xspace}
\newcommand{\GNonce}{\textsc{GNonce}\xspace}
\newcommand{\GVar}{\textsc{GVar}\xspace}
\newcommand{\GKey}{\textsc{GKey}\xspace}

\newcommand{\SRefl}{\textsc{SRefl}\xspace}
\newcommand{\SHigh}{\textsc{SHigh}\xspace}
\newcommand{\STrans}{\textsc{STrans}\xspace}
\newcommand{\SPairL}{\textsc{SPairL}\xspace}
\newcommand{\SPairS}{\textsc{SPairS}\xspace}
\newcommand{\SPair}{\textsc{SPair}\xspace}
\newcommand{\SPairSp}{\textsc{SPairS'}\xspace}
\newcommand{\SKey}{\textsc{SKey}\xspace}
\newcommand{\SEnc}{\textsc{SEnc}\xspace}
\newcommand{\SAenc}{\textsc{SAenc}\xspace}

\newcommand{\DDecH}{\textsc{DDecH}\xspace}
\newcommand{\DDecL}{\textsc{DDecL}\xspace}
\newcommand{\DDecT}{\textsc{DDecT}\xspace}
\newcommand{\DAdecH}{\textsc{DAdecH}\xspace}
\newcommand{\DAdecL}{\textsc{DAdecL}\xspace}
\newcommand{\DAdecT}{\textsc{DAdecT}\xspace}
\newcommand{\DCheckH}{\textsc{DCheckH}\xspace}
\newcommand{\DCheckL}{\textsc{DCheckL}\xspace}
\newcommand{\DFst}{\textsc{DFst}\xspace}
\newcommand{\DFstL}{\textsc{DFstL}\xspace}
\newcommand{\DSnd}{\textsc{DSnd}\xspace}
\newcommand{\DSndL}{\textsc{DSndL}\xspace}

\newcommand{\rep}{rep}
\newcommand{\repP}[2]{\rep^P_{#1}(#2)}
\newcommand{\repC}[2]{\rep^C_{#1}(#2)}
\newcommand{\repc}[2]{\rep^c_{#1}(#2)}
\newcommand{\rept}[2]{\rep^t_{#1}(#2)}
\newcommand{\repv}[2]{\rep^v_{#1}(#2)}
\newcommand{\repn}[2]{\rep^n_{#1}(#2)}
\newcommand{\repCinv}[2]{{\rep^C}^{-1}_{#1}(#2)}
\newcommand{\repcinv}[2]{{\rep^c}^{-1}_{#1}(#2)}
\newcommand{\reptinv}[2]{{\rep^t}^{-1}_{#1}(#2)}
\newcommand{\repvinv}[2]{{\rep^v}^{-1}_{#1}(#2)}
\newcommand{\repninv}[2]{{\rep^n}^{-1}_{#1}(#2)}
\newcommand{\bijn}[2]{b^n_{#1}(#2)}
\newcommand{\bijv}[2]{b^v_{#1}(#2)}
\newcommand{\bijninv}[2]{{b^n}^{-1}_{#1}(#2)}
\newcommand{\bijvinv}[2]{{b^v}^{-1}_{#1}(#2)}

%% file: introduction.tex
\section{Introduction}
\label{sec:intro}

Formal methods proved to be indispensable  tools for the  analysis of advanced cryptographic 
protocols such as those for key distribution~\cite{SchmidtMCB12},
mobile payments~\cite{CFGT17},
e-voting~\cite{DKR-jcs08,Backes:2008:AVR,Cortier:2015:TVE}, and e-health~\cite{Maffei:2013:SPD}. In the last years, mature push-button analysis tools 
have  emerged and have been successfully applied to many protocols from the
literature in the context of \emph{trace properties} such as
authentication or confidentiality. These tools employ a variety of analysis techniques, such as model checking (e.g.,   Avispa~\cite{avispa} and Scyther \cite{scyther}),   Horn clause resolution (e.g., 
ProVerif~\cite{proverif}), term rewriting (e.g.,  Scyther~\cite{scyther} and Tamarin~\cite{tamarin}), and type systems~\cite{Bengtson:2011,FM:2011,Backes:2014:UIR,Bugliesi:2015:ART}.

A current and very active topic is the adaptation of these techniques to the
more involved case of \emph{trace equivalence} properties. These 
are the natural symbolic counterpart of cryptographic indistinguishability properties, and they are at the heart of privacy properties such as ballot
privacy~\cite{DKR-jcs08}, untraceability~\cite{MyrtoRFID09}, or
anonymity~\cite{private-auth,passport}. They are also used to express stronger
forms of confidentiality, such as strong secrecy~\cite{Cortier2006}, or game-based like
properties~\cite{CLC-CCS2008}. 

\mypar{Related Work} Numerous model checking-based  tools have recently been proposed for the case of a bounded number of sessions, i.e., when protocols are executed a bounded
number of times. These tools encompass SPEC~\cite{SPEC},
APTE~\cite{APTE,APTE-por}, Akiss~\cite{akiss}, or
SAT-Equiv~\cite{SAT-equiv}. 
These tools   vary in the class of
cryptographic primitives and the class of protocols they can
consider. However, due to the complexity of the problem, they all
suffer from the state explosion problem and most of them  can typically
analyse no more than 3-4 sessions of (relatively small) protocols,
with the exception of SAT-Equiv which can more easily reach about 10
sessions. The only tools that can verify equivalence properties for an
unbounded number of sessions are  ProVerif~\cite{proverif-equiv}. Maude-NPA~\cite{maude-equiv}, and  Tamarin~\cite{tamarin-equiv}. ProVerif checks a  property that is stronger than trace equivalence, namely diff equivalence, which  works well in practice provided that protocols have a
similar structure. However, as for trace properties, the internal
design of ProVerif renders the tool unable to distinguish between
exactly one  session and infinitely many: this
over-approximation often yields false attacks, in particular when the security of a
protocol relies on the fact that some action is only performed once. 
\veroniquebis{Maude-NPA also checks diff-equivalence but often does not terminate.}
Tamarin can handle an unbounded number of sessions and is
very flexible in terms of supported  protocol classes but it often requires human
interactions. Finally, some recent work has started to leverage type systems to enforce relational properties for programs, exploring this approach also in the context of cryptographic protocol implementation~\cite{BartheFGSSB14}: like ProVerif,  the resulting tool  is unable to distinguish between
exactly one  session and infinitely many, and furthermore it is only
semi-automated, in that it often requires non-trivial 
\veroniquebis{lemmas to guide the tool}
and a specific programming discipline. 


Many recent results have been obtained in the area of relational verification of
programs using Relational Hoare
Logic~\cite{Benton04,Yang07,BartheGB09}.
The results hold on programs and cannot be directly applied  to
cryptographic protocols due to the more special treatment of the primitives.

\mypar{Our contribution} In this paper, we consider a novel type checking-based approach. 
Intuitively, a type system over-approximates protocol behavior. Due
to this over-approximation, it is no longer possible to \emph{decide}
security properties but the types typically convey 
sufficient information to \emph{prove} security.
Extending this approach to equivalence properties is a delicate
task. Indeed, two protocols $P$ and $Q$ are in equivalence if
(roughly) any trace of $P$ has an equivalent trace in $Q$ (and
conversely). Over-approximating behavior may not preserve
equivalence. 

Instead, we develop a somewhat hybrid approach: we design a type
system to over-approximate the set of possible traces and we collect
the set of sent messages into \emph{constraints}. We then propose a
procedure for proving (static) equivalence of the constraints. 
These do not only contain sent messages but also reflect
internal checks made by the protocols, which  is crucial to guarantee
that whenever a message is accepted by $P$, it is also accepted by $Q$
(and conversely). 

As a result, we provide a sound type system for proving equivalence of
protocols 
for both a bounded and
an unbounded number of sessions, or a mix of both. This is
particularly convenient to analyse systems where some actions are
limited (e.g., no revote, or limited access to some resource).
\veroniquebis{
More specifically, we show that whenever two protocols $P$ and $Q$ are
type-checked to be equivalent, then they are in trace equivalence, for
the standard notion of trace equivalence~\cite{trace-equivalence}, against a full
Dolev-Yao attacker. In particular, one advantage of our approach is
that it proves security directly in a security model that is similar to the ones
used by the other popular tools, in contrast to many other security
proofs based on type systems.
Our result holds for protocols with all standard primitives
(symmetric and asymetric encryption, signatures, pairs, hash), with
atomic long-term keys (no fresh keys) and no private channels.
Similarly to ProVerif, we need the two protocols $P$ and $Q$ to have a
rather similar structure.
}

We provide a prototype implementation of our type system, that we
evaluate on several protocols of the literature. In the case of a
bounded number of sessions, our tool provides a significant speed-up
(less than one second to analyse a dozen of sessions while other tools typically
do not answer
within 12 hours, with a few exceptions). To be fair, let us emphasize
that these tools can \emph{decide} equivalence while our tool checks
sufficient conditions by the means of our type system.
In the case of an unbounded number of sessions, the performance of our
prototype tool is comparable to ProVerif. In contrast to ProVerif, our tool can consider a mix of
bounded and unbounded number of sessions. As an application, we can
prove for the first time ballot privacy of the well-known Helios e-voting protocol~\cite{helios},
without assuming a reliable channel between honest voters and the
ballot box. ProVerif fails in this case as ballot privacy only holds
under the assumption that honest voters vote at most once,  otherwise
the protocol is subject to a copy attack~\cite{Roenne}. For similar reasons, also Tamarin fails to verify this protocol.
 
\veroniquebis{In most of our example, only a few straightforward type
  annotations were needed, such as indicated which keys are supposed
  to be secret or public. The case of the helios protocol is more
  involved and requires to describe the form of encrypted ballots that
can be sent by a voter.}

Our prototype, the protocol models, as well as a technical report are
available here~\cite{oursite}.

%% file: overview.tex

\section{Overview of our Approach}
\label{sec:overview}

In this section, we introduce the key ideas underlying our approach on a simplified  version of the Helios voting protocol.
Helios~\cite{helios} is a verifiable voting protocol that has been used in
various elections, including the election of the rector of the
University of Louvain-la-Neuve.
Its behavior is
 depicted below:
\begin{align*}
S &\rightarrow V_i: \quad r_{i}\\
V_i &\rightarrow S: \quad {[{\{v_i\}}^{r_{i},r'_i}_{\PUBK{k_s}}]}_{k_i}\\
S &\rightarrow V_1,\dots,V_n : \; v_1  , \ldots, v_n 
\end{align*}
where $\{m\}^r_{\PUBK{k}}$ denotes the asymmetric encryption of message $m$ with the key $\PUBK{k}$ randomized with the nonce $r$,
and $[m]_k$ denotes the signature of $m$ with key $k$.
$v_i$ is a value in the set $\{0,1\}$, which represents the candidate  $V_i$ votes for. In the first  step, the voter casts her vote, encrypted with the election's public key $\PUBK{k_s}$ and then signed. 
Since generating a good random number is difficult for the voter's client (typically a JavaScript run in a browser), a typical trick is to input some randomness ($r_i$) from the server and to add it to its own randomness ($r_i'$).
In the second step the server outputs the tally (i.e., a randomized permutation of the valid votes received in the voting phase). 
Note that the original Helios protocol does not assume signed ballots. Instead, voters authenticate themselves through a login mechanism. For simplicity, we abstract this authenticated channel by a signature.

	A voting protocol provides vote privacy~\cite{DKR-jcs08} if an attacker is not able to know which voter voted for which candidate. Intuitively, this can be modeled as the following trace equivalence property, which  requires the attacker not to be able to distinguish  $A$ voting 0 and $B$ voting 1 from $A$ voting 1 and $B$ voting 0. Notice that the attacker may control an unbounded number of voters:
	\begin{align*}
	&Voter(k_a,0) \;\PAR\; Voter(k_b,1) \;\PAR\;  CompromisedVoters  \PAR\; S\\
	\equivTrace \;& Voter(k_a,1) \;\PAR\; Voter(k_b,0) \;\PAR\; CompromisedVoters  \PAR\; S
	\end{align*}
	Despite its simplicity, this protocol has a few interesting features that make its analysis particularly challenging. First of all,  the server is supposed to discard ciphertext duplicates, otherwise a malicious eligible voter   $E$ could intercept $A$'s ciphertext, sign it, and send it to the server~\cite{Helios-CSF11}, as exemplified below:
	\begin{align*}
	A &\rightarrow S: \quad {[{\{v_a\}}^{r_a,r'_a}_{\PUBK{k_s}}]}_{k_a}\\
		E &\rightarrow S: \quad {[{\{v_a\}}^{r_a,r'_a}_{\PUBK{k_s}}]}_{k_e}\\
	B &\rightarrow S: \quad {[{\{v_b\}}^{r_b,r'_b}_{\PUBK{k_s}}]}_{k_b}\\
	S &\rightarrow A,B : \; v_a  ,  v_b, v_a
	\end{align*}
	This would make the two tallied results distinguishable, thereby breaking trace equivalence since $v_a, v_b, v_a \not\equivTrace	v_b, v_a, v_b$

	Even more interestingly, each voter is supposed to be able to vote \emph{only once}, otherwise the same attack would apply~\cite{Roenne} even if the server discards ciphertext duplicates (as the randomness used by the voter in the two ballots would be different). This makes the analysis particularly challenging, and in particular out of scope of existing cryptographic protocol analyzers like ProVerif, which abstract away from the number of protocol sessions.

	With our type system, we can successfully verify the aforementioned privacy property using the following types:
		\begin{align*}
    r_a&:\noncetypelab{\L}{1}{r_a},\; r_b:\noncetypelab{\L}{1}{r_b},
	r'_a:\noncetypelab{\S}{1}{r'_a},\; r'_b:\noncetypelab{\S}{1}{r'_b}\\
	k_a&:\skey{\S}{\aencT{\LRTnewnew{\noncetypelab{\L}{1}{0}}{\noncetypelab{\L}{1}{1}}*\H*\noncetypelab{\S}{1}{r'_a}}{k_s}}\\
	k_b&:\skey{\S}{\aencT{\LRTnewnew{\noncetypelab{\L}{1}{1}}{\noncetypelab{\L}{1}{0}}*\H*\noncetypelab{\S}{1}{r'_b}}{k_s}}\\
  k_s&:\skeyB{\S}{\begin{array}{l}(\LRTnewnew{\noncetypelab{\L}{1}{0}}{\noncetypelab{\L}{1}{1}}*\H*\noncetypelab{\S}{1}{r'_a}) ~ \orT ~ \\(\LRTnewnew{\noncetypelab{\L}{1}{1}}{\noncetypelab{\L}{1}{0}}*\H*\noncetypelab{\S}{1}{r'_b})\end{array}}
    \end{align*}
	We assume standard security labels: $ \S$ stands for high confidentiality and high integrity, $ \H$ for high confidentiality and low integrity, and $\L$ for low confidentiality and low integrity (for simplicity, we omit the low confidentiality and high integrity type, since we do not need it in our examples).
	The type $\noncetypelab{l}{1}{i}$ describes randomness of security label $l$ produced by the randomness generator  at position $i$ in the program, which can be invoked at most once. 
$\noncetypelab{l}{\infty}{i}$ is similar, with the difference that the randomness generator can be invoked an unbounded number of times. These types induce a partition on random values, in which each set contains at most one element or an unbounded number of elements, respectively. This turns out to be useful, as explained below, to type-check protocols, like Helios, in which the number of times messages of a certain shape are produced is relevant for the security of the protocol.

The  type of $k_a$ (resp. $k_b$) says that this key is supposed to encrypt 0 and 1 (resp. 1 and 0) on the left- and right-hand side of the equivalence relation, further describing the type of the randomness. The type of $k_s$   inherits the two  payload types, which are combined in disjunctive form.  In fact, public key types  implicitly convey an additional payload type, the one characterizing messages encrypted by the attacker: these are of low confidentiality and turn out to be the same on the left- and right-hand side. 
Key types are crucial to type-check the server code: we verify the signatures produced by  $A$ and $B$ and can then use the ciphertext type derived from the type of $k_a$ and $k_b$ to infer after decryption the vote cast by $A$ and $B$, respectively. While processing the other ballots, the server discards the ciphertexts produced with randomness matching the one used by $A$ or $B$: given that these random values are used only once,    we know that the remaining ciphertexts must come from the attacker and thus convey the same vote on the left- and on the right-hand side. This suffices to type-check the final output, since the two tallied results on the left- and right-hand side are the same, and thus fulfill trace equivalence. 

The type system generates a set of constraints, which, if ``consistent'', suffice to prove that the protocol is trace equivalent. Intuitively, these constraints characterize the indistinguishability of the messages output by the process.  The constraints generated for this simplified version of Helios are reported below: 
\[\renewcommand{\arraystretch}{1.2}\begin{array}{l@{}l}
C = \{(\{& \SIGN{\AENC{\PAIR{0}{\PAIR{x}{r'_a}}}{\PUBK{k_S}}}{k_a} \eqC\\
        &\hfill\SIGN{\AENC{\PAIR{1}{\PAIR{x}{r'_a}}}{\PUBK{k_S}}}{k_a},\\ 
		& \AENC{\PAIR{0}{\PAIR{x}{r'_a}}}{\PUBK{k_S}} \eqC \AENC{\PAIR{1}{\PAIR{x}{r'_a}}}{\PUBK{k_S}},\\
		& \SIGN{\AENC{\PAIR{1}{\PAIR{y}{r'_b}}}{\PUBK{k_S}}}{k_b} \eqC\\
		& \hfill \SIGN{\AENC{\PAIR{0}{\PAIR{y}{r'_b}}}{\PUBK{k_S}}}{k_b},\\ 
		& \AENC{\PAIR{1}{\PAIR{y}{r'_b}}}{\PUBK{k_S}} \eqC \AENC{\PAIR{0}{\PAIR{y}{r'_b}}}{\PUBK{k_S}}\},\\
		& [x:\L, y:\L])\}
\end{array}\]
These constraints are  consistent if the set of left messages of the constraints is in (static) equivalence with the set of the right messages of the constraints.
This is clearly the case here, since encryption hides the content of the plaintext. Just to give an example of non-consistent constraints, consider the following ones:
\[C' =  \{\{ \HASH{n_1} \eqC \HASH{n_2},\;\; \HASH{n_1} \eqC \HASH{n_1}\}\}\]
where $n_1$, $n_2$ are two confidential nonces.
While the first constraint alone is consistent, since $n_1$ and $n_2$ are of high confidentiality and the attacker cannot 
thus distinguish between $\HASH{n_1}$ and $\HASH{n_2}$, the two constraints all together are not consistent, since the attacker can clearly notice if the two terms output by the process are the same or not. 
We developed a dedicated procedure to check the consistency of such constraints. 
		

%% file: calculus.tex
\section{Framework}
\label{sec:calculus}

In symbolic models, security protocols are typically modeled as processes of a process algebra, such as the applied pi-calculus~\cite{AbadiFournet2001}. We present here a calculus close to~\cite{Length-cav2013} 
inspired from the calculus underlying the ProVerif tool~\cite{BlanchetFnTPS16}.

\subsection{Terms}
Messages are  modeled as terms. We assume an infinite set of names $\N$ for nonces, further partitioned into the set $\FN$ of free nonces (created by the attacker) and the set $\BN$ of bound nonces (created by the protocol parties), an infinite set of names $\K$ for keys, ranged over by $k$, and an infinite set of variables $\V$.
Cryptographic primitives are modeled through a \emph{signature} $\F$, that is a set of function symbols, given with their arity (that is, the number of arguments).
Here, we will consider the following signature:
\[
\F_c= \{\PUBKNA,\VKNA,\ENCNA,\AENCNA,\SIGNNA,\PAIRNA,\HASHNA\}
\]
that models respectively public and verification key, symmetric and asymmetric encryption, concatenation and hash.
The companion primitives (symmetric and asymmetric decryption, signature check, and projections) are represented by the following signature:
\[
\F_d = \{\DECNA,\ADECNA,\CHECKNA,\FSTNA,\SNDNA\}
\]
We also consider  a set $\CST$ of (public) constants (used as agents names for instance). 
Given a signature $\F $, a set of names $\N$ and a set of variables $\V$, the set of \emph{terms} $\T(\F,\V,\N)$ is the set inductively defined by applying functions to variables in $\V$ and names in $\N$. 
We denote by $\names{t}$ (resp.  $\var{t}$) the set of names (resp. variables) occurring in $t$.
A term is \emph{ground} if it does not contain variables.

Here, we will consider the set $\T(\F_c\cup\F_d\cup\CST,\V,\N\cup\K)$
of \emph{cryptographic terms}, simply called \emph{terms}.
\emph{Messages} are terms from $\T(\F_c\cup\CST,\V,\N\cup\K)$ with atomic keys, that is, a term $t\in \T(\F_c\cup\CST,\V,\N\cup\K)$ is a message if any subterm of $t$ of the form $\PUBK{t'}$, $\VK{t'}$, $\ENC{t_1}{t'}$, $\AENC{t_1}{t_2}$, or $\SIGN{t_1}{t'}$ is such that $t'\in \K $ and $t_2=\PUBK{t_2'}$ with  $t_2'\in \K $. 
We assume the set of variables to be split into two subsets $\V = \X \uplus \AX$ where $\X $ are variables used in processes while $\AX$ are variables used to store messages.
An \emph{attacker term} is a term from $\T(\F_c\cup\F_d\cup\CST,\AX,\FN)$.

A \emph{substitution} $\sigma = \{M_1 / x_1, \dots, M_k / x_k\}$ is a mapping from
variables $x_1, \dots, x_k \in \V$ to messages $M_1, \dots, M_k$. 
We let $\dom{\sigma} = \{x_1 , \dots , x_k\}$.
We say that $\sigma$ is ground if all messages $M_1, \dots, M_k$ are ground.
We let $\names{\sigma} = \bigcup_{1 \leq i \leq k} \names{M_i}$.
The application of a substitution $\sigma $ to a term $t$ is denoted $t\sigma$ and is defined as usual.

The \emph{evaluation} of a term $t$, denoted $\eval{t}$, corresponds to the application of the cryptographic primitives. For example, the decryption succeeds only if the right decryption key is used. Formally, $\eval{t}$ is recursively defined as follows.
\[
\begin{array}{r@{\;}c@{\;}ll}
\eval{u} & = & u & \text{if $u\in\N\cup\V\cup\K\cup\CST$}\\
\eval{\PUBK{t}} & = & \PUBK{\eval{t}} & \text{if $\eval{t}\in\K$}\\ 
\eval{\VK{t}} & = & \VK{\eval{t}} & \text{if $\eval{t}\in\K$}\\ 
\eval{\HASH{t}} & = & \HASH{\eval{t}} & \text{if $\eval{t}\neq\bot$}\\ 
\eval{\PAIR{t_1}{t_2}} & = & \PAIR{\eval{t_1}}{\eval{t_2}} & \text{if $\eval{t_1} \neq \bot$ and $\eval{t_2} \neq \bot$}\\
 \eval{\ENC{t_1}{t_2}} & = & \ENC{\eval{t_1}}{\eval{t_2}} & \text{if $\eval{t_1} \neq \bot$ and $\eval{t_2} \in\K$}\\
 \eval{\SIGN{t_1}{t_2}} & = & \SIGN{\eval{t_1}}{\eval{t_2}} & \text{if $\eval{t_1} \neq \bot$ and $\eval{t_2} \in\K$}\\
\eval{\AENC{t_1}{t_2}} & = & \AENC{\eval{t_1}}{\eval{t_2}} & \text{if $\eval{t_1} \neq \bot$ and $\eval{t_2} =\PUBK{k}$}\\
         &&\qquad \text{for some $k\in\K$}\\
\end{array}\]
\[\begin{array}{r@{\;}c@{\;}ll}
\eval{\FST{t}} & = & t_1 & \text{if $\eval{t}=\PAIR{t_1}{t_2}$}\\
\eval{\SND{t}} & = & t_2 & \text{if $\eval{t}=\PAIR{t_1}{t_2}$}\\
 \eval{\DEC{t_1}{t_2}} & = & t_3 & \text{if $\eval{t_1} = \ENC{t_3}{t_4}$ and $t_4=\eval{t_2}$}\\
 \eval{\ADEC{t_1}{t_2}} & = & t_3 & \text{if $\eval{t_1} = \AENC{t_3}{\PUBK{t_4}}$ and $t_4=\eval{t_2}$}\\
 \eval{\CHECK{t_1}{t_2}} & = & t_3 & \text{if $\eval{t_1} = \SIGN{t_3}{t_4}$ and $\eval{t_2} = \VK{t_4}$}\\
\eval{t} & = & \bot & \text{otherwise}
\end{array}
\]
Note that the evaluation of term $t$ succeeds only if the underlying keys are atomic and always returns a message or $\bot$.
We write $t \eqEv t'$ if $\eval{t} = \eval{t'}$.

\subsection{Processes}
Security protocols describe how messages should be exchanged between participants. We model them through a process algebra, whose syntax is displayed in Figure~\ref{fig:syntax}.
\begin{figure}
\[
\begin{array}{rll}
\multicolumn{3}{l}{\text{Destructors used in processes:}}\\
\multicolumn{3}{l}{d::= 	 \DEC{\cdot}{k} ~|~ 
        \ADEC{\cdot}{k} ~|~ 
         \CHECK{\cdot}{\VK{k}} ~|~ 
		 \FST{\cdot} ~|~ 
		 \SND{\cdot} }\\
\\
\multicolumn{2}{l}{\text{Processes:}}\\
P,Q &::= \\
&\ZERO &\\
		| & \NEWN{n}.P &\text{ for } n\in\BN (n \text{ bound in } P)
  \\
		| & \OUT{M}.P &\\
		| & \IN{x}.P &\text{ for } x\in\X (x \text{ bound in } P)\\
		| & P \PAR Q \\
		| & \LET{x}{d(y)}{P}{Q} &\text{ for } x, y \in \X (x \text{ bound in } P)\\\
		| & \ITE{M}{N}{P}{Q} \\
		| & !P \\
\end{array}
\]
where $M,N$ are messages.
\caption{Syntax for processes.}
\label{fig:syntax}
\end{figure}
We identify processes up to $\alpha$-renaming, i.e., capture avoiding substitution of bound names and variables, which are defined as usual. Furthermore, we assume that all bound names and variables in the process are distinct. 

A \emph{configuration} of the system is a quadruple $(\E; \PP; \phi ; \sigma)$ where:
\begin{itemize}
\item $\PP$ is a multiset of processes that represents the current active processes;
\item $\E$ is a set of names, which represents the private names of the processes;
\item $\phi$ is a substitution  with $\dom{\phi}\subseteq \AX$ and for
  any $x\in\dom{\phi}$, $\phi(x)$ (also denoted $x\phi$) is a message
  that only contains variables in $\dom{\sigma}$. $\phi$ represents the terms already output.
  \item $\sigma$ is a ground substitution;
\end{itemize}
The semantics of processes is given through a transition relation  $\redAction{\alpha}$ on the quadruples provided in Figure~\ref{fig:semantics} ($\silentAction$ denotes a silent action). The relation $\redWord{w}$ is defined as the reflexive transitive closure of $\redAction{\alpha}$, where $w$ is the concatenation of all actions.
We also write equality up to silent actions $\eqSilent$.

\begin{figure*}
\begin{minipage}{\textwidth}
\[
\begin{array}{rclr}
(\E; \{P_1 \PAR P_2\} \cup \PP; \phi ;\sigma) 	& \redAction{\silentAction}	& (\E; \{P_1, P_2\} \cup \PP; \phi; \sigma)
	& \text{\textsc{Par}}\\
%
%
(\E; \{\ZERO\} \cup \PP; \phi; \sigma)	& \redAction{\silentAction}	& (\E; \PP; \phi; \sigma)
	& \text{\textsc{Zero}}\\
%
%
(\E; \{\NEWN{n}.P\} \cup \PP; \phi; \sigma) 	&\redAction{\silentAction} 	& (\E \cup \{n\}; \{P\} \cup \PP; \phi; \sigma)
	& \text{\textsc{New}}\\
%
%
%
(\E ; \{\OUT{t}.P\}\cup \PP ; \phi; \sigma) 	&\redAction{\NEWN{ax_n}.\OUT{ax_n}} 	& (\E ; \{P\} \cup \PP ; \phi \cup \{t/ax_n\}; \sigma)
	& \text{\textsc{Out}}\\
\multicolumn{3}{r}{\text{if $t\sigma$ is a ground term, } {ax}_n \in \AX \text{ and } n = |\phi | + 1} & \\
%
%
(\E; \{\IN{x}. P\} \cup \PP ; \phi; \sigma)	&\redAction{\IN{R}}	& (\E; \{ P\} \cup \PP ; \phi; \sigma \cup \{\eval{(R\phi\sigma)}/x\})
	& \text{\textsc{In}}\\
\multicolumn{3}{r}{\text{if $R$ is an attacker term such that } \var{R} \subseteq \dom{\phi}, 
} & \\
\multicolumn{3}{r}{\text{and } \eval{(R \phi \sigma)} \neq \bot} & \\
%
%
{(\E; \{\LET{x}{d(M)}{P}{Q}\} \cup \PP ; \phi ; \sigma)} 
	& \redAction{\silentAction} & (\E ; \{P\} \cup \PP ; \phi ; \sigma \cup \{\eval{d(M\sigma)}/x\}) 	& \text{\textsc{Let-In}}\\
\multicolumn{3}{r}{\text{if $M\sigma$ is ground and } \eval{d(M\sigma)} \neq \bot} & \\
%
%
{(\E; \{\LET{x}{d(M)}{P}{Q}\} \cup \PP ; \phi ; \sigma)}
	& \redAction{\silentAction}& (\E ; \{Q\} \cup \PP ; \phi ; \sigma) 	& \text{\textsc{Let-Else}}\\
\multicolumn{3}{r}{\text{if $M\sigma$ is ground and $\eval{d(M\sigma)} = \bot$, \ie $d$ cannot be applied to $M\sigma$}} & \\
%
%
{(\E; \{\ITE{M}{N}{P}{Q}\} \cup \PP ; \phi ; \sigma)}
	& \redAction{\silentAction} &(\E ; \{P\} \cup \PP ; \phi ; \sigma) 	& \text{\textsc{If-Then}}\\
\multicolumn{3}{r}{\text{if $M$, $N$ are messages such that $M\sigma$, $N\sigma$ are ground and $M\sigma = N\sigma$}} & \\
%
%
{(\E; \{\ITE{M}{N}{P}{Q}\} \cup \PP ; \phi; \sigma)}
	& \redAction{\silentAction} & (\E ; \{Q\} \cup \PP ; \phi ; \sigma) 	& \text{\textsc{If-Else}}\\
	\multicolumn{3}{r}{\text{if $M$, $N$ are messages such that $M\sigma$, $N\sigma$ are ground and $M\sigma \neq N\sigma$}} & \\
	{(\E; \{!P\} \cup \PP ; \phi; \sigma)}
	& \redAction{\silentAction} & (\E ; \{P,!P\} \cup \PP ; \phi ; \sigma) 	& \text{\textsc{Repl}}\\
\end{array}
\]
\end{minipage}
\caption{Semantics}
\label{fig:semantics}
\end{figure*}

Intuitively, process $\NEWN{n}.P$ creates a fresh nonce, stored in $\E$, and behaves like $P$. Process $\OUT{M}.P $ emits $M$ and behaves like $P$. Process $\IN{x}.P$ inputs any term computed by the attacker provided it evaluates as a message and then behaves like $P$. 
Process $P \PAR Q$ corresponds to the parallel composition of $P$ and $Q$.
Process $ \LET{x}{d(y)}{P}{Q}$ behaves like $P$ in which $x$ is replaced by $d(y)$ if $d(y)$ can be successfully evaluated and behaves like $Q$ otherwise. Process $\ITE{M}{N}{P}{Q}$ behaves like $P$ if $M$ and $N$ correspond to two equal messages and behaves like $Q$ otherwise. The replicated process $!P$ behaves as an unbounded number of copies of $P$.

A \emph{trace} of a process $P$ is any possible sequence of
transitions in the presence of an attacker that may read, forge, and
send messages. Formally, the set of traces $\trace{P} $ is defined as follows.
\[
\trace{P} = \{(w,\NEWN{\E}. \phi,\sigma) | 
(\emptyset; \{P\} ; \emptyset; \emptyset)
\redWord{w} (\E ; \PP ; \phi; \sigma)\}\
\]

\begin{example}
\label{ex:voting-protocol}
Consider the Helios protocol presented in Section~\ref{sec:overview}.
For simplicity, we describe here a
simplified version with only two (honest) voters $A$ and $B$ and a
voting server $S$. 
%
This (simplified) protocol can be modeled  by the process:
\[\NEWN{r_a}. Voter(k_a,v_a,r_a) \;\PAR\;\NEWN{r_b}. Voter(k_b,v_b,r_b) \;\PAR\; P_S\]
where $Voter(k, v, r)$ represents voter $k$ willing to vote for $v$
using randomness $r$ while $P_S$ represents the voting server.
%
$Voter(k, v, r)$ simply outputs a signed encrypted
vote.
\[Voter(k, v, r) = \OUT{\SIGN{\AENC{\PAIR{v}{r}}{\PUBK{k_S}}}{k}}\]

The voting server receives ballots from $A$ and $B$ and then outputs
the decrypted ballots, after some mixing.
\begin{align*}
P_S = & \IN{x_1}.\IN{x_2}.\\
&\LETIN{y_1}{\CHECK{x_1}{\VK{k_a}}}\\
&\LETIN{y_2}{\CHECK{x_2}{\VK{k_b}}}\\
&\LETIN{z_1}{\ADEC{y_1}{k_s}}\quad \LETIN{z_1'}{\FST{z_1}}\\
&\LETIN{z_2}{\ADEC{y_2}{k_s}}\quad \LETIN{z_2'}{\FST{z_2}} \\
&\quad(\OUT{z_1'} \PAR \OUT{z_2'})
\end{align*}
\end{example}

\subsection{Equivalence}
When processes evolve, sent messages are stored in a substitution $\phi $ while private names are stored in $\E $. A \emph{frame} is simply an expression of the form $\NEWN{\E}. \phi$ where $\dom{\phi} \subseteq \AX$. We define $\dom{\NEWN{\E}. \phi}$ as $\dom{\phi}$. Intuitively, a frame represents the knowledge of an attacker.

Intuitively, two sequences of messages are indistinguishable to an attacker if he cannot perform any test that could distinguish them. This is typically modeled as static equivalence~\cite{AbadiFournet2001}. Here, we consider of variant of~\cite{AbadiFournet2001} where the attacker is also given the ability to observe when the evaluation of a term fails, as defined for example in~\cite{Length-cav2013}.

\begin{definition}[Static Equivalence]
Two ground frames $\NEWN{\E}.\phi$ and $\NEWN{\E'}.\phi'$ are statically equivalent if and only if they have the same domain, and for all attacker terms $R, S$
with variables in $\dom{\phi} = \dom{\phi'}$, we have
\[
(R \phi \eqEv S \phi) \iff (R \phi' \eqEv S \phi')
\]
\end{definition}

Then two processes $P$ and $Q$ are in equivalence if no matter how the adversary interacts with $P$, a similar interaction may happen with $Q$, with equivalent resulting frames.

\begin{definition}[Trace Equivalence]
Let $P$, $Q$ be two processes. We write $P \incTrace Q$ if for all $(s,\psi,
\sigma)\in \trace{P}$, there exists $(s',\psi', \sigma')\in\trace{Q}$ such that
$s \eqSilent s'$ and $\psi\sigma$ and $\psi'\sigma'$ are statically equivalent.
We say that $P$ and $Q$ are trace equivalent, and we write $P \equivTrace Q$, if $P \incTrace Q$ and $Q \incTrace P$.
\end{definition}

Note that this definition already includes the attacker's behavior,
since  processes may input any message forged by the attacker.

\begin{example}
As explained in Section~\ref{sec:overview},
ballot privacy is typically modeled as an equivalence
property~\cite{DKR-jcs08} 
that requires
that an attacker
cannot distinguish when Alice is voting $0$ and Bob is voting $1$
from the scenario where the two votes are swapped.

Continuing Example~\ref{ex:voting-protocol}, ballot privacy of Helios
can be expressed as follows:
\begin{align*}
&\NEWN{r_a}. Voter(k_a,0,r_a) \;\PAR\;\NEWN{r_b}. Voter(k_b,1,r_b) \;\PAR\; P_S\\
\equivTrace \;&\NEWN{r_a}. Voter(k_a,1,r_a) \;\PAR\;\NEWN{r_b}. Voter(k_b,0,r_b) \;\PAR\; P_S
\end{align*}




\end{example}

%% file: typing.tex

\section{Typing}
We now introduce a  type system to statically check trace equivalence between  processes. Our typing judgements thus capture properties of pairs of terms or 
 processes, which we will refer to as \emph{left} and \emph{right} term or process, respectively. 
 
%

\subsection{Types}

\begin{figure}
\[
  \begin{array}{llll}
     l ::= & & \L ~|~  \H ~|~ \S \\ 
    T ::= & & l ~|~  T * T ~|~  \skey{l}{T}  ~|~ \encT{T}{k}  ~|~  \aencT{T}{k} &  \\
    &| &  \LRTnewnew{\noncetypelab{l}{\oneorinf}{n}}{\noncetypelab{l'}{\oneorinf}{m}} \text{ with } \oneorinf \in \{1,\infty\} ~|~   T \orT T &
  \end{array}
\]
\caption{Types for terms  (selected)}
\label{fig:types}
\end{figure}
A selection of the types for messages are defined in Figure~\ref{fig:types} and
explained below.  We assume three security labels (namely, $\S,\H,\L$), ranged
over by $l$, whose first (resp. second) component denotes the confidentiality
(resp. integrity) level. Intuitively, messages of high confidentiality cannot
be learned by the attacker, while messages of high integrity cannot originate
from the attacker. Pair types  describe the type of their components, as usual. 
Type $\skey{l}{T}$ describes keys of security level $l$ used to encrypt (or sign) messages of type $T$. 
 The type $ \encT{T}{k} $ (resp. $ \aencT{T}{k} $) describes  symmetric (resp. asymmetric) encryptions with key
    $k$  of a message of type $T$. 
      The type
    $\noncetypelab{l}{\oneorinf}{i}$ describes  nonces and constants of
    security level $l$: the label $\oneorinf$ ranges over $\{\infty,1\}$,
    denoting whether the nonce is bound within a replication or not (constants
    are always typed with $\oneorinf=1$). We assume a different identifier  $i$
    for each constant and restriction in the process. The type
    $\noncetypelab{l}{1}{i}$ is populated by a single name, (i.e., $i$
    describes a constant or a non-replicated nonce) and
    \smash{$\noncetypelab{l}{\infty}{i}$} is a special type, that is instantiated to
    \smash{$\noncetypelab{l}{1}{i_j}$} in the $jth$ replication of the process.
  Type $\LRTnewnew{\noncetypelab{l}{\oneorinf}{n}}{\noncetypelab{l'}{\oneorinf}{m}}$ is a refinement type that restricts the set
    of values which  can be taken by a message to values of type $\noncetypelab{l}{\oneorinf}{n}$ on the left and type $\noncetypelab{l'}{\oneorinf}{m}$ on the right.
For a refinement type 
    $\LRTnewnew{\noncetypelab{l}{\oneorinf}{n}}{\noncetypelab{l}{\oneorinf}{n}}$
    with equal types on both sides we simply write
    $\noncetypelab{l}{\oneorinf}{n}$.
  Messages of type $T \orT T'$ are messages that can have type $T$ or
    type $T'$.


\subsection{Constraints}
When typing messages, we generate constraints of the form  $(M \eqC N)$, meaning that the attacker sees $M$ and $N$ in the left and right process, respectively, and these two messages are thus required to be indistinguishable.  

\subsection{Typing Messages}

\begin{figure*}
\begin{framed}
\[ \infer[TNonce]
  {\Gamma(n)=\noncetypelab{l}{\oneorinf}{n} \\ 
  \Gamma(m)=\noncetypelab{l}{\oneorinf}{m} \\ 
  l \in \{\S, \H\} }
  {\teqTcnew{\Gamma}{n}{m}{l}{\emptyset}}
 \hspace{10pt}
 \infer[TNonceL]
  {\Gamma(n)=\noncetypelab{\L}{\oneorinf}{n}}
  {\teqTcnew{\Gamma}{n}{n}{\L}{\emptyset}}
 \hspace{10pt}
 \infer[TCstFN]
  {a \in \CST \cup \FN}
  {\teqTcnew{\Gamma}{a}{a}{\L}{\emptyset}}
\]
\[
 \infer[TPubKey]
  {k\in \dom{\Gamma}}
  {\teqTcnew{\Gamma}{\PUBK{k}}{\PUBK{k}}{\L}{\emptyset}}
 \hspace{10pt}
 \infer[TVKey]
  {k\in \dom{\Gamma}}
  {\teqTcnew{\Gamma}{\VK{k}}{\VK{k}}{\L}{\emptyset}}
 \hspace{10pt}
 \infer[TKey]
  { \Gamma(k) = T}
  {\teqTcnew{\Gamma}{k}{k}{T}{\emptyset}}
\]
\[
\infer[TVar]
  { \Gamma(x) = T}
  {\teqTcnew{\Gamma}{x}{x}{T}{\emptyset}}
\hspace{10pt}
 \infer[TPair]
  {\teqTcnew{\Gamma}{M}{N}{T}{c} \\
   \teqTcnew{\Gamma}{M'}{N'}{T'}{c'}}
  {\teqTcnew{\Gamma}{\PAIR{M}{M'}}{\PAIR{N}{N'}}{T * T'}{c\cup c'}}
\]
\[ \infer[TEnc]
  {\teqTcnew{\Gamma}{M}{N}{T}{c}}
  {\teqTcnew{\Gamma}{\ENC{M}{k}}{\ENC{N}{k}}{\encT{T}{k}}{c}}
\hspace{10pt}
 \infer[TEncH]
  {\teqTcnew{\Gamma}{M}{N}{\encT{T}{k}}{c}\\
   \Gamma(k) = \skey{\S}{T}}
  {\teqTcnew{\Gamma}{M}{N}{\L}{c \cup \{M \eqC N\}}}
\]
\[ \infer[TEncL]
  {\teqTcnew{\Gamma}{M}{N}{\encT{\L}{k}}{c}\\
   \Gamma(k) = \skey{\L}{T}}
  {\teqTcnew{\Gamma}{M}{N}{\L}{c}}
\hspace{10pt}
  \infer[TAenc]
  {\teqTcnew{\Gamma}{M}{N}{T}{c}}
  {\teqTcnew{\Gamma}{\AENC{M}{\PUBK{k}}}{\AENC{N}{\PUBK{k}}}{\aencT{T}{k}}{c}}
\]
\[ \infer[TAencH]
  {\teqTcnew{\Gamma}{M}{N}{\aencT{T}{k}}{c}\\
   \Gamma(k) = \skey{\S}{T}}
  {\teqTcnew{\Gamma}{M}{N}{\L}{c \cup \{M \eqC N\}}}
\hspace{10pt}
  \infer[TAencL]
  {\teqTcnew{\Gamma}{M}{N}{\aencT{\L}{k}}{c}\\
   k \in \dom{\Gamma}}
  {\teqTcnew{\Gamma}{M}{N}{\L}{c}}
\]
\[ \infer[TSignH]
  {\teqTcnew{\Gamma}{M}{N}{T}{c}\\
   \teqTcnew{\Gamma}{M}{N}{\L}{c'}\\
   \Gamma(k) = \skey{\S}{T}}
  {\teqTcnew{\Gamma}{\SIGN{M}{k}}{\SIGN{N}{k}}{\L}{c \cup c' \cup \{\SIGN{M}{k} \eqC \SIGN{N}{k}\}}}
\]
\[ \infer[TSignL]
  {\teqTcnew{\Gamma}{M}{N}{\L}{c}\\
   \Gamma(k) = \skey{\L}{T} }
  {\teqTcnew{\Gamma}{\SIGN{M}{k}}{\SIGN{N}{k}}{\L}{c}}
\hspace{10pt}
\infer[THash]
  {\names{M} \cup \names{N} \cup \var{M} \cup \var{N} \subseteq \dom{\Gamma} \cup \FN}
  {\teqTcnew{\Gamma}{\HASH{M}}{\HASH{N}}{\L}{\{\HASH{M}\eqC \HASH{N}\}}}
\]
\[
 \infer[THashL]
  {\teqTcnew{\Gamma}{M}{N}{\L}{c}}
  {\teqTcnew{\Gamma}{\HASH{M}}{\HASH{N}}{\L}{c}}
\hspace{10pt}
\infer[THigh]
  {\names{M} \cup \names{N} \cup \var{M} \cup \var{N} \subseteq \dom{\Gamma} \cup \FN}
  {\teqTcnew{\Gamma}{M}{N}{\H}{\emptyset}}
\]
\[ \infer[TSub]
  {\teqTcnew{\Gamma}{M}{N}{T'}{c}\\
   T' \subtyp T}
  {\teqTcnew{\Gamma}{M}{N}{T}{c}}
\hspace{10pt}
 \infer[TOr]
  {\teqTcnew{\Gamma}{M}{N}{T}{c}}
  {\teqTcnew{\Gamma}{M}{N}{T \orT T'}{c}}
\]
\[
 \infer[TLR$^1$]
  {\Gamma(m)=\noncetypelab{l}{1}{m} \quad\text{or}\quad m\in\FN\cup\CST \;\wedge\; l = \L\\\\
  \Gamma(n)=\noncetypelab{l'}{1}{n} \quad\text{or}\quad n\in\FN\cup\CST \;\wedge\; l' = \L} 
  {\teqTcnew{\Gamma}{m}{n}{\LRTnewnew{\noncetypelab{l}{1}{m}}{\noncetypelab{l'}{1}{n}}}{\emptyset}}
\hspace{10pt}
 \infer[TLR$^\infty$]
  {\Gamma(m)=\noncetypelab{l}{\infty}{m} \\
  \Gamma(n)=\noncetypelab{l'}{\infty}{n}} 
  {\teqTcnew{\Gamma}{m}{n}{\LRTnewnew{\noncetypelab{l}{\infty}{m}}{\noncetypelab{l'}{\infty}{n}}}{\emptyset}}
\]
\[
\infer[TLR']
  {\teqTcnew{\Gamma}{M}{N}{\LRTnewnew{\noncetypelab{l}{\oneorinf}{m}}{\noncetypelab{l}{\oneorinf}{n}}}{c} \\ 
   l \in \{\H,\S\}
  }
  {\teqTcnew{\Gamma}{M}{N}{l}{c}}
\hspace{10pt}
\infer[TLRL']
  {\teqTcnew{\Gamma}{M}{N}{\LRTnewnew{\noncetypelab{\L}{\oneorinf}{n}}{\noncetypelab{\L}{\oneorinf}{n}}}{c}}
  {\teqTcnew{\Gamma}{M}{N}{\L}{c}}
\]
\[ \infer[TLRVar]
  {\teqTcnew{\Gamma}{x}{x}{\LRTnewnew{\noncetypelab{l}{1}{m}}{\noncetypelab{l'}{1}{n}}}{\emptyset}\\
  \teqTcnew{\Gamma}{y}{y}{\LRTnewnew{\noncetypelab{l''}{1}{m'}}{\noncetypelab{l'''}{1}{n'}}}{\emptyset}}
  {\teqTcnew{\Gamma}{x}{y}{\LRTnewnew{\noncetypelab{l}{1}{m}}{\noncetypelab{l'''}{1}{n'}}}{\emptyset}}
\]
%
%
\end{framed}
\caption{Rules for Messages}
\label{fig:termtypingrules}
\end{figure*}

Typing judgments are parametrized over a typing environment $\Gamma$, which is a list of mappings from names and variables to types. 
The  typing judgement for messages is of the form
the form
$
  \teqTcnew{\Gamma}{M}{N}{T}{c}
$
which reads as follows: under the environment $\Gamma$,  $M$ and $N$ are of type $T$ and either this is a high confidentiality type (i.e.,  $M$ and $N$ are not disclosed to the attacker) or $M$ and $N$ are indistinguishable for the attacker assuming the set of constraints $c$ holds true. 
We present an excerpt of the  typing rules for messages in
Figure~\ref{fig:termtypingrules} and comment on them in the following. 
%
%
%

 Confidential nonces (i.e. nonces with label $l=\S$ or $l=\H$) are
    typed with their label from the typing environment. As the attacker may
    not observer them, they may be different in the left and the right message
    and we do not add any constraints (\textsc{TNonce}).
  Public terms are given type $\L$ 
    if they are the same in the left and the right message (\textsc{TNonceL, TCstFN, TPubKey, TVKey}). 
 We require keys and variables to be the same in the two processes, deriving their type from the environment (\textsc{TKey} and \textsc{TVar}).  
 The rule for pairs operates recursively component-wise (\textsc{TPair}).
  
   For symmetric key encryptions (\textsc{TEnc}), we have to make sure that
    the payload type matches the key type (which is achieved by  rule \textsc{TEncH}). We add the generated ciphertext to
    the set of constraints, because even though the attacker cannot read the
    plaintext, he can perform an equality check on the  ciphertext that he
    observed. If we type an encryption with a key that is of low confidentiality
    (i.e., the attacker has access to it), then we need to make sure the
    payload is of type $\L$, because the attacker can simply decrypt the
    message and recover the plaintext (\textsc{TEncL}). 
    The rules for asymmetric encryption are the same, with the only difference
    that we can always chose to ignore the key type and use type $\L$ to check
    the payload.
    This allows us to type messages produced by  the attacker, which has access to the
    public key but does not  need to respect its type.
    Signatures are also handled similarly, the difference here is that we
    need to type the payload with $\L$ even if an honest key is used, as the
    signature does not hide the content. 
  The  first typing rule for hashes
    (\textsc{THash}) gives them type $\L$ and adds the term to the constraints,
    without looking at the arguments of the hash function: intuitively this is
    justified, because the hash function makes it impossible to recover the
    argument. 
    The second rule (\textsc{THashL}) gives type $\L$ only if we can also give
    type $\L$ to the argument of the hash function, but does not add any
    constraints on its own, it is just passing on the constraints created for
    the arguments. This means we are typing the message  as if the hash
    function would not have been applied and use the message without the hash,
    which is a strictly stronger result.
    Both rules have their applications: while the former has to be used
    whenever we hash a secret, the latter may be useful to avoid the creation
    of unnecessary constraints when hashing terms like constants or public
    nonces.
   Rule \textsc{THigh} states that we can give type $\H$ to every message,
    which intuitively means that we can treat every message as if it were
   confidential. 
   Rule \textsc{TSub} allows us to type  messages according to the  subtyping relation, which is standard and defined in 
    Figure~\ref{fig:subtypingrules}.
   Rule \textsc{TOr} allows us to give a union type to messages, if they
    are typable with at least one of the two types.
   \textsc{TLR$^1$} and \textsc{TLR$^\infty$} are the introduction rules for refinement types, while  \textsc{TLR'} and \textsc{TLRL'} are the corresponding elimination rules. 
   Finally,  \textsc{TLRVar} allows to derive a new refinement type for two
    variables for which we have singleton refinement types, by taking the left refinement
    of the left variable and the right refinement of the right variable. We will see application of this rule in the e-voting protocol, where we use it to  combine A's vote (0 on the left, 1 on the right) and B's vote (1 on the
    left, 0 on the right), into a message that is the same on both sides.

\begin{figure}
\begin{framed}
\[
\infer[SRefl]
  { }
  {T \subtyp T}
\hspace{10pt}
\infer[SHigh]
  { }
  {T \subtyp \H}
\]
\[
\infer[STrans]
  {T \subtyp T' \\ T' \subtyp T''}
  {T \subtyp T''}
\hspace{10pt}
\infer[SPairL]
  { }
  {\L * \L \subtyp \L}
\]
\[
\infer[SPair]
  {T_1 \subtyp T_1' \\
   T_2 \subtyp T_2'}
  {T_1 * T_2 \subtyp T_1' * T_2'}
\hspace{10pt}
\infer[SPairS]
  { }
  {\S * T \subtyp \S}
\]
\[
\infer[SPairS']
  { }
  {T * \S \subtyp \S}
\hspace{10pt}
\infer[SKey]
  { }
  {\skey{l}{T} \subtyp l}
\]
\[
\infer[SEnc]
  {T \subtyp T'}
  {\encT{T}{k} \subtyp \encT{T'}{k}}
\hspace{10pt}
\infer[SAenc]
  {T \subtyp T'}
  {\aencT{T}{k} \subtyp \aencT{T'}{k}}
\]
\end{framed}
\caption{Subtyping Rules}
\label{fig:subtypingrules}
\end{figure}

\subsection{Typing Processes}

The typing judgement for processes is of the form 
$
  \teqPnew{\Gamma}{P}{Q}{C}
$
and can be interpreted as follows:
If two processes $P$ and $Q$ can be typed in $\Gamma$ and if the generated
constraint set $C$ is consistent, then $P$ and $Q$ are trace
equivalent. 
We assume in this section  that $P$ and $Q$ do not contain replication
and that variables and names are renamed to avoid any capture. We also
assume processes to be given with type annotations for nonces.

When typing processes, the typing environment $\Gamma$ is passed down and
extended from the root towards the leafs of the syntax tree of the process,
i.e., following the execution semantics. 
The generated constraints $C$ however, are passed up from the leafs towards the
root, so that at the root we get all generated  constraints, modeling the
attacker's global view on the process execution.

More precisely, each possible execution path of the process - there may be
multiple paths because of conditionals - creates its own set of constraints $c$
together with the typing environment $\Gamma$ that contains types for all names
and variables appearing in $c$. Hence a \emph{constraint set} $C$ is a set
elements of the form $(c,\Gamma)$ for a set of constraints $c$. The typing
environments are required in the constraint checking procedure, as they helps us
to be more precise when checking the consistency of constraints.

\begin{figure*}
\begin{framed}
\[
\infer[PZero]
  {\tewf{\Gamma}\\ \Gamma\text{ does not contain union types}}
  {\teqPnew{\Gamma}{\ZERO}{\ZERO}{(\emptyset,\Gamma)}}
\]
\[
\infer[POut]
  {\teqPnew{\Gamma}{P}{Q}{C} \\
   \teqTcnew{\Gamma}{M}{N}{\L}{c}}
  {\teqPnew{\Gamma}{\OUT{M}.P}{\OUT{N}.Q}{C \UnionAll c}}
\hspace{10pt}
\infer[PIn]
  {\teqPnew{\Gamma,x:\L}{P}{Q}{C}}
  {\teqPnew{\Gamma}{\IN{x}.P}{\IN{x}.Q}{C}}
\]
\[
\infer[PNew]
  {
   \teqPnew{\Gamma,{n:\noncetypelab{l}{\oneorinf}{n}}}{P}{Q}{C}}
  {\teqPnew{\Gamma}{\NEWnew{n}{\noncetypelab{l}{\oneorinf}{n}}.P}{\NEWnew{n}{\noncetypelab{l}{\oneorinf}{n}}.Q}{C}}
\]
\[
\infer[PPar]
  {\teqPnew{\Gamma}{P}{Q}{C} \\
   \teqPnew{\Gamma}{P'}{Q'}{C'}}
  {\teqPnew{\Gamma}{P \PAR P'}{Q \PAR Q'}{C \UnionCart C'}}
\hspace{10pt}
\infer[POr]
  {\teqPnew{\Gamma,x:T}{P}{Q}{C}\\
  \teqPnew{\Gamma,x:T'}{P}{Q}{C'}}
  {\teqPnew{\Gamma, x:T\orT T'}{P}{Q}{C \cup C'}}
\]
\[
\infer[PLet]
  {\tDestnew{\Gamma}{d(y)}{T}\\
   \teqPnew{\Gamma,x:T}{P}{Q}{C}\\
   \teqPnew{\Gamma}{P'}{Q'}{C'}}
  {\teqPnew{\Gamma}{\LET{x}{d(y)}{P}{P'}}{\LET{x}{d(y)}{Q}{Q'}}{C \cup C'}}
\]
\[
\infer[PLetLR]
   {\Gamma(y) = \LRTnewnew{\noncetypelab{l}{\oneorinf}{n}}{\noncetypelab{l'}{\oneorinf}{m}}\\
   \teqPnew{\Gamma}{P'}{Q'}{C'}}
  {\teqPnew{\Gamma}{\LET{x}{d(y)}{P}{P'}}{\LET{x}{d(y)}{Q}{Q'}}{C'}}
\]
\[
\infer[PIfL]
  {\teqPnew{\Gamma}{P}{Q}{C}\\
  \teqPnew{\Gamma}{P'}{Q'}{C'}\\
  \teqTcnew{\Gamma}{M}{N}{\L}{c}\\
  \teqTcnew{\Gamma}{M'}{N'}{\L}{c'}}
  {\teqPnew{\Gamma}{\ITE{M}{M'}{P}{P'}}{\ITE{N}{N'}{Q}{Q'}}{\left( C \cup C'\right) \UnionAll (c \cup c')}}
\]
\ifallrules
\[
\infer[PIfP]
  {\teqPnew{\Gamma}{P}{Q}{C}\\
  \teqPnew{\Gamma}{P'}{Q'}{C'}\\\\
  \teqTcnew{\Gamma}{M}{N}{\L}{c}\\
  \teqTcnew{\Gamma}{t}{t}{\L}{c'} \\
  t \in \K \cup \N \cup \CST}
  {\teqPnew{\Gamma}{\ITE{M}{t}{P}{P'}}{\ITE{N}{t}{Q}{Q'}}{C \cup C'}}
\]
\fi
\[
\infer[PIfLR]
  {\teqTcnew{\Gamma}{M_1}{N_1}{\LRTnewnew{\noncetypelab{l}{1}{m}}{\noncetypelab{l'}{1}{n}}}{\emptyset}\\
  \teqTcnew{\Gamma}{M_2}{N_2}{\LRTnewnew{\noncetypelab{l''}{1}{m'}}{\noncetypelab{l'''}{1}{n'}}}{\emptyset}\\\\
  b = (\noncetypelab{l}{1}{m} \overset{?}{=} \noncetypelab{l''}{1}{m'}) \\ 
  b' = (\noncetypelab{l'}{1}{n} \overset{?}{=} \noncetypelab{l'''}{1}{n'}) \\ 
  \teqPnew{\Gamma}{P_b}{Q_{b'}}{C}}
  {\teqPnew{\Gamma}{\ITE{M_1}{M_2}{P_\top}{P_\bot}}{\ITE{N_1}{N_2}{Q_\top}{Q_\bot}}{C}}
\]
\[
\infer[PIfS]
  {\teqPnew{\Gamma}{P'}{Q'}{C'}\\
  \teqTcnew{\Gamma}{M}{N}{\L}{c}\\
  \teqTcnew{\Gamma}{M'}{N'}{\S}{c'}}
  {\teqPnew{\Gamma}{\ITE{M}{M'}{P}{P'}}{\ITE{N}{N'}{Q}{Q'}}{C'}}
\]
\ifallrules
\[
\infer[PIfI]
  {\teqPnew{\Gamma}{P'}{Q'}{C'}\\
  \teqTcnew{\Gamma}{M}{N}{T*T'}{c}\\
  \teqTcnew{\Gamma}{M'}{N'}{\LRT{M''}{N''}}{c'}\\
  \text{$M''$, $N''$ not pairs}}
  {\teqPnew{\Gamma}{\ITE{M}{M'}{P}{P'}}{\ITE{N}{N'}{Q}{Q'}}{C'}}
\]
\joseph{this could be much more general: we only type the else branch as soon as the types of $M$, $N$ and $M'$, $N'$ are incompatible, \ie $M$, $N$ are of type $T*T'$ and
$M'$, $N'$ of type $T+T'$, $\encT{T}{k}$, $\aencT{T}{k}$, or $\LRT{M''}{N''}$ with $M''$, $N''$ ground and not pairs;
and similarly if $M$, $N$ are of type $T+T'$, $\encT{T}{k}$ or $\aencT{T}{k}$.
}
\fi
\ifallrules
\[
\infer[PLetLR*]
  {x \notin \dom{\Gamma}\\
   \Gamma(y) = \LRT{\nrep{m}}{\nrep{n}}\\
   \teqPnew{\Gamma}{P'}{Q'}{C'}}
  {\teqPnew{\Gamma}{\LET{x}{d(y)}{P}{P'}}{\LET{x}{d(y)}{Q}{Q'}}{C'}}
\]
\fi
\[
\infer[PIfLR*]
  {\teqTcnew{\Gamma}{M_1}{N_1}{\LRTnewnew{\noncetypelab{l}{\infty}{m}}{\noncetypelab{l'}{\infty}{n}}}{\emptyset}\\
  \teqTcnew{\Gamma}{M_2}{N_2}{\LRTnewnew{\noncetypelab{l}{\infty}{m}}{\noncetypelab{l'}{\infty}{n}}}{\emptyset}\\\\
  \teqPnew{\Gamma}{P}{Q}{C}\\
  \teqPnew{\Gamma}{P'}{Q'}{C'}}
  {\teqPnew{\Gamma}{\ITE{M_1}{M_2}{P}{P'}}{\ITE{N_1}{N_2}{Q}{Q'}}{C\cup C'}}
\]
\ifallrules
\[
\infer[PIfLR'*]
  {\teqTcnew{\Gamma}{M_1}{N_1}{\LRT{M_1'}{N_1'}}{c_1}\\
  \teqTcnew {\Gamma}{M_2}{N_2}{\LRT{\nrep{m}}{\nrep{n}}}{c_2}\\\\
  M_1'\neq m, N_1'\neq n, M_1', N_1' \text{ are ground}\\
  \teqPnew{\Gamma}{P'}{Q'}{C}}
  {\teqPnew{\Gamma}{\ITE{M_1}{M_2}{P}{P'}}{\ITE{N_1}{N_2}{Q}{Q'}}{C}}
\]
\[
\infer[PIfI*]
  {\teqPnew{\Gamma}{P'}{Q'}{C'}\\
  \teqTcnew{\Gamma}{M}{N}{T*T'}{c}\\
  \teqTcnew{\Gamma}{M'}{N'}{\LRT{\nrep{m}}{\nrep{n}}}{c'}}
  {\teqPnew{\Gamma}{\ITE{M}{M'}{P}{P'}}{\ITE{N}{N'}{Q}{Q'}}{C'}}
\]
\joseph{this could be more general, as soon as $M$ $N$ are of a type which can't be a nonce we know the test fails on both sides}
\fi
\end{framed}
\caption{Rules for processes}
\label{fig:processtypingrules}
\end{figure*}
An excerpt of our typing rules for processes is  presented in
Figure~\ref{fig:processtypingrules} and explained in the following.
Rule \textsc{PZero} copies the current typing environment in the constraints and checks the well-formedness of the environment ($\Gamma \vdash \diamond$), which is defined as expected. 
 Messages output on the network are possibly learned by the attacker, so they have to be of   type $\L$ (\textsc{POut}). 
The generated constraints are added to each element of the constraint set for the continuation process, using the
operator $\UnionAll$ defined as
\[C \UnionAll c' := \left\{ (c \cup c',\Gamma) \;|\; (c,\Gamma) \in C \right\}.\]
Conversely, messages input from the network are given type $\L$ (\textsc{PIn}). 
 Rule \textsc{PNew} introduces a new nonce, which may be used in the
    continuation processes. 
 While typing parallel composition  (\textsc{PPar}), we type the individual subprocesses and take the product union
    of the generated constraint sets as the new constraint set. 
    The \emph{product union} of constraint sets is defined as 
\begin{align*}
  &C \UnionCart C' := \{  (c \cup c', \Gamma \cup \Gamma') \;|\;\\
   &\quad (c,\Gamma) \in C \,\wedge\, (c',\Gamma') \in C'
   \,\wedge\, \Gamma, \Gamma' \text{ are compatible} \}
\end{align*}
    where \emph{compatible} environments are those that agree on the type of all
    arguments of the shared domain.   This operation models the fact that a
    process $P \PAR P'$ can have every trace that is a combination of any trace
    of $P$ with any trace of $P'$.
    The branches that are discarded due to incompatible environments correspond
   to  impossible executions (e.g., taking the left branch in $P$ and the right
    branch in $P'$ in  two conditionals with the same guard).
   \textsc{POr} is the elimination rule for union types, which requires  the continuation process to be well-typed with both types. 

    To ensure that the destructor application fails or succeeds equally in the
    two processes, we allow only the same destructor to be applied to the same
    variable in both processes (\textsc{PLet}). As usual, we then type-check the then as well as the else branch and then take the union of the corresponding constraints.  The typing rules for destructors are
    presented in Figure~\ref{fig:destructorrules}. 
    These are mostly standard: for instance, after decryption,  the type of the payload is determined by the one of  the decryption key, as long as this is of high integrity (\textsc{DDecH}). We can as well exploit strong types for ciphertexts, typically introduced by verifying a surrounding signature (see, e.g., the types for Helios) to derive the type of the payload (\textsc{DDecT}). 
    In the case of public key encryption, we have to be careful, since the public encryption key is accessible to the attacker: we thus give the payload type $T \orT
    \L$  (rule \textsc{DAdecH}). For operations involving corrupted
    keys (label $\L$) we know that the payload is public and hence give the derived message type $\L$.

In the  special case in which we know that the 
    concrete value of the argument of the destructor application is a nonce or constant due to a refinement type, and we
    know statically that any destructor application will fail,  we only need to type-check  the else branch (\textsc{PLetLR}).
    As  for destructor applications, 
    the difficulty while typing conditionals is to make sure  that the same branch is taken in both
    processes (\textsc{PIfL}). To ensure this we use a trick: We type both the left and the
    right operands of the conditional with type $\L$ and add both generated
     sets of constraints to the constraint set.  Intuitively, this means that the attacker could perform the equality test himself, since the guard is of type $\L$, which means that the conditional must take the same branch on the left and on the right. 
In the special case  in which we can statically determine
    the concrete value of the terms in the conditional (because the corresponding type is populated by a singleton), we  have to typecheck only the single combination of branches that will be
    executed (\textsc{PIfLR}). 
   Another special case  is if  the messages on the right  are of type $\S$ and the ones on the left of type $\L$. As a secret of high integrity  can never be equal
    to a public value of low integrity, we know that both
    processes will take the else branch (\textsc{PIfS}). 
     This rule is crucial, since it may allow us
    to prune the low typing branch of asymmetric decryption. 
   The last special case for conditionals is 
    when we have a refinement type with replication for both operands of the
    equality check (\textsc{PIfLR*}). Although we know that the nonces on both sides are of the
    same type and hence both are elements of the same set, we cannot assume
    that they are equal, as the sets are infinite, unlike in rule
    \textsc{PIfLR}.
    Yet, 
\veroniquebis{concrete instantiations of nonces will have the same
  index for the left  and the right process. This is because we check
  for a variant of diff-equivalence.}
    This ensures
    that the equality check always yields the same result in the two processes. All these special cases highlight how a careful treatment of  names in terms of equivalence classes (statically captured by types) is a powerful device to enhance the expressiveness of the analysis.

  Finally, notice that we do not have any typing rule for replication: this is in line with our general idea of typing a bounded number of sessions and then extending this result  to the unbounded case in the constraint checking phase, as detailed in  Section~\ref{sec:results}.

\begin{figure}
\begin{framed}
\[ \infer[DDecH]
  {\Gamma(k) = \skey{\S}{T}\\
   \Gamma(x) = \L}
  {\tDestnew{\Gamma}{\DEC{x}{k}}{T}}
\]
\[ \infer[DDecL]
  {\Gamma(k) = \skey{\L}{T}\\
   \Gamma(x) = \L}
  {\tDestnew{\Gamma}{\DEC{x}{k}}{\L}}
\]
\[ \infer[DDecT]
  {\Gamma(x) = \encT{T}{k}}
  {\tDestnew{\Gamma}{\DEC{x}{k}}{T}}
\]
%
%
\[ \infer[DAdecH]
  {\Gamma(k) = \skey{\S}{T}\\
   \Gamma(x) = \L}
  {\tDestnew{\Gamma}{\ADEC{x}{k}}{T \orT \L}}
\]
\[ \infer[DAdecL]
  {\Gamma(k) = \skey{\L}{T}\\
   \Gamma(x) = \L}
  {\tDestnew{\Gamma}{\ADEC{x}{k}}{\L}}
\]  
\[ \infer[DAdecT]
  {\Gamma(x) = \aencT{T}{k}}
  {\tDestnew{\Gamma}{\ADEC{x}{k}}{T}}
\]
\[ \infer[DCheckH]
  {\Gamma(k) = \skey{\S}{T}\\
   \Gamma(x) = \L}
  {\tDestnew{\Gamma}{\CHECK{x}{\VK{k}}}{T}}
\]
\[ \infer[DCheckL]
  {\Gamma(k) = \skey{\L}{T}\\ 
   \Gamma(x) = \L}
  {\tDestnew{\Gamma}{\CHECK{x}{\VK{k}}}{\L}}
\]  
%
  
\[ \infer[DFst]
  {\Gamma(x) = T * T'}
  {\tDestnew{\Gamma}{\FST{x}}{T}}
\hspace{10pt}
\infer[DSnd]
  {\Gamma(x) = T * T'}
  {\tDestnew{\Gamma}{\SND{x}}{T'}}
\]
\[ \infer[DFstL]
  {\Gamma(x) = \L}
  {\tDestnew{\Gamma}{\FST{x}}{\L}}
\hspace{10pt}
 \infer[DSndL]
  {\Gamma(x) = \L}
  {\tDestnew{\Gamma}{\SND{x}}{\L}}
\]
\end{framed}
\caption{Destructor Rules}
\label{fig:destructorrules}
\end{figure}

%% file: consistency.tex
\section{Consistency of Constraints}

Our type system  guarantees trace equivalence of two processes only if the 
generated constraints are \emph{consistent}. In this section we give a slightly
simplified definition of consistency of constraints and explain how it
captures the attacker's capability to distinguish processes based on their
outputs.

To define consistency, we need the following ingredients:
\begin{itemize}
  \item $\phiL{c}$ and $\phiR{c}$ denote the frames that are composed of the
    left and the right terms of the constraints respectively (in the same order).
  \item $\phiEnew{\Gamma}$ denotes the frame that is composed of all low
    confidentiality nonces and keys in $\Gamma$, as well as all public
    encryption keys and verification keys in $\Gamma$. \veroniquebis{This
    intuitively corresponds to the initial knowledge of the attacker.}
  \item Let $\E_{\Gamma}$ be the set of all nonces occurring in $\Gamma$.
  \item Two ground substitutions $\sigma, \sigma'$ are well-formed in $\Gamma$
    if they preserve the types for variables in $\Gamma$ (i.e.,
    $\teqTcnew{\Gamma}{\sigma(x)}{\sigma'(x)}{\Gamma(x)}{c_x}$).
\end{itemize} 
\begin{definition}[Consistency]
A set of constraints $c$ is \emph{consistent} in an environment
$\Gamma$ if for all substitutions $\sigma$,$\sigma'$ well-typed in $\Gamma$ the
frames $\NEWN{\E_{\Gamma}}.(\phiEnew{\Gamma} \cup \phiL{c}\sigma)$ and
$\NEWN{\E_{\Gamma}}.(\phiEnew{\Gamma} \cup \phiR{c}\sigma')$ are statically
equivalent.
We say that $(c,\Gamma)$ is consistent if $c$ is consistent in $\Gamma$ and
that a constraint set $C$ is consistent in $\Gamma$ if each element $(c, \Gamma) \in C$
is consistent.
\end{definition}
We define consistency of constraints in terms of static equivalence, as this
notion exactly captures all capabilities of our attacker: to distinguish two
processes, he can arbitrarily apply constructors and destructors on observed
messages to create new terms, on which he can then perform equality tests or
check the applicability of destructors.
We require that this property holds for any well-typed substitutions, to soundly
cover that fact that we do not know the content of variables statically, except
for the information we get by typing. In Section \ref{sec:consistency-check} we introduce an algorithm  to check consistency of constraints.

%% file: results.tex
\section{Main results}
\label{sec:results}
In this
section, we state our two main soundness theorems, entailing trace equivalence by typing for the bounded and
unbounded case, and we explain how to automatically check consistency.
\subsection{Soundness of the type system}

Our type system soundly enforces trace equivalence: if we can typecheck $P$ and $Q$ then $P$ and
$Q$ are equivalent, provided that the corresponding constraint set is consistent.
\begin{theorem}[Typing implies trace equivalence]
\label{thm:typing-sound}
For all $P$, $Q$, and $C$, for all $\Gamma$ containing only keys, if~ $\teqPnew{\Gamma}{P}{Q}{C}$
and $C$ is consistent, then
$P \equivTrace Q$.
\end{theorem}

\josephbis{%
To prove this theorem, we first show that typing is preserved by reduction,
and guarantees that the same actions can be observed on both sides.
More precisely, we show that if $\PP$ and $\QQ$ are multisets of processes which are
pairwise typably equivalent (with consistent constraints),
and if a reduction step with action $\alpha$ can be performed to reduce $\PP$ into $\PP'$,
then $\QQ$ can be reduced in one or several steps, with the same action $\alpha$,
to some multiset $\QQ'$ such that the processes in $\PP'$ and $\QQ'$ are still typably equivalent (with consistent constraints).
This is done by carefully examining all the possible typing rules used to type the processes in $\PP$ and $\QQ$.
In addition we show that the frames of messages output when reducing $\PP$ and $\QQ$
are typably equivalent with consistent constraints; and that this entails their static equivalence.

This implies that if $P$ and $Q$ are typable with a consistent constraint,
then for each trace of $P$, by induction on the length of the trace, there exists a trace of $Q$ with the same sequence of
actions, and with a statically equivalent frame.
That is to say $P \incTrace Q$.
Similarly we show $Q \incTrace P$, and we thus have $P \equivTrace Q$.
}

Since we do not have typing rules for  replication,
Theorem~\ref{thm:typing-sound} only allows us to prove equivalence of
protocols for a \emph{finite} number of sessions. An arguably surprising result, however, is that, thanks to  our infinite nonce types, 
we can prove  equivalence for an
\emph{unbounded} number of sessions, as detailed in the  next section.

\subsection{Typing replicated processes}


For more clarity, in this section, without loss of generality
we consider that 
for each infinite nonce type $\noncetypelab{l}{\infty}{m}$ appearing in the processes,
the set of names $\BN$ contains an infinite number of fresh names $ \{m_i\;|\; i\in\mathbb{N}\}$ which do not appear
in the processes or environments.
We similarly assume that for all the variables $x$ appearing in the processes, the set $\X$ of all variables also contains
fresh variables $\{x_i\;|\;i\in\mathbb{N}\}$ which do not appear in the processes or environments.

Intuitively, whenever we can typecheck a process of the form $\NEWnew{n}{\noncetypelab{l}{1}{n}}.\;
\NEWnew{m}{\noncetypelab{l}{\infty}{m}}.\; P$, we can actually typecheck
\[\NEWnew{n}{\noncetypelab{l}{1}{n}}.\; (\NEWnew{m_1}{\noncetypelab{l}{1}{m_1}}. P_1 \PAR \dots \PAR \NEWnew{m_k}{\noncetypelab{l}{1}{m_k}}. P_k)\]
where 
in $P_i$, the nonce $m$ has been replaced by $m_i$ and variables $x$ have
been renamed to $x_i$.

Formally, we denote by $\instTerm{t}{i}{\Gamma}$, the term $t$ in which names
$n$ such that $\Gamma(n)=\noncetypelab{l}{\infty}{n}$ for some $l$ are replaced by $n_i$,
and variables $x$ are replaced by $x_i$.

Similarly, when 
a term is of type $\LRTnewnew{\noncetypelab{l}{\infty}{m}}{\noncetypelab{l'}{\infty}{p}}$, it can
be of type $\LRTnewnew{\noncetypelab{l}{1}{m_i}}{\noncetypelab{l'}{1}{p_i}}$ for any $i$.
The nonce type $\noncetypelab{l}{\infty}{m}$ represents infinitely many nonces (one for each session). 
That is, for $n$ sessions, the  type $\LRTnewnew{\noncetypelab{l}{\infty}{m}}{\noncetypelab{l'}{\infty}{p}}$
represents all  $\LRTnewnew{\noncetypelab{l}{1}{m_i}}{\noncetypelab{l'}{1}{p_i}}$. Formally, given a type
$T$, we define its expansion  to $n$ sessions, denoted
$\instTyp{T}{n}$, as follows.
\[
\begin{array}{r@{\;}c@{\;\;}l}
\instTyp{l}{n} &=& l \\
\instTyp{T * T'}{n} &=& \instTyp{T}{n}*\instTyp{T'}{n} \\
\instTyp{T + T'}{n} &=& \instTyp{T}{n}+\instTyp{T'}{n} \\
\instTyp{\skey{l}{T}}{n} &=& \skey{l}{\instTyp{T}{n}} \\
\instTyp{\encT{T}{k}}{n} &=& \encT{\instTyp{T}{n}}{k} \\
\instTyp{\aencT{T}{k}}{n} &=& \aencT{\instTyp{T}{n}}{k} \\
\instTyp{T \orT T'}{n} &=& \instTyp{T}{n}\orT\instTyp{T'}{n}\\
\instTyp{\LRTnewnew{\noncetypelab{l}{1}{m}}{\noncetypelab{l'}{1}{p}}}{n} &=& \LRTnewnew{\noncetypelab{l}{1}{m}}{\noncetypelab{l'}{1}{p}}\\
\instTyp{\LRTnewnew{\noncetypelab{l}{\infty}{m}}{\noncetypelab{l'}{\infty}{p}}}{n} &=& \bigvee_{j=1}^{n} \LRTnewnew{\noncetypelab{l}{1}{m_j}}{\noncetypelab{l'}{1}{p_j}} \\
\end{array}
\]
where $l, l'\in\{\L, \S, \H\}$, $k\in\K$. 
Note that the size of the expanded type $\instTyp{T}{n}$ depends on $n$.

We need to adapt typing environments accordingly.
For any typing environment $\Gamma$, we define its renaming for session $i$ as:
\begin{align*}
\instGr{\Gamma}{i} = \;&
\{x_i : T \;|\; \Gamma(x) = T\} \cup \{k : T \;|\; \Gamma(k) = T\}\\
&\cup \;\{m : \noncetypelab{l}{1}{m} \;|\; \Gamma(m) = \noncetypelab{l}{1}{m}\}\\
&\cup \;\{m_i : \noncetypelab{l}{1}{m_i} \;|\; \Gamma(m) = \noncetypelab{l}{\infty}{m}\}.
\end{align*}
and then its expansion to $n$ sessions as
\begin{align*}
\instG{\Gamma}{i}{n} = \;&
\{x_i : \instTyp{T}{n} \;|\; \instGr{\Gamma}{i}(x_i) = T\} \cup \{k : \instTyp{T}{n} \;|\; \instGr{\Gamma}{i}(k) = T\}\\
&\cup \;\{m : \noncetypelab{l}{1}{m} \;|\; \instGr{\Gamma}{i}(m) = \noncetypelab{l}{1}{m}\}.
\end{align*}

Note that in $\instG{\Gamma}{i}{n}$, due to the expansion, the size of the types depends on $n$.

By construction, the environments contained in the constraints generated by typing do not
contain union types.
However, refinement types with infinite nonce types introduce union types when expanded.
In order to recover environments without union types after expanding, which, as we will explain in the next subsection,
is needed for our consistency checking procedure,
we define $\branch{\instG{\Gamma}{i}{n}}$ as the set of all $\Gamma'$,
with the same domain as $\instG{\Gamma}{i}{n}$, such that for all $x$, $\Gamma'(x)$ is not a union type, and
either
\begin{itemize}
  \item $\instG{\Gamma}{i}{n}(x) = \Gamma'(x)$; 
  
\item or there exist types $T_1$,\dots,$T_k$,$T'_1$,\dots,$T'_{k'}$ such that
\[\instG{\Gamma}{i}{n}(x) = T_1 \orT \dots \orT T_k \orT \Gamma'(x) \orT T'_1 \orT\dots\orT T'_{k'}\]
\end{itemize}

Finally, when typechecking two processes containing nonces with infinite nonce types,
we collect constraints that represent families of constraints.

Given a set of constraints $c$, and an environment $\Gamma$, we define the renaming of $c$ for session $i$ in $\Gamma$ as 
$\instConst{c}{i}{\Gamma} = \{\instTerm{u}{i}{\Gamma} \eqC \instTerm{v}{i}{\Gamma}\;|\; u \eqC v \in c\}$.
This is propagated to constraint sets as follows:
the renaming of $C$ for session $i$ is
$\instCstr{C}{i} = \{(\instConst{c}{i}{\Gamma}, \instGr{\Gamma}{i})\;|\;(c,\Gamma)\in C\}$
and its expansion to $n$ sessions is
$\instCst{C}{i}{n} = \{(\instConst{c}{i}{\Gamma}, \Gamma')\;|\;\exists \Gamma.\; (c, \Gamma)\in C \;\wedge\; \Gamma' \in \branch{\instG{\Gamma}{i}{n}}\}$.

Again, note that the size of $\instCstr{C}{i}$ does not depend on the number of sessions considered, while
the size of the types present in $\instCst{C}{i}{n}$ does.
For example, for
$C = \{(\{\HASH{x} \eqC\HASH{x}\}, [x:\LRTnewnew{\noncetypelab{\S}{\infty}{m}}{\noncetypelab{\S}{\infty}{p}}])\}$,
we have
$\instCstr{C}{i} = \{(\{\HASH{x_i} \eqC\HASH{x_i}\}, [x_i:\LRTnewnew{\noncetypelab{\S}{\infty}{m}}{\noncetypelab{\S}{\infty}{p}}])\}$
and
$\instCst{C}{i}{n} = \{(\{\HASH{x_i} \eqC\HASH{x_i}\}, [x_i:\bigvee_{j=1}^{n} \LRTnewnew{\noncetypelab{\S}{1}{m_j}}{\noncetypelab{\S}{1}{p_j}}])\}$.


\newcommand{\Deltan}{\instD{\Delta}{n}}
\newcommand{\En}{\instE{\E}{n}}
\newcommand{\EEn}{\instE{\E'}{n}}
\newcommand{\Cin}{\instCst{C}{i}{n}}
\newcommand{\Ci}{\instCst{C}{1}{n}}
\newcommand{\Cii}{\instCst{C}{2}{n}}
\newcommand{\UCn}{\UnionCart_{1\leq i \leq n}\instCst{C}{i}{n}}
\newcommand{\Cud}{\instCst{C}{1}{n}\UnionCart\instCst{C}{2}{n}}
\newcommand{\unn}{\llbracket 1, n\rrbracket}
\newcommand{\nm}{\overline{n}}

Our type system is sound for replicated processes provided that the
collected constraint sets are consistent, when instantiated with all
possible instantiations of the nonces and keys.
\begin{theorem}
\label{thm:typing-sound-replicated}
Consider $P$, $Q$, $P'$ ,$Q'$, $C$, $C'$,
such that $P$, $Q$ and $P'$, $Q'$ do not share any variable.
Consider $\Gamma$, containing only keys and nonces with types of the form $\noncetypelab{l}{1}{n}$.

Assume that $P$ and $Q$ only bind nonces with infinite nonce types,
\ie using $\NEWnew{m}{\noncetypelab{l}{\infty}{m}}$ for some label $l$;
while $P'$ and $Q'$ only bind nonces with finite types, \ie using $\NEWnew{m}{\noncetypelab{l}{1}{m}}$.

Let us abbreviate by $\NEWN{\nm}$ the sequence of declarations of each nonce $m\in\dom{\Gamma}$.
If
\begin{itemize}
\item $\teqPnew{\Gamma}{P}{Q}{C}$,
\item $\teqPnew{\Gamma}{P'}{Q'}{C'}$,
\item $C'\UnionCart(\UCn)$ is consistent for all $n$,
\end{itemize}
then \hfill$\NEWN{\nm}. \;((!P)\PAR P') \equivTrace
\NEWN{\nm}. \;((!Q)\PAR Q')$.\hfill~
\end{theorem}
Theorem~\ref{thm:typing-sound} requires to check consistency of one constraint set. Theorem~\ref{thm:typing-sound-replicated} now
requires to check consistency of an infinite family of contraint sets. 
Instead of \emph{deciding} consistency, we provide a procedure that
checks a slightly stronger condition.

\subsection{Procedure for consistency}
\label{sec:consistency-check}
Checking consistency of a set of constraints amounts to checking static equivalence of the
corresponding frames. Our procedure follows the spirit
of~\cite{AbadiRogaway} for checking computational
indistinguishability: we first open encryption, signatures and pairs
as much as possible. Note that the type of a key indicates whether it
is public or secret. The two resulting frames should have the same
shape. Then, for unopened components, we simply need to check that
they satisfy the same equalities.

From now on, we only consider constraint sets that can actually be generated
when typing processes, as these are the only ones for which we need to check consistency.

Formally, the procedure $\checkconststar$ is described in
Figure~\ref{fig:checkconststar}. It consists of four steps.
First, we replace variables with refinements of finite nonce types by their left and right values.
In particular a variable with a union type is not associated with a single value and thus cannot be replaced.
This is why the branching operation needs to be performed when expanding environments
containing refinements with types of the form $\noncetypelab{l}{\infty}{n}$.
Second, we recursively open the constraints as much as possible.
Third, we check that the resulting constraints have the same shape.
Finally, as soon as two constraints $M \eqC M'$ and $N \eqC N'$ are such that $M$, $N$ are
unifiable, we must have $M'=N'$, and conversely. The condition is
slightly more involved, especially when the constraints contain
variables of refined types with infinite nonce types.

\begin{figure}
\begin{framed}\begin{minipage}{8cm}
\small
\newcommand{\separ}{\hfill\rule{0.5\linewidth}{0.5pt}\hfill~\vspace{0.3em}}
$\stepI_\Gamma(c) := \inst{c}{\sigma_F}{\sigma_F'}$, with
\begin{align*}
F := \{&x\in\dom{\Gamma} \;|\;\\ &\exists m,n,l,l'.\;\Gamma(x) = \LRTnewnew{\noncetypelab{l}{1}{m}}{\noncetypelab{l'}{1}{n}}\}
\end{align*}
and $\sigma_F,\sigma_F'$ defined by
\[\left\{\begin{array}{l}
\bullet\;\dom{\sigma_F} = \dom{\sigma_F'} = F\\
\bullet\;\forall x \in F.\;\forall m,n,l,l'.\\
\quad\LRTnewnew{\noncetypelab{l}{1}{m}}{\noncetypelab{l'}{1}{n}} \Rightarrow \sigma_F(x) = m\;\wedge\;\sigma_F'(x) = n
\end{array}\right.\]

\separ

$\stepII_\Gamma(c)$ is recursively defined by, for all $M$, $N$, $M'$, $N'$:
\begin{itemize}[leftmargin=0.5cm, nosep]
\item $\stepII_\Gamma(\{\PAIR{M}{N} \eqC \PAIR{M'}{N'}\} \cup c') :=$\par
$\stepII_\Gamma(\{M\eqC M', N\eqC N'\} \cup c')$

\item For all $k\in\K$, if $\exists T.\Gamma(k) = \skey{\L}{T}$:
\[\renewcommand{\arraystretch}{1.1}\begin{array}{l}
\bullet\;\stepII_\Gamma(\{\ENC{M}{k} \eqC \ENC{M'}{k}\} \cup c, c') := \\
\quad\stepII_\Gamma(\{M\eqC M'\} \cup c')\\
\bullet\;\stepII_\Gamma(\{\AENC{M}{\PUBK{k}} \eqC \AENC{M'}{\PUBK{k}}\} \cup c, c') :=\\
\quad \stepII_\Gamma(\{M\eqC M'\} \cup c')\\
\bullet\;\stepII_\Gamma(\{\SIGN{M}{k} \eqC \SIGN{M'}{k}\} \cup c') := \\
\quad \stepII_\Gamma(\{M\eqC M'\} \cup c')\\
\end{array}\]

\item For all $k\in\K$, if $\exists T.\Gamma(k) = \skey{\S}{T}$:
\[\begin{array}{l}
\stepII_{\Gamma}(\{\SIGN{M}{k} \eqC \SIGN{M'}{k}\} \cup c') := \\
\quad\{\SIGN{M}{k} \eqC \SIGN{M'}{k}\}\cup\stepII_{\Gamma}(\{M\eqC M'\} \cup c')\\
\end{array}\]

\item For all other terms $M,N$:
\[\stepII_{\Gamma}(\{M\eqC N\}\cup c') := \{M\eqC N\} \cup \stepII_{\Gamma}(c').\]
\end{itemize}

\separ

$\stepIII_\Gamma(c):=$ check that for all $M \eqC N \in c$, $M$ and $N$ are both
\begin{itemize}[leftmargin=0.5cm]
\item a key $k\in\K$ such that $\exists T. \Gamma(k)=\skey{\L}{T}$;
\item nonces $m,n\in\N$ such that
\[\exists a\in\{1,\infty\}.\;
\Gamma(n) = \noncetypelab{\L}{\oneorinf}{n}\;\wedge\;\Gamma(m)=\noncetypelab{\L}{\oneorinf}{n},\]

\item or public keys, verification keys, or constants;
\item or $\ENC{M'}{k}$, $\ENC{N'}{k}$ such that $\exists T. \Gamma(k)=\skey{\S}{T}$;
\item or either $\HASH{M'}$, $\HASH{N'}$ or $\AENC{M'}{\PUBK{k}}$, $\AENC{N'}{\PUBK{k}}$, where $\exists T. \Gamma(k)=\skey{\S}{T}$; 
such that $M'$ and $N'$ contain directly under pairs some $n$ with $\Gamma(n) = \S$ or $k$ such that $\exists T. \Gamma(k)=\skey{\S}{T}$;
\item or $\SIGN{M'}{k}$, $\SIGN{N'}{k}$ such that $\exists T. \Gamma(k)=\skey{\S}{T}$.
\end{itemize}

\separ

$\stepIV_\Gamma(c) :=$ If for all $M \eqC M'$ and $N \eqC N'\in c$
such that $M$, $N$ are unifiable with a most general unifier $\mu$,
and such that
\[\begin{array}{l}
\forall x\in\dom{\mu}.\exists l, l', m, p.\; 
(\Gamma(x)=\LRTnewnew{ \noncetypelab{l}{\infty}{m}}{ \noncetypelab{l'}{\infty}{p}}) \Rightarrow\\
\quad (x\mu\in\X\;\vee\;\exists i.\; x\mu=m_i)
\end{array}\]
we have
\[M'\alpha\theta = N'\alpha\theta\]

where
\[\begin{array}{l}
\forall x\in\dom{\mu}.\forall l, l', m, p, i.\\ 
\quad(\Gamma(x)=\LRTnewnew{ \noncetypelab{l}{\infty}{m}}{ \noncetypelab{l'}{\infty}{p}} \;\wedge\;\mu(x)=m_i)\Rightarrow\theta(x)=p_i
\end{array}\]

and $\alpha$ is the restriction of $\mu$ to
$\{x\in\dom{\mu}\;|\; \Gamma(x)=\L \;\wedge\; \mu(x)\in\N\}$;

and if the symmetric condition for the case where $M'$, $N'$ are unifiable holds as well,
then return \texttt{true}.

\separ

$\checkconststar(C):=$ for all $(c,\Gamma)\in C$, 
 let $c_1 := \stepII_\Gamma(\stepI_\Gamma(c))$ and check that $\stepIII_\Gamma(c_1) = \mathtt{true}$ and $\stepIV_\Gamma(c_1) = \mathtt{true}$.
 
\end{minipage}\end{framed}
\caption{Procedure for checking consistency.}
\label{fig:checkconststar}
\end{figure}

\josephbis{
\begin{example}
\label{ex:voting-consistency}
Continuing Example~\ref{ex:voting-protocol}, when typechecked with appropriate key types, the simplified model of Helios
yields constraint sets containing notably the following two constraints.
\begin{align*}
\{ \; \AENC{\PAIR{0}{r_a}}{\PUBK{k_s}} &\eqC \AENC{\PAIR{1}{r_a}}{\PUBK{k_s}},\\
     \AENC{\PAIR{1}{r_b}}{\PUBK{k_s}} &\eqC \AENC{\PAIR{0}{r_b}}{\PUBK{k_s}} \; \}
\end{align*}

For simplicity, consider the set $c$ containing only these two constraints,
together with a typing environment $\Gamma$ where $r_a$ and $r_b$ are respectively given types
$\noncetypelab{\S}{1}{r_a}$ and $\noncetypelab{\S}{1}{r_b}$, and $k_s$ is given type $\skey{\S}{T}$ for some $T$.

The procedure $\checkconststar(\{(c,\Gamma)\})$ can detect that the
constraint $c$ is consistent and returns $\mathtt{true}$.
Indeed, as $c$ does not contain variables, $\stepI_{\Gamma}(c)$ simply returns $(c,\Gamma)$.
Since $c$ only contains messages encrypted with secret keys, $\stepII_{\Gamma}(c)$ also leaves $c$ unmodified.
$\stepIII_{\Gamma}(c)$ then returns $\mathtt{true}$,
since the messages appearing in $c$ are messages asymmetrically encrypted with secret keys, which contain a secret nonce
($r_a$ or $r_b$) directly under pairs.
Finally $\stepIV_{\Gamma}(c)$ trivially returns $\mathtt{true}$, as the messages
$\AENC{\PAIR{0}{r_a}}{\PUBK{k_s}}$ and $\AENC{\PAIR{1}{r_b}}{\PUBK{k_s}}$ cannot be unified,
as well as the messages
$\AENC{\PAIR{1}{r_a}}{\PUBK{k_s}}$ and $\AENC{\PAIR{0}{r_b}}{\PUBK{k_s}}$.

Consider now the following set $c'$, where encryption has not been randomised:
\begin{align*}
c'=\{ \; \AENC{0}{\PUBK{k_s}} &\eqC \AENC{1}{\PUBK{k_s}},\\
     \AENC{1}{\PUBK{k_s}} &\eqC \AENC{0}{\PUBK{k_s}} \; \}
\end{align*}

The procedure $\checkconststar(\{(c',\Gamma)\})$ returns $\mathtt{false}$.
Indeed, contrary to the case of $c$, $\stepIII_{\Gamma}(c')$ fails,
as the encrypted message do not contain a secret nonce.
Actually, the corresponding frames are indeed not statically
equivalent since the adversary can reconstruct the encryption of $0$
and $1$ with the key $\PUBK{k_s}$ (in his initial knowledge), and
check for equality.
\end{example}
}


For constraint sets without infinite nonce types, $\checkconststar$ entails consistency.
\begin{theorem}
\label{thm:soundness-proc-consistency}
Let $C$ be a set of constraints such that
\[\forall (c, \Gamma)\in C.\; \forall l, l', m, p.\; 
\Gamma(x)
\neq\LRTnewnew{\noncetypelab{l}{\infty}{m}}{\noncetypelab{l'}{\infty}{p}}.\]
If $\checkconststar(C) = \mathtt{true}$,
then $C$ is consistent.
\end{theorem}

\josephbis{%
We prove this theorem by showing that, for each of the first two steps of the procedure,
if $\mathsf{step}i_{\Gamma}(c)$ is consistent in $\Gamma$, then $c$ is consistent in $\Gamma$.
It then suffices to check the consistency of the constraint $\stepII_\Gamma(\stepI_\Gamma(c))$ in $\Gamma$.
Provided that $\stepIII_\Gamma$ holds, we show that this constraint is saturated in the sense that
any message obtained by the attacker by decomposing terms in the constraint already occurs in the constraint;
and the constraint only contains messages which cannot be reconstructed by the attacker from the rest of the constraint.
Using this property, we finally prove that the simple unification tests performed in $\stepIV$ are
sufficient to ensure static equivalence of each side of the constraint for any well-typed instantiation of the variables.
}

As a direct consequence of Theorems~\ref{thm:typing-sound}
and~\ref{thm:soundness-proc-consistency}, we now have a procedure to
prove trace equivalence of processes without replication.

For proving trace equivalence of processes with replication, we need
to check consistency of an infinite family of constraint sets, as
prescribed by Theorem~\ref{thm:typing-sound-replicated}. 
As mentioned earlier, not only the number of constraints is unbounded,
but the size of the type of some
(replicated) variables is also unbounded (\ie of the form
$\bigvee_{j=1}^{n} \LRTnewnew{\noncetypelab{l}{1}{m_j}}{\noncetypelab{l'}{1}{p_j}}$).
We use here two ingredients: we first show that it is sufficient to
apply our procedure to two constraints only.
Second, we show that our procedure applied to variables with replicated types,
\ie nonce types of the form $\noncetypelab{l}{\infty}{n}$ implies consistency
of the corresponding constraints with types of unbounded size.

\subsection{Two constraints suffice}
Consistency of a constraint set $C$ does not guarantee consistency of
$\UCn$. For example, consider
\[C = \{(\{\HASH{m} \eqC\HASH{p}\}, [m:\noncetypelab{\S}{\infty}{m}, p:\noncetypelab{\S}{1}{p}])\}\]
which can be obtained when typing
\[\begin{array}{l}
\NEWnew{m}{\noncetypelab{\S}{\infty}{m}}.\;\NEWnew{p}{\noncetypelab{\S}{1}{p}}.\;\OUT{\HASH{m}} \eqC\\
\NEWnew{m}{\noncetypelab{\S}{\infty}{m}}.\;\NEWnew{p}{\noncetypelab{\S}{1}{p}}.\;\OUT{\HASH{p}}.
\end{array}\]
$C$ is consistent: since $m$, $p$ are secret, the attacker cannot distinguish between their hashes.
However $\UCn$ contains (together with some environment):
\[\{\HASH{m_1} \eqC\HASH{p}, \HASH{m_2} \eqC\HASH{p}, \dots ,\HASH{m_n} \eqC\HASH{p}\}\]
which is not, since the attacker can notice that the value on the right is always the same, while the value on the left is not.

Note however that the
inconsistency of $\UCn$ would have been discovered when checking
the consistency of two copies of the constraint set only.
Indeed, $\instConst{C}{1}{n}\UnionCart \instConst{C}{2}{n}$ contains (together with some environment):
\[\{\HASH{m_1} \eqC\HASH{p}, \HASH{m_2} \eqC\HASH{p}\}\]
which is already inconsistent, for the same reason.

Actually, checking consistency (with our procedure) of two constraints
$\instCst{C}{1}{n}$ and $\instCst{C}{2}{n}$ entails consistency of $\UCn$.
Note that this does not mean that consistency of $\instCst{C}{1}{n}$
and $\instCst{C}{2}{n}$ implies consistency of $\UCn$. Instead, our
procedure ensures a stronger property, for which two constraints suffice.
\begin{theorem}
\label{thm:consistency-two-to-n}
Let $C$ and $C'$ be two constraint sets, which do not contain any common variables.
For all $n\in\mathbb{N}$, 
\[\begin{array}{l}
\checkconststar(\instCst{C}{1}{n} \UnionCart \instCst{C}{2}{n} \UnionCart \instCst{C'}{1}{n}) = \mathtt{true}\;\Rightarrow\\
\quad\checkconststar((\UnionCart_{1\leq i \leq n}\instCst{C}{i}{n})\UnionCart\instCst{C'}{1}{n}) = \mathtt{true}.
\end{array}\]
\end{theorem}

\josephbis{%
It is rather easy to show that if
\[\checkconststar(\instCst{C}{1}{n} \UnionCart \instCst{C}{2}{n} \UnionCart \instCst{C'}{1}{n}) = \mathtt{true},\]
then the first three steps of the procedure $\checkconststar$ can be
successfully applied to each element of
$(\UnionCart_{1\leq i \leq n}\instCst{C}{i}{n})\UnionCart\instCst{C'}{1}{n}$.
However the case of the fourth step is more intricate.
When applying the procedure $\checkconststar$ to an element of
the constraint set
$(\UnionCart_{1\leq i \leq n}\instCst{C}{i}{n})\UnionCart\instCst{C'}{1}{n}$,
if $\stepIV$ fails,
then the constraint contains an inconsistency, \ie 
elements $M\eqC M'$ and $N\eqC N'$ for which the unification condition from $\stepIV$ does not hold.
Intuitively, the property holds by contraposition thanks to the fact that
a similar inconsistency, up to reindexing the nonces and variables, can then already be found
when considering only the first two constraint sets, \ie in
$\instCst{C}{1}{n} \UnionCart \instCst{C}{2}{n} \UnionCart \instCst{C'}{1}{n}$.
The actual proof requires a careful examination of the structure of the constraint set  
$(\UnionCart_{1\leq i \leq n}\instCst{C}{i}{n})\UnionCart\instCst{C'}{1}{n}$,
to establish this reindexing.
}

\subsection{Reducing the size of types}
The
procedure $\checkconststar$ applied to replicated types implies consistency of
corresponding constraints with unbounded types.
\begin{theorem}
\label{lem:consistency-stars-sound}
Let $C$ be a constraint set.
Then for all $i$,
\[\begin{array}{l}
\checkconststar(\instCstr{C}{i}) = \mathtt{true} \;\Rightarrow\;\\
\quad\forall n\geq 1.\checkconststar(\instCst{C}{i}{n})= \mathtt{true}
\end{array}\]
\end{theorem}

\josephbis{%
Again here, it is rather easy to show that if $\checkconststar(\instCstr{C}{i}) = \mathtt{true}$
then the first three steps of the procedure $\checkconststar$ can successfully be applied to each
element of $\instCst{C}{i}{n}$.
The case of $\stepIV$ is more involved.
The property holds thanks to the condition on the most general unifier expressed in $\stepIV$.
Intuitively, this condition is written in such a way that if, when applying $\stepIV$ to an element of
$\instCst{C}{i}{n}$, two messages can be unified, then the corresponding messages (with replicated types)
in $\instCstr{C}{i}$ can be unified with a most general unifier $\mu$ satisfying the condition.
The proof uses this idea to show that if $\stepIV$ succeeds on all elements of
$\instCstr{C}{i}$, then it also succeeds on the elements of $\instCst{C}{i}{n}$.
}

\subsection{Checking the consistency of the infinite constraint}

Theorems~\ref{thm:typing-sound-replicated},
\ref{thm:consistency-two-to-n}, and~\ref{lem:consistency-stars-sound}
provide a sound procedure for checking trace equivalence of processes
with and without replication.
\begin{theorem}
Let $C$, and $C'$ be two constraint sets without any common variable.
\[\begin{array}{l}
\checkconststar(\instCstr{C}{1}\UnionCart\instCstr{C}{2}\UnionCart\instCstr{C'}{1}) = \mathtt{true}\;\Rightarrow\\
\forall n. \;\instCst{C'}{1}{n}\UnionCart(\UCn)\text{ is consistent.}
\end{array}\]
\end{theorem}

All detailed proofs are available online~\cite{oursite}.

%% file: experiments.tex

\section{Experimental results}
We have implemented a prototype type-checker \TypeEq~and applied it on various
examples briefly described below.

{\bf Symmetric key protocols.} For the sake of comparison, we consider 5 symmetric key
protocols taken from the benchmark of~\cite{SAT-equiv}, and described in~\cite{clark1997}: Denning-Sacco, Wide Mouth Frog,
Needham-Schroeder, Yahalom-Lowe, 
and Otway-Rees. All these
protocols aim at exchanging a key $k$. We prove strong secrecy of the
key, as defined in~\cite{Abadi2000}, i.e., $P(k_1)\approx_t P(k_2)$
where $k_1$ and $k_2$ are public names. Intuitively, an attacker
should not be able to tell which key is used even if he knows the two possible values in
advance. For some of the protocols, we truncated the last step, when
it consists in using the exchanged key for encryption, since our
framework currently covers only encryption with long-term (fixed)
keys.

{\bf Asymmetric key protocols.} In addition to the symmetric key
protocols, we consider the well-known Needham-Schroeder-Lowe (NSL)
protocol~\cite{lowe96breaking} and we again prove strong secrecy of the nonce
sent by the receiver (Bob).

{\bf Helios.} We model the Helios protocol for two honest voters and
infinitely many dishonest ones, as informally described in
Section~\ref{sec:overview}. The corresponding process includes a non
trivial else branch, used to express the weeding phase~\cite{Helios-CSF11}, where
dishonest ballots equal to some honest one are discarded. As
emphasised in Section~\ref{sec:overview}, Helios is secure only if
honest voters vote at most once. Therefore the protocol includes non
replicated processes (for voters) as well as a replicated process (to
handle dishonest ballots). 

All our experiments have been run on a single Intel Xeon E5-2687Wv3 3.10GHz core, with 378GB of RAM (shared with the 19 other cores).
All corresponding files can be found online at~\cite{oursite}.
\subsection{Bounded number of sessions}

We first compare our tool with  tools designed for a bounded number
of sessions: SPEC~\cite{SPEC},
APTE (and its APTE-POR variant)~\cite{APTE,APTE-por}, Akiss~\cite{akiss}, or
SAT-Equiv~\cite{SAT-equiv}.
The protocol models may  slightly differ due to the
subtleties of each tool. For example, several of these tools require
\emph{simple} processes where each sub-process emits on a distinct
channel. We do not need such an assumption. In addition, SAT-Equiv only
covers symmetric encryption and therefore could not be applied to
Helios or NSL.
SAT-Equiv further assumes protocols to be well-typed,
which sometimes requires to tag protocols. Since we consider only
untagged versions (following the original description of each
protocol), SAT-Equiv failed to prove the Otway-Rees protocol.
Moreover, Helios involves non-trivial else branches, which are
only supported by APTE.

The number of sessions we consider denotes the number of processes in parallel in each scenario.
For symmetric key protocols, we start with a
simple scenario with only two honest participants A, B and a honest server~S (3 sessions).
We consider increasingly more complex scenarios (6, 7, 10, 12, and 14 sessions) featuring a dishonest agent C.
In the complete scenario (14 sessions) each agent among A, B (and C) runs the protocol once as the initiator,
and once as the responder with each other agent (A, B, C).
In the case of NSL, we similarly consider a scenario with two honest agents A, B running the protocol once (2 sessions),
and two scenarios with an additional dishonest agent C, up to the complete scenario (8 sessions) where each agent runs NSL
once as initiator, once as responder, with each agent.
For Helios, we consider 2 honest voters, and one dishonest voter only, as well as a ballot box.
The corresponding results are reported in Figure~\ref{fig:exp-bounded}.
We write TO for Time Out (12 hours), MO for Memory Out (more than 64 GB of RAM),
SO for Stack Overflow, BUG in the case of APTE, when the proof failed due to bugs in the tool, and x when the
tool could not handle the protocol for the reasons discussed
previously.
In all cases, our tool is almost instantaneous and outperforms by orders of magnitude the competitors.

\begin{figure}
\begin{center}
\resizebox{8.3cm}{!}{
\begin{tabular}{|c|c|c|c|c|c|c|c|}%
\hline
\multicolumn{2}{|c|}{Protocols (\# sessions)} &\!\!Akiss\!\!& APTE &\!\!APTE-POR\!\!& Spec & Sat-Eq & \TypeEq\\
\hline
                               & 3  & 0.08s & 0.32s & 0.02s & 9s   & 0.09s & 0.002s\\
                               & 6  & 3.9s  & TO    & 1.6s  & 191m & 0.3s  & 0.003s\\ 
Denning -                      & 7  & 29s   &       & 3.6s  & TO   & 0.8s  & 0.004s\\ 
Sacco                          & 10 & SO    &       & 12m   &      & 1.8s  & 0.004s\\
                               & 12 &       &       & TO    &      & 3.4s  & 0.005s\\
                               & 14 &       &       &       &      & 5s    & 0.006s\\
\hline
                & 3  & 0.03s & 0.05s & 0.009s & 8s  & 0.06s & 0.002s \\
                & 6  & 0.4s  & 28m   & 0.4s   & 52m & 0.2s  & 0.003s \\ 
 Wide Mouth     & 7  & 1.4s  & TO    & 1.9s   & MO  & 2.3s  & 0.003s \\ 
 Frog           & 10 & 46s   &       & 5m31s  &     & 5s    & 0.004s \\
                & 12 & 71m   &       & TO     &     & 1m    & 0.005s \\
                & 14 & TO    &       &        &     & 4m20s & 0.006s \\

\hline
                  & 3  & 0.1s & 0.4s & 0.02s & 52s & 0.5s  & 0.003s \\
                  & 6  & 20s  & TO   & 4s    & MO  & 4s    & 0.003s \\ 
Needham -         & 7  & 2m   &      & 8m    &     & 36s   & 0.003s \\ 
Schroeder         & 10 & SO   &      & TO    &     & 1m50s & 0.005s \\
                  & 12 &      &      &       &     & 4m47s & 0.005s \\
                  & 14 &      &      &       &     & 11m   & 0.007s \\
\hline
              & 3  & 0.16s & 3.6s & 0.03s & 6s   & 1.4s & 0.003s \\
              & 6  & 33s   & TO   & 44s   & 132m & 1m   & 0.004s \\ 
Yahalom -     & 7  & 11m   &      & 36m   & MO   & 17m  & 0.004s \\ 
Lowe          & 10 & SO    &      & TO    &      & 63m  & 0.009s \\
              & 12 &       &      &       &      & TO   & 0.04s  \\
              & 14 &       &      &       &      &      & 0.05s  \\
\hline
\multirow{6}{*}{Otway-Rees}
              & 3  & 2m12s & BUG & 1.7s & 27m & x & 0.004s \\
              & 6  & TO    &     & SO   & MO  &   & 0.011s \\ 
              & 7  &       &     &      &     &   & 0.012s \\ 
              & 10 &       &     &      &     &   & 0.02s  \\
              & 12 &       &     &      &     &   & 0.03s  \\
              & 14 &       &     &      &     &   & 0.1s   \\
\hline
Needham-          & 2 & 0.1s & 4s  & 0.06s & 31s & x & 0.003s \\
Schroeder-        & 4 & 2m   & BUG & BUG   & MO  &   & 0.003s \\ 
Lowe              & 8 & TO   &     &       &     &   & 0.007s \\ 
\hline
\multirow{1}{*}{Helios} & 3 & x & TO & BUG & x & x & 0.002s \\
\hline
\end{tabular}}
\end{center}
\caption{Experimental results for the  bounded case}
\label{fig:exp-bounded}
\end{figure}

\subsection{Unbounded numbers of sessions}
We then compare our type-checker with ProVerif~\cite{proverif-equiv},
for an unbounded number of sessions, on three examples: Helios,
Denning-Sacco, and NSL. As expected, ProVerif cannot prove Helios
secure since it cannot express that voters vote only once. 
This may sound surprising, since proofs of Helios in ProVerif already
exist (e.g. \cite{Helios-CSF11,Myrto3Voters}). Interestingly, these
models actually implicitly assume a reliable channel between honest
voters and the voting server: whenever a voter votes, she first sends
her vote to the voting server on a secure channel, before letting the
attacker see it.
This model prevents an attacker from reading and blocking a message,
while this
can be easily done in practice (by breaking the connection). 
We also failed to prove (automatically) Helios in Tamarin~\cite{tamarin-equiv}.
The reason is that the weeding procedure makes Tamarin enter a loop where
it cannot detect that, as soon as a ballot is not weed, it has been forged by
the adversary.

For the sake of comparison, we run both tools (ProVerif and \TypeEq)
on a symmetric protocol (Denning-Sacco) and an  asymmetric protocol
(Needham-Schroeder-Lowe). The execution times are very similar.
The corresponding results are reported in Figure~\ref{fig:exp-unbounded}.
\begin{figure}
\begin{tabular}{|c|c|c|}
\hline
Protocols & ProVerif & \TypeEq\\
\hline
Helios                 & x     & 0.003s \\
Denning-Sacco          & 0.05s & 0.05s\\
Needham-Schroeder-Lowe & 0.08s & 0.09s\\
\hline
\end{tabular}
\caption{Experimental results for unbounded numbers of sessions}
\label{fig:exp-unbounded}
\end{figure}

%% file: conclusion.tex
\section{Conclusion}

We presented a novel type system for verifying trace equivalence in security protocols. It can be applied to various protocols, with support for else branches,  standard cryptographic primitives, as well as  a bounded and an unbounded number of sessions. 
We believe that our prototype implementation demonstrates that this approach is promising and opens the way to the development of an efficient technique for proving equivalence properties in even larger classes of protocols.

Several interesting problems remain to be studied. For example, a limitation of ProVerif is that it cannot properly handle global states. We plan to explore this case by enriching our types to express the fact that an event is ``consumed''. 
Also, for the moment, our type system only applies to protocols $P,Q$ that have the same structure. One advantage  of a type system is its modularity: it is relatively easy to add a few rules without redoing the whole proof. We plan to add rules to cover protocols with different structures (e.g. when branches are swapped). 
Another direction is the treatment of primitives with algebraic properties (e.g. Exclusive Or, or homomorphic encryption). It seems possible to extend the type system and discharge the difficulty to the consistency of the constraints, which seems easier to handle (since this captures the static case). 
Finally, our type system is sound w.r.t. equivalence in a symbolic model. An interesting question is whether it also entails computational indistinguishability. Again, we expect that an advantage of our type system is the possibility to discharge most of the difficulty to the constraints.

%% file: proofs/typing2.tex

\section{Typing rules and definitions}

We give on Figures~
\ref{fig-proof:processtypingrules} and \ref{fig-proof:wellformednesstyping} 
a complete version of our typing rules for processes, as well as the formal definition of the well-formedness judgement for typing environments.

\begin{figure}
\begin{framed}
\[
\infer[PZero]
  {\tewf{\Gamma}\\ \branch{\Gamma} = \{\Gamma\}}
  {\teqPnew{\Gamma}{\ZERO}{\ZERO}{(\emptyset,\Gamma)}}
\]
\[
\infer[POut]
  {\teqPnew{\Gamma}{P}{Q}{C} \\
   \teqTcnew{\Gamma}{M}{N}{\L}{c}}
  {\teqPnew{\Gamma}{\OUT{M}.P}{\OUT{N}.Q}{C \UnionAll c}}
\hspace{10pt}
\infer[PIn]
  {\teqPnew{\Gamma,x:\L}{P}{Q}{C}}
  {\teqPnew{\Gamma}{\IN{x}.P}{\IN{x}.Q}{C}}
\]
\[
\infer[PNew]
  {
   \teqPnew{\Gamma,{n:\noncetypelab{l}{\oneorinf}{n}}}{P}{Q}{C}}
  {\teqPnew{\Gamma}{\NEWnew{n}{\noncetypelab{l}{\oneorinf}{n}}.P}{\NEWnew{n}{\noncetypelab{l}{\oneorinf}{n}}.Q}{C}}
\]
\[
\infer[PPar]
  {\teqPnew{\Gamma}{P}{Q}{C} \\
   \teqPnew{\Gamma}{P'}{Q'}{C'}}
  {\teqPnew{\Gamma}{P \PAR P'}{Q \PAR Q'}{C \UnionCart C'}}
\hspace{10pt}
\infer[POr]
  {\teqPnew{\Gamma,x:T}{P}{Q}{C}\\
  \teqPnew{\Gamma,x:T'}{P}{Q}{C'}}
  {\teqPnew{\Gamma, x:T\orT T'}{P}{Q}{C \cup C'}}
\]
\[
\infer[PLet]
  {\tDestnew{\Gamma}{d(y)}{T}\\
   \teqPnew{\Gamma,x:T}{P}{Q}{C}\\
   \teqPnew{\Gamma}{P'}{Q'}{C'}}
  {\teqPnew{\Gamma}{\LET{x}{d(y)}{P}{P'}}{\LET{x}{d(y)}{Q}{Q'}}{C \cup C'}}
\]
\[
\infer[PLetLR]
   {\Gamma(y) = \LRTnewnew{\noncetypelab{l}{\oneorinf}{n}}{\noncetypelab{l'}{\oneorinf}{m}}\\
   \teqPnew{\Gamma}{P'}{Q'}{C'}}
  {\teqPnew{\Gamma}{\LET{x}{d(y)}{P}{P'}}{\LET{x}{d(y)}{Q}{Q'}}{C'}}
\]
\[
\infer[PIfL]
  {\teqPnew{\Gamma}{P}{Q}{C}\\
  \teqPnew{\Gamma}{P'}{Q'}{C'}\\
  \teqTcnew{\Gamma}{M}{N}{\L}{c}\\
  \teqTcnew{\Gamma}{M'}{N'}{\L}{c'}}
  {\teqPnew{\Gamma}{\ITE{M}{M'}{P}{P'}}{\ITE{N}{N'}{Q}{Q'}}{\left( C \cup C'\right) \UnionAll (c \cup c')}}
\]

\[
\infer[PIfLR]
  {\teqTcnew{\Gamma}{M_1}{N_1}{\LRTnewnew{\noncetypelab{l}{1}{m}}{\noncetypelab{l'}{1}{n}}}{\emptyset}\\
  \teqTcnew{\Gamma}{M_2}{N_2}{\LRTnewnew{\noncetypelab{l''}{1}{m'}}{\noncetypelab{l'''}{1}{n'}}}{\emptyset}\\\\
  b = (\noncetypelab{l}{1}{m} \overset{?}{=} \noncetypelab{l''}{1}{m'}) \\ 
  b' = (\noncetypelab{l'}{1}{n} \overset{?}{=} \noncetypelab{l'''}{1}{n'}) \\ 
  \teqPnew{\Gamma}{P_b}{Q_{b'}}{C}}
  {\teqPnew{\Gamma}{\ITE{M_1}{M_2}{P_\top}{P_\bot}}{\ITE{N_1}{N_2}{Q_\top}{Q_\bot}}{C}}
\]
\[
\infer[PIfS]
  {\teqPnew{\Gamma}{P'}{Q'}{C'}\\
  \teqTcnew{\Gamma}{M}{N}{\L}{c}\\
  \teqTcnew{\Gamma}{M'}{N'}{\S}{c'}}
  {\teqPnew{\Gamma}{\ITE{M}{M'}{P}{P'}}{\ITE{N}{N'}{Q}{Q'}}{C'}}
\]

\[
\infer[PIfLR*]
  {\teqTcnew{\Gamma}{M_1}{N_1}{\LRTnewnew{\noncetypelab{l}{\infty}{m}}{\noncetypelab{l'}{\infty}{n}}}{\emptyset}\\
  \teqTcnew{\Gamma}{M_2}{N_2}{\LRTnewnew{\noncetypelab{l}{\infty}{m}}{\noncetypelab{l'}{\infty}{n}}}{\emptyset}\\\\
  \teqPnew{\Gamma}{P}{Q}{C}\\
  \teqPnew{\Gamma}{P'}{Q'}{C'}}
  {\teqPnew{\Gamma}{\ITE{M_1}{M_2}{P}{P'}}{\ITE{N_1}{N_2}{Q}{Q'}}{C\cup C'}}
\]

\[
\infer[PIfP]
  {\teqPnew{\Gamma}{P}{Q}{C}\\
  \teqPnew{\Gamma}{P'}{Q'}{C'}\\\\
  \teqTcnew{\Gamma}{M}{N}{\L}{c}\\
  \teqTcnew{\Gamma}{t}{t}{\L}{c'} \\
  t \in \K \cup \N \cup \CST}
  {\teqPnew{\Gamma}{\ITE{M}{t}{P}{P'}}{\ITE{N}{t}{Q}{Q'}}{C \cup C'}}
\]

\[
\infer[PIfI]
  {\teqPnew{\Gamma}{P'}{Q'}{C'}\\
  \teqTcnew{\Gamma}{M}{N}{T*T'}{c}\\
  \teqTcnew{\Gamma}{M'}{N'}{\LRTnewnew{\noncetypelab{l}{a}{m}}{\noncetypelab{l'}{a}{n}}}{\emptyset}}
  {\teqPnew{\Gamma}{\ITE{M}{M'}{P}{P'}}{\ITE{N}{N'}{Q}{Q'}}{C'}}
\]

\[
\infer[PIfLR'*]
  {\teqTcnew{\Gamma}{M_1}{N_1}{\LRTnewnew{\noncetypelab{l}{a}{m}}{\noncetypelab{l'}{a}{n}}}{\emptyset}\\
   \teqTcnew{\Gamma}{M_2}{N_2}{\LRTnewnew{\noncetypelab{l''}{a'}{m'}}{\noncetypelab{l'''}{a'}{n'}}}{\emptyset}\\\\
  \noncetypelab{l}{a}{m}\neq \noncetypelab{l''}{a'}{m'}, \noncetypelab{l'}{a}{n}\neq \noncetypelab{l'''}{a'}{n'}\\
  \teqPnew{\Gamma}{P'}{Q'}{C}}
  {\teqPnew{\Gamma}{\ITE{M_1}{M_2}{P}{P'}}{\ITE{N_1}{N_2}{Q}{Q'}}{C}}
\]

\end{framed}
\caption{Rules for processes}
\label{fig-proof:processtypingrules}
\end{figure}

\begin{figure}
\begin{framed}
\[
  \infer[GNil]
  {}
  {\tewf{[]}}
\hspace{10pt}
  \infer[GNonce]
  {\tewf{\Gamma}\\ n\in\BN}
  {\tewf{\Gamma,n:\noncetypelab{l}{\oneorinf}{n}}}
\]
\[
  \infer[GVar]
  {\tewf{\Gamma}\\
  \keys{T} \subseteq \dom{\Gamma}}
  {\tewf{\Gamma,x:T}}
\]
\[
  \infer[GKey]
  {\tewf{\Gamma}\\
  \keys{T} \subseteq \dom{\Gamma}}
  {\tewf{\Gamma,k:\skeyB{l}{T}}}
\]
\end{framed}
\caption{Well-formedness of the typing environment}

\label{fig-proof:wellformednesstyping}
\end{figure}

In this section, we also provide additional definitions (or more precise versions of previous definitions)
regarding constraints, and especially their consistency,
that the proofs require.

\begin{definition}[Constraint]
A constraint is defined as a couple of messages, separated by the symbol $\eqC$:
\[u \eqC v\]
\end{definition}

We will consider sets of constraints, which we usually denote $c$.
We will also consider couples $(c, \Gamma)$ composed of such a set, and a typing environment $\Gamma$.
Finally we will denote sets of such tuples $C$, and call them \emph{constraint sets}.

\begin{definition}[Compatible environments]
We say that two typing environments $\Gamma$, $\Gamma'$ are compatible if they are equal on the intersection of their domains, \ie if
\[\forall x \in \dom{\Gamma} \cap \dom{\Gamma'}.\; \Gamma(x) = \Gamma'(x)\]
\end{definition}

\begin{definition}[Union of environments]
Let $\Gamma$, $\Gamma'$ be two compatible environments.
Their union $\Gamma \cup \Gamma'$ is defined by
\begin{itemize}
\item $\dom{\Gamma \cup \Gamma'} = \dom{\Gamma}\cup\dom{\Gamma'}$
\item $\forall x \in \dom{\Gamma}.\quad (\Gamma \cup \Gamma')(x) = \Gamma(x)$
\item $\forall x \in \dom{\Gamma'}.\quad (\Gamma \cup \Gamma')(x) = \Gamma'(x)$
\end{itemize}
Note that this function is well defined since $\Gamma$ and $\Gamma'$ are assumed to be compatible.
\end{definition}

\begin{definition}[Operations on constraint sets]
We define two operations on constraints.
\begin{itemize}
\item the product union of constraint sets:
\begin{align*}
  &C \UnionCart C' := \{  (c \cup c', \Gamma \cup \Gamma') \;|\;\\
   &\quad (c,\Gamma) \in C \,\wedge\, (c',\Gamma') \in C'
   \,\wedge\, \Gamma, \Gamma' \text{ are compatible} \}
\end{align*}

  \item the addition of a set of constraints $c'$ to all elements of a constraint set $C$:
\begin{align*}
C \UnionAll c' := &C \UnionCart \{(c',\emptyset)\} \\
 =& \left\{ (c \cup c',\Gamma) \;|\; (c,\Gamma) \in C \right\}
\end{align*}
\end{itemize}
\end{definition}

\begin{definition}
For any typing environment $\Gamma$, we denote by $\onlyvar{\Gamma}$ its restriction to variables,
by $\novar{\Gamma}$ its restriction to names and keys,
and by $\EG$ the set of the names it contains, \ie $\N\cap\dom{\Gamma}$.
\end{definition}

\begin{definition}[Well-typed substitutions]
Let $\Gamma$ 
be a typing environment, $\theta$, $\theta'$ two substitutions, and $c$ a set of constraints.
We say that $\theta$, $\theta'$ are well-typed in $\Gamma$, 
and write $\wtcDE{\theta}{\theta'}{\Gamma}{c}$, if they are ground and 
\begin{itemize}
\item $\dom{\theta} = \dom{\theta'} = \dom{\onlyvar{\Gamma}}$,
\item and
\[\forall x \in \dom{\onlyvar{\Gamma}}, \; \teqTc{\novar{\Gamma}}{\Delta}{\E}{\E'}{\theta(x)}{\theta'(x)}{\Gamma(x)}{c_x}\]
for some $c_x$ such that $c = \bigcup_{x\in\dom{\onlyvar{\Gamma}}} c_x$.
\end{itemize}
\end{definition}

\begin{definition}[$\L$ substitutions]
Let $\Gamma$ be an environment, $\phi$, $\phi'$ two substitutions and $c$ a set of constraints.
We say that $\phi$, $\phi'$ have type $\L$ in $\Gamma$ with constraint $c$,
and write $\teqTc{\Gamma}{}{}{}{\phi}{\phi'}{\L}{c}$ if
\begin{itemize}
\item $\dom{\phi} = \dom{\phi'}$;
\item for all $x\in\dom{\phi}$ there exists $c_x$ such that $\teqTc{\Gamma}{}{}{}{\phi(x)}{\phi'(x)}{\L}{c_x}$ 
and $c = \bigcup_{x\in\dom{\phi}} c_x$.
\end{itemize}
\end{definition}

\begin{definition}[Frames associated to a set of constraints]
If $c$ is a set of constraints, let $\phiL{c}$ and $\phiR{c}$ be the frames composed of the terms respectively on the left and on the right of the $\eqC$ symbol in the constraints of $c$ (in the same order).
\end{definition}

\begin{definition}[Instantiation of constraints]
If $c$ is a set of constraints, and $\sigma$, $\sigma'$ are two substitutions, let $\inst{c}{\sigma}{\sigma'}$ be the instantiation of $c$ by $\sigma$ on the left and $\sigma'$ on the right, \ie
\[\inst{c}{\sigma}{\sigma'} = \{M\sigma \eqC N\sigma' \;|\; M\eqC N \in c\}.\]

Similarly we write for a constraint set $C$
\[\inst{C}{\sigma}{\sigma'} = \{(\inst{c}{\sigma}{\sigma'}, \Gamma) \;|\; (c,\Gamma)\in C\}.\]
\end{definition}

\begin{definition}[Frames associated to environments]
If $\Gamma$ is a typing environment,
we denote $\phiEE$ the frame containing all the keys $k$ such that $\Gamma(k) = \skey{\L}{T}$ for
some $T$, all the public keys $\PUBK{k}$ and $\VK{k}$ for $k\in\dom{\Gamma}$, and all the nonces $n$ such that $\Gamma(n) = \noncetypelab{\L}{a}{n}$ (for $a\in\{\infty,1\}$).
\end{definition}

\begin{definition}[Branches of a type]
If $T$ is a type, we write $\branch{T}$ the set of all types $T'$ such that $T'$ is not a union type, and either
\begin{itemize}
\item $T = T'$; 
\item or there exist types $T_1$,\dots,$T_k$,$T'_1$,\dots,$T'_{k'}$ such that
\[T = T_1 \orT \dots \orT T_k \orT T \orT T'_1 \orT\dots\orT T'_{k'}\]
\end{itemize}
\end{definition}

\begin{definition}[Branches of an environment]
For a typing environment $\Gamma$, we write $\branch{\Gamma}$ the sets of all environments $\Gamma'$ such that
\begin{itemize}
\item $\dom{\Gamma'} = \dom{\Gamma}$
\item $\forall x \in \dom{\Gamma}.\quad \Gamma'(x) \in \branch{\Gamma(x)}$.
\end{itemize}
\end{definition}

\begin{definition}[Consistency]
We say that $c$ is consistent in a typing environment $\Gamma$, 
if for any subsets $c' \subseteq c$ and $\Gamma' \subseteq \Gamma$ such that $\novar{\Gamma'} = \novar{\Gamma}$ and $\var{c'} \subseteq \dom{\Gamma'}$,
for any ground substitutions $\sigma$, $\sigma'$ well-typed in $\Gamma'$, 
the frames $\NEWN{\EG}.(\phiEE \cup \phiL{c'}\sigma)$ and $\NEWN{\EG}.(\phiEE \cup \phiR{c'}\sigma')$ are statically equivalent.

We say that $(c,\Gamma)$ is consistent 
if $c$ is consistent in $\Gamma$. 

We say that a constraint set $C$ is consistent 
if each element $(c, \Gamma) \in C$ is consistent. 
\end{definition}

%% file: proofs/newproofs.tex

\section{Proofs}

In this section, we provide the detailed proofs to all of our theorems.

Unless specified otherwise, the environments $\Gamma$ considered in the lemmas are implicitly assumed to be well-formed.

\subsection{General results and soundness}

In this subsection, we prove soundness for non replicated processes, as well as several results
regarding the type system that this proof uses.

\begin{lemma}[Subtyping properties]
\label{lem-proof:subtyping}
The following properties of subtyping hold:
\begin{enumerate}
\item $\forall T.\quad \H \subtyp T \implies T = \H$
\item $\forall T.\quad \L \subtyp T \implies T=\L \;\vee\; T=\H$
\item $\forall T.\quad \S \subtyp T \implies T=\S \;\vee\; T=\H$
\item $\forall T_1, T_2, T_3.\quad T_1 * T_2 \subtyp T_3 \implies T_3 = \L \;\vee\; T_3 = \H \;\vee\; T_3 = \S \;\vee\; (\exists T_4, T_5.\quad T_3 = T_4*T_5)$ \ie $T_3$ is $\L$, $\H$, $\S$ or a pair type.
\item $\forall T, T_1, T_2.\quad T \subtyp T_1*T_2 \implies (\exists T_1', T_2'.\quad T = T_1'*T_2' \;\wedge\; T_1' \subtyp T_1 \;\wedge\;T_2' \subtyp T_2$)
\item $\forall T_1, T_2.\quad T_1 * T_2 \subtyp \L \implies T_1 \subtyp \L \;\wedge\; T_2 \subtyp \L$
\item $\forall T_1, T_2.\quad T_1 * T_2 \subtyp \S \implies T_1 \subtyp \S \;\vee\; T_2 \subtyp \S$
\item $\forall T_1, T_2, k.\quad T_1 \subtyp \encT{T_2}{k} \implies (\exists T_3 \subtyp T_2. \quad T_1 = \encT{T_3}{k})$
\item $\forall T_1, T_2, k.\quad T_1 \subtyp \aencT{T_2}{k} \implies (\exists T_3 \subtyp T_2. \quad T_1 = \aencT{T_3}{k})$
\item $\forall T_1, T_2, k.\quad \encT{T_1}{k} \subtyp T_2 \implies T_2 = \H \;\vee\; (\exists T_3. T_1 \subtyp T_3\;\wedge\; T_2 = \encT{T_3}{k})$
\item $\forall T_1, T_2, k.\quad \aencT{T_1}{k} \subtyp T_2 \implies 
T_2 = \H \;\vee\; (\exists T_3. T_1 \subtyp T_3\;\wedge\; T_2 = \aencT{T_3}{k})$
\item $\forall T, m,n,l,l'.\quad T \subtyp \LRTn{l}{a}{m}{l'}{a}{n} \implies T = \LRTn{l}{a}{m}{l'}{a}{n}$
\item $\forall T, m,n,l,l'.\quad \LRTn{l}{a}{m}{l'}{a}{n} \subtyp T \implies T = \H \;\vee\; T = \LRTn{l}{a}{m}{l'}{a}{n}$
\item $\forall T_1, T_2. \quad T_1 \subtyp T_2 \implies \text{ neither $T_1$ nor $T_2$ are union types unless $T_2 = \H$ or $T_1 = T_2$}$.
\item $\forall T, l, T'.\quad T \subtyp \skey{l}{T'} \implies T = \skey{l}{T'}$.
\item $\forall T. \quad T \subtyp \L \implies
 \text{$T$ is a pair type}
 \;\vee\; (\exists T'.\quad T = \skey{\L}{T'})\;\vee\; T = \L$.
\item $\forall T. \quad T \subtyp \S \implies
 \text{$T$ is a pair type} \;\vee\; (\exists T'.\quad T = \skey{\S}{T'})\;\vee\; T = \S$.
\end{enumerate}
\end{lemma}
\begin{proof}
\begin{itemize}
\item Points 1 to 3 are immediate by induction on the subtyping proof (the proofs of the second and third points use the first one in the \STrans case).
\item Point 4 is immediate by induction on the proof of $T_1 * T_2 \subtyp T_3$ (the \STrans case uses the first three points).
\item Point 5 is proved by induction on the proof of $T \subtyp T_1*T_2$ (the \STrans case uses the induction hypothesis).
\item Points 6 and 7 are proved by induction on the proof of $T_1 * T_2 \subtyp \L$ (resp. $\S$; using the previous points in the \STrans case).
\item Point 8 and 9 immediate by induction on the subtyping proof.
\item Points 10 to 13 are proved by induction on the subtyping proof (using the first two points
in the \STrans case).
\item Points 14 and 15 are immediate by induction on the subtyping proof.
\item Points 16 and 17 are proved by induction on the subtyping proof, using points 5, 8, 9 in the \STrans case.
\end{itemize}
\end{proof}

\begin{lemma}[Terms of type $T \orT T'$]
\label{lem-proof:ground-term-ortype}
For all $\Gamma$, $T$, $T'$, for all ground terms $t$, $t'$, for all $c$, if
\[\teqTc{\Gamma}{\Delta}{\E}{\E'}{t}{t'}{T \orT T'}{c}\]
then 
\[\teqTc{\Gamma}{\Delta}{\E}{\E'}{t}{t'}{T}{c} \quad\text{or}\quad \teqTc{\Gamma}{\Delta}{\E}{\E'}{t}{t'}{T'}{c}\]
\end{lemma}
\begin{proof}
We prove this property by induction on the derivation of $\teqTc{\Gamma}{\Delta}{\E}{\E'}{t}{t'}{T \orT T'}{c}$.

The last rule of the derivation cannot be \TNonce, \TNonceL, \TCst, \TPair, 
\TKey, \TPubkey, \TVkey, \TEnc, \TEncH, \TEncL, \TAenc, \TAencH, \TSignH, \TSignL, \THash, \THashL, \THigh, \TLRone, \TLRinf, \TLRVar, \TLRp, or \TLRLp since the type in their conclusion cannot be $T \orT T'$.
It cannot be \TVar since $t$, $t'$ are ground.


In the \TSub case we know that $\teqTc{\Gamma}{\Delta}{\E}{\E'}{t}{t'}{T''}{c}$ (with a shorter derivation) for some $T'' \subtyp T \orT T'$; thus, by Lemma~\ref{lem-proof:subtyping}, $T'' = T \orT T'$, and the claim holds by the induction hypothesis.

Finally in the \TOr case, the premise of the rule directly proves the claim.
\end{proof}

\begin{lemma}[Terms and branch types]
\label{lem-proof:term-branch}
For all $\Gamma$, $T$, $c$, for all ground terms $t$, $t'$, if
	\[\teqTc{\Gamma}{\Delta}{\E}{\E'}{t}{t'}{T}{c}\]
then there exists $T' \in \branch{T}$ such that
\[\teqTc{\Gamma}{\Delta}{\E}{\E'}{t}{t'}{T'}{c}\]
\end{lemma}
\begin{proof}
This property is a corollary of Lemma~\ref{lem-proof:ground-term-ortype}.
We indeed prove it by successively applying this lemma to $\teqTc{\Gamma}{\Delta}{\E}{\E'}{t}{t'}{T}{c}$ until $T$ is not a union type.
\end{proof}

\begin{lemma}[Substitutions type in a branch]
\label{lem-proof:ground-subst-branch}
For all $\Gamma$, $c$, for all ground substitutions $\sigma$, $\sigma'$, if
\[\wtc{\Delta}{\E}{\E'}{\sigma}{\sigma'}{\Gamma}{c}\]
then there exists $\Gamma' \in \branch{\Gamma}$ such that
\[\wtc{\Delta}{\E}{\E'}{\sigma}{\sigma'}{\Gamma'}{c}\]
\end{lemma}
\begin{proof}
This property follows from Lemma~\ref{lem-proof:term-branch}.
Indeed, by definition, $c = \bigcup_{x \in \dom{\Gamma}} c_x$ for some $c_x$ such that for all $x \in \dom{\Gamma} (=\dom{\sigma} = \dom{\sigma'})$,
\[\teqTc{\Gamma}{\Delta}{\E}{\E'}{\sigma(x)}{\sigma'(x)}{\Gamma(x)}{c_x}\]
Hence by applying Lemma~\ref{lem-proof:term-branch} we obtain a type $T_x \in \branch{\Gamma(x)}$ such that 
\[\teqTc{\Gamma}{\Delta}{\E}{\E'}{\sigma(x)}{\sigma'(x)}{T_x}{c_x}\]
Thus if we denote $\Gamma''$ by $\forall x \in \dom{\onlyvar{\Gamma}}.\Gamma''(x) = T_x$,
and $\Gamma' = \novar{\Gamma}\cup\Gamma''$,
we have $\Gamma' \in \branch{\Gamma}$ and
$\wtc{\Delta}{\E}{\E'}{\sigma}{\sigma'}{\Gamma'}{c}$.
\end{proof}

\begin{lemma}[Typing terms in branches]
\label{lem-proof:type-terms-branches}
For all $\Gamma$, $T$, $c$, for all terms $t$, $t'$, for all $\Gamma' \in \branch{\Gamma}$,
if $\teqTc{\Gamma}{\Delta}{\E}{\E'}{t}{t'}{T}{c}$
then
$\teqTc{\Gamma'}{\Delta}{\E}{\E'}{t}{t'}{T}{c}$.

Corollary: in that case, there exists $T' \in\branch{T}$ such that $\teqTc{\Gamma'}{\Delta}{\E}{\E'}{t}{t'}{T'}{c}$.
\end{lemma}
\begin{proof}
We prove this property by induction on the derivation of $\teqTc{\Gamma}{\Delta}{\E}{\E'}{t}{t'}{T}{c}$.
In most cases for the last rule applied, $\Gamma(x)$ is not directly involved in the premises, for any variable $x$.
Rather, $\Gamma$ appears only in other typing judgements, or is used in $\Gamma(k)$ or $\Gamma(n)$ for some key $k$ or
nonce $n$, and keys or nonces cannot have union types.
Hence, since the typing rules for terms do 
not change $\Gamma$, the claim directly follows from the induction hypothesis.
For instance in the \TPair case, we have $t = \PAIR{t_1}{t_2}$, $t' = \PAIR{t_1'}{t_2'}$, $T = T_1 * T_2$, $c = c_1 \cup c_2$, 
$\teqTc{\Gamma}{\Delta}{\E}{\E'}{t_1}{t_1'}{T_1}{c_1}$, and $\teqTc{\Gamma}{\Delta}{\E}{\E'}{t_2}{t_2'}{T_2}{c_2}$.
Thus by the induction hypothesis,
$\teqTc{\Gamma'}{\Delta}{\E}{\E'}{t_1}{t_1'}{T_1}{c_1}$, and $\teqTc{\Gamma'}{\Delta}{\E}{\E'}{t_2}{t_2'}{T_2}{c_2}$; and therefore
by rule \TPair, $\teqTc{\Gamma'}{\Delta}{\E}{\E'}{t}{t'}{T}{c}$.
The cases of rules 
\TEnc, \TEncH, \TEncL, \TAenc, \TAencH, \TAencL, \THashL, \TSignH, \TSignL, \TLRp, \TLRLp, \TLRVar, \TSub, \TOr are similar.

The cases of rules \TNonce, \TNonceL, TCst, \TKey, \TPubkey, \TVkey, \THash, \THigh, \TLRone, and \TLRinf are immediate since these rules use neither $\Gamma$ nor another typing judgement in their premise.

Finally, in the \TVar case, $t = t' = x$ for some variable $x$ such that $\Gamma(x) = T$, and $c = \emptyset$.
Rule \TVar also proves that $\teqTc{\Gamma'}{\Delta}{\E}{\E'}{x}{x}{\Gamma'(x)}{\emptyset}$.
Since $\Gamma'(x)\in\branch{\Gamma(x)}$, by applying rule \TOr as many times as necessary, we have 
$\teqTc{\Gamma'}{\Delta}{\E}{\E'}{x}{x}{\Gamma(x)}{\emptyset}$, \ie
$\teqTc{\Gamma'}{\Delta}{\E}{\E'}{x}{x}{T}{\emptyset}$, which proves the claim.

\bigskip

The corollary then follows, again by induction on the typing derivation.
If $T$ is not a union type, $\branch{T} = \{T\}$ and the claim is directly the previous property.
Otherwise, the last rule applied in the typing derivation can only be \TVar, \TSub, or \TOr.
The \TSub case follows trivially from the induction hypothesis; since $T$ is a union type, it is its own only subtype.
In the \TVar case, $t = t' = x$ for some variable $x$ such that $\Gamma(x) = T$.
Hence, by definition, $\Gamma'(x)\in\branch{T}$, and by rule \TVar we have $\teqTc{\Gamma'}{\Delta}{\E}{\E'}{t}{t'}{\Gamma'(x)}{c}$.
Finally, in the \TOr case, we have $T = T_1 \orT T_2$ for some $T_1$, $T_2$ such that
$\teqTc{\Gamma}{\Delta}{\E}{\E'}{t}{t'}{T_1}{c}$.
By the induction hypothesis, there exists $T_1'\in\branch{T_1}$ such that
$\teqTc{\Gamma'}{\Delta}{\E}{\E'}{t}{t'}{T_1'}{c}$.
Since, by definition, $\branch{T_1}\subseteq\branch{T_1\orT T_2}$, this proves the claim.
\end{proof}

\begin{lemma}[Typing destructors in branches]
\label{lem-proof:type-dest-branches}
For all $\Gamma$, 
$T$, $d$, $x$, for all $\Gamma' \in \branch{\Gamma}$,
if $\tDestnew{\Gamma}{d(x)}{T}$
then
$\tDestnew{\Gamma'}{d(x)}{T}$.
\end{lemma}
\begin{proof}
This property is immediate by examining the typing rules for destructors.
Indeed, $\Gamma$ and $\Gamma'$ only differ on variables, and the rules for destructors only
involve $\Gamma(x)$ for $x\in\X$ in conditions of the form $\Gamma(x) = T$ for some type $T$ which is not a union type.

Hence in these cases $\Gamma'(x)$ is also $T$, and the same rule can be applied to $\Gamma'$ to prove the claim.
\end{proof}

\begin{lemma}[Typing processes in branches]
\label{lem-proof:type-processes-branches}
For all $\Gamma$,
$C$, for all processes $P$, $Q$, for all $\Gamma' \in \branch{\Gamma}$,
if $\teqP{\Gamma}{\Delta}{\E}{\E'}{P}{Q}{C}$
then there exists $C' \subseteq C$ such that
$\teqP{\Gamma'}{\Delta}{\E}{\E'}{P}{Q}{C'}$.
\end{lemma}
\begin{proof}
We prove this lemma by induction on the derivation of $\teqP{\Gamma}{\Delta}{\E}{\E'}{P}{Q}{C}$.
In all the cases for the last rule applied in this derivation, we can show that the conditions of this rule still hold in
$\Gamma'$ (instead of $\Gamma$) using
\begin{itemize}
\item Lemma~\ref{lem-proof:type-terms-branches} for the conditions of the form $\teqTc{\Gamma}{\Delta}{\E}{\E'}{M}{N}{T}{c}$;
\item Lemma~\ref{lem-proof:type-dest-branches} for the conditions of the form $\tDestnew{\Gamma}{d(y)}{T}$;
\item the fact that if $\Gamma(x)$ is not a union type,
then $\Gamma'(x) = \Gamma(x)$, for conditions such as "$\Gamma(x) = \L$", "$\Gamma(x) = \LRTn{l}{a}{m}{l'}{a}{n}$" (in the \PLetLR case);
\item the induction hypothesis for the conditions of the form $\teqP{\Gamma}{\Delta}{\E}{\E'}{P'}{Q'}{C''}$.
In this case, the induction hypothesis produces a $C''' \subseteq C''$, which can then be used to show $C' \subseteq C$, since $C'$ and $C$ are usually respectively $C'''$ and $C''$ with some terms added.
\end{itemize}

We detail here the cases of rules \POut, \PPar, and \POr. The other cases are similar, as explained above.

If the last rule is \POut, then we have $P = \OUT{M}. P'$, $Q = \OUT{N}. Q'$, $C = C''\UnionAll c$ for some
$P'$, $Q'$, $M$, $N$, $C''$, $c$, such that $\teqP{\Gamma}{\Delta}{\E}{\E'}{P'}{Q'}{C''}$ and $\teqTc{\Gamma}{\Delta}{\E}{\E'}{M}{N}{\L}{c}$.
Hence by Lemma~\ref{lem-proof:type-terms-branches}, $\teqTc{\Gamma'}{\Delta}{\E}{\E'}{M}{N}{\L}{c}$, and by the induction hypothesis
applied to $P'$, $Q'$, $\teqP{\Gamma'}{\Delta}{\E}{\E'}{P'}{Q'}{C'''}$ for some $C'''$ such that $C'''\subseteq C''$.
Therefore by rule \POut, $\teqP{\Gamma'}{\Delta}{\E}{\E'}{P}{Q}{C''' \UnionAll c}$, and since 
$C''' \UnionAll c \subseteq C'' \UnionAll c (=C)$, this proves the claim.

\medskip

If the last rule is \PPar, then we have $P = P_1 \PAR P_2$, $Q = Q_1 \PAR Q_2$, $C = C_1 \UnionCart C_2$
for some $P_1$, $P_2$, $Q_1$, $Q_2$, $C_1$, $C_2$ such that $\teqP{\Gamma}{\Delta}{\E}{\E'}{P_1}{Q_1}{C_1}$ and
$\teqP{\Gamma}{\Delta}{\E}{\E'}{P_2}{Q_2}{C_2}$.
Thus by applying the induction hypothesis twice, we have $\teqP{\Gamma'}{\Delta}{\E}{\E'}{P_1}{Q_1}{C_1'}$
and $\teqP{\Gamma'}{\Delta}{\E}{\E'}{P_2}{Q_2}{C_2'}$ with $C_1' \subseteq C_1$ and $C_2' \subseteq C_2$.
Therefore by rule \PPar, $\teqP{\Gamma'}{\Delta}{\E}{\E'}{P_1 \PAR P_2}{Q_1 \PAR Q_2}{C_1' \UnionCart C_2'}$,
and since $C_1'\UnionCart C_2' \subseteq C_1 \UnionCart C_2 (= C)$, this proves the claim.

\medskip

If the last rule is \POr, then there exist $\Gamma''$, $x$, $T_1$, $T_2$, $C_1$ and $C_2$ such that 
$\Gamma = \Gamma'', x:T_1\orT T_2$, $C = C_1 \cup C_2$,
$\teqP{\Gamma'', x:T_1}{\Delta}{\E}{\E'}{P}{Q}{C_1}$ and
$\teqP{\Gamma'', x:T_2}{\Delta}{\E}{\E'}{P}{Q}{C_2}$.
By definition of branches, it is clear that $\branch{\Gamma} = \branch{\Gamma'', x:T_1\orT T_2} = \branch{\Gamma'', x:T_1} \cup \branch{\Gamma'', x:T_2}$.
Thus, since $\Gamma'\in\branch{\Gamma}$, we know that $\Gamma'\in\branch{\Gamma'', x:T_1}$ or $\Gamma'\in\branch{\Gamma'', x:T_2}$.
We write the proof for the case where $\Gamma'\in\branch{\Gamma'', x:T_1}$, the other case is analogous.
By applying the induction hypothesis to 
$\teqP{\Gamma'', x:T_1}{\Delta}{\E}{\E'}{P}{Q}{C_1}$,
there exists $C_1' \subseteq C_1$ such that $\teqP{\Gamma'}{\Delta}{\E}{\E'}{P}{Q}{C_1'}$.
Since $C_1\subseteq C$, this proves the claim.
\end{proof}

\begin{lemma}[Environments in the constraints]
\label{lem-proof:env-constraints-bound}
For all $\Gamma$,
$C$, for all processes $P$, $Q$,
if
\[\teqP{\Gamma}{\Delta}{\E}{\E'}{P}{Q}{C}\]
then for all $(c, \Gamma') \in C$,
\[\dom{\Gamma'} \subseteq \dom{\Gamma} \cup \bvars{P} \cup \bvars{Q}\cup\nnames{P}\cup\nnames{Q}\]
(where $\bvars{P}$, $\nnames{P}$ respectively denote the sets of bound variables and names in $P$).
\end{lemma}
\begin{proof}
We prove this lemma by induction on the typing derivation of $\teqP{\Gamma}{\Delta}{\E}{\E'}{P}{Q}{C}$.

If the last rule applied in this derivation is \PZero, we have $C = \{(\emptyset,\Gamma)\}$,
and the claim clearly holds.

\medskip

In the \PPar case, we have $P = P_1 \PAR P_2$, $Q = Q_1 \PAR Q_2$, and $C = C_1 \UnionCart C_2$ for some $P_1$, $P_2$, $Q_1$, $Q_2$, $C_1$, $C_2$ such that $\teqP{\Gamma}{\Delta}{\E}{\E'}{P_1}{Q_1}{C_1}$
and $\teqP{\Gamma}{\Delta}{\E}{\E'}{P_2}{Q_2}{C_2}$.
Thus any element of $C$ is of the form $(c_1 \cup c_2, \Gamma_1 \cup \Gamma_2)$ where $(c_1, \Gamma_1) \in C_1$, 
$(c_2, \Gamma_2) \in C_2$, and $\Gamma_1$, $\Gamma_2$ are compatible.
By the induction hypothesis,
$\dom{\Gamma_1} \subseteq \dom{\Gamma} \cup \bvars{P_1} \cup \bvars{Q_1} \cup \nnames{P_1}\cup\nnames{Q_1} \subseteq \dom{\Gamma} \cup \bvars{P} \cup \bvars{Q} \cup \nnames{P}\cup\nnames{Q}$, and similarly for $\Gamma_2$.
Therefore, since $\dom{\Gamma_1 \cup \Gamma_2} = \dom{\Gamma_1}\cup\dom{\Gamma_2}$ (by definition), the claim holds.

\medskip

In the \PIn and \PLet 
cases, the typing judgement appearing in the condition of the rule uses $\Gamma$ extended with an additional variable, which is bound in $P$ and $Q$.
We detail the \PIn case, the other case is similar.
We have $P = \IN{x}. P'$, $Q = \IN{x}. Q'$ for some $x$, $P'$, $Q'$ such that $x \notin \dom{\Gamma}$ and
$\teqP{\Gamma, x:\L}{\Delta}{\E}{\E'}{P'}{Q'}{C}$.
Hence by the induction hypothesis, if $(c, \Gamma')\in C$, $\dom{\Gamma'} \subseteq \dom{\Gamma, x:\L} \cup \bvars{P'} \cup \bvars{Q'}\cup\nnames{P'}\cup\nnames{Q'}$.
Since $\bvars{P} = \{x \} \cup \bvars{P'}$ and $\bvars{Q} = \{x \} \cup \bvars{Q'}$, this proves the claim.

\medskip

The case of rule \PNew is similar, extending $\Gamma$ with a nonce instead of a variable.

\medskip

In the \POut case, there exist $P'$, $Q'$, $M$, $N$, $C'$, $c$ such that $P = \OUT{M}. P'$, $Q = \OUT{N}. Q'$,
$C = C' \UnionAll c$, $\teqTc{\Gamma}{}{}{}{M}{N}{\L}{c}$ and $\teqP{\Gamma}{}{}{}{P'}{Q'}{C'}$.
If $(c',\Gamma')\in C$, by definition of $\UnionAll$ there exists $c''$ such that $(c'',\Gamma')\in C'$ and
$c' = c\cup c''$.
By the induction hypothesis, we thus have
\[\dom{\Gamma'} \subseteq \dom{\Gamma} \cup \bvars{P'} \cup \bvars{Q'}\cup\nnames{P'}\cup\nnames{Q'}\]
and since $\bvars{P'}=\bvars{P}$, $\nnames{P'}=\nnames{P}$, and similarly for $Q$, this proves the claim.

\medskip

In the \PIfL case, there exist $P'$, $P''$, $Q'$, $Q''$, $M$, $N$, $M'$, $N'$, $C'$, $C''$, $c$, $c'$
such that $P = \ITE{M}{M'}{P'}{P''}$, $Q = \ITE{N}{N'}{Q'}{Q''}$,
$C = (C'\cup C'') \UnionAll (c\cup c')$, $\teqTc{\Gamma}{}{}{}{M}{N}{\L}{c}$, $\teqTc{\Gamma}{}{}{}{M'}{N'}{\L}{c'}$,
$\teqP{\Gamma}{}{}{}{P'}{Q'}{C'}$, and $\teqP{\Gamma}{}{}{}{P''}{Q''}{C''}$.
If $(c'',\Gamma')\in C$, by definition of $\UnionAll$ there exist $c'''$, such that $(c''',\Gamma')\in C'\cup C''$ and
$c'' = c'''\cup c\cup c'$.
We write the proof for the case where $(c''',\Gamma')\in C'$, the other case is analogous.
By the induction hypothesis, we thus have
\[\dom{\Gamma'} \subseteq \dom{\Gamma} \cup \bvars{P'} \cup \bvars{Q'}\cup\nnames{P'}\cup\nnames{Q'}\]
and since $\bvars{P'}\subseteq\bvars{P}$, $\nnames{P'}\subseteq\nnames{P}$, and similarly for $Q$, this proves the claim.

\medskip

The cases of rules \POr, \PLetLR, \PIfLR, \PIfS, \PIfLRinf, \PIfP, \PIfI, and \PIfLRp
remain. All these cases are similar, we write the proof for the \PIfLRinf case.
In this case, there exist $P'$, $P''$, $Q'$, $Q''$, $M$, $N$, $M'$, $N'$, $C'$, $C''$, $l$, $l'$, $m$, $n$
such that $P = \ITE{M}{M'}{P'}{P''}$, $Q = \ITE{N}{N'}{Q'}{Q''}$,
$C = C'\cup C''$, $\teqTc{\Gamma}{}{}{}{M}{N}{\LRTn{l}{\infty}{m}{l'}{\infty}{n}}{\emptyset}$,
$\teqTc{\Gamma}{}{}{}{M'}{N'}{\LRTn{l}{\infty}{m}{l'}{\infty}{n}}{\emptyset}$,
$\teqP{\Gamma}{}{}{}{P'}{Q'}{C'}$, and $\teqP{\Gamma}{}{}{}{P''}{Q''}{C''}$.
If $(c,\Gamma')\in C$, we thus know that $(c,\Gamma')\in C'$ or $(c,\Gamma')\in C''$.
We write the proof for the case where $(c,\Gamma')\in C'$, the other case is analogous.
By the induction hypothesis, we thus have
\[\dom{\Gamma'} \subseteq \dom{\Gamma} \cup \bvars{P'} \cup \bvars{Q'}\cup\nnames{P'}\cup\nnames{Q'}\]
and since $\bvars{P'}\subseteq\bvars{P}$, $\nnames{P'}\subseteq\nnames{P}$, and similarly for $Q$, this proves the claim.
\end{proof}

\begin{lemma}[Environments in the constraints do not contain union types]
\label{lem-proof:env-constr-union}
For all $\Gamma$,
$C$, for all processes $P$, $Q$,
if
\[\teqP{\Gamma}{\Delta}{\E}{\E'}{P}{Q}{C}\]
then for all $(c, \Gamma') \in C$,
\[\branch{\Gamma'} = \{\Gamma'\}\]
\ie for all $x\in\dom{\Gamma'}$, $\Gamma'(x)$ is not a union type.
\end{lemma}
\begin{proof}
This property is immediate by induction on the typing derivation.
\end{proof}

\begin{lemma}[Typing is preserved by extending the environment]
\label{lem-proof:typing-contextinclusion}
For all $\Gamma$, $\Gamma'$,
$P$, $Q$, $C$, $c$, $t$, $t'$, $T$, $c$, if $\tewf{\Gamma}$ and $\tewf{\Gamma\cup\Gamma'}$ (we do not require that
$\Gamma'$ is well-formed):
\begin{itemize}
\item if $\dom{\Gamma} \cap \dom{\Gamma'} = \emptyset$, 
and if $\teqTc{\Gamma}{\Delta}{\E_1}{\E_2}{t}{t'}{T}{c}$, then
$\teqTc{\Gamma \cup \Gamma'}{\Delta}{\E_1}{\E_2}{t}{t'}{T}{c}$.
\item if $\dom{\Gamma} \cap \dom{\Gamma'} = \emptyset$,
and if $\tDestnew{\Gamma}{d(y)}{T}$, then
$\tDestnew{\Gamma \cup \Gamma'}{d(y)}{T}$.
\item if $(\dom{\Gamma} \cup \bvars{P} \cup \bvars{Q}\cup\nnames{P}\cup\nnames{Q}) \cap \dom{\Gamma'} = \emptyset$,
and if $\teqP{\Gamma}{\Delta}{\E_1}{\E_2}{P}{Q}{C}$, then
$\teqP{\Gamma \cup \Gamma'}{\Delta}{\E_1}{\E_2}{P}{Q}{C'}$.
where
$C' = \{(c,\Gamma_c \cup \Gamma'')|(c,\Gamma_c)\in C \wedge \Gamma'' \in \branch{\Gamma'}\}$
(note that the union is well defined, \ie $\Gamma_c$ and $\Gamma''$ are compatible,
thanks to Lemma~\ref{lem-proof:env-constraints-bound})
\end{itemize}
\end{lemma}
\begin{proof}
\begin{itemize}
\item The first point is immediate by induction on the type derivation.
\item The second point is immediate by examining the typing rules for destructors.
\item The third point is immediate by induction on the type derivation of the processes.
In the \PZero case, to satisfy the condition that the environment is its own only branch, rule \POr needs to be applied first, in order to split all the
union types in $\Gamma'$, which yields the environments $\branch{\Gamma\cup\Gamma'}$ in the constraints.
\end{itemize}
\end{proof}

\begin{lemma}[Consistency for Subsets]
\label{lem-proof:cons-subset}
The following statements about constraints hold:
\begin{enumerate}
\item If $(c,\Gamma)$ is consistent, 
and $c' \subseteq c$ then $(c',\Gamma)$ is consistent. 
\item Let $C$ be a consistent constraint set. 
Then every subset
 $C' \subseteq C$ is also consistent.
\item If $C \UnionAll c'$ is consistent 
then $C$ also is.
\item If $C_1 \subseteq C_2$ and $C'_1 \subseteq C'_2$, then $C_1 \UnionCart C'_1 \subseteq C_2 \UnionCart C'_2$.
\item $\inst{\cdot}{\sigma}{\sigma'}$ commutes with $\cup$, $\UnionCart$, $\UnionAll$, \ie for all $C$, $C'$,
$\sigma$, $\sigma'$,
$\inst{C \UnionCart C'}{\sigma}{\sigma'} = \inst{C}{\sigma}{\sigma'} \UnionCart \inst{C'}{\sigma}{\sigma'}$
and similarly for $\cup$, $\UnionAll$.
\item If $\sigma_1$ and $\sigma_1'$ are ground and have disjoint domains, as well as $\sigma_2$ and $\sigma_2'$,
then for all $c$,
$\inst{\inst{c}{\sigma_1}{\sigma_2}}{\sigma_1'}{\sigma_2'} = \inst{c}{\sigma_1 \cup \sigma_1'}{\sigma_2 \cup \sigma_2'}$
\item if $C$ is consistent, 
if $\wtc{\Delta}{\E}{\E'}{\sigma}{\sigma'}{\Gamma}{c}$ for some $c$,
and if for all $(c', \Gamma') \in C$, $\Gamma \subseteq \Gamma'$, then
$\inst{C}{\sigma}{\sigma'}$ is consistent. 


\end{enumerate}
\end{lemma}
\begin{proof}
Points 1 and 2 follow immediately from the definition of consistency and of static equivalence.

Point 3 follows from the point 1: for every $(c,\Gamma) \in C$,
$(c \cup c',\Gamma)$ is in $C \UnionAll c'$; which is consistent since $C \UnionAll c'$ is;
therefore $(c,\Gamma)$ also is.

Point 4 follows from the definition of $\UnionCart$. If $(c, \Gamma) \in C_1 \UnionCart C'_1$, there exists
$(c_1,\Gamma_1) \in C_1$, $(c'_1,\Gamma'_1) \in C'_1$ such that
$(c, \Gamma) = (c_1 \cup c'_1, \Gamma_1 \cup \Gamma'_1)$ (and $\Gamma_1$, $\Gamma'_1$ are compatible).
Since $C_1 \subseteq C_2$, $(c_1,\Gamma_1) \in C_2$. Similarly, $(c'_1,\Gamma'_1) \in C'_2$.
Therefore $(c, \Gamma) \in C_2 \UnionCart C'_2$.

Points 5 and 6 follow from the definitions of $\inst{\cdot}{\sigma}{\sigma'}$, $\UnionCart$, $\UnionAll$.

\medskip

Point 7 follows from the definitions of $\inst{\cdot}{\sigma}{\sigma'}$, and of consistency.
Indeed, let $(c'',\Gamma'')\in \inst{C}{\sigma}{\sigma'}$.
There exists $c'''$ such that $c'' = \inst{c'''}{\sigma}{\sigma'}$, and $(c''',\Gamma'')\in C$.
Let $c_1\subseteq c''$ and $\Gamma_1\subseteq \Gamma''$ such that $\novar{\Gamma_1}=\novar{\Gamma''}$ and
$\var{c_1}\subseteq\dom{\Gamma_1}$.
Let $\theta$, $\theta'$ be well-typed in $\Gamma_1$.
Since $c'' = \inst{c'''}{\sigma}{\sigma'}$, there exists $c_2\subseteq c'''$ such that $c_1 = \inst{c_2}{\sigma}{\sigma'}$.
If we show that $\sigma\theta$ and $\sigma'\theta'$ are well-typed in $\Gamma_1\cup\Gamma$,
it will follow from the consistency of $C$ that
$\NEWN{\E_{\Gamma_2}}.(\phiEnew{\Gamma_2} \cup \phiL{c_2}\sigma\theta)$ and 
$\NEWN{\E_{\Gamma_2}}.(\phiEnew{\Gamma_2} \cup \phiR{c_2}\sigma'\theta')$ are statically equivalent,
where $\Gamma_2 = \Gamma_1\cup\Gamma \subseteq \Gamma''$.
Since $\novar{\Gamma_1}=\novar{\Gamma''}$ and $\Gamma\subseteq \Gamma''$, we have $\E_{\Gamma_2} = \E_{\Gamma_1}$,
and $\phiEnew{\Gamma_2} = \phiEnew{\Gamma_1}$.
Hence 
$\NEWN{\E_{\Gamma_1}}.(\phiEnew{\Gamma_1} \cup \phiL{c_1}\theta)$ and 
$\NEWN{\E_{\Gamma_1}}.(\phiEnew{\Gamma_1} \cup \phiR{c_1}\theta')$ are statically equivalent,

It only remains to be proved that $\sigma\theta$ and $\sigma'\theta'$ are well-typed in $\Gamma_2$.

Since $\sigma$ is ground, $\sigma\theta = \sigma \cup \theta|_{\dom{\Gamma_1}\backslash\dom{\Gamma}}$, and similarly for
$\sigma'\theta'$. Hence, since $\sigma$, $\sigma'$ are well-typed in $\Gamma$, and $\theta$, $\theta'$ are well-typed in $\Gamma_1$, their compositions also are, which concludes the proof.
\end{proof}

\begin{lemma}[Environments in constraints contain a branch of the typing environment]
\label{lem-proof:env-const-branch}
For all $\Gamma$, 
$C$, for all processes $P$, $Q$,
if $\teqP{\Gamma}{\Delta}{\E}{\E'}{P}{Q}{C}$ then
for all $(c,\Gamma') \in C$,
there exists $\Gamma'' \in \branch{\Gamma}$ such that $\Gamma'' \subseteq \Gamma'$.
\end{lemma}
\begin{proof}
We prove this property by induction on the type derivation of $\teqP{\Gamma}{\Delta}{\E}{\E'}{P}{Q}{C}$.
In the \PZero case, $C = \{(\emptyset, \Gamma)\}$, and by assumption $\branch{\Gamma}=\{\Gamma\}$, hence
the claim trivially holds.

In the \PPar case, we have $P = P_1 \PAR P_2$, $Q = Q_1 \PAR Q_2$, and $C = C_1 \UnionCart C_2$ for some $P_1$, $P_2$, $Q_1$, $Q_2$, $C_1$, $C_2$ such that $\teqP{\Gamma}{\Delta}{\E}{\E'}{P_1}{Q_1}{C_1}$
and $\teqP{\Gamma}{\Delta}{\E}{\E'}{P_2}{Q_2}{C_2}$.
Thus any element of $C$ is of the form $(c_1 \cup c_2, \Gamma_1 \cup \Gamma_2)$ where $(c_1, \Gamma_1) \in C_1$, 
$(c_2, \Gamma_2) \in C_2$, and $\Gamma_1$, $\Gamma_2$ are compatible.
By the induction hypothesis, both $C_1$ and $C_2$ contain a branch of $\Gamma$.
The claim holds, as these are necessarily the same branch, since $\Gamma_1$ and $\Gamma_2$ are compatible.

\medskip

In the \POr case, we have $\Gamma = \Gamma'', x:T_1\orT T_2$ for some $x$, $\Gamma''$, $T_1$, $T_2$ such that
$\teqP{\Gamma'', x:T_1}{\Delta}{\E}{\E'}{P}{Q}{C_1}$ and $\teqP{\Gamma'', x:T_2}{\Delta}{\E}{\E'}{P}{Q}{C_2}$, and $C = C_1 \cup C_2$.
Thus by the induction hypothesis, if $(c,\Gamma') \in C_i$ (for $i \in\{1,2\}$), then $\Gamma'$ contains some $\Gamma''' \in \branch{\Gamma'', x:T_i} \subseteq \branch{\Gamma}$, and the claim holds.

\medskip

In the \POut case, there exist $P'$, $Q'$, $M$, $N$, $C'$, $c$ such that $P = \OUT{M}. P'$, $Q = \OUT{N}. Q'$,
$C = C' \UnionAll c$, $\teqTc{\Gamma}{}{}{}{M}{N}{\L}{c}$ and $\teqP{\Gamma}{}{}{}{P'}{Q'}{C'}$.
If $(c',\Gamma')\in C$, by definition of $\UnionAll$ there exists $c''$ such that $(c'',\Gamma')\in C'$ and
$c' = c\cup c''$.
Hence by applying the induction hypothesis to $\teqP{\Gamma}{}{}{}{P'}{Q'}{C'}$, there exists $\Gamma''\in\branch{\Gamma}$ such that $\Gamma''\subseteq \Gamma'$.

\medskip

In the \PIfL case, there exist $P'$, $P''$, $Q'$, $Q''$, $M$, $N$, $M'$, $N'$, $C'$, $C''$, $c$, $c'$
such that $P = \ITE{M}{M'}{P'}{P''}$, $Q = \ITE{N}{N'}{Q'}{Q''}$,
$C = (C'\cup C'') \UnionAll (c\cup c')$, $\teqTc{\Gamma}{}{}{}{M}{N}{\L}{c}$, $\teqTc{\Gamma}{}{}{}{M'}{N'}{\L}{c'}$,
$\teqP{\Gamma}{}{}{}{P'}{Q'}{C'}$, and $\teqP{\Gamma}{}{}{}{P''}{Q''}{C''}$.
If $(c'',\Gamma')\in C$, by definition of $\UnionAll$ there exist $c'''$, such that $(c''',\Gamma')\in C'\cup C''$ and
$c'' = c'''\cup c\cup c'$.
We write the proof for the case where $(c''',\Gamma')\in C'$, the other case is analogous.
By applying the induction hypothesis to $\teqP{\Gamma}{}{}{}{P'}{Q'}{C'}$, there exists $\Gamma''\in\branch{\Gamma}$ such that $\Gamma''\subseteq \Gamma'$, which proves the claim.

\medskip

All remaining cases are similar. We write the proof for the \PIfLRinf case.
In this case, there exist $P'$, $P''$, $Q'$, $Q''$, $M$, $N$, $M'$, $N'$, $C'$, $C''$, $l$, $l'$, $m$, $n$
such that $P = \ITE{M}{M'}{P'}{P''}$, $Q = \ITE{N}{N'}{Q'}{Q''}$,
$C = C'\cup C''$, $\teqTc{\Gamma}{}{}{}{M}{N}{\LRTn{l}{\infty}{m}{l'}{\infty}{n}}{\emptyset}$,
$\teqTc{\Gamma}{}{}{}{M'}{N'}{\LRTn{l}{\infty}{m}{l'}{\infty}{n}}{\emptyset}$,
$\teqP{\Gamma}{}{}{}{P'}{Q'}{C'}$, and $\teqP{\Gamma}{}{}{}{P''}{Q''}{C''}$.
If $(c,\Gamma')\in C$, we thus know that $(c,\Gamma')\in C'$ or $(c,\Gamma')\in C''$.
We write the proof for the case where $(c,\Gamma')\in C'$, the other case is analogous.
By applying the induction hypothesis to $\teqP{\Gamma}{}{}{}{P'}{Q'}{C'}$, there exists $\Gamma''\in\branch{\Gamma}$ such that $\Gamma''\subseteq \Gamma'$, which proves the claim.
\end{proof}

\begin{lemma}[All branches are represented in the constraints]
\label{lem-proof:all-branch-const}
For all $\Gamma$, $\Delta$, 
$C$, for all processes $P$, $Q$,
if $\teqP{\Gamma}{\Delta}{\E}{\E'}{P}{Q}{C}$ then for all $\Gamma' \in \branch{\Gamma}$,
there exists $(c,\Gamma'') \in C$,
such that $\Gamma' \subseteq \Gamma''$.
\end{lemma}
\begin{proof}
We prove this property by induction on the type derivation of $\teqP{\Gamma}{\Delta}{\E}{\E'}{P}{Q}{C}$.
In the \PZero case, $C = \{(\emptyset, \Gamma)\}$, and by assumption $\branch{\Gamma}=\{\Gamma\}$,
hence the claim trivially holds.

In the \PPar case, we have $P = P_1 \PAR P_2$, $Q = Q_1 \PAR Q_2$, and $C = C_1 \UnionCart C_2$ for some $P_1$, $P_2$, $Q_1$, $Q_2$, $C_1$, $C_2$ such that $\teqP{\Gamma}{\Delta}{\E}{\E'}{P_1}{Q_1}{C_1}$
and $\teqP{\Gamma}{\Delta}{\E}{\E'}{P_2}{Q_2}{C_2}$.
By the induction hypothesis, there exists $(c_1, \Gamma_1) \in C_1$ and $(c_2, \Gamma_2) \in C_2$ such that
$\Gamma' \subseteq \Gamma_1$ and $\Gamma' \subseteq \Gamma_2$.
By Lemma~\ref{lem-proof:env-constraints-bound}, $\dom{\Gamma_1}$ and $\dom{\Gamma_2}$ only contain $\dom{\Gamma} (= \dom{\Gamma'})$ and variables and names in $\bvars{P_1}\cup\bvars{Q_1}\cup\nnames{P_1}\cup\nnames{Q_1}$ and $\bvars{P_2}\cup\bvars{Q_2}\cup\nnames{P_2}\cup\nnames{Q_2}$ respectively.
Since $\Gamma_1(x) = \Gamma_2(x) = \Gamma'(x)$ for all $x \in \dom{\Gamma'}$, and 
since the sets $\bvars{P_1}\cup\bvars{Q_1}\cup\nnames{P_1}\cup\nnames{Q_1}$ and $\bvars{P_2}\cup\bvars{Q_2}\cup\nnames{P_2}\cup\nnames{Q_2}$ are disjoint by well formedness of the processes
$P_1 \PAR P_2$ and $Q_1 \PAR Q_2$, $\Gamma_1$ and $\Gamma_2$ are compatible.
Thus $(c_1 \cup c_2, \Gamma_1 \cup \Gamma_2)\in C_1 \UnionCart C_2 (= C)$, and the claim holds since
$\Gamma' \subseteq \Gamma_1 \cup \Gamma_2$.

\medskip

In the \POr case, we have $\Gamma = \Gamma'', x:T_1\orT T_2$ for some $x$, $\Gamma''$, $T_1$, $T_2$ such that
$\teqP{\Gamma'', x:T_1}{\Delta}{\E}{\E'}{P}{Q}{C_1}$ and $\teqP{\Gamma'', x:T_2}{\Delta}{\E}{\E'}{P}{Q}{C_2}$, and $C = C_1 \cup C_2$.
Since $\branch{\Gamma} = \branch{\Gamma'', x:T_1} \cup \branch{\Gamma'', x:T_2}$, 
we know that $\Gamma' \in \branch{\Gamma'', x:T_i}$ for some $i$.
We conclude this case directly by applying the induction hypothesis to $\teqP{\Gamma'', x:T_i}{\Delta}{\E}{\E'}{P}{Q}{C_i}$.

\medskip

In the \POut case, there exist $P'$, $Q'$, $M$, $N$, $C'$, $c$ such that $P = \OUT{M}. P'$, $Q = \OUT{N}. Q'$,
$C = C' \UnionAll c$, $\teqTc{\Gamma}{}{}{}{M}{N}{\L}{c}$ and $\teqP{\Gamma}{}{}{}{P'}{Q'}{C'}$.
By applying the induction hypothesis to $\teqP{\Gamma}{}{}{}{P'}{Q'}{C'}$, there exists $(c'',\Gamma'')\in C'$ such that
$\Gamma'\subseteq \Gamma''$.
By definition of $\UnionAll$, $(c''\cup c,\Gamma'')\in C$, which proves the claim.

\medskip

In the \PIfL case, there exist $P'$, $P''$, $Q'$, $Q''$, $M$, $N$, $M'$, $N'$, $C'$, $C''$, $c$, $c'$
such that $P = \ITE{M}{M'}{P'}{P''}$, $Q = \ITE{N}{N'}{Q'}{Q''}$,
$C = (C'\cup C'') \UnionAll (c\cup c')$, $\teqTc{\Gamma}{}{}{}{M}{N}{\L}{c}$, $\teqTc{\Gamma}{}{}{}{M'}{N'}{\L}{c'}$,
$\teqP{\Gamma}{}{}{}{P'}{Q'}{C'}$, and $\teqP{\Gamma}{}{}{}{P''}{Q''}{C''}$.
By applying the induction hypothesis to $\teqP{\Gamma}{}{}{}{P'}{Q'}{C'}$, there exists $(c'',\Gamma'')\in C'$ such that $\Gamma'\subseteq \Gamma''$.
By definition of $\UnionAll$, $(c''\cup c \cup c',\Gamma'')\in C$, which proves the claim.

\medskip

All remaining cases are similar. We write the proof for the \PIfLRinf case.
In this case, there exist $P'$, $P''$, $Q'$, $Q''$, $M$, $N$, $M'$, $N'$, $C'$, $C''$, $l$, $l'$, $m$, $n$
such that $P = \ITE{M}{M'}{P'}{P''}$, $Q = \ITE{N}{N'}{Q'}{Q''}$,
$C = C'\cup C''$, $\teqTc{\Gamma}{}{}{}{M}{N}{\LRTn{l}{\infty}{m}{l'}{\infty}{n}}{\emptyset}$,
$\teqTc{\Gamma}{}{}{}{M'}{N'}{\LRTn{l}{\infty}{m}{l'}{\infty}{n}}{\emptyset}$,
$\teqP{\Gamma}{}{}{}{P'}{Q'}{C'}$, and $\teqP{\Gamma}{}{}{}{P''}{Q''}{C''}$.
By applying the induction hypothesis to $\teqP{\Gamma}{}{}{}{P'}{Q'}{C'}$, there exists $(c'',\Gamma'')\in C'$ such that $\Gamma'\subseteq \Gamma''$.

If $(c,\Gamma')\in C$, we thus know that $(c,\Gamma')\in C'$ or $(c,\Gamma')\in C''$.
We write the proof for the case where $(c,\Gamma')\in C'$, the other case is analogous.
By applying the induction hypothesis to $\teqP{\Gamma}{}{}{}{P'}{Q'}{C'}$, there exists $\Gamma''\in\branch{\Gamma}$ such that $\Gamma''\subseteq \Gamma'$, which proves the claim.

\end{proof}

\begin{lemma}[Refinement types]
\label{lem-proof:lr-ground}
For all $\Gamma$, 
for all terms $t$, $t'$, for all $m$, $n$, $a$, $l$, $l'$, $c$, if
$\teqTc{\Gamma}{\Delta}{\E}{\E'}{t}{t'}{\LRTn{l}{a}{m}{l'}{a}{n}}{c}$ then $c = \emptyset$ and
\begin{itemize}
\item either $t = m$, $t' = n$, $a = \infty$ and $\Gamma(m) = \noncetypelab{l}{a}{m}$ and $\Gamma(n) = \noncetypelab{l'}{a}{n}$;
\item or $t = m$, $t' = n$, $a = 1$,
and $(\Gamma(m) = \noncetypelab{l}{a}{m}) \vee (m\in\FN\cup\CST \wedge l = \L)$,
and $(\Gamma(n) = \noncetypelab{l'}{a}{n}) \vee (n\in\FN\cup\CST \wedge l' = \L)$;

\item or $t$ and $t'$ are variables $x,y\in\X$ and there exist labels $l''$, $l'''$, and names $m'$, $n'$ such that
$\Gamma(x) = \LRTn{l}{a}{m}{l''}{a}{n'}$ and $\Gamma(y) = \LRTn{l'''}{a}{m'}{l'}{a}{n}$.
\end{itemize}
In particular if $t$, $t'$ are ground then only the first case can occur.
\end{lemma}
\begin{proof}
The proof of this property is immediate by induction on the typing derivation for the terms.
Indeed, because of the form of the type, and by well-formedness of $\Gamma$,
the only rules which can lead to $\teqTc{\Gamma}{\Delta}{\E}{\E'}{t}{t'}{\LRTn{l}{a}{m}{l'}{a}{n}}{c}$ are \TVar, \TLRone, \TLRinf, \TLRVar, and \TSub.

In the \TVar, \TLRone, \TLRinf cases the claim directly follows from the premises of the rule.

In the \TLRVar case, $t$ and $t'$ are necessarily variables, and their types in $\Gamma$ are obtained directly
by applying the induction hypothesis to the premises of the rule.


Finally in the \TSub case, $\teqTc{\Gamma}{\Delta}{\E}{\E'}{t}{t'}{T}{c}$ and $T\subtyp \LRTn{l}{a}{m}{l'}{a}{n}$.
By Lemma~\ref{lem-proof:subtyping}, $T = \LRTn{l}{a}{m}{l'}{a}{n}$ and we conclude by the induction hypothesis.
\end{proof}

\begin{lemma}[Encryption types]
\label{lem-proof:enc-types}
For all environment $\Gamma$, type $T$, key $k\in\K$, messages $M$, $N$, and set of constraints $c$:

\begin{enumerate}
\item If $\teqTc{\Gamma}{\Delta}{\E}{\E'}{M}{N}{\encT{T}{k}}{c}$ then
  \begin{itemize}
    \item either there exist $M'$, $N'$, such that $M = \ENC{M'}{k}$, $N = \ENC{N'}{k}$, and
    $\teqTc{\Gamma}{\Delta}{\E}{\E'}{M'}{N'}{T}{c}$ with a shorter derivation (than the one for
    $\teqTc{\Gamma}{\Delta}{\E}{\E'}{M}{N}{\encT{T}{k}}{c}$);

    \item or $M$ and $N$ are variables.
  \end{itemize}

\item If $\teqTc{\Gamma}{\Delta}{\E}{\E'}{M}{N}{\aencT{T}{k}}{c}$ then
  \begin{itemize}
    \item either there exist $M'$, $N'$, such that $M = \AENC{M'}{\PUBK{k}}$, $N = \AENC{N'}{\PUBK{k}}$, and
    $\teqTc{\Gamma}{\Delta}{\E}{\E'}{M'}{N'}{T}{c}$ with a shorter derivation (than the one for
    $\teqTc{\Gamma}{\Delta}{\E}{\E'}{M}{N}{\aencT{T}{k}}{c}$);

    \item or $M$ and $N$ are variables.
  \end{itemize}

\item If $T \subtyp \L$ and $\teqTc{\Gamma}{\Delta}{\E}{\E'}{\ENC{M}{k}}{N}{T}{c}$ then $T = \L$.

\item If $T \subtyp \L$ and $\teqTc{\Gamma}{\Delta}{\E}{\E'}{\AENC{M}{\PUBK{k}}}{N}{T}{c}$ then $T = \L$.

\item If $\teqTc{\Gamma}{\Delta}{\E}{\E'}{\ENC{M}{k}}{N}{\L}{c}$ then
there exists $N'$ such that $N = \ENC{N'}{k}$, and
  \begin{itemize}
    \item either there exist $T'$ and $c'$ such that
    $\Gamma(k) = \skey{\S}{T'}$, $c = \{\ENC{M}{k}\eqC N\} \cup c'$, and $\teqTc{\Gamma}{\Delta}{\E}{\E'}{M}{N'}{T'}{c'}$;
    \item or there exists $T'$ such that $\Gamma(k) = \skey{\L}{T'}$ and $\teqTc{\Gamma}{\Delta}{\E}{\E'}{M}{N'}{\L}{c}$.
  \end{itemize}

\item If $\teqTc{\Gamma}{\Delta}{\E}{\E'}{\AENC{M}{\PUBK{k}}}{N}{\L}{c}$ then
there exists $N'$ such that $N = \AENC{N'}{\PUBK{k}}$, and
  \begin{itemize}
    \item either there exist $T'$ and $c'$ such that
    $\Gamma(k) = \skey{\S}{T'}$, $c = \{\AENC{M}{\PUBK{k}}\eqC N\} \cup c'$, and $\teqTc{\Gamma}{\Delta}{\E}{\E'}{M}{N'}{T'}{c'}$;
    \item or $k\in\dom{\Gamma}$ and $\teqTc{\Gamma}{\Delta}{\E}{\E'}{M}{N'}{\L}{c}$.
  \end{itemize}

\item The symmetric properties to the previous four points, \ie when the term on the right is an encryption, also hold.
\end{enumerate}
\end{lemma}
\begin{proof}
We prove point 1 by induction on the derivation of $\teqTc{\Gamma}{\Delta}{\E}{\E'}{M}{N}{\encT{T}{k}}{c}$.
Because of the form of the type, and by well-formedness of $\Gamma$,
the only possibilities for the last rule applied are 
\TVar, \TEnc, and \TSub.
The claim clearly holds in the \TVar and \TEnc cases.
In the \TSub case, we have $\teqTc{\Gamma}{\Delta}{\E}{\E'}{M}{N}{T' \subtyp \encT{T}{k}}{c}$,
and by Lemma~\ref{lem-proof:subtyping}, there exists $T'' \subtyp T$ such that $T' = \encT{T''}{k}$.
Therefore, by applying the induction hypothesis to $\teqTc{\Gamma}{\Delta}{\E}{\E'}{M}{N}{T'}{c}$
\begin{itemize}
\item either $M$ and $N$ are either two variables, and the claim holds;
\item or there exist $M'$, $N'$ such that $M = \ENC{M'}{k}$, $N = \ENC{N'}{k}$,
and $\teqTc{\Gamma}{\Delta}{\E}{\E'}{M'}{N'}{T''}{c}$, with a derivation shorter than the one for
$\teqTc{\Gamma}{\Delta}{\E}{\E'}{M}{N}{T'}{c}$. Thus by subtyping (rule \TSub),
$\teqTc{\Gamma}{\Delta}{\E}{\E'}{M'}{N'}{T}{c}$ with a shorter derivation that $\teqTc{\Gamma}{\Delta}{\E}{\E'}{M}{N}{\encT{T}{k}}{c}$, which proves the property.
\end{itemize}

\medskip

Point 2 has a similar proof to point 1.

\bigskip

We now prove point 3 by induction on the proof of
$\teqTc{\Gamma}{\Delta}{\E}{\E'}{\ENC{M}{k}}{N}{T}{c}$.
Because of the form of the terms, the last rule applied can only be
\THigh, \TOr, \TEnc, \TEncH, \TEncL, \TAencH, \TAencL, \TLRp, \TLRLp or \TSub.

The \THigh, \TLRp, \TOr, \TEnc cases are actually impossible by Lemma~\ref{lem-proof:subtyping}, since $T \subtyp \L$.
%
%
In the \TSub case, we have $\teqTc{\Gamma}{\Delta}{\E}{\E'}{\ENC{M}{k}}{N}{T'}{c}$ for some $T'$ such that $T'\subtyp T$.
By transitivity of $\subtyp$, $T'\subtyp \L$, and the induction hypothesis proves the claim.
In all other cases, $T = \L$ and the claim holds.

\medskip

Point 4 has a similar proof to point 3.

\bigskip

We prove point 5 by induction on the proof of
$\teqTc{\Gamma}{\Delta}{\E}{\E'}{\ENC{M}{k}}{N}{\L}{c}$.
Because of the form of the terms and of the type (\ie $\L$) the last rule applied can only be
\TEncH, \TEncL, \TAencH, \TAencL, \TLRLp or \TSub.

The \TLRLp case is impossible, since by Lemma~\ref{lem-proof:lr-ground} it would imply that $\ENC{M}{k}$ is either a variable or a nonce.

In the \TSub case, we have $\teqTc{\Gamma}{\Delta}{\E}{\E'}{\ENC{M}{k}}{N}{T'}{c}$ for some $T'$ such that $T'\subtyp \L$.
By point 3, $T' = \L$, and the premise of the rule thus gives a shorter derivation
of $\teqTc{\Gamma}{\Delta}{\E}{\E'}{\ENC{M}{k}}{N}{\L}{c}$.
The induction hypothesis applied to this shorter derivation proves the claim.

The \TAencH and \TAencL cases are impossible, since the condition of the rule
would then imply $\teqTc{\Gamma}{\Delta}{\E}{\E'}{\ENC{M}{k}}{N}{\aencT{T}{k}}{c'}$ for some $T$, $k$, $c'$, which is not 
possible by point 2.

Finally, in the \TEncH and \TEncL cases, the premises of the rule directly proves the claim.

\medskip

Point 6 has a similar proof to point 5.

\bigskip

The symmetric properties, as described in point 7, have analogous proofs.
\end{proof}

\begin{lemma}[Signature types]
\label{lem-proof:sign-types}
For all environment $\Gamma$, type $T$, key $k\in\K$, messages $M$, $N$, and set of constraints $c$:

\begin{enumerate}
\item If $T \subtyp \L$ and $\teqTc{\Gamma}{\Delta}{\E}{\E'}{\SIGN{M}{k}}{N}{T}{c}$ then $T = \L$.

\item If $\teqTc{\Gamma}{\Delta}{\E}{\E'}{\SIGN{M}{k}}{N}{\L}{c}$ then
there exists $N'$ such that $N = \SIGN{N'}{k}$, and
  \begin{itemize}
    \item either there exist $T'$, $c'$ and $c''$ such that
    $\Gamma(k) = \skey{\S}{T'}$, $c = \{\SIGN{M}{k}\eqC N\} \cup c' \cup c''$, $\teqTc{\Gamma}{\Delta}{\E}{\E'}{M}{N'}{T'}{c'}$, and
    $\teqTc{\Gamma}{\Delta}{\E}{\E'}{M}{N'}{\L}{c''}$;
    \item or there exists $T'$ such that $\Gamma(k) = \skey{\L}{T'}$ and $\teqTc{\Gamma}{\Delta}{\E}{\E'}{M}{N'}{\L}{c}$.
  \end{itemize}

\item The symmetric properties to the previous four points, \ie when the term on the right is a signature, also hold.
\end{enumerate}
\end{lemma}
\begin{proof}
We prove point 1 by induction on the proof of
$\teqTc{\Gamma}{\Delta}{\E}{\E'}{\SIGN{M}{k}}{N}{T}{c}$.
Because of the form of the terms, the last rule applied can only be
\THigh, \TOr, \TEncH, \TEncL, \TAencH, \TAencL, \TSignH, \TSignL, \TLRp, \TLRLp or \TSub.

The \THigh, \TLRp, \TOr cases are actually impossible by Lemma~\ref{lem-proof:subtyping}, since $T \subtyp \L$.
%
%
In the \TSub case, we have $\teqTc{\Gamma}{\Delta}{\E}{\E'}{\SIGN{M}{k}}{N}{T'}{c}$ for some $T'$ such that $T'\subtyp T$.
By transitivity of $\subtyp$, $T'\subtyp \L$, and the induction hypothesis proves the claim.
In all other cases, $T = \L$ and the claim holds.

\bigskip

We prove point 2 by induction on the proof of
$\teqTc{\Gamma}{\Delta}{\E}{\E'}{\SIGN{M}{k}}{N}{\L}{c}$.
Because of the form of the terms and of the type (\ie $\L$) the last rule applied can only be
\TEncH, \TEncL, \TAencH, \TAencL, \TSignH, \TSignL, \TLRLp or \TSub.

The \TLRLp case is impossible, since by Lemma~\ref{lem-proof:lr-ground} it would imply that $\SIGN{M}{k}$ is either a variable or a nonce.

In the \TSub case, we have $\teqTc{\Gamma}{\Delta}{\E}{\E'}{\SIGN{M}{k}}{N}{T'}{c}$ for some $T'$ such that $T'\subtyp \L$.
By point 3, $T' = \L$, and the premise of the rule thus gives a shorter derivation
of $\teqTc{\Gamma}{\Delta}{\E}{\E'}{\SIGN{M}{k}}{N}{\L}{c}$.
The induction hypothesis applied to this shorter derivation proves the claim.

The \TEncH, \TEncL, \TAencH and \TAencL cases are impossible, since the condition of the rule
would then imply $\teqTc{\Gamma}{\Delta}{\E}{\E'}{\SIGN{M}{k}}{N}{\encT{T}{k}}{c'}$ (or $\aencT{T}{k}$) for some $T$, $k$, $c'$, which is not possible by Lemma~\ref{lem-proof:enc-types}.

Finally, in the \TSignH and \TSignL cases, the premises of the rule directly proves the claim.

\bigskip

The symmetric properties, as described in point 3, have analogous proofs.
\end{proof}

\begin{lemma}[Pair types]
\label{lem-proof:pair-types} For all environment $\Gamma$, for all $M$, $N$, $T$, $c$:
\begin{enumerate}
\item For all $T_1$, $T_2$, if $\teqTc{\Gamma}{\Delta}{\E}{\E'}{M}{N}{T_1*T_2}{c}$ then
  \begin{itemize}
    \item either there exist $M_1$, $M_2$, $N_1$, $N_2$, $c_1$, $c_2$ such that
    $M = \PAIR{M_1}{M_2}$, $N = \PAIR{N_1}{N_2}$, $c = c_1 \cup c_2$, and
    $\teqTc{\Gamma}{\Delta}{\E}{\E'}{M_1}{N_1}{T_1}{c_1}$ and $\teqTc{\Gamma}{\Delta}{\E}{\E'}{M_2}{N_2}{T_2}{c_2}$;
    \item or $M$ and $N$ are variables.
  \end{itemize}

\item For all $M_1$, $M_2$, if $T \subtyp \L$ and $\teqTc{\Gamma}{\Delta}{\E}{\E'}{\PAIR{M_1}{M_2}}{N}{T}{c}$ then
either $T = \L$ or there exists $T_1$, $T_2$ such that $T = T_1*T_2$.

\item For all $M_1$, $M_2$, if $\teqTc{\Gamma}{\Delta}{\E}{\E'}{\PAIR{M_1}{M_2}}{N}{\L}{c}$ then
there exist $N_1$, $N_2$, $c_1$, $c_2$, such that $c = c_1 \cup c_2$, $N = \PAIR{N_1}{N_2}$,
$\teqTc{\Gamma}{\Delta}{\E}{\E'}{M_1}{N_1}{\L}{c_1}$ and $\teqTc{\Gamma}{\Delta}{\E}{\E'}{M_2}{N_2}{\L}{c_2}$.

\item The symmetric properties to the previous two points (\ie when the term on the right is a pair) also hold.
\end{enumerate}
\end{lemma}
\begin{proof}
Let us prove point 1 by induction on the typing derivation $\teqTc{\Gamma}{\Delta}{\E}{\E'}{M}{N}{T_1*T_2}{c}$.
Because of the form of the type, and by well-formedness of $\Gamma$,
the only possibilities for the last rule applied are 
\TVar, \TPair, and \TSub.

The claim clearly holds in the \TVar and \TPair cases.


In the \TSub case, $\teqTc{\Gamma}{\Delta}{\E}{\E'}{M}{N}{T'}{c}$ for some $T' \subtyp T_1*T_2$, and by
Lemma~\ref{lem-proof:subtyping}, $T' = T_1'*T_2'$ for some $T_1'$, $T_2'$ such that $T_1' \subtyp T_1$ and $T_2' \subtyp T_2$.
Therefore, by applying the induction hypothesis to
$\teqTc{\Gamma}{\Delta}{\E}{\E'}{M}{N}{T_1'*T_2'}{c}$, $M$ and $N$ are either two variables, and the claim holds;
or two pairs, \ie there exist $M_1$, $M_2$, $N_1$, $N_2$, $c_1$, $c_2$ such that $M = \PAIR{M_1}{M_2}$,
$N = \PAIR{N_1}{N_2}$, $c = c_1 \cup c_2$, and for $i \in \{1,2\}$,
$\teqTc{\Gamma}{\Delta}{\E}{\E'}{M_i}{N_i}{T_i'}{c_i}$.
Hence, by subtyping, $\teqTc{\Gamma}{\Delta}{\E}{\E'}{M_i}{N_i}{T_i}{c_i}$,
and the claim holds.

\bigskip

We now prove point 2 by induction on the proof of
$\teqTc{\Gamma}{\Delta}{\E}{\E'}{\PAIR{M_1}{M_2}}{N}{T}{c}$.
Because of the form of the terms, the last rule applied can only be
\THigh, \TOr, \TPair, \TEncH, \TEncL, \TAencH, \TAencL, \TLRp, \TLRLp or \TSub.

The \THigh, \TLRp, 
and \TOr cases are actually impossible by Lemma~\ref{lem-proof:subtyping}, since $T \subtyp \L$.

The \TLRLp and case is also impossible, since by Lemma~\ref{lem-proof:lr-ground} it would imply that $\PAIR{M_1}{M_2}$ is either a variable or a nonce.

The \TEncH, \TEncL, \TAencH, \TAencL cases are impossible, since the condition of the rule
would then imply $\teqTc{\Gamma}{\Delta}{\E}{\E'}{\PAIR{M_1}{M_2}}{N}{\encT{T}{k}}{c'}$ (or $\aencT{T}{k}$) for some $T$, $k$, $c'$, which is not possible by Lemma~\ref{lem-proof:enc-types}.

In the \TPair case, the claim clearly holds.


Finally, in the \TSub case, we have $\teqTc{\Gamma}{\Delta}{\E}{\E'}{\PAIR{M_1}{M_2}}{N}{T'}{c}$ for some $T'$ such that $T'\subtyp T$.
By transitivity of $\subtyp$, $T'\subtyp \L$, and we may apply the induction hypothesis to
$\teqTc{\Gamma}{\Delta}{\E}{\E'}{\PAIR{M_1}{M_2}}{N}{T'}{c}$.
Hence either $T' = \L$ or $T' = T_1'*T_2'$ for some $T_1'$, $T_2'$.
By Lemma~\ref{lem-proof:subtyping}, this implies in the first case that $T = \L$ and in the second case that $T = \L$ or $T$ is also a pair type ($T \neq \H$ and $T \neq \S$ in both cases, since we already know that $T \subtyp \L$).

\bigskip

We prove point 3 as a consequence of the first two points, by induction on the derivation of $\teqTc{\Gamma}{\Delta}{\E}{\E'}{\PAIR{M_1}{M_2}}{N}{\L}{c}$.
The last rule in this derivation can only be \TEncH, \TEncL, \TAencH, \TAencL, \TLRp, \TLRLp or \TSub by the form of the types and terms,
but similarly to the previous point \TEncH, \TEncL, \TAencH, \TAencL, \TLRp and \TLRLp are actually not possible.

Hence the last rule of the derivation is \TSub.
We have $\teqTc{\Gamma}{\Delta}{\E}{\E'}{\PAIR{M_1}{M_2}}{N}{T}{c}$ for some $T$ such that $T\subtyp \L$.
By point 2, either $T = \L$ or there exist $T_1$, $T_2$ such that $T = T_1*T_2$.
If $T = \L$, we have a shorter proof of $\teqTc{\Gamma}{\Delta}{\E}{\E'}{\PAIR{M_1}{M_2}}{N}{\L}{c}$
and we conclude by the induction hypothesis.
Otherwise, since $T \subtyp \L$, by Lemma~\ref{lem-proof:subtyping}, $T_1 \subtyp \L$ and $T_2 \subtyp \L$.
Moreover by the first property,
there exist $N_1$, $N_2$, $c_1$, $c_2$ such that $N = \PAIR{N_1}{N_2}$, $c = c_1 \cup c_2$,
$\teqTc{\Gamma}{\Delta}{\E}{\E'}{M_1}{N_1}{T_1}{c_1}$, and $\teqTc{\Gamma}{\Delta}{\E}{\E'}{M_2}{N_2}{T_2}{c_2}$.

Thus by subtyping, $\teqTc{\Gamma}{\Delta}{\E}{\E'}{M_1}{N_1}{\L}{c_1}$ and
$\teqTc{\Gamma}{\Delta}{\E}{\E'}{M_2}{N_2}{\L}{c_2}$,
which proves the claim.
\end{proof}

\begin{lemma}[Type for keys, nonces and constants]
\label{lem-proof:type-key-nonce}
For all environment $\Gamma$, for all messages $M$, $N$, for all key $k\in \K$, for all nonce or constant 
$n\in\N\cup\CST$, for all $c$, $l$,
the following properties hold:
\begin{enumerate}
\item For all $T$, if $\teqTc{\Gamma}{\Delta}{\E}{\E'}{M}{N}{\skey{l}{T}}{c}$,
then $c = \emptyset$; and either $M = N$ are in $\K$ and $\Gamma(M) = \skey{l}{T}$; or $M$ and $N$ are variables.

\item If $l \in \{\L, \S\}$, and $\teqTc{\Gamma}{\Delta}{\E}{\E'}{k}{N}{l}{c}$,
then $N = k$, $c=\emptyset$, and there exists $T$ such that $\Gamma(k) = \skey{l}{T}$.

\item If $\teqTc{\Gamma}{}{}{}{\PUBK{k}}{N}{\L}{c}$, then
$k \in\dom{\Gamma}$ and $N=\PUBK{k}$.

\item If $\teqTc{\Gamma}{}{}{}{\VK{k}}{N}{\L}{c}$, then
$k \in\dom{\Gamma}$ and $N=\VK{k}$.

\item If $\teqTc{\Gamma}{\Delta}{\E}{\E'}{n}{N}{\S}{c}$,
then $n\in\BN$, $c = \emptyset$ and either $\Gamma(n) = \noncetypelab{\S}{1}{n}$ or $\noncetypelab{\S}{\infty}{n}$.

\item If $\teqTc{\Gamma}{\Delta}{\E}{\E'}{n}{N}{\L}{c}$,
then $N = n$, $c = \emptyset$, and either there exists $a\in\{1,\infty\}$ such that $\Gamma(n) = \noncetypelab{\L}{a}{n}$,
or $n\in\FN\cup\CST$.

\item The symmetric properties to the previous five points (\ie with $k$ (resp. $\PUBK{k}$, $\VK{k}$, $n$) on the right) also hold.
\end{enumerate}
\end{lemma}
\begin{proof}
Point 1 is easily proved by induction on the derivation of
$\teqTc{\Gamma}{\Delta}{\E}{\E'}{M}{N}{\skey{l}{T}}{c}$.
Indeed, by the form of the type the last rule can only be \TKey, \TVar, or \TSub.
In the \TKey and \TVar cases the claim clearly holds.
In the \TSub case, by Lemma~\ref{lem-proof:subtyping}, $\skey{l}{T}$ is its only subtype, thus there exists a shorter
derivation of $\teqTc{\Gamma}{\Delta}{\E}{\E'}{M}{N}{\skey{l}{T}}{c}$, an the claim holds by the induction hypothesis.

\bigskip

We prove point 2 by induction on the derivation of $\teqTc{\Gamma}{\Delta}{\E}{\E'}{k}{N}{l}{c}$.
Because of the form of the terms and type, and by well-formedness of $\Gamma$,
the last rule applied can only be \TEncH, \TEncL, \TAencH, \TAencL, \TLRp, \TLRLp or \TSub.

The \TEncH, \TEncL, \TAencH, \TAencL cases are impossible since they would imply that $\teqTc{\Gamma}{\Delta}{\E}{\E'}{k}{N}{\encT{T'}{k'}}{c'}$
(or $\aencT{T'}{k'}$) for some $T'$, $k'$, $c'$, which is impossible by Lemma~\ref{lem-proof:enc-types}.

The \TLRp and \TLRLp cases are impossible. Indeed in these cases, we have $\teqTc{\Gamma}{\Delta}{\E}{\E'}{k}{N}{\LRTn{l}{a}{m}{l'}{a}{n}}{\emptyset}$ for some $m$, $n$.
Lemma~\ref{lem-proof:lr-ground} then implies that $m = k$ (and $n = N$), which is contradictory.

Finally, in the \TSub case, we have $\teqTc{\Gamma}{\Delta}{\E}{\E'}{k}{N}{T}{c}$ for some $T$ such that $T\subtyp l$.
By Lemma~\ref{lem-proof:subtyping}, this implies that $T$ is either a pair type,
a key type, or $l$. The first case is impossible by Lemma~\ref{lem-proof:pair-types}, since $k\in\K$.
The last case is trivial by the induction hypothesis.
Only the case where $T = \skey{l}{T'}$ (for some $T'$) remains.
By point 1, in that case, since $k$ is not a variable, we have 
$\Gamma(k) = \skey{l}{T'}$ and $k = N$, and therefore the claim holds.

\bigskip

Similarly, we prove point 3 by induction on the derivation of $\teqTc{\Gamma}{\Delta}{\E}{\E'}{\PUBK{k}}{N}{\L}{c}$.
Because of the form of the terms and type, and by well-formedness of $\Gamma$,
the last rule applied can only be \TEncH, \TEncL, \TAencH, \TAencL, \TLRp, \TLRLp, \TPubkey or \TSub.

The \TEncH, \TEncL, \TAencH, \TAencL cases are impossible since they would imply that $\teqTc{\Gamma}{\Delta}{\E}{\E'}{\PUBK{k}}{N}{\encT{T'}{k'}}{c'}$
(or $\aencT{T'}{k'}$) for some $T'$, $k'$, $c'$, which is impossible by Lemma~\ref{lem-proof:enc-types}.

The \TLRp and \TLRLp cases are impossible. Indeed in these cases, we have $\teqTc{\Gamma}{\Delta}{\E}{\E'}{\PUBK{k}}{N}{\LRTn{l}{a}{m}{l'}{a}{n}}{\emptyset}$ for some $m$, $n$.
Lemma~\ref{lem-proof:lr-ground} then implies that $m = \PUBK{k}$ (and $n = N$), which is contradictory.

In the \TSub case, we have $\teqTc{\Gamma}{\Delta}{\E}{\E'}{\PUBK{k}}{N}{T}{c}$ for some $T$ such that $T\subtyp l$.
By Lemma~\ref{lem-proof:subtyping}, this implies that $T$ is either a pair type,
a key type, or $l$. Just as in the previous point, the first case is impossible and the last one is trivial.
The case where $T = \skey{l}{T'}$ (for some $T'$) is also impossible by point 1, since $\PUBK{k}$ is not in $\K\cup\X$.

Finally in the \TPubkey case, the claim clearly holds.

\medskip

Point 4 has a similar proof to point 3.

\bigskip

The remaining properties have similar proofs to point 2.
For point 5, \ie if $\teqTc{\Gamma}{\Delta}{\E}{\E'}{n}{t}{\S}{c}$, only the \TNonce, \TSub, and \TLRp cases are possible.
The claim clearly holds in the \TNonce case.

In the \TLRp case, we have $\teqTc{\Gamma}{\Delta}{\E}{\E'}{n}{t}{\LRTn{\S}{a}{m}{\S}{a}{p}}{\emptyset}$ for some $m$, $p$.
Lemma~\ref{lem-proof:lr-ground} then implies that $m = n$, and $p = t$, and $\Gamma(n)=\noncetypelab{\S}{a}{m}$, and 
$\Gamma(p)=\noncetypelab{\S}{a}{p}$, which proves the claim.

In the \TSub case, $\teqTc{\Gamma}{\Delta}{\E}{\E'}{n}{t}{T}{c}$ for some $T \subtyp \S$, thus by Lemma~\ref{lem-proof:subtyping} $T$ is either a pair type (impossible by Lemma~\ref{lem-proof:pair-types}), a key type (impossible by point 1), or $\S$ (and we conclude by the induction hypothesis).

\medskip

For point 6, similarly, only the \TNonceL, \TCst, \TSub, \TLRLp cases are possible.
The \TSub case is proved in the same way as for the third property.
The \TLRLp case is proved similarly to the previous point.
Finally the claim clearly holds in the \TNonceL and \TCst cases.

\bigskip

The symmetric properties, as described in point 7, have analogous proofs.
\end{proof}

\begin{lemma}[Type $\L$ implies same head symbol]
\label{lem-proof:l-same-head}
For all $\Gamma$, $M$, $N$, $c$,
if $\teqTc{\Gamma}{\Delta}{\E}{\E'}{M}{N}{\L}{c}$ then either $M$ and $N$ have the same head symbol and use the same key if this symbol is $\ENCNA$, $\AENCNA$ or $\SIGN$; or $M$, $N$ both are variables. 
\end{lemma}
\begin{proof}
We prove a slightly more general property: for all $T \subtyp \L$,
if $\teqTc{\Gamma}{\Delta}{\E}{\E'}{M}{N}{T}{c}$ then either $M$ and $N$ have the same head symbol and use the same key if 
this symbol is $\ENCNA$, $\AENCNA$ or $\SIGN$; or $M$, $N$ both are variables.

This is proved by induction on the typing derivation.
Many of the cases for the last rule applied are immediate, since they directly state that the two terms have the same
head symbol (with the same key) or are variables.
This covers rules \TNonceL, \TCst, \TPubkey, \TVkey, \TKey, \TVar, \TPair, \TEnc, \TAenc, \TSignH, \TSignL, \THash, \THashL, \TLRVar.
Among the remaining cases, some are also immediate thanks to the assumption that $T\subtyp \L$, as they contradict it
(which we prove using Lemma~\ref{lem-proof:subtyping}).
This covers rules \TNonce, \THigh, \TOr, \TLRone, \TLRinf, \TLRp.
Moreover in the case of rule \TSub the claim follows directly from the application of the induction hypothesis to the premise of the rule.

Only the cases of rules \TEncH, \TEncL, \TAencH, \TAencL, and \TLRLp remain.
In the \TEncH case, $\teqTc{\Gamma}{\Delta}{\E}{\E'}{M}{N}{\encT{T}{k}}{c}$ for some $T$, $k$, and therefore by Lemma~\ref{lem-proof:enc-types} $M$ and $N$ are either two variables or some terms encrypted with $k$, and in both cases
the claim holds.
The \TEncL, \TAencH, \TAencL cases are similar, using Lemma~\ref{lem-proof:enc-types}.
In the \TLRLp case, $\teqTc{\Gamma}{\Delta}{\E}{\E'}{M}{N}{\LRTn{\L}{a}{m}{\L}{a}{m}}{c}$ for some $m$.
Thus by Lemma~\ref{lem-proof:lr-ground}, either $M$ and $N$ are two variables or $M=N=m$, and in any case the claim holds, which concludes this proof.
\end{proof}

\begin{lemma}[Application of destructors]
\label{lem-proof:app-dest}
For all $\Gamma$, $y$, such that $y\notin\dom{\Gamma}$, for all $d$, $T$, $T'$, $c$, for all ground messages $M$, $N$,
if $\tDestnew{\Gamma, y:T}{d(y)}{T'}$ and $\teqTc{\Gamma}{\Delta}{\E}{\E'}{M}{N}{T}{c}$, then:
\begin{enumerate}
\item\label{item:app_desta} We have:
\[\eval{(d(M))} = \bot \Longleftrightarrow \eval{(d(N))} = \bot\]
\item\label{item:app_destb} And if $\eval{(d(M))} \neq \bot$ then there exists $c' \subseteq c$ such that 
\[\teqTc{\Gamma}{\Delta}{\E}{\E'}{\eval{(d(M))}}{\eval{(d(N))}}{T'}{c'}\]
\end{enumerate}
\end{lemma}
\begin{proof}

We distinguish four cases for $d$.

\begin{itemize}
\item \case{$d = \DEC{\cdot}{k}$.}
We know that $\tDestnew{\Gamma, y:T}{d(y)}{T'}$, which can be proved using either rule \DDecH, rule \DDecL, or rule \DDecT. In the first two cases, $T=\L$, and in the last case $T = \encT{T'}{k}$.

\begin{itemize}
\item Let us prove \ref{item:app_desta}) by contraposition.
Assume $\eval{d(M)} \neq \bot$.
Hence, $M=\ENC{M'}{k}$ for some $M'$. Lemma~\ref{lem-proof:l-same-head} in the \DDecH and \DDecL cases
(where $\teqTc{\Gamma}{\Delta}{\E}{\E'}{M}{N}{\L}{c}$), and Lemma~\ref{lem-proof:enc-types} in the \DDecT case
(where $\teqTc{\Gamma}{\Delta}{\E}{\E'}{M}{N}{\encT{T'}{k}}{c}$), guarantee that there exists $N'$ such that
$N = \ENC{N'}{k}$.
Therefore $\eval{d(N)} \neq \bot$ which proves the first direction of \ref{item:app_desta}). The other direction is analogous.

\item Moreover, still assuming $\eval{d(M)} \neq \bot$, and keeping the notations from the previous point, we have $\eval{d(M)} = M'$ and $\eval{d(N)} = N'$.
The destructor typing rule applied to prove $\tDestnew{\Gamma, y:\L}{d(y)}{T'}$ can be \DDecT, \DDecH, or \DDecL.
\begin{itemize}
\item \case{In the \DDecT case} we have $T=\encT{T'}{k}$ and therefore $\teqTc{\Gamma}{\Delta}{\E}{\E'}{M}{N}{\encT{T'}{k}}{c}$.
Lemma~\ref{lem-proof:enc-types} (point 1) then guarantees that $\teqTc{\Gamma}{\Delta}{\E}{\E'}{M'}{N'}{T'}{c}$, which proves point \ref{item:app_destb}).

\item \case{In the \DDecH case} we have $T=\L$ and $\Gamma(k)=\skey{\S}{T'}$.
Thus, we have $\teqTc{\Gamma}{\Delta}{\E}{\E'}{\ENC{M'}{k}}{\ENC{N'}{k}}{\L}{c}$, and by Lemma~\ref{lem-proof:enc-types} (point 5),
we know that there exists $c'$ such that $c = c' \cup \{M\eqC N\}$ and
$\teqTc{\Gamma}{\Delta}{\E}{\E'}{M'}{N'}{T'}{c'}$, which proves point \ref{item:app_destb}).

\item \case{In the \DDecL case} we have $T = T' =\L$ and there exists $T''$ such that $\Gamma(k)=\skey{\L}{T''}$.
Thus, we have $\teqTc{\Gamma}{\Delta}{\E}{\E'}{\ENC{M'}{k}}{\ENC{N'}{k}}{\L}{c}$, and by Lemma~\ref{lem-proof:enc-types} (point 5),
we know that
$\teqTc{\Gamma}{\Delta}{\E}{\E'}{M'}{N'}{\L}{c}$, which proves point \ref{item:app_destb}).
\end{itemize}
In all cases, point \ref{item:app_destb}) holds, which concludes this case.
\end{itemize}

\item \case{$d = \ADEC{\cdot}{k}$.}
We know that $\tDestnew{\Gamma, y:T}{d(y)}{T'}$, which can be proved using either rule \DAdecH, rule \DAdecL, or rule \DAdecT. In the first two cases, $T=\L$, and in the last case $T = \aencT{T'}{k}$.

\begin{itemize}
\item Let us prove \ref{item:app_desta}) by contraposition.
Assume $\eval{d(M)} \neq \bot$.
Hence, $M=\AENC{M'}{\PUBK{k}}$ for some $M'$. Lemma~\ref{lem-proof:l-same-head} in the \DAdecH and \DAdecL cases
(where $\teqTc{\Gamma}{\Delta}{\E}{\E'}{M}{N}{\L}{c}$), and Lemma~\ref{lem-proof:enc-types} in the \DAdecT case
(where $\teqTc{\Gamma}{\Delta}{\E}{\E'}{M}{N}{\aencT{T'}{k}}{c}$), guarantee that there exists $N'$ such that
$N = \AENC{N'}{\PUBK{k}}$.
Therefore $\eval{d(N)} \neq \bot$ which proves the first direction of \ref{item:app_desta}). The other direction is analogous.

\item Moreover, still assuming $\eval{d(M)} \neq \bot$, and keeping the notations from the previous point, we have $\eval{d(M)} = M'$ and $\eval{d(N)} = N'$.
The destructor typing rule applied to prove $\tDestnew{\Gamma, y:\L}{d(y)}{T'}$ can be \DAdecT, \DAdecH, or \DAdecL.

\begin{itemize}
\item \case{In the \DAdecT case} we have $T=\aencT{T'}{k}$ and thus $\teqTc{\Gamma}{\Delta}{\E}{\E'}{M}{N}{\aencT{T'}{k}}{c}$.
Lemma~\ref{lem-proof:enc-types} then guarantees that $\teqTc{\Gamma}{\Delta}{\E}{\E'}{M'}{N'}{T'}{c}$ which proves the claim.

\item \case{In the \DAdecH case} we have $T = \L$, and there exists $T''$ such that $T' = T'' \orT \L$ and 
$\Gamma(k)=\skey{\S}{T''}$.
Thus, we have $\teqTc{\Gamma}{\Delta}{\E}{\E'}{\AENC{M'}{\PUBK{k}}}{\AENC{N'}{\PUBK{k}}}{\L}{c}$.
Hence by Lemma~\ref{lem-proof:enc-types} (point 5),
we know that either there exists $c'$ such that $c = c' \cup \{M\eqC N\}$ and
$\teqTc{\Gamma}{\Delta}{\E}{\E'}{M'}{N'}{T''}{c'}$; or $\teqTc{\Gamma}{\Delta}{\E}{\E'}{M'}{N'}{\L}{c}$.
By rule \TOr, we then have $\teqTc{\Gamma}{\Delta}{\E}{\E'}{M'}{N'}{T''\orT \L}{c'}$ (resp. $c$),
which proves point \ref{item:app_destb}).

\item \case{In the \DAdecL case} we have $T = T' =\L$ and there exists $T''$ such that $\Gamma(k)=\skey{\L}{T''}$.
Thus, we have $\teqTc{\Gamma}{\Delta}{\E}{\E'}{\AENC{M'}{\PUBK{k}}}{\AENC{N'}{\PUBK{k}}}{\L}{c}$,
and by Lemma~\ref{lem-proof:enc-types} (point 5), we know that
$\teqTc{\Gamma}{\Delta}{\E}{\E'}{M'}{N'}{\L}{c}$, which proves point \ref{item:app_destb}).
\end{itemize}
In all cases, point \ref{item:app_destb}) holds, which concludes this case.
\end{itemize}

\item \case{$d = \CHECK{\cdot}{\VK{k}}$.}
We know that $\tDestnew{\Gamma, y:T}{d(y)}{T'}$, which can be proved using either rule \DCheckH or rule \DCheckL.
In both cases, $T=\L$.

\begin{itemize}
\item Let us prove \ref{item:app_desta}) by contraposition.
Assume $\eval{d(M)} \neq \bot$.
Hence, $M=\SIGN{M'}{k}$ for some $M'$. Lemma~\ref{lem-proof:l-same-head}
(applied to $\teqTc{\Gamma}{\Delta}{\E}{\E'}{M}{N}{\L}{c}$) guarantees that there exists $N'$ such that
$N = \SIGN{N'}{k}$.
Therefore $\eval{d(N)} \neq \bot$ which proves the first direction of \ref{item:app_desta}). The other direction is analogous.

\item Moreover, still assuming $\eval{d(M)} \neq \bot$, and keeping the notations from the previous point, we have $\eval{d(M)} = M'$ and $\eval{d(N)} = N'$.
The destructor typing rule applied to prove $\tDestnew{\Gamma, y:\L}{d(y)}{T'}$ can be \DCheckH or \DCheckL.

\begin{itemize}
\item \case{In the \DCheckH case} we have $T = \L$, and $\Gamma(k)=\skey{\S}{T'}$.
Thus we have $\teqTc{\Gamma}{\Delta}{\E}{\E'}{\SIGN{M'}{k}}{\SIGN{N'}{k}}{\L}{c}$.
Hence by Lemma~\ref{lem-proof:sign-types} (point 2),
we know that there exist $c'$, $c''$ such that $c = c' \cup c'' \cup \{M\eqC N\}$,
$\teqTc{\Gamma}{\Delta}{\E}{\E'}{M'}{N'}{T'}{c'}$, and $\teqTc{\Gamma}{\Delta}{\E}{\E'}{M'}{N'}{\L}{c''}$.
This proves point \ref{item:app_destb}).

\item \case{In the \DCheckL case} we have $T = \L$, and there exists $T''$ such that $\Gamma(k)=\skey{\L}{T''}$.
Hence by Lemma~\ref{lem-proof:sign-types} (point 2),
we know that 
$\teqTc{\Gamma}{\Delta}{\E}{\E'}{M'}{N'}{\L}{c}$.
This proves point \ref{item:app_destb}).
\end{itemize}
In all cases, point \ref{item:app_destb}) holds, which concludes this case.
\end{itemize}

\item \case{$d = \FSTNA$.}
We know that $\tDestnew{\Gamma, y:T}{d(y)}{T'}$, which can be proved using either rule \DFst or \DFstL.
In the first case, $T = T_1 * T_2$ is a pair type, and in the second case $T = \L$.

\begin{itemize}
\item We prove \ref{item:app_desta}) by contraposition.
Assume $\eval{d(M)} \neq \bot$.
Hence, $M=\PAIR{M_1}{M_2}$ for some $M_1, M_2$.
Thus, by applying Lemma~\ref{lem-proof:pair-types} to $\teqTc{\Gamma}{\Delta}{\E}{\E'}{M}{N}{T}{c}$,
in any case we know that there exist $N_1$, $N_2$ such that $N = \PAIR{N_1}{N_2}$.
Therefore $\eval{d(N)} \neq \bot$ which proves the first direction of \ref{item:app_desta}). The other direction is analogous.

\item Moreover, still assuming $\eval{d(M)} \neq \bot$, and keeping the notations from the previous point, we have $\eval{d(M)} = M_1$ and $\eval{d(N)} = N_1$.
In addition, we know that $\tDestnew{\Gamma, y:T}{d(y)}{T'}$, which can be proved using either rule \DFst or \DFstL.
Lemma~\ref{lem-proof:pair-types}, which we applied in the previous point, also implies that
there exist $c_1$, $c_2$, such that $c=c_1\cup c_2$ and for $i\in\{1,2\}$,
$\teqTc{\Gamma}{\Delta}{\E}{\E'}{M_i}{N_i}{T_i}{c_i}$ (in the \DFst case) or
$\teqTc{\Gamma}{\Delta}{\E}{\E'}{M_i}{N_i}{\L}{c_i}$ (in the \DFstL case).

We distinguish two cases for the rule applied to prove $\tDestnew{\Gamma, y:T}{d(y)}{T'}$.

\begin{itemize}
\item \case{\DFst:} Then $T=T_1*T_2$ and $T' = T_1$, and $\teqTc{\Gamma}{\Delta}{\E}{\E'}{M_1}{N_1}{T_1}{c_1 (\subseteq c)}$
proves~\ref{item:app_destb}).
\item \case{\DFstL:} Then $T = T' = \L$, and $\teqTc{\Gamma}{\Delta}{\E}{\E'}{M_1}{N_1}{\L}{c_1 (\subseteq c)}$
proves~\ref{item:app_destb}).
\end{itemize}
In both cases, point \ref{item:app_destb}) holds, which concludes this case.
\end{itemize}

\item \case{$d = \SNDNA$.}
This case is similar to the previous one.
\end{itemize}
\end{proof}

\begin{lemma}[$\L$ type is preserved by attacker terms]
\label{lem-proof:l-type-recipe}
For all $\Gamma$,
for all frames $\psi = \NEWN{\E_\Gamma}.\phi$ and $\psi' = \NEWN{\E_\Gamma}.\phi'$ with
$\teqTc{\Gamma}{\Delta}{\E}{\E'}{\phi}{\phi'}{\L}{c}$, for all attacker term $R$ such that $\var{R} \subseteq \dom{\phi}$,\\
either there exists $c' \subseteq c$ such that
\[
 \teqTc{\Gamma}{\Delta}{\E}{\E'}{\eval{R \phi}}{\eval{R \phi'}}{\L}{c'}
\]
or
\[\eval{R \phi} = \eval{R \phi'} = \bot.\]
\end{lemma}
\begin{proof}
Let us recall that $\E_\Gamma$ denotes the set of names in $\Gamma$.

We show this property by induction over the attacker term $R$.

Induction Hypothesis: the statement holds for all subterms of $R$.
There are several cases for $R$. The base cases are the cases where $R$ is a variable, a name in $\FN$ or a constant in $\CST$.
\begin{enumerate}

\item \case{$R=x$}
 Since $\var{R}\subseteq \dom{\phi}$, we have $x \in \dom{\phi} = \dom{\phi'}$, hence $\eval{R \phi} = \phi(x)$ and $\eval{R \phi'}= \phi'(x)$.
 Since $\teqTc{\Gamma}{\Delta}{\E}{\E'}{\phi}{\phi'}{\L}{c}$, we have $\teqTc{\Gamma}{\Delta}{\E}{\E'}{\phi(x)}{\phi'(x)}{\L}{c_x}$ for some $c_x \subseteq c$, and the claim holds.

\item \case{$R=a$}
 with $a\in \CST\cup\FN$. Then $\eval{R\phi}=\eval{R\phi'}=a$ and by rule \TCst, we have $\teqTc{\Gamma}{\Delta}{\E}{\E'}{a}{a}{\L}{\emptyset}$. Hence the claim holds.

\item \case{$R=\PUBK{K}$}
 We apply the induction hypothesis to $K$ and distinguish three cases.
 \begin{enumerate}
 \item \case{If $\eval{K\phi} = \bot$} then $\eval{K\phi'} = \bot$, hence $\eval{R\phi} = \eval{R\phi'} = \bot$. 
 \item \case{If $\eval{K\phi} \neq \bot$ and is not a key}
	then $\eval{K\phi'} \neq \bot$ (by IH), and
  by IH we have $\teqTc{\Gamma}{\Delta}{\E}{\E'}{\eval{K \phi}}{\eval{K\phi'}}{\L}{c'}$ for some $c' \subseteq c$.
  Then by Lemma~\ref{lem-proof:l-same-head},
   $\eval{K\phi'}$ is not a key either.
	Hence $\eval{R\phi} = \eval{R\phi'} = \bot$.
 \item \case{If $\eval{K\phi}$ is a key}, then by IH there exists $c'\subseteq c$ such that
  $\teqTc{\Gamma}{\Delta}{\E}{\E'}{\eval{K \phi}}{\eval{K\phi'}}{\L}{c'}$.
  Hence by Lemma~\ref{lem-proof:type-key-nonce}
  $\eval{K\phi} = \eval{K\phi'}$, and
  $\Gamma(\eval{K\phi}) = \Gamma(\eval{K\phi'}) = \skey{\L}{T}$ for some $T$.
  Therefore by rule \TPubkey,
  $\teqTc{\Gamma}{\Delta}{\E}{\E'}{\eval{R \phi}}{\eval{R\phi'}}{\L}{\emptyset}$ and the claim holds.
 \end{enumerate}

\item \case{$R=\VK{K}$}
 We apply the induction hypothesis to $K$ and distinguish three cases.
 \begin{enumerate}
 \item \case{If $\eval{K\phi} = \bot$} then $\eval{K\phi'} = \bot$, hence $\eval{R\phi} = \eval{R\phi'} = \bot$. 
 \item \case{If $\eval{K\phi} \neq \bot$ and is not a key}
	then $\eval{K\phi'} \neq \bot$ (by IH), and by IH we have
  $\teqTc{\Gamma}{\Delta}{\E}{\E'}{\eval{K \phi}}{\eval{K\phi'}}{\L}{c'}$ for some $c' \subseteq c$.
  Then, by Lemma~\ref{lem-proof:l-same-head},
  $\eval{K\phi'}$ is not a key either.
	Hence $\eval{R\phi} = \eval{R\phi'} = \bot$.
 \item \case{If $\eval{K\phi}$ is a key}, then by IH there exists $c'\subseteq c$ such that
  $\teqTc{\Gamma}{\Delta}{\E}{\E'}{\eval{K \phi}}{\eval{K\phi'}}{\L}{c'}$.
  Hence by Lemma~\ref{lem-proof:type-key-nonce}
  $\eval{K\phi} = \eval{K\phi'}$, and
  $\Gamma(\eval{K\phi}) = \Gamma(\eval{K\phi'}) = \skey{\L}{T}$ for some $T$.
  Therefore by rule \TVkey,
  $\teqTc{\Gamma}{\Delta}{\E}{\E'}{\eval{R \phi}}{\eval{R\phi'}}{\L}{\emptyset}$ and the claim holds.
 \end{enumerate}

\item \case{$R = \PAIR{R_1}{R_2}$}
 where $R_1$ and $R_2$ are also attacker terms.
 We then apply the induction hypothesis to the same frames and $R_1$, $R_2$.
 We distinguish two cases:
 \begin{enumerate}
 \item \case{$\eval{R_1 \phi} = \bot \vee \eval{R_2 \phi} = \bot$}
  In this case we also have $\eval{R_1 \phi'} = \bot \vee \eval{R_2\phi'} = \bot$ and
  therefore $\eval{R \phi} = \eval{R \phi'} = \bot$.
 \item \case{$\eval{R_1 \phi} \neq \bot \wedge \eval{R_2 \phi} \neq \bot$}
  In this case, by the induction hypothesis, we also have $\eval{R_1 \phi'} \neq \bot \wedge \eval{R_2\phi'} \neq \bot$,
  and we also know that there exist $c_1 \subseteq c$ and $c_2 \subseteq c$ such that
  $\teqTc{\Gamma}{\Delta}{\E}{\E'}{\eval{R_1 \phi}}{\eval{R_1\phi'}}{\L}{c_1}$ and
  $\teqTc{\Gamma}{\Delta}{\E}{\E'}{\eval{R_2 \phi}}{\eval{R_2\phi'}}{\L}{c_2}$.
    
  Thus, by the rule \TPair followed by \TSub,
  $\teqTc{\Gamma}{\Delta}{\E}{\E'}{\eval{R \phi}}{\eval{R \phi'}}{\L}{c_1 \cup c_2}$.
  Since $c_1 \cup c_2 \subseteq c$, this proves the case.
 \end{enumerate}

\item \case{$R = \ENC{S}{K}$}
\label{proof:case_enc}
 We apply the induction hypothesis to $K$ and distinguish three cases.
 \begin{enumerate}
 \item \case{If $\eval{K\phi} = \bot$} then $\eval{K\phi'} = \bot$, hence $\eval{R\phi} = \eval{R\phi'} = \bot$. 
 \item \case{If $\eval{K\phi} \neq \bot$ and is not a key}
	then $\eval{K\phi'} \neq \bot$ (by IH), and by IH we have
  $\teqTc{\Gamma}{\Delta}{\E}{\E'}{\eval{K \phi}}{\eval{K\phi'}}{\L}{c'}$ for some $c' \subseteq c$.
  Then, by Lemma~\ref{lem-proof:l-same-head},
  $\eval{K\phi'}$ is not a key either.
	Hence $\eval{R\phi} = \eval{R\phi'} = \bot$.
 \item \case{If $\eval{K\phi}$ is a key}, then by IH there exists $c' \subseteq c$ such that
  $\teqTc{\Gamma}{\Delta}{\E}{\E'}{\eval{K \phi}}{\eval{K\phi'}}{\L}{c'}$.
  Hence by Lemma~\ref{lem-proof:type-key-nonce}
  $\eval{K\phi} = \eval{K\phi'}$, and
  $\Gamma(\eval{K\phi}) = \Gamma(\eval{K\phi'}) = \skey{\L}{T}$ for some $T$.
  We then apply the IH to $S$, and either $\eval{S\phi} = \eval{S\phi'} = \bot$,
  in which case $\eval{R\phi} = \eval{R\phi'} = \bot$; or there exists $c'' \subseteq c$ such that
  $\teqTc{\Gamma}{\Delta}{\E}{\E'}{\eval{S \phi}}{\eval{S\phi'}}{\L}{c''}$.
  Since $\eval{R\phi} = \ENC{\eval{S\phi}}{\eval{K\phi}}$, and similarly for $\phi'$,
  by rule \TEnc, we have
  $\teqTc{\Gamma}{\Delta}{\E}{\E'}{\eval{R \phi}}{\eval{R\phi'}}{\encT{\L}{\eval{K\phi}}}{c''}$,
  and then by rule \TEncL,
  $\teqTc{\Gamma}{\Delta}{\E}{\E'}{\eval{R \phi}}{\eval{R\phi'}}{\L}{c''}$.
 \end{enumerate}

\item \case{$R = \AENC{S}{K}$}
\label{proof:case_aenc}
 We apply the induction hypothesis to $K$ and distinguish three cases.
 \begin{enumerate}
 \item \case{If $\eval{K\phi} = \bot$} then $\eval{K\phi'} = \bot$, hence $\eval{R\phi} = \eval{R\phi'} = \bot$. 
 \item \case{If $\eval{K\phi} \neq \bot$ and is not $\PUBK{k}$ for some $k\in\K$}
  then $\eval{K\phi'} \neq \bot$ (by IH), and
  by IH there exists $c'\subseteq c$ such that
  $\teqTc{\Gamma}{\Delta}{\E}{\E'}{\eval{K \phi}}{\eval{K\phi'}}{\L}{c'}$.
  Then, by Lemma~\ref{lem-proof:l-same-head}, $\eval{K\phi'}$ is not a public key either.
  Hence $\eval{R\phi} = \eval{R\phi'} = \bot$.
 \item \case{If $\eval{K\phi}=\PUBK{k}$ for some $k\in\K$}, then by IH there exists $c'\subseteq c$ such that
  $\teqTc{\Gamma}{\Delta}{\E}{\E'}{\PUBK{k}}{\eval{K\phi'}}{\L}{c'}$.
  Thus, by Lemma~\ref{lem-proof:type-key-nonce}, $\eval{K\phi'} = \PUBK{k}$ and $k\in \dom{\Gamma}$.
  We then apply the IH to $S$, and either $\eval{S\phi} = \eval{S\phi'} = \bot$,
  in which case $\eval{R\phi} = \eval{R\phi'} = \bot$; or there exists $c'' \subseteq c$ such that
  $\teqTc{\Gamma}{\Delta}{\E}{\E'}{\eval{S \phi}}{\eval{S\phi'}}{\L}{c''}$.
  Therefore, by rule \TAenc,
  $\teqTc{\Gamma}{\Delta}{\E}{\E'}{\eval{R \phi}}{\eval{R\phi'}}{\aencT{\L}{k}}{c''}$,
  and by rule \TAencL we have
  $\teqTc{\Gamma}{\Delta}{\E}{\E'}{\eval{R \phi}}{\eval{R\phi'}}{\L}{c''}$.
 \end{enumerate}

\item \case{$R = \SIGN{S}{K}$}
\label{proof:case_sign}
 We apply the induction hypothesis to $K$ and distinguish three cases.
 \begin{enumerate}
 \item \case{If $\eval{K\phi} = \bot$} then $\eval{K\phi'} = \bot$, hence $\eval{R\phi} = \eval{R\phi'} = \bot$. 
 \item \case{If $\eval{K\phi} \neq \bot$ and is not a key $k\in\K$}
	then $\eval{K\phi'} \neq \bot$ (by IH), and by IH we have
  $\teqTc{\Gamma}{\Delta}{\E}{\E'}{\eval{K \phi}}{\eval{K\phi'}}{\L}{c'}$ for some $c' \subseteq c$.
  Then, by Lemma~\ref{lem-proof:l-same-head},
  $\eval{K\phi'}$ is not a key either.
  Hence $\eval{R\phi} = \eval{R\phi'} = \bot$.
 \item \case{If $\eval{K\phi}=k$ for some $k\in\K$}, then by IH there exists $c'\subseteq c$ such that
  $\teqTc{\Gamma}{\Delta}{\E}{\E'}{\eval{K \phi}}{\eval{K\phi'}}{\L}{c'}$.
  Hence by Lemma~\ref{lem-proof:type-key-nonce}, $\eval{K\phi'} = k$ and $\Gamma(k) = \skey{\L}{T}$ for some $T$.
  We then apply the IH to $S$, and either $\eval{S\phi} = \eval{S\phi'} = \bot$,
  in which case $\eval{R\phi} = \eval{R\phi'} = \bot$; or there exists $c'' \subseteq c$ such that
  $\teqTc{\Gamma}{\Delta}{\E}{\E'}{\eval{S \phi}}{\eval{S\phi'}}{\L}{c''}$.
  Therefore by rule \TSignL,
  $\teqTc{\Gamma}{\Delta}{\E}{\E'}{\eval{R \phi}}{\eval{R\phi'}}{\L}{c''}$.
 \end{enumerate}

\item \case{$R = \HASH{S}$}
\label{proof:case_hash}
 We apply the induction hypothesis to $S$. We distinguish two cases:
 \begin{enumerate}
 \item \case{$\eval{S \phi} = \bot$}
  In this case we also have $\eval{S \phi'} = \bot$ and
  therefore $\eval{R \phi} = \eval{R \phi'} = \bot$.
 \item \case{$\eval{S \phi} \neq \bot$}
  In this case, by the induction hypothesis, we also have $\eval{S \phi'}\neq \bot$,
  and we also know that there exists $c' \subseteq c$ such that
  $\teqTc{\Gamma}{\Delta}{\E}{\E'}{\eval{S \phi}}{\eval{S\phi'}}{\L}{c'}$
  Thus, by rule \THashL,
  $\teqTc{\Gamma}{\Delta}{\E}{\E'}{\eval{R \phi}}{\eval{R \phi'}}{\L}{c'}$, which proves this case.
 \end{enumerate}

\item \case{$R = \FST{S}$} 
 \label{proof:case_fst}
 We apply the induction hypothesis to $S$ and distinguish three cases.
 \begin{enumerate}
 \item \case{$\eval{S \phi} = \bot$}
  Then $\eval{S \phi'} = \bot$ (by IH), hence $\eval{R \phi} = \eval{R \phi'} = \bot$.
 \item \case{$\eval{S \phi} \neq \bot$ and is not a pair}
 	Then by IH there exists $c'\subseteq c$ such that
	$\teqTc{\Gamma}{\Delta}{\E}{\E'}{\eval{S \phi}}{\eval{S\phi'}}{\L}{c'}$,
    which implies that $\eval{R\phi'}$ and is not a pair either by Lemma~\ref{lem-proof:l-same-head}.
    Hence $\eval{R \phi} = \eval{R \phi'} = \bot$.
 \item \case{$\eval{S \phi} = \PAIR{t_1}{t_2}$ is a pair}
  Then by IH there exists $c'\subseteq c$ such that
  $\teqTc{\Gamma}{\Delta}{\E}{\E'}{\eval{S \phi}}{\eval{S\phi'}}{\L}{c'}$.
  This implies, by Lemma~\ref{lem-proof:pair-types}, that $\eval{S \phi'} = \PAIR{t_1'}{t_2'}$ is also a pair,
  and that $\teqTc{\Gamma}{\Delta}{\E}{\E'}{t_1}{t_1'}{\L}{c''}$ for some $c'' \subseteq c'$.
  Since $\eval{R \phi} = t_1$ and $\eval{R \phi'} = t_1'$, this proves the case.
 \end{enumerate}

\item \case{$R = \SND{S}$} 
 This case is analogous to the case \ref{proof:case_fst}.

\item \case{$R = \DEC{S}{K}$} 
 We apply the induction hypothesis to $K$ and, similarly to the case \ref{proof:case_enc}, we distinguish several cases.
 \begin{enumerate}
 \item \case{If $\eval{K\phi} = \bot$ or is not a key} then, as in case \ref{proof:case_enc}, $\eval{R\phi} = \eval{R\phi'} = \bot$. 
 \item \case{If $\eval{K\phi}$ is a key}, then
 	similarly to case \ref{proof:case_enc} we can show that $\eval{K\phi} = \eval{K\phi'}$, and
  $\Gamma(\eval{K\phi}) = \Gamma(\eval{K\phi'}) = \skey{\L}{T}$ for some $T$.
  We then apply the IH to $S$, which creates two cases.
  Either $\eval{S\phi} = \eval{S\phi'} = \bot$, or there exists $c' \subseteq c$ such that
  $\teqTc{\Gamma}{\Delta}{\E}{\E'}{\eval{S \phi}}{\eval{S\phi'}}{\L}{c'}$.
  In the first case, the claim holds, since $\eval{R\phi} = \eval{R\phi'} = \bot$.
  In the second case, by Lemma~\ref{lem-proof:l-same-head}, we know that
  $\eval{S\phi}$ is an encryption by $\eval{K\phi}$ if and only if $\eval{S\phi'}$ also is an encryption by this key.
  Consequently, if $\eval{S\phi}$ is not an encryption by $\eval{K\phi}$
  ($=\eval{K\phi'}$), then it is the same for $\eval{S\phi'}$; and $\eval{R\phi} = \eval{R\phi'} = \bot$.
  Otherwise, $\eval{S\phi} = \ENC{t}{\eval{K\phi}}$ and $\eval{S\phi'}=\ENC{t'}{\eval{K\phi'}}$ for some $t$, $t'$.
  In that case, by IH, we have $\teqTc{\Gamma}{\Delta}{\E}{\E'}{\ENC{t}{\eval{K\phi}}}{\ENC{t'}{\eval{K\phi'}}}{\L}{c'}$.
  Therefore, by Lemma~\ref{lem-proof:enc-types} (point 5), $\teqTc{\Gamma}{\Delta}{\E}{\E'}{t}{t'}{\L}{c'}$,
  which is to say $\teqTc{\Gamma}{\Delta}{\E}{\E'}{\eval{R\phi}}{\eval{R\phi'}}{\L}{c'}$.
  Hence the claim holds in this case.
 \end{enumerate}

\item \case{$R = \ADEC{S}{K}$} 
 We apply the induction hypothesis to $K$ and, similarly to the case \ref{proof:case_enc}, we distinguish several cases.
 \begin{enumerate}
 \item \case{If $\eval{K\phi} = \bot$ or is not a key} then, as in case \ref{proof:case_enc}, $\eval{R\phi} = \eval{R\phi'} = \bot$. 
 \item \case{If $\eval{K\phi}$ is a key}, then
 	similarly to case \ref{proof:case_enc} we can show that $\eval{K\phi} = \eval{K\phi'}$, and
  $\Gamma(\eval{K\phi}) = \Gamma(\eval{K\phi'}) = \skey{\L}{T}$ for some $T$.
  We then apply the IH to $S$, which creates two cases.
  Either $\eval{S\phi} = \eval{S\phi'} = \bot$, or there exists $c' \subseteq c$ such that
  $\teqTc{\Gamma}{\Delta}{\E}{\E'}{\eval{S \phi}}{\eval{S\phi'}}{\L}{c'}$.
  In the first case, the claim holds, since $\eval{R\phi} = \eval{R\phi'} = \bot$.
  In the second case, by Lemma~\ref{lem-proof:l-same-head}, we know that
  $\eval{S\phi}$ is an asymmetric encryption by $\PUBK{\eval{K\phi}}$ if and only if $\eval{S\phi'}$ also is an encryption by this key.
  Consequently, if $\eval{S\phi}$ is not an encryption by $\PUBK{\eval{K\phi}}$
  ($=\PUBK{\eval{K\phi'}}$), then it is the same for $\eval{S\phi'}$, and $\eval{R\phi} = \eval{R\phi'} = \bot$.
  Otherwise, $\eval{S\phi} = \AENC{t}{\PUBK{\eval{K\phi}}}$ and $\eval{S\phi'}=\AENC{t'}{\PUBK{\eval{K\phi'}}}$ for some $t$, $t'$.
  Thus by IH we have $\teqTc{\Gamma}{\Delta}{\E}{\E'}{\AENC{t}{\PUBK{\eval{K\phi}}}}{\AENC{t'}{\PUBK{\eval{K\phi'}}}}{\L}{c'}$.
  Therefore,
  by Lemma~\ref{lem-proof:enc-types} (point 6), we know that
  $\teqTc{\Gamma}{\Delta}{\E}{\E'}{t}{t'}{\L}{c'}$,
  which is to say $\teqTc{\Gamma}{\Delta}{\E}{\E'}{\eval{R\phi}}{\eval{R\phi'}}{\L}{c'}$.
  Hence the claim holds in this case.
 \end{enumerate}

\item \case{$R = \CHECK{S}{K}$} 
 We apply the induction hypothesis to $K$ and, similarly to the case \ref{proof:case_aenc}, we distinguish several cases.
 \begin{enumerate}
 \item \case{If $\eval{K\phi} = \bot$ or is not a verification key} then, as in case \ref{proof:case_aenc}, we can show that $\eval{R\phi} = \eval{R\phi'} = \bot$. 
 \item \case{If $\eval{K\phi}$ is a verification key $\VK{k}$ for some $k\in \K$}, then
 	similarly to case \ref{proof:case_aenc} we can show that $\eval{K\phi} = \eval{K\phi'}$, and
  $k\in\dom{\Gamma}$.
  We then apply the IH to $S$, which creates two cases.
  Either $\eval{S\phi} = \eval{S\phi'} = \bot$, or there exists $c' \subseteq c$ such that
  $\teqTc{\Gamma}{\Delta}{\E}{\E'}{\eval{S \phi}}{\eval{S\phi'}}{\L}{c'}$.
  In the first case, the claim holds, since $\eval{R\phi} = \eval{R\phi'} = \bot$.
  In the second case, by Lemma~\ref{lem-proof:l-same-head}, we know that
  $\eval{S\phi}$ is a signature by $k (=\eval{K\phi})$ if and only if $\eval{S\phi'}$ also is a signature by this key.
  Consequently, if $\eval{S\phi}$ is not signed by $k$,
  then neither is $\eval{S\phi'}$, and $\eval{R\phi} = \eval{R\phi'} = \bot$.
  Otherwise, $\eval{S\phi} = \SIGN{t}{k}$ and $\eval{S\phi'}=\SIGN{t'}{k}$ for some $t$, $t'$.
  Thus by IH we have $\teqTc{\Gamma}{\Delta}{\E}{\E'}{\SIGN{t}{k}}{\SIGN{t'}{k}}{\L}{c'}$.
  Therefore, by Lemma~\ref{lem-proof:sign-types} (point 2), we know that there exists $c'' \subseteq c'$ such that
  $\teqTc{\Gamma}{\Delta}{\E}{\E'}{t}{t'}{\L}{c''}$.
  That is to say $\teqTc{\Gamma}{\Delta}{\E}{\E'}{\eval{R\phi}}{\eval{R\phi'}}{\L}{c''}$.
  Hence the claim holds in this case.
 \end{enumerate}
\end{enumerate} 
\end{proof}

\begin{lemma}[Substitution preserves typing]
\label{lem-proof:subst-typing}
For all $\Gamma$, $\Gamma'$, $\Gamma''$, such that $\tewf{\Gamma\cup\Gamma'\cup\Gamma''}$,
(we do not require that $\Gamma$ and $\Gamma'$ are well-formed),
for all $M$, $N$, $T$, $c_\sigma$, $c$,
for all ground substitutions $\sigma$, $\sigma'$,
if
\begin{itemize}
\item $\Gamma$, $\Gamma'$ only contains variables, and have disjoint domains;
\item $\Gamma''$ only contains names and keys;
\item for all $x\in\dom{\Gamma'}$, $\Gamma'(x)$ is not of the form $\LRTn{l}{1}{m}{l'}{1}{n}$,
\item $\wtc{\Delta}{\E}{\E'}{\sigma}{\sigma'}{(\Gamma\cup\Gamma'')}{c_\sigma}$,
\item and $\teqTc{\Gamma\cup\Gamma'\cup\Gamma''}{\Delta}{\E}{\E'}{M}{N}{T}{c}$
\end{itemize}
then there exists $c' \subseteq \inst{c}{\sigma}{\sigma'} \cup c_{\sigma}$ such that
\[\teqTc{\Gamma'\cup\Gamma''}{\Delta}{\E}{\E'}{M\sigma}{N\sigma'}{T}{c'}.\]

In particular, if we have $\Gamma = \onlyvar{\Gamma'''}$, $\Gamma' = \emptyset$,
and $\Gamma'' = \novar{\Gamma'''}$ for some $\Gamma'''$, then the first three conditions trivially hold.
\end{lemma}
\begin{proof}
Note that $\novar{\Gamma}=\novar{\Gamma'} = \emptyset$, $\onlyvar{\Gamma}=\Gamma$, $\onlyvar{\Gamma'}=\Gamma'$,
$\novar{\Gamma''}=\Gamma''$ and $\onlyvar{\Gamma''}=\emptyset$.
This proof is done by induction on the typing derivation for the terms.
The claim clearly holds in the \TNonce, \TNonceL, \TCst, \TPubkey, \TVkey, \TKey, \THash, \THigh, \TLRone, \TLRinf
since their conditions do not use $\Gamma(x)$ (for any variable $x$) or another type judgement, and they still apply to the 
messages $M\sigma$ and $N\sigma'$.

\bigskip

It follows directly from the induction hypothesis in all other cases except the \TVar and \TLRVar cases,
which are the base cases.

\bigskip

In the \TVar case, the claim also holds, since $M = N = x$ for some variable $x\in\dom{\Gamma}\cup\dom{\Gamma'}$.
If $x\in\dom{\Gamma'}$, then $x\sigma = x\sigma' = x$, and $T = \Gamma'(x)$.
Thus, by rule \TVar, $\teqTc{\Gamma'\cup\Gamma''}{\Delta}{\E}{\E'}{x\sigma}{x\sigma'}{\Gamma'(x)}{\emptyset}$
and the claim holds.
If $x\in\dom{\Gamma}$, then $T = \Gamma(x)$,
and, since by hypothesis the substitutions are well-typed, there exists $c_x \subseteq c_{\sigma,\sigma'}$ such that 
$\teqTc{\novar{(\Gamma\cup\Gamma'')}}{\Delta}{\E}{\E'}{\sigma(x)}{\sigma'(x)}{\Gamma(x)}{c_x}$.
Thus, since $\novar{(\Gamma\cup\Gamma'')} = \Gamma''$, and by applying Lemma~\ref{lem-proof:typing-contextinclusion} to
$\Gamma'$,
$\teqTc{\Gamma'\cup\Gamma''}{\Delta}{\E}{\E'}{\sigma(x)}{\sigma'(x)}{\Gamma(x)}{c_x}$ and the claim holds.

\bigskip

Finally, in the \TLRVar case, there exist two variables $x$, $y$, and types $\noncetypelab{l}{1}{m}$,
$\noncetypelab{l'}{1}{n}$, $\noncetypelab{l''}{1}{m'}$, $\noncetypelab{l'''}{1}{n'}$, such that
$M=x$, $N=y$, $c = \emptyset$,
$\teqTc{\Gamma\cup\Gamma'\cup\Gamma''}{\Delta}{\E}{\E'}{x}{x}{\LRTn{l}{1}{m}{l'}{1}{n}}{\emptyset}$,
$\teqTc{\Gamma\cup\Gamma'\cup\Gamma''}{\Delta}{\E}{\E'}{y}{y}{\LRTn{l''}{1}{m'}{l'''}{1}{n'}}{\emptyset}$, and
$T = \LRTn{l}{1}{m}{l'''}{1}{n'}$.

By Lemma~\ref{lem-proof:lr-ground}, this implies that $(\Gamma\cup\Gamma'\cup\Gamma'')(x) = \LRTn{l}{1}{m}{l'}{1}{n}$ and $(\Gamma\cup\Gamma'\cup\Gamma'')(y) = \LRT{l''}{1}{m'}{l'''}{1}{n'}$.
Hence, since by hypothesis $\Gamma'$ does not contain such types, and $\Gamma''$ does not contain variables, 
$x\in\dom{\Gamma}$ and $y\in\dom{\Gamma}$.

Moreover, by the induction hypothesis, there exist $c'$, $c''\subseteq c_\sigma$ such that
$\teqTc{\Gamma'\cup\Gamma''}{\Delta}{\E}{\E'}{x\sigma}{x\sigma'}{\LRTn{l}{1}{m}{l'}{1}{n}}{c'}$, and
$\teqTc{\Gamma'\cup\Gamma''}{\Delta}{\E}{\E'}{y\sigma}{y\sigma'}{\LRTn{l''}{1}{m'}{l'''}{1}{n'}}{c''}$.
That is to say, since $x, y\in\dom{\Gamma}=\dom{\sigma}=\dom{\sigma'}$, that
$\teqTc{\Gamma'\cup\Gamma''}{\Delta}{\E}{\E'}{\sigma(x)}{\sigma'(x)}{\LRTn{l}{1}{m}{l'}{1}{n}}{c'}$, and
$\teqTc{\Gamma'\cup\Gamma''}{\Delta}{\E}{\E'}{\sigma(y)}{\sigma'(y)}{\LRTn{l''}{1}{m'}{l'''}{1}{n'}}{c''}$.
Hence, by Lemma~\ref{lem-proof:lr-ground}, and since $\sigma$, $\sigma'$ are ground, we have
$\sigma(x) = m$, $\sigma'(x) = n$, $\sigma(y) = m'$, and $\sigma'(y) = n'$,
and $\Gamma''(m)=\noncetypelab{l}{1}{m}$ and $\Gamma''(n') = \noncetypelab{l'''}{1}{n'}$.

Thus, by rule \TLRone, $\teqTc{\Gamma'}{\Delta}{\E}{\E'}{\sigma(x)}{\sigma'(y)}{\LRTn{l}{1}{m}{l'''}{1}{n'}}{\emptyset}$,
which proves the claim.
\end{proof}

\begin{lemma}[Types $\L$ and $\S$ are disjoint]
\label{lem-proof:LS-disjoint}
For all $\Gamma$, for all ground terms $M$, $M'$, $N$, $N'$, for all sets of constraints $c$, $c'$,
if $\teqTc{\Gamma}{\Delta}{\E}{\E'}{M}{N}{\L}{c}$ and $\teqTc{\Gamma}{\Delta}{\E}{\E'}{M'}{N'}{\S}{c'}$
then $M \neq M'$ and $N \neq N'$.
\end{lemma}
\begin{proof}
First, it is easy to see by induction on the type derivation that for all ground terms $M$, $N$, for all $c$, if
$\teqTc{\Gamma}{\Delta}{\E}{\E'}{M}{N}{\S}{c}$ then either
\begin{itemize}
\item $M$ is a nonce $m\in\N$ such that $\Gamma(m) = \noncetypelab{\S}{a}{m}$ for some $a\in\{\infty, 1\}$;
\item or $M$ is a key and $\Gamma(M) = \skey{\S}{T}$ for some $T$;
\item or $\teqTc{\Gamma}{\Delta}{\E}{\E'}{M}{N}{\S*T}{c'}$ for some $T$, $c'$;
\item or $\teqTc{\Gamma}{\Delta}{\E}{\E'}{M}{N}{T*\S}{c'}$ for some $T$, $c'$.
\end{itemize}

Indeed, (as $\Gamma$ is well-formed) the only possible cases are \TNonce, \TSub, and \TLRp.
In the \TNonce case the claim clearly holds.
In the \TSub case we use Lemma~\ref{lem-proof:subtyping} followed by Lemma~\ref{lem-proof:type-key-nonce}.
In the \TLRp case we apply Lemma~\ref{lem-proof:lr-ground} and the claim directly follows.

\bigskip

Let us now show that for all $M$, $N$, $N'$ ground, for all $c$, $c'$,
$\teqTc{\Gamma}{\Delta}{\E}{\E'}{M}{N}{\L}{c}$ and $\teqTc{\Gamma}{\Delta}{\E}{\E'}{M}{N'}{\S}{c'}$ cannot both hold.
(This corresponds, with the notations of the statement of the lemma, to proving by contradiction that $M \neq M'$.
The proof that $N \neq N'$ is analogous.)

We show this property by induction on the size of $t_1$.

Since $\teqTc{\Gamma}{\Delta}{\E}{\E'}{M}{N'}{\S}{c'}$, by the property stated in the beginning of this proof,
we can distinguish four cases.

\begin{itemize}
\item \case{If $M$ is a nonce and $\Gamma(M) = \noncetypelab{\S}{a}{M}$}:
then this contradicts Lemma~\ref{lem-proof:type-key-nonce}.
Indeed, this lemma (point 5) implies that $M\in\BN$, but also (by point 6),
since $\teqTc{\Gamma}{\Delta}{\E}{\E'}{M}{N}{\L}{c}$,
that either $\Gamma(M) = \noncetypelab{\L}{a}{M}$ for some $a\in\{1,\infty\}$, or $M\in\FN\cup\CST$.

\item \case{If $M$ is a key and $\Gamma(M) = \skey{\S}{T}$ for some $T$}:
then by Lemma~\ref{lem-proof:type-key-nonce}, since $\teqTc{\Gamma}{\Delta}{\E}{\E'}{M}{N}{\L}{c}$,
there exists $T'$ such that $\Gamma(M) = \skey{\L}{T'}$. This contradicts $\Gamma(M) = \skey{\S}{T}$.

\item \case{If $\teqTc{\Gamma}{\Delta}{\E}{\E'}{M}{N'}{\S*T}{c''}$ for some $T$, $c''$}:
then by Lemma~\ref{lem-proof:pair-types}, since $M$, $N'$ are ground,
there exist $M_1$, $M_2$, $N'_1$, $N'_2$, $c_1'$ such that
$M = \PAIR{M_1}{M_2}$, $N' = \PAIR{N'_1}{N'_2}$,
and $\teqTc{\Gamma}{\Delta}{\E}{\E'}{M_1}{N'_1}{\S}{c_1'}$.
Moreover, since $\teqTc{\Gamma}{\Delta}{\E}{\E'}{M}{N}{\L}{c}$, also by Lemma~\ref{lem-proof:pair-types},
there exist $N_1$, $N_2$, $c_1$ such that $N = \PAIR{N_1}{N_2}$
and $\teqTc{\Gamma}{\Delta}{\E}{\E'}{M_1}{N_1}{\L}{c_1}$.
However, by the induction hypothesis, $\teqTc{\Gamma}{\Delta}{\E}{\E'}{M_1}{N'_1}{\S}{c_1'}$ and
$\teqTc{\Gamma}{\Delta}{\E}{\E'}{M_1}{N_1}{\L}{c_1}$ is impossible.

\item \case{If $\teqTc{\Gamma}{\Delta}{\E}{\E'}{M}{N'}{T*\S}{c''}$ for some $T$, $c''$}:
this case is similar to the previous one.
\end{itemize}
\end{proof}

\begin{lemma}[Low terms are recipes on their constraints]
\label{lem-proof:low-recipe}
For all ground messages $M$, $N$, for all $T \subtyp \L$, for all $\Gamma$, $c$,
if $\teqTc{\Gamma}{\Delta}{\E}{\E'}{M}{N}{T}{c}$ then there exists an attacker recipe $R$ without destructors such that
$M = R (\phiL{c}\cup\phiEE)$ and
$N = R (\phiR{c}\cup\phiEE)$.
\end{lemma}
\begin{proof}
We prove this lemma by induction on the typing derivation of $\teqTc{\Gamma}{\Delta}{\E}{\E'}{M}{N}{T}{c}$.
We distinguish several cases for the last rule in this derivation.

\begin{itemize}
\item \case{\TNonce, \TEnc, \TAenc, \THigh, \TOr, \TLRone, \TLRinf, \TLRp, \TLRVar}: these cases are not possible, since the type they give to terms is never a subtype of $\L$ by Lemma~\ref{lem-proof:subtyping}.

\item \case{\TVar}: this case is not possible since $M$, $N$ are ground.

\item \case{\TSub}: this case is directly proved by applying the induction hypothesis to the judgement
$\teqTc{\emptyset}{\Delta}{\E}{\E'}{M}{N}{T'}{c}$ where $T' \subtyp T \subtyp \L$, which appears in the conditions of this rule, and has a shorter derivation.

\item \case{\TLRLp}: in this case, $\teqTc{\Gamma}{\Delta}{\E}{\E'}{M}{N}{\LRTn{\L}{a}{n}{\L}{a}{n}}{c'}$ for some nonce $n$, some $a\in\{\infty,1\}$,
some $c'$, and $c = \emptyset$. By Lemma~\ref{lem-proof:lr-ground}, this implies that $M=N=n$, and
$\Gamma(n) = \noncetypelab{\L}{a}{n}$.
Thus, by definition, there exists $x$ such that $\phiEE(x) = n$ and the claim holds with
$R=x$.

\item \case{\TNonceL}: in this case $M = N = n$ for some $n \in \N$ such that $\Gamma(n) = \noncetypelab{\L}{a}{n}$ for some $a\in\{1,\infty\}$.
Hence, by definition, there exists $x$ such that $\phiEE(x) = n$ and the claim holds with $R=x$.

\item \case{\TCst}: then $M = N = a \in \CST\cup\FN$, and the claim holds with $R = a$.

\item \case{\TKey}: then $M = N = k \in \K$ and there exists $T'$ such that $\Gamma(k) = \skey{\L}{T'}$.
By definition, there exists $x$ such that $\phiEE(x) = k$ and the claim holds with
$R=x$.

\item \case{\TPubkey, \TVkey}: then $M = N = \PUBK{k}$ (resp. $\VK{k}$) for some $k\in\dom{\Gamma}$.
By definition, there exists $x$ such that $\phiEE(x) = \PUBK{k}$ (resp. $\VK{k}$) and the claim holds 
with $R=x$.

\item \case{\TPair, \THashL}: these cases are similar.
We detail the \TPair case. In that case, $T = T_1 * T_2$ for some $T_1$, $T_2$.
By Lemma~\ref{lem-proof:subtyping}, $T_1$, $T_2$ are subtypes of $\L$.
In addition, there exist $M_1$, $M_2$, $N_1$, $N_2$, $c_1$, $c_2$ such
that $\teqTc{\Gamma}{\Delta}{\E}{\E'}{M_i}{N_i}{T_i}{c_i}$ (for $i\in\{1,2\}$). By applying the induction hypothesis to these two judgements (which have shorter proofs), we obtain $R_1$, $R_2$ such that for all $i$,
$M_i = R_i (\phiL{c_i}\cup\phiEE)$ and
$N_i = R_i (\phiR{c_i}\cup\phiEE)$.
Therefore the claim holds with $R = \PAIR{R_1}{R_2}$.

\item \case{\TEncH, \TAencH, \THash, \TSignH}: these four cases are similar. In each case, by the form of the typing rule,
we have $c = \{M \eqC N\} \cup c'$ for some $c'$.
Therefore by definition of $\phiL{c}$, $\phiR{c}$, there exists $x$ such that $\phiL{c}(x) = M$ and $\phiR{c}(N) = N$.
The claim holds with $R = x$.

\item \case{\TEncL, \TAencL}: these two cases are similar, we write the proof for the \TEncL case.
The form of this rule application is:
\[ \inferrule
 {\inferrule*{\Pi}{\teqTc{\Gamma}{\Delta}{\E}{\E'}{M}{N}{\encT{\L}{k}}{c}}\\
  \Gamma(k) = \skey{\L}{T'}}
 {\teqTc{\Gamma}{\Delta}{\E}{\E'}{M}{N}{\L}{c}}
\]
By Lemma~\ref{lem-proof:enc-types}, there exist $M'$, $N'$ such that $M = \ENC{M'}{k}$, $N = \ENC{N'}{k}$, and
$\teqTc{\Gamma}{\Delta}{\E}{\E'}{M'}{N'}{\L}{c}$ with a proof shorter that $\Pi$.
Thus by applying the induction hypothesis to this judgement, there exists $R$ such that
$M' = R (\phiL{c}\cup\phiEE)$ and
$N' = R (\phiR{c}\cup\phiEE)$.

Moreover, since $\Gamma(k) = \skey{\L}{T'}$, by definition, there exists $x$ such that $\phiEE(x) = k$
(in the asymmetric case, the messages are encrypted with a public key $\PUBK{k}$, which also appears in this frame).
Therefore, the claim holds with the recipe $\ENC{R}{x}$.

\item \case{\TSignL}: the form of this rule application is:
\[ \inferrule
 {\inferrule*{\Pi}{\teqTc{\Gamma}{\Delta}{\E}{\E'}{M'}{N'}{\L}{c}}\\
  \Gamma(k) = \skey{\L}{T'}}
 {\teqTc{\Gamma}{\Delta}{\E}{\E'}{\SIGN{M'}{k}}{\SIGN{N'}{k}}{\L}{c}}
\]
with $M = \SIGN{M'}{k}$, $N = \SIGN{N'}{k}$.
Thus by applying the induction hypothesis to $\teqTc{\Gamma}{\Delta}{\E}{\E'}{M'}{N'}{\L}{c}$,
there exists $R$ such that
$M' = R (\phiL{c}\cup\phiE{\Delta}{\E}{\E'})$ and
$N' = R (\phiR{c}\cup\phiE{\Delta}{\E}{\E'})$.

Moreover, since $\Gamma(k) = \skey{\L}{T'}$, by definition, there exists $x$ such that $\phiEE(x) = k$.
Therefore, the claim holds with the recipe $\SIGN{R}{x}$.
\end{itemize}
\end{proof}

\begin{lemma}[Low frames with consistent constraints are statically equivalent]
\label{lem-proof:l-cons-stat-eq}
For all ground $\phi$, $\phi'$, for all $c$, $\Gamma$, if
\begin{itemize}
\item $\teqTc{\Gamma}{\Delta}{\E}{\E'}{\phi}{\phi'}{\L}{c}$
\item and $c$ is consistent in $\novar{\Gamma}$,
\end{itemize}
then
$\NEWN{\E_\Gamma}.\phi$ and $\NEWN{\E_\Gamma}.\phi'$ are statically equivalent.
\end{lemma}
\begin{proof}
We can first notice that since $\phi$ and $\phi'$ are ground, so is $c$ (this is easy to see by examining the typing rules for terms).
Let $R$, $R'$ be two attacker recipes, such that
$\var{R} \cup \var{R'} \subseteq \dom{\phi} (= \dom{\phi'})$.

For all $x\in \dom{\phi} (= \dom{\phi'})$, by assumption, there exists $c_x \subseteq c$ such that
$\teqTc{\Gamma}{\Delta}{\E}{\E'}{\phi(x)}{\phi'(x)}{\L}{c_x}$.
By Lemma~\ref{lem-proof:low-recipe}, there exists a recipe $R_x$ such that
$\phi(x) = R_x (\phiL{c_x}\cup\phiEE)$ and
$\phi'(x) = R_x (\phiR{c_x}\cup\phiEE)$.

Since $c_x \subseteq c$, we also have $\phi(x) = R_x (\phiL{c}\cup\phiEE)$ and
$\phi'(x) = R_x (\phiR{c}\cup\phiEE)$.

Let $\overline{R}$ and $\overline{R}'$ be the recipes obtained by replacing every occurence of $x$
with $R_x$ in respectively $R$ and $R'$, for all variable $x \in \dom{\phi} (= \dom{\phi'})$.

We then have
$R\phi = \overline{R}(\phiL{c}\cup\phiEE)$ and 
$R'\phi = \overline{R}'(\phiL{c}\cup\phiEE)$;
and similarly
$R\phi' = \overline{R}(\phiR{c}\cup\phiEE)$ and 
$R'\phi' = \overline{R}'(\phiR{c}\cup\phiEE)$.

Since $c$ is ground, and consistent in $\novar{\Gamma}$,
by definition of consistency,
the frames $\NEWN{\E_\Gamma}.\phiL{c}\cup\phiEE$ and $\NEWN{\E_\Gamma}.\phiR{c}\cup\phiEE$
are statically equivalent.
Hence, by definition of static equivalence,
\[\overline{R}(\phiL{c}\cup\phiEE) = \overline{R}'(\phiL{c}\cup\phiEE)
\quad \Longleftrightarrow \quad
\overline{R}(\phiR{c}\cup\phiEE) = \overline{R}'(\phiR{c}\cup\phiEE)\]
\ie
\[R\phi = R'\phi
\quad \Longleftrightarrow \quad
R\phi' = R'\phi'\]

Therefore, $\NEWN{\E_\Gamma}.\phi$ and $\NEWN{\E_\Gamma}.\phi'$ are statically equivalent.
\end{proof}

\begin{lemma}[Invariant]
\label{lem-proof:invariant}
For all $\Gamma$, 
$\phi_P$, $\phi_P'$, $\phi_Q$, $\sigma_P$, $\sigma_P'$, $\sigma_Q$, $\E'$, 
$c_\phi$, $c_\sigma$,
for all multisets of processes $\PP$, $\PP'$, $\QQ$, where the processes in $\PP$, $\PP'$, $\QQ$ are noted $\{P_i\}$, $\{P'_i\}$, $\{Q_i\}$;
for all constraint sets $\{C_i\}$, if:

\begin{itemize}
\item $|\PP| = |\QQ|$
\item $\dom{\phi_P}=\dom{\phi_Q}$
\item $\forall i,$ there is a derivation $\Pi_i$ of $\teqP{\Gamma}{\Delta}{\EPO}{\EQO}{P_i}{Q_i}{C_i}$,
\item $\teqTc{\Gamma}{\Delta}{\EPO}{\EQO}{\phi_P}{\phi_Q}{\L}{c_\phi}$
\item for all $i\neq j$, the sets of bound variables in $P_i$ and $P_j$ (resp. $Q_i$ and $Q_j$) are disjoint, and similarly for the names 
\item $\sigma_P$, $\sigma_Q$ are ground, and $\wtc{\Delta}{\EPO}{\EQO}{\sigma_P}{\sigma_Q}{\Gamma}{c_\sigma}$,
\item $\inst{({\UnionCart}_i C_i) \UnionAll c_\phi}{\sigma_P}{\sigma_Q} \UnionAll c_\sigma$ is consistent, 
\item $(\E_\Gamma,\PP,\phi_P,\sigma_P) \redAction{\alpha} (\E',\PP',\phi_P',\sigma_P')$,
\end{itemize}
then there exist a word $w$, a multiset $\QQ' = \{Q'_i\}$, constraint sets $\{C_i'\}$, a frame $\phi_Q'$, a substitution $\sigma_Q'$, an environment $\Gamma'$, 
constraints $c_\phi'$, and $c_\sigma'$ such that:
\begin{itemize}
\item $w \eqSilent \alpha$
\item $|\PP'|=|\QQ'|$
\item for all $i\neq j$, the sets of bound variables in $P_i'$ and $P_j'$ (resp. $Q_i'$ and $Q_j'$) are disjoint, and similarly for the bound names;
\item $\teqTc{\Gamma'}{\Delta}{\EPO}{\EQO}{\phi_P'}{\phi_Q'}{\L}{c_\phi'}$
\item $\E' = \E_{\Gamma'}$
\item $(\E_\Gamma,\QQ,\phi_Q,\sigma_Q) \redWord{w} (\E_{\Gamma'},\QQ',\phi_Q',\sigma_Q')$,
\item $\forall i, \teqP{\Gamma'}{\Delta}{\EPO}{\EQO}{P'_i}{Q'_i}{C'_i}$,
\item $\sigma_P'$, $\sigma_Q'$ are ground and $\wtc{\Delta}{\EPO}{\EQO}{\sigma_P'}{\sigma_Q'}{\Gamma'}{c_\sigma'}$,
\item $\dom{\phi_P'} = \dom{\phi_Q'}$,
\item $\inst{({\UnionCart}_i C_i') \UnionAll c_\phi'}{\sigma_P'}{\sigma_Q'} \UnionAll c_\sigma'$ is consistent. 
\end{itemize}
\end{lemma}

\begin{proof}

First, we show that it is sufficient to prove this lemma in the case where $\Gamma$ does not contain any union types.
Indeed, assume we know the property holds in that case.
Let us show that the lemma then also holds in the other case, \ie if $\Gamma$ contains union types.
By hypothesis, $\sigma_P$, $\sigma_Q$ are ground, and $\wtc{\Delta}{\EPO}{\EQO}{\sigma_P}{\sigma_Q}{\Gamma}{c_\sigma}$.
Hence we know by Lemma~\ref{lem-proof:ground-subst-branch} that there exists a branch $\Gamma''\in\branch{\Gamma}$
(thus $\Gamma''$ does not contain union types), such that $\wtc{\Delta}{\EPO}{\EQO}{\sigma_P}{\sigma_Q}{(\Gamma'')}{c_\sigma}$.

Moreover, by Lemma~\ref{lem-proof:type-processes-branches}, $\forall i, \teqP{\Gamma''}{\Delta}{\EPO}{\EQO}{P_i}{Q_i}{C_i'' \subseteq C_i}$;
and by Lemma~\ref{lem-proof:type-terms-branches},
$\teqTc{\Gamma''}{\Delta}{\EPO}{\EQO}{\phi_P}{\phi_Q}{\L}{c_\phi}$.
In addition by Lemma~\ref{lem-proof:cons-subset},
$\inst{({\UnionCart}_i C_i'') \UnionAll c_\phi}{\sigma_P}{\sigma_Q} \UnionAll c_\sigma$
is a subset of $\inst{({\UnionCart}_i C_i) \UnionAll c_\phi}{\sigma_P}{\sigma_Q} \UnionAll c_\sigma$ and is
therefore consistent. 
Thus, if the lemma holds when the environment does not contain union types, it can be applied to the same processes, frames, substitutions and to $\Gamma''$, which directly concludes the proof.

Therefore, we may assume that $\Gamma$ does not contain any union types.


Note that the assumption on the disjointness of the sets of bound variables (and names) in the processes implies, using Lemma~\ref{lem-proof:env-constraints-bound}, that since
$({\UnionCart}_i \inst{C_i}{\sigma_P}{\sigma_Q} \UnionAll c_\sigma)$ is consistent 
(by Lemma~\ref{lem-proof:cons-subset}), each of the $\inst{C_i}{\sigma_P}{\sigma_Q} \UnionAll c_\sigma$ is also consistent. 
Moreover, this disjointness property for $\PP'$ and $\QQ'$ follows from the other points, as it is easily proved by examining the reduction rules that it is preserved by reduction.


\vspace{1em}
By hypothesis, $(\E_\Gamma, \PP, \phi_P, \sigma_P)$ reduces to $(\E_{\Gamma'}, \PP', \phi_P', \sigma_P')$
We know from the form of the reduction rules that exactly one process $P_i \in \PP$ is reduced, while the others are unchanged.
By the assumptions, there is a corresponding process $Q_i \in \QQ$ and a derivation $\Pi_i$ of
$\teqP{\Gamma}{\Delta}{\EPO}{\EQO}{P_i}{Q_i}{C_i}$.

We continue the proof by a case disjunction on the last rule of $\Pi_i$.
Let us first consider the cases of the rules \PZero, \PPar, \PNew, 
 and \POr.

\begin{itemize}
\item \case{ \PZero}:
then $P_i = Q_i = \ZERO$. Hence, the reduction rule applied to $\PP$ is Zero, and $\PP' = \PP \backslash \{P_i\}$,
$\E_P' = \E_P$, $\phi_P' = \phi_P$, and $\sigma_P' = \sigma_P$.
The same reduction can be performed in $\QQ$:
\[(\E_Q, \QQ, \phi_Q, \sigma_Q) \redAction{\silentAction} (\E_Q, \QQ\backslash\{Q_i\}, \phi_Q, \sigma_Q)\]

Since the other processes, the frames, environments and substitutions do not change in this reduction, all the claims
clearly hold in this case (with $c_\phi' = c_\phi$, $c_\sigma' = c_\sigma$).
In particular, the consistency of the constraints follow from the consistency hypothesis.
Indeed,
\begin{align*}
({\UnionCart}_{j\neq i} C_j) \UnionCart C_i \UnionAll c_\phi &=
({\UnionCart}_{j\neq i} C_j) \UnionCart \{(\emptyset,\Gamma)\} \UnionAll c_\phi\\
&= ({\UnionCart}_{j\neq i} C_j) \UnionAll c_\phi,
\end{align*}
since $\Gamma$ is already contained in the environments appearing in each $C_j$ (by Lemma~\ref{lem-proof:env-const-branch}).
Thus
\[\inst{({\UnionCart}_{j} C_j') \UnionAll c_\phi'}{\sigma_P'}{\sigma_Q'} \UnionCart c_\sigma' =
 \inst{({\UnionCart}_{j} C_j) \UnionAll c_\phi}{\sigma_P}{\sigma_Q} \UnionCart c_\sigma\]

\item \case{\PPar}:
then $P_i = P_i^1 \PAR P_i^2$, $Q_i = Q_i^1 \PAR Q_i^2$. Hence, the reduction rule applied to $\PP$ is Par:
\[(\EG, \PP, \phi_P, \sigma_P) \redAction{\silentAction} (\EG, \PP\backslash\{P_i\} \cup \{P_i^1, P_i^2\}, \phi_P,
\sigma_P).\]
We choose $\Gamma' = \Gamma$.

In addition
\[\Pi_i =
\inferrule
 {\inferrule*{\Pi^1}{\teqP{\Gamma}{\Delta}{\EPO}{\EQO}{P_i^1}{Q_i^1}{C_i^1}} \\
  \inferrule*{\Pi^2}{\teqP{\Gamma}{\Delta}{\EPO}{\EQO}{P_i^2}{Q_i^2}{C_i^2}}}
 {\teqP{\Gamma}{\Delta}{\EPO}{\EQO}{P_i}{Q_i}{C_i = C_i^1 \UnionCart C_i^2}}.
\]

The same reduction rule can be applied to $\QQ$:
\[(\EG, \QQ, \phi_Q, \sigma_Q) \redAction{\silentAction}
(\EG, \QQ\backslash\{Q_i\}\cup\{Q_i^1, Q_i^2\}, \phi_Q, \sigma_Q)\]

In this case again, the claims on the substitutions and frames hold since they do not change in the reduction.
Moreover the processes in $\PP'$ and $\QQ'$ are still pairwise typably equivalent.
Indeed, all the processes from $\PP$ and $\QQ$ are unchanged, except for $P_i$ and $Q_i$ which are reduced to $P_i^1$, 
$P_i^2$, $Q_i^1$, $Q_i^2$, and those are typably equivalent using $\Pi^1$ and $\Pi^2$.

Finally the constraint set is still consistent, since: 
\begin{align*}
\inst{({\UnionCart}_j C_j') \UnionAll c_\phi'}{\sigma_P'}{\sigma_Q'} \UnionAll c_\sigma' &=
\inst{({\UnionCart}_{j\neq i} C_j) \UnionCart C_i^1 \UnionCart C_i^2 \UnionAll c_\phi}{\sigma_P}{\sigma_Q} \UnionAll c_\sigma\\
&= \inst{({\UnionCart}_j C_j) \UnionAll c_\phi}{\sigma_P}{\sigma_Q} \UnionAll c_\sigma
\end{align*}

\item \case{\PNew}:
then $P_i = \NEW{n}{\noncetypelab{l}{\oneorinf}{n}}. P_i'$ and $Q_i = \NEW{n}{\noncetypelab{l}{\oneorinf}{n}}. Q_i'$.
$P_i$ is reduced to $P_i'$ by rule New:

\[(\EG, \PP, \phi_P, \sigma_P) \redAction{\silentAction} (\EG\cup\{n\}, \PP\backslash\{P_i\} \cup \{P'_i\}, \phi_P,
\sigma_P).\]

In addition
\[\Pi_i =
\inferrule
 {\inferrule*{\Pi_i'}{\teqP{\Gamma, n : \noncetypelab{l}{\oneorinf}{n}}{\Delta}{\EPO}{\EQO}{P_i'}{Q'_i}{C_i}}}
 {\teqP{\Gamma}{\Delta}{\EPO}{\EQO}{P_i}{Q_i}{C_i}}.
\]

We choose $\Gamma' = \Gamma, n:\noncetypelab{l}{\oneorinf}{n}$, and we have $\EGG = \EG\cup\{n\}$.

The same reduction rule can be applied to $\QQ$:
\[(\EG, \QQ, \phi_Q, \sigma_Q) \redAction{\silentAction}
(\EGG, \QQ\backslash\{Q_i\}\cup\{Q_i'\}, \phi_Q, \sigma_Q)\]

The claim clearly holds: the processes are still pairwise typable (using $\Pi_i'$ in the case of $P_i'$ and $Q_i'$; and $\Pi_j$ for $j\neq i$ as these processes are unchanged by the reduction),
and all the frames, substitutions, and constraints are unchanged, and since $\sigma$, $\sigma'$ are well-typed in $\Gamma'$ if and only if they are well-typed in $\Gamma$.

%
%

\item \case{POr}:
this case is not possible, since we have already eliminated the case where $\Gamma$ contains union types.

\end{itemize}

\vspace{1em}
In all the other cases for the last rule in $\Pi_i$, we know that the head symbol of $P_i$ is not $\PAR$, $\ZERO$
or $\NEWNA$.

Hence, the form of the reduction rules implies that $P_i \in \PP$
is reduced to exactly one process $P_i' \in \PP'$, while the other processes in $\PP$ do not change (\ie $P_j' = P_j$ for $j\neq i$), and $\E' = \EG$.
If we show in each case that the same reduction rule that is applied to $P_i$ can be applied to reduce $\QQ$ to a multiset $\QQ'$ by reducing process $Q_i$ into $Q_i'$,
we will also have $Q_j' = Q_j$ for $j\neq i$.
Therefore the claim on the cardinality of the processes multisets will hold.

Since $P_i$, $Q_i$ can be typed and the head symbol of $P_i$ is not $\NEWNA$, it is clear by examining the typing rules that the head symbol of $Q_i$ is not $\NEWNA$ either.
Hence, we will choose a $\Gamma'$ containing the same nonces as $\Gamma$, and we will have $\EGG = \EG$.

The proofs for theses cases follow the same structure:
\begin{itemize}
\item The typing rule gives us information on the form of $P_i$ and $Q_i$.
\item The form of $P_i$ gives us information on which reduction rule was applied to $\PP$.
\item The form of $Q_i$ is the same as $P_i$. Hence (additional conditions may need to be checked depending on
the rule) $Q_i$ can be reduced to some process $Q_i'$ by applying the same reduction rule that was applied to $P_i$ (except in the \PIfLR case).
\item thus $\QQ$ can be reduced too, with the same actions as $\PP$. We then check the additional conditions on the 
typing of the processes, frames and substitutions, and the consistency condition.
\end{itemize}

\vspace{1em}

First, let us consider the $ \POut$ case.
\begin{itemize}
\item \case{\POut}: then $P_i = \OUT{M}. P_i'$ and reduces to $P_i'$ via the Out rule, and $Q_i = \OUT{N}.Q_i'$ for some $N$ and $Q_i'$.
In addition 
\[\Pi_i =
\inferrule
 {\inferrule*{\Pi}{\teqP{\Gamma}{\Delta}{\EPO}{\EQO}{P_i'}{Q_i'}{C_i'}} \\
  \inferrule*{\Pi'}{\teqTc{\Gamma}{\Delta}{\EPO}{\EQO}{M}{N}{\L}{c}}}
 {\teqP{\Gamma}{\Delta}{\EPO}{\EQO}{P_i}{Q_i}{C_i = C_i' \UnionAll c}}.
\]
We have $\E' = \EG$, $\sigma_P' = \sigma_P$, $\phi_P' = \phi_P \cup \{M / ax_n\}$, and $\alpha = \NEWN ax_n. \OUT{ax_n}$.

The same reduction rule Out can be applied to reduce the process $Q_i$ into $Q_i'$,
hence the claim on the reduction of $\QQ$ holds.
We choose $\Gamma' = \Gamma$.
We have $\EGG = \EG$, $\sigma_Q' = \sigma_Q$, and $\phi_Q' = \phi_Q \cup \{N / ax_n\}$.
We also choose $c_\phi' = c_\phi \cup c$ and $c_\sigma' = c_\sigma$.
The substitutions $\sigma_P$, $\sigma_Q$ are not extended by the reduction, and the typing environment does not change,
which trivially proves the claim regarding the substitutions.

Moreover, since only $M$ and $N$ are added to the frames in the reduction, $\Pi'$ suffices to prove the claim that $\teqTc{\Gamma}{\Delta}{\EPO}{\EQO}{\phi_P'}{\phi_Q'}{\L}{c_\phi'}$.
Since all processes other that $P_i$ and $Q_i$ are unchanged by the reduction (and since the typing environment is also unchanged), $\Pi$ suffices to proves the claim
that $\forall j.\; \teqP{\Gamma'}{\Delta}{\EPO}{\EQO}{P_j'}{Q_j'}{C_j'}$ (with $C_j' = C_j$ for $j\neq i$).

Thus, in this case, it only remains to be proved that
$\inst{({\UnionCart}_j C_j') \UnionAll c_\phi'}{\sigma_P'}{\sigma_Q'}\UnionAll c_\sigma'$ is consistent. 
This constraint set is equal to
\[\inst{({\UnionCart}_{j\neq i} C_j) \UnionCart C_i' \UnionAll (c_\phi \cup c)}{\sigma_P}{\sigma_Q} \UnionAll c_\sigma\]
\ie to
\[\inst{({\UnionCart}_{j\neq i} C_j) \UnionCart (C_i' \UnionAll c) \UnionAll c_\phi}{\sigma_P}{\sigma_Q} \UnionAll c_\sigma\]
\ie
\[\inst{({\UnionCart}_{j} C_j) \UnionAll c_\phi}{\sigma_P}{\sigma_Q} \UnionAll c_\sigma\]
which is consistent 
by hypothesis. Hence the claim holds in this case.
\end{itemize}

\vspace{1em}
In the remaining cases, from the form of the typing rules for processes, the head symbol of neither $P_i$ nor $Q_i$ is $\OUTNA$.
Thus, the reduction applied to $P_i$ (from the assumption), as well as the one applied to $Q_i$ (we still must show it exists and is the same as for $P_i$, except in case \PIfLR, where $P_i$ can follow one branch of the conditional while $Q_i$ follows the other), cannot be Out.
Therefore no new term is output on either side, and $\phi_P' = \phi_P$ and $\phi_Q' = \phi_Q$. Hence the claim on the
domains of the frames holds by assumption.
Moreover, as we will see, in all cases $\Gamma'$ is either $\Gamma$, or $\Gamma, x:T$ where $x$ is a variable declared in
(the head of) $P_i$ and $Q_i$, and $T$ is not a union type.
The condition that $\EG = \EGG$ formulated previously will thus hold.

We choose $c_\phi'=c_\phi$.
The claim that $\teqTc{\Gamma'}{\Delta}{\EPO}{\EQO}{\phi_P'}{\phi_Q'}{\L}{c_\phi'}$ is then actually that
$\teqTc{\Gamma'}{\Delta}{\EPO}{\EQO}{\phi_P}{\phi_Q}{\L}{c_\phi}$,
which is true by Lemma~\ref{lem-proof:typing-contextinclusion},
since by hypothesis $\teqTc{\Gamma}{\Delta}{\EPO}{\EQO}{\phi_P}{\phi_Q}{\L}{c_\phi}$.

Besides, in the cases where we choose $\Gamma' = \Gamma$ then it is true (by hypothesis) that for $j\neq i$, $\teqP{\Gamma'}{\Delta}{\EPO}{\EQO}{P_j'}{Q_j'}{C_j}$.
In the cases where we choose $\Gamma' = \Gamma,x:T$, where $x$ is bound in $P_i$ and $Q_i$, then, since the processes are assumed to use different variable names, $x$ does not appear in $P_j$ or $Q_j$ (for $j\neq i$).
Hence, if $j \neq i$, using the assumption that $\teqP{\Gamma}{\Delta}{\EPO}{\EQO}{P_j}{Q_j}{C_j}$, by
Lemma~\ref{lem-proof:typing-contextinclusion}, we have
$\teqP{\Gamma'}{\Delta}{\EPO}{\EQO}{P_j'}{Q_j'}{C_j'}$,
where $C_j' = \{(c,\Gamma_c \cup \{x:T\})|(c,\Gamma_c)\in C_j\}$.

Hence, for each remaining possible last rule of $\Pi_i$, we only have to show that:
\begin{enumerate}[a)]
\item
\label{item:eqt-red} The same reduction rule can be applied to $Q_i$ as to $P_i$, with the same action.
(Except in the case of the rule \PIfLR, as we will see, where rule If-Then may be applied on one side while rule If-Else 
is applied on the other side, but this has no influence on the argument, as these two rules both represent a silent
action, and have a very similar form.)
\item
\label{item:eqt-subst}$(\sigma_P',\sigma_Q')$ are ground, and $\wtc{\Delta}{\EPO}{\EQO}{\sigma_P'}{\sigma_Q'}{\Gamma'}{c_\sigma'}$ for some set of constraints $c_\sigma'$.
Since at most one variable $x$ is added to the substitutions in the reduction, we only have to check that condition on this 
variable, \ie $\teqTc{\novar{\Gamma'}}{\Delta}{\EPO}{\EQO}{\sigma_P'(x)}{\sigma_Q'(x)}{\Gamma'(x)}{c_x}$ for some $c_x$.
We can then choose $c_\sigma' = c_\sigma \cup c_x$.
As we will see in the proof, we will always have $c_x \subseteq \inst{c_\phi}{\sigma_P}{\sigma_Q} \cup c_\sigma$.
\item
\label{item:eqt-proc}the new processes obtained by reducing $P_i$ and $Q_i$ are typably equivalent in
$\Gamma'$, 
 with a constraint $C_i'$, such that

\[\inst{({\UnionCart}_{j\neq i} C_j) \UnionCart C_i' \UnionAll c_\phi}{\sigma_P}{\sigma_Q} \UnionAll c_\sigma\]
is consistent. 

\bigskip 

The actual claim, from the statement of the lemma, is that
\[\inst{({\UnionCart}_{j\neq i} C_j') \UnionCart C_i' \UnionAll c_\phi}{\sigma_P'}{\sigma_Q'} \UnionAll c_\sigma'\]
is consistent, 
 but we can show that the previous condition is sufficient.

In the case where $\Gamma = \Gamma'$, we have $\sigma_P' = \sigma_P$, $\sigma_Q' = \sigma_Q$,
$C_j' = C_j$ for $j\neq i$, and $c_\sigma' = c_\sigma$. Thus the proposed condition is clearly sufficient (it is even necessary in this case).

In the case where $\Gamma' = \Gamma, x:T$ for some $T$ which is not a union type, and the substitutions $\sigma_P'$, $\sigma_Q'$ are $\sigma_P$, $\sigma_Q$ extended with a term associated to $x$, the proof that the condition is sufficient
is more involved.
First, we show that $({\UnionCart}_{j\neq i} C_j') \UnionCart C_i' = ({\UnionCart}_{j\neq i} C_j) \UnionCart C_i'$.
Indeed, if $S$ denotes the set $({\UnionCart}_{j\neq i} C_j') \UnionCart C_i'$, we have
\begin{align*}
S	&= \{(\bigcup_{j} c'_j,\bigcup_{j} \Gamma'_j) \,|\,
\forall j.\; (c'_j, \Gamma'_j)\in C_j' \;\wedge\;\forall j, j'.\; \Gamma'_j \text{ and } \Gamma'_{j'} \text{ are compatible})\}\\
	&= \{(c_i' \cup \bigcup_{j\neq i} c_j,
	   \Gamma_i' \cup \bigcup_{j\neq i} (\Gamma_j, x:T)\,|\,
	   (c_i',\Gamma_i') \in C_i'
	   \;\wedge\;
	   (\forall j\neq i.\; (c_j, \Gamma_j)\in C_j)
	   \;\wedge\\
	&\qquad
	(\forall j\neq i, j'\neq i.\; (\Gamma_j, x:T) \text{ and } (\Gamma_{j'}, x:T) \text{ are compatible})
	\;\wedge\; 
	(\forall j\neq i.\; \Gamma_i' \text{ and } (\Gamma_{j}, x:T) \text{ are compatible}))\}
\end{align*}
since we already know that for $j\neq i$, $C_j' = \{(c,\Gamma_c \cup \{x:T\})|(c,\Gamma_c)\in C_j\}$.
Assuming we show that $\teqP{\Gamma,x:T}{\Delta}{\EPO}{\EQO}{P_i'}{Q_i'}{C_i'}$, by Lemma~\ref{lem-proof:env-const-branch}, we will also have that all the $\Gamma_i'$ appearing in the elements of $C_i'$ contain $x:T$ (since $T$ is not a union type).
Hence:
\begin{align*}
S	&= \{(c_i' \cup \bigcup_{j\neq i} c_j,
	   \Gamma_i' \cup (\bigcup_{j\neq i} \Gamma_j)\,|\,
	   (c_i',\Gamma_i') \in C_i' \;\wedge\; (\forall j\neq i.\; (c_j, \Gamma_j)\in C_j) \;\wedge\\
  &\qquad (\forall j\neq i, j'\neq i.\; \Gamma_j \text{ and } \Gamma_{j'} \text{ are compatible})
	\;\wedge\; 
	(\forall j\neq i.\; \Gamma_i' \text{ and } \Gamma_{j} \text{ are compatible}))\}\\
	&= ({\UnionCart}_{j\neq i} C_j) \UnionCart C_i'
\end{align*}

It is thus sufficient to ensure the consistency of
\[\inst{S \UnionAll c_\phi}{\sigma_P'}{\sigma_Q'} \UnionAll c_\sigma \UnionAll c_x,\]
\ie, since $c_\sigma$ and $c_x$ are ground (since the substitutions are), that
\[\inst{S \UnionAll c_\phi \UnionAll c_\sigma \UnionAll c_x}{\sigma_P'}{\sigma_Q'}\]
is consistent. 
Using Lemma~\ref{lem-proof:cons-subset}, since $\sigma_P' = \sigma_P \cup \{\sigma_P'(x)/x\}$ (and similarly for $Q$),
it then suffices to show the consistency of
\[\inst{\inst{S \UnionAll c_\phi}{\sigma_P}{\sigma_Q} \UnionAll c_\sigma \UnionAll c_x}{\sigma_P'(x)/x}{\sigma_Q'(x)/x}\]
which is equal to
\[\inst{\inst{S \UnionAll c_\phi}{\sigma_P}{\sigma_Q} \UnionAll c_\sigma}{\sigma_P'(x)/x}{\sigma_Q'(x)/x}\]
since $c_x \subseteq \inst{c_\phi}{\sigma_P}{\sigma_Q} \cup c_\sigma$
(by point~\ref{item:eqt-subst}).
Moreover, as we observed previously, the environments in all elements in 
$\inst{S \UnionAll c_\phi}{\sigma_P}{\sigma_Q} \UnionAll c_\sigma$
contain $x:T$.

Therefore by Lemma~\ref{lem-proof:cons-subset}, since $\teqTc{\novar{\Gamma}}{\Delta}{\EPO}{\EQO}{\sigma_P'(x)}{\sigma_Q'(x)}{T}{c_x}$
(as we will show, as point~\ref{item:eqt-subst}), it suffices to ensure that
\[\inst{S \UnionAll c_\phi}{\sigma_P}{\sigma_Q} \UnionAll c_\sigma\]
is consistent,
to prove the claim.
This is the condition stated at the beginning of this point, since
$S = ({\UnionCart}_{j\neq i} C_j) \UnionCart C_i'$.

\end{enumerate}

\bigskip

\vspace{1em}
We can now prove the remaining cases for the last rule of $\Pi_i$,
that is to say the cases of the rules \PIn, \PLet, \PLetLR, \PIfL, \PIfLR, \PIfS, \PIfLRinf, \PIfP, \PIfI, \PIfLRp.

\begin{itemize}
\item \case{\PIn}: then $P_i = \IN{x}. P_i'$ and reduces to $P_i'$ via the In rule, and $Q_i = \IN{x}.Q_i'$ for some $Q_i'$.
In addition 
\[\Pi_i =
\inferrule
 {\inferrule*{\Pi}{\teqP{\Gamma, x:\L}{\Delta}{\EPO}{\EQO}{P_i'}{Q_i'}{C'_i}}}
 {\teqP{\Gamma}{\Delta}{\EPO}{\EQO}{P_i}{Q_i}{C_i = C_i'}}.
\]
We have $\alpha = \IN{R}$ for some attacker recipe $R$ such that $\var{R} \subseteq \dom{\phi_P}$,
and $\eval{R\phi_P\sigma_P} \neq \bot$.
We also have $\E' = \EG$, $\sigma_P' = \sigma_P \cup \{\eval{R\phi_P\sigma_P} / x\}$, $\phi_P' = \phi_P$.

The same reduction rule In can be applied to reduce the process $Q_i$ into $Q'$.
Indeed,
\begin{itemize}
\item $\var{R} \subseteq \dom{\phi_Q}$ since $\dom{\phi_Q} = \dom{\phi_P}$ by hypothesis;
\item $\eval{R\phi_Q\sigma_Q} \neq \bot$. This follows from Lemma \ref{lem-proof:l-type-recipe}, using 
the fact that by Lemma \ref{lem-proof:subst-typing}, $\teqTc{\novar{\Gamma}}{\Delta}{\EPO}{\EQO}{\phi_P\sigma_P}{\phi_Q\sigma_Q}{\L}{c}$, for some $c \subseteq \inst{c_\phi}{\sigma_P}{\sigma_Q} \cup c_\sigma$.
\end{itemize}
Therefore point \ref{item:eqt-red} holds.

We choose $\Gamma' = \Gamma, x:\L$.
We have $\sigma_Q' = \sigma_Q \cup \{\eval{R\phi_Q\sigma_Q}/x\}$.

Lemmas \ref{lem-proof:subst-typing} and \ref{lem-proof:l-type-recipe}, previously evoked, guarantee that 
\[\teqTc{\novar{\Gamma}}{\Delta}{\E_P}{\E_Q}{\eval{R\phi_P\sigma_P}}{\eval{R\phi_Q\sigma_Q}}{\L}{c'}\]
for some $c' \subseteq \inst{c_\phi}{\sigma_P}{\sigma_Q} \cup c_\sigma$.
This proves point \ref{item:eqt-subst}.

Moreover, $\Pi$ and the fact that
\[\inst{({\UnionCart}_{j\neq i} C_j) \UnionCart C_i' \UnionAll c_\phi}{\sigma_P}{\sigma_Q} \UnionAll c_\sigma =
\inst{({\UnionCart}_{j} C_j) \UnionAll c_\phi}{\sigma_P}{\sigma_Q} \UnionAll c_\sigma\]

which is consistent 
by hypothesis, prove point \ref{item:eqt-proc} and conclude this case.

\item \case{\PLet}: then $P_i = \LET{x}{d(y)}{P_i'}{P_i''}$ and $Q_i = \LET{x}{d(y)}{Q_i'}{Q_i''}$ for some $P_i'$, $P_i''$, 
$Q_i'$, $Q_i''$.
$P_i$ reduces to either $P_i'$ via the Let-In rule, or $P_i''$ via the Let-Else rule.
In addition 
\[\Pi_i =
\inferrule
 {x\notin \dom{\Gamma} \\
 \inferrule*{\Pi}{\tDestnew{\Gamma}{d(y)}{T}} \\
 \inferrule*{\Pi'}{\teqP{\Gamma, x:T}{\Delta}{\EPO}{\EQO}{P_i'}{Q_i'}{C_i'}} \\
 \inferrule*{\Pi''}{\teqP{\Gamma}{\Delta}{\EPO}{\EQO}{P_i''}{Q_i''}{C_i''}}}
 {\teqP{\Gamma}{\Delta}{\EPO}{\EQO}{P_i}{Q_i}{C_i = C_i' \cup C_i''}}.
\]

We have $\alpha = \silentAction$.

Since $\teqTc{\novar{\Gamma}}{\Delta}{\EPO}{\EQO}{\sigma_P(y)}{\sigma_Q(y)}{\Gamma(y)}{c_y}$ (for some $c_y \subseteq c_\sigma$, by hypothesis), and using $\Pi$, by
Lemma \ref{lem-proof:app-dest}, we have:
\[\eval{d(\sigma_P(y))} \neq \bot \Longleftrightarrow \eval{d(\sigma_Q(y))} \neq \bot\]
Therefore, if rule Let-In is applied to $P_i$ then it can also be applied to reduce $Q_i$ into $Q_i'$,
and if the rule applied to $P_i$ is Let-Else then it can also be applied to reduce $Q_i$ into $Q_i''$.
This proves point \ref{item:eqt-red}.
We prove here the Let-In case. The Let-Else case is similar (although slightly easier, since no new variable is added to
the substitutions).

In this case we have $\sigma_P' = \sigma_P \cup \{\eval{d(\sigma_P(y))} / x\}$ and $\sigma_Q' = \sigma_Q \cup \{\eval{d(\sigma_Q(y))} / x\}$.

By Lemma~\ref{lem-proof:app-dest}, we know in this case that there exists $c \subseteq c_y$ such that
$\teqTc{\novar{\Gamma}}{\Delta}{\EPO}{\EQO}{\eval{d(\sigma_P(y))}}{\eval{d(\sigma_Q(y))}}{T}{c}$.
Thus, by Lemma~\ref{lem-proof:ground-term-ortype}, there exists $T' \in \branch{T}$ such that
$\teqTc{\novar{\Gamma}}{\Delta}{\EPO}{\EQO}{\eval{d(\sigma_P(y))}}{\eval{d(\sigma_Q(y))}}{T'}{c}$.

We choose $\Gamma' = \Gamma, x:T'$.
Since $\Gamma$ does not contain union types, $\Gamma' \in \branch{\Gamma, x:T}$.

Since $c \subseteq c_y \subseteq c_\sigma$ and $\teqTc{\novar{\Gamma}}{\Delta}{\EPO}{\EQO}{\eval{d(\sigma_P(y))}}{\eval{d(\sigma_Q(y))}}{T'}{c}$, point~\ref{item:eqt-subst} holds.

We now prove that point~\ref{item:eqt-proc} holds.
Using $\Pi'$, we have $\teqP{\Gamma, x:T}{\Delta}{\EPO}{\EQO}{P_i'}{Q_i'}{C_i'}$.
Hence, by Lemma~\ref{lem-proof:type-processes-branches}, there exists $C_i''' \subseteq C_i' (\subseteq C_i)$
such that $\teqP{\Gamma'}{\Delta}{\EPO}{\EQO}{P_i'}{Q_i'}{C_i'''}$.

Since $C_i''' \subseteq C_i$, we have
\[\inst{({\UnionCart}_{j\neq i} C_j) \UnionCart C_i''' \UnionAll c_\phi}{\sigma_P}{\sigma_Q} \UnionAll c_\sigma
\subseteq
\inst{({\UnionCart}_j C_j) \UnionAll c_\phi}{\sigma_P}{\sigma_Q} \UnionAll c_\sigma.
\]
This last constraint set is consistent 
by hypothesis.
Hence, by Lemma \ref{lem-proof:cons-subset},
$\inst{({\UnionCart}_{j\neq i} C_j) \UnionCart C_i''' \UnionAll c_\phi}{\sigma_P}{\sigma_Q} \UnionAll c_\sigma$
is also consistent. 
This proves point \ref{item:eqt-proc} and concludes this case.

\item \case{\PLetLR}: then $P_i = \LET{x}{d(y)}{P_i'}{P_i''}$ and $Q_i = \LET{x}{d(y)}{Q_i'}{Q_i''}$ for some $P_i'$, $P_i''$, $Q_i'$, $Q_i''$.

$P_i$ reduces to either $P_i'$ via the Let-In rule, or $P_i''$ via the Let-Else rule.

In addition 
\[\Pi_i =
\inferrule
 {x\notin \dom{\Gamma} \\
 \Gamma(y) = \LRTn{l}{\oneorinf}{m}{l'}{\oneorinf}{n} \\
 \inferrule*{\Pi''}{\teqP{\Gamma}{\Delta}{\EPO}{\EQO}{P_i''}{Q_i''}{C_i''}}}
 {\teqP{\Gamma}{\Delta}{\EPO}{\EQO}{P_i}{Q_i}{C_i (= C_i'')}}.
\]

We have $\alpha = \silentAction$.

By hypothesis, $\sigma_P$, $\sigma_Q$ are ground and $\wtc{\Delta}{\E_P}{\E_Q}{\sigma_P}{\sigma_Q}{\Gamma}{c_\sigma}$.
Hence, by definition of the well-typedness of substitutions, there exists $c_y \subseteq c_\sigma$ such that
$\teqTc{\novar{\Gamma}}{\Delta}{\EPO}{\EQO}{\sigma_P(y)}{\sigma_Q(y)}{\LRTn{l}{\oneorinf}{m}{l'}{\oneorinf}{n}}{c_y}$.
Therefore by Lemma~\ref{lem-proof:lr-ground},
$\sigma_P(y) = m$ and $\sigma_Q(y) = n$.

Since $m$, $n$ are nonces, $\eval{d(m)} = \eval{d(n)} = \bot$, and
we thus have $\eval{d(\sigma_P(y))} = \bot$.
Therefore the reduction rule applied to $P_i$ can only be Let-Else, and $P_i$ is reduced to $P_i''$.
Since we also have $\eval{d(\sigma_Q(y))} = \bot$, this rule can also be applied to reduce $Q_i$ into $Q_i''$.
This proves point \ref{item:eqt-red}.

We therefore have $\sigma_P' = \sigma_P$ and $\sigma_Q' = \sigma_Q$.
We choose $\Gamma' = \Gamma$.

Since the substitutions and typing environments are unchanged by the reduction, point \ref{item:eqt-subst} clearly holds.

Moreover, $\Pi''$, and the fact that
\[\inst{({\UnionCart}_{j\neq i} C_j) \UnionCart C_i'' \UnionAll c_\phi}{\sigma_P}{\sigma_Q} \UnionAll c_\sigma = 
\inst{({\UnionCart}_j C_j) \UnionAll c_\phi}{\sigma_P}{\sigma_Q} \UnionAll c_\sigma
\]
which is consistent by hypothesis, prove point \ref{item:eqt-proc} and conclude this case.

%
\item \case{\PIfL}: then $P_i = \ITE{M}{M'}{P_i^\top}{P_i^\bot}$ and $Q_i = \ITE{N}{N'}{Q_i^\top}{Q_i^\bot}$ for some $Q_i^\top$, $Q_i^\bot$.
$P_i$ reduces to $P_i'$ which is either $P_i^\top$ via the If-Then rule, or $P_i^\bot$ via the If-Else rule.
In addition 
\[\Pi_i =
\inferrule{%
 \inferrule*{\Pi^\top}{\teqP{\Gamma}{\Delta}{\EPO}{\EQO}{P_i^\top}{Q_i^\top}{C_i^\top}}\\
 \inferrule*{\Pi^\bot}{\teqP{\Gamma}{\Delta}{\EPO}{\EQO}{P_i^\bot}{Q_i^\bot}{C_i^\bot}}\\\\
 \inferrule*{\Pi}{\teqTc{\Gamma}{\Delta}{\EPO}{\EQO}{M}{N}{\L}{c}}\\
 \inferrule*{\Pi'}{\teqTc{\Gamma}{\Delta}{\EPO}{\EQO}{M'}{N'}{\L}{c'}}}
 {\teqP{\Gamma}{\Delta}{\EPO}{\EQO}{P_i}{Q_i}{C_i = (C_i^\top \cup C_i^\bot) \UnionAll (c \cup c')}}
\]

We have $\alpha = \silentAction$, and $\E' = \EG$.


Since $\teqTc{\Gamma}{\Delta}{\EPO}{\EQO}{M}{N}{\L}{c}$,
by Lemma~\ref{lem-proof:subst-typing}, there exists $c'' \subseteq \inst{c}{\sigma_P}{\sigma_Q} \cup c_\sigma$ such that
$\teqTc{\novar{\Gamma}}{\Delta}{\EPO}{\EQO}{M\sigma_P}{N\sigma_Q}{\L}{c''}$.
Similarly, there exists $c''' \subseteq \inst{c'}{\sigma_P}{\sigma_Q} \cup c_\sigma$ such that
$\teqTc{\novar{\Gamma}}{\Delta}{\EPO}{\EQO}{M'\sigma_P}{N'\sigma_Q}{\L}{c'''}$.

Let $\phi = \{M\sigma_P/x, M'\sigma_P/y\}$ and $\phi' = \{N\sigma_Q/x, N'\sigma_Q/y\}$.
We then have $\teqTc{\novar{\Gamma}}{\Delta}{\EPO}{\EQO}{\phi}{\phi'}{\L}{c''\cup c'''}$.

Let us prove that $c \cup c'''$ is consistent in some typing environment.
By hypothesis, $\sigma_P$, $\sigma_Q$ are ground and $\wtc{\Delta}{\EPO}{\EQO}{\sigma_P}{\sigma_Q}{\Gamma}{c_\sigma}$.
Hence, by Lemma~\ref{lem-proof:ground-subst-branch}, there exists $\Gamma'' \in \branch{\Gamma}$ such that $\wtc{\Delta}{\EPO}{\EQO}{\sigma_P}{\sigma_Q}{\Gamma''}{c_\sigma}$.
By Lemma~\ref{lem-proof:all-branch-const}, there exists $(c_1,\Gamma''')\in C_i$ such that $\Gamma'' \subseteq \Gamma'''$.
Since $C_i = (C_i^\top \cup C_i^\bot) \UnionAll (c\cup c')$,
$c_1$ is of the form $c_2 \cup c \cup c'$ for some $c_2$.

As we noted previously, $\inst{C_i}{\sigma_P}{\sigma_Q} \UnionAll c_\sigma$ is consistent. 
Therefore, by Lemma~\ref{lem-proof:cons-subset}, $\{(\inst{c \cup c'}{\sigma_P}{\sigma_Q} \cup c_\sigma, \Gamma''')\}$ is consistent. 
Hence, by the same Lemma, $c'' \cup c'''$ is also consistent in $\Gamma'''$. 

Thus, by Lemma~\ref{lem-proof:l-cons-stat-eq}, $\phi$ and $\phi'$ are statically equivalent.
Hence, in particular, $M\sigma_P = M'\sigma_P \Longleftrightarrow N\sigma_Q = N'\sigma_Q$.

Therefore, if rule If-Then is applied to $P_i$ then it can also be applied to reduce $Q_i$ into $Q_i^\top$,
and if the rule applied to $P_i$ is If-Else then it can also be applied to reduce $Q_i$ into $Q_i^\bot$.
This proves point \ref{item:eqt-red}.
We prove here the If-Then case. The If-Else case is similar.

We choose $\Gamma' = \Gamma$.
We have $\sigma_P' = \sigma_P$ and $\sigma_Q' = \sigma_Q$.

Since the substitutions and environments do not change in this reduction, point \ref{item:eqt-subst}
trivially holds.

Moreover, by hypothesis, 
\[\inst{({\UnionCart}_{j\neq i} C_j) \UnionCart (C_i^\top \cup C_i^\bot) \UnionAll (c \cup c' \cup c_\phi)}{\sigma_P}{\sigma_Q} \UnionAll c_\sigma\]
is consistent. 
Thus by Lemma~\ref{lem-proof:cons-subset},
\[
\inst{({\UnionCart}_{j\neq i} C_j) \UnionCart (C_i^\top \cup C_i^\bot) \UnionAll c_\phi}{\sigma_P}{\sigma_Q} \UnionAll c_\sigma
\]
is also consistent. 
Since, using
$C_i' = C_i^\top$ and $C_i = (C_i^\top \cup C_i^\bot) \UnionAll \left\{ M \eqC N, M' \eqC N'\right\}$,
\[\inst{({\UnionCart}_{j\neq i} C_j) \UnionCart C_i' \UnionAll c_\phi}{\sigma_P}{\sigma_Q} \subseteq 
\inst{({\UnionCart}_{j\neq i} C_j) \UnionCart (C_i^\top \cup C_i^\bot) \UnionAll c_\phi}{\sigma_P}{\sigma_Q},
\]
we have by Lemma~\ref{lem-proof:cons-subset}
that $\inst{({\UnionCart}_{j\neq i} C_j) \UnionCart C_i' \UnionAll c_\phi}{\sigma_P}{\sigma_Q} \UnionAll c_\sigma$ is consistent. 
$\Pi^\top$ and this fact prove point \ref{item:eqt-proc} and conclude this case.
%
%
%
%
%
\item \case{\PIfLR}: then $P_i = \ITE{M_1}{M_2}{P_i^\top}{P_i^\bot}$ and $Q_i = \ITE{N_1}{N_2}{Q_i^\top}{Q_i^\bot}$ for some $Q_i^\top$, $Q_i^\bot$.
$P_i$ reduces to $P_i'$ which is either $P_i^\top$ via the If-Then rule, or $P_i^\bot$ via the If-Else rule.
In addition 
\[\Pi_i =
\inferrule{%
 \inferrule*{\Pi}{\teqTc{\Gamma}{\Delta}{\EPO}{\EQO}{M_1}{N_1}{\LRTn{l}{1}{m}{l'}{1}{n}}{\emptyset}}\\
 \inferrule*{\Pi'}{\teqTc{\Gamma}{\Delta}{\EPO}{\EQO}{M_2}{N_2}{\LRTn{l''}{1}{m'}{l'''}{1}{n'}}{c'}}\\
 b = (\noncetypelab{l}{1}{m} \overset{?}{=} \noncetypelab{l''}{1}{m'})\\
 b' = (\noncetypelab{l'}{1}{n} \overset{?}{=} \noncetypelab{l'''}{1}{n'})\\
 \inferrule*{\Pi''}{\teqP{\Gamma}{\Delta}{\EPO}{\EQO}{P_i^b}{Q_i^{b'}}{C_i'}}}
 {\teqP{\Gamma}{\Delta}{\EPO}{\EQO}{P_i}{Q_i}{C_i = C_i'}}
\]

We have $\alpha = \silentAction$ in any case.

By hypothesis, $\sigma_P$, $\sigma_Q$ are ground and $\wtc{\Delta}{\E_P}{\E_Q}{\sigma_P}{\sigma_Q}{\Gamma}{c_\sigma}$.
Hence, by Lemma~\ref{lem-proof:subst-typing}, using $\Pi$, there exists $c''$ such that
$\teqTc{\novar{\Gamma}}{\Delta}{\EPO}{\EQO}{M_1\sigma_P}{N_1\sigma_Q}{\LRTn{l}{1}{m}{l'}{1}{n}}{c''}$.
Therefore by Lemma~\ref{lem-proof:lr-ground},
$M_1\sigma_P = m$ and $N_1\sigma_Q = n$.
Similarly we can show that $M_2\sigma_P = m'$ and $N_2\sigma_Q = n'$. 

There are four cases for $b$ and $b'$, which are all similar. We write the proof for the case where $b = \top$ and
$b' = \bot$, \ie $\noncetypelab{l}{1}{m} = \noncetypelab{l''}{1}{m'}$ and $\noncetypelab{l'}{1}{n} \neq \noncetypelab{l'''}{1}{n'}$.

Thus the reduction rule applied to $P_i$ is If-Then and $P_i' = P_i^\top$.
On the other hand, rule If-Else can be applied to reduce $Q_i$ into $Q_i' = Q_i^\bot$.
This proves point \ref{item:eqt-red} (these rules both correspond to silent actions).

We choose $\Gamma' = \Gamma$.
We have $\sigma_P' = \sigma_P$ and $\sigma_Q' = \sigma_Q$.

Since the substitutions and environments do not change in this reduction, point \ref{item:eqt-subst}
trivially holds.

Moreover, $\Pi''$ and the fact that 
\[\inst{({\UnionCart}_{j\neq i} C_j) \UnionCart C_i' \UnionAll c_\phi}{\sigma_P}{\sigma_Q} \UnionAll c_\sigma = 
\inst{({\UnionCart}_{j} C_j) \UnionAll c_\phi}{\sigma_P}{\sigma_Q} \UnionAll c_\sigma\]
prove point \ref{item:eqt-proc} and conclude this case.

\item \case{\PIfS}: then $P_i = \ITE{M}{M'}{P_i^\top}{P_i^\bot}$ and $Q_i = \ITE{N}{N'}{Q_i^\top}{Q_i^\bot}$ for some $Q_i^\top$, $Q_i^\bot$.
$P_i$ reduces to $P_i'$ which is either $P_i^\top$ via the If-Then rule, or $P_i^\bot$ via the If-Else rule.
In addition 
\[\Pi_i =
\inferrule{%
 \inferrule*{\Pi^\bot}{\teqP{\Gamma}{\Delta}{\EPO}{\EQO}{P_i^\bot}{Q_i^\bot}{C_i'}}\\
 \inferrule*{\Pi}{\teqTc{\Gamma}{\Delta}{\EPO}{\EQO}{M}{N}{\L}{c}}\\
 \inferrule*{\Pi'}{\teqTc{\Gamma}{\Delta}{\EPO}{\EQO}{M'}{N'}{\S}{c'}}}
 {\teqP{\Gamma}{\Delta}{\EPO}{\EQO}{P_i}{Q_i}{C_i = C_i'}}
\]

We have $\alpha = \silentAction$ in any case.

By hypothesis, $\sigma_P$, $\sigma_Q$ are ground and $\wtc{\Delta}{\E_P}{\E_Q}{\sigma_P}{\sigma_Q}{\Gamma}{c_\sigma}$.
Hence, by Lemma~\ref{lem-proof:subst-typing}, using $\Pi$, there exists $c''$ such that
$\teqTc{\novar{\Gamma}}{\Delta}{\EPO}{\EQO}{M\sigma_P}{N\sigma_Q}{\L}{c''}$.
Similarly we can show that $\teqTc{\novar{\Gamma}}{\Delta}{\EPO}{\EQO}{M'\sigma_P}{N'\sigma_Q}{\S}{c'''}$ for some $c'''$.

Therefore by Lemma~\ref{lem-proof:LS-disjoint},
$M\sigma_P \neq M'\sigma_P$ and $N\sigma_Q \neq N'\sigma_Q$.
Hence the reduction for $P_i$ is necessarily If-Else, which is also applicable to reduce $Q_i$ to $Q_i^\bot$.
This proves point \ref{item:eqt-red}.

We choose $\Gamma' = \Gamma$.
We have $\sigma_P' = \sigma_P$ and $\sigma_Q' = \sigma_Q$.

Since the substitutions and typing environments do not change in this reduction, point \ref{item:eqt-subst}
trivially holds.

Moreover, $\Pi''$ and the fact that 
\[\inst{({\UnionCart}_{j\neq i} C_j) \UnionCart C_i' \UnionAll c_\phi}{\sigma_P}{\sigma_Q} \UnionAll c_\sigma = 
\inst{({\UnionCart}_{j} C_j) \UnionAll c_\phi}{\sigma_P}{\sigma_Q}\]
prove point \ref{item:eqt-proc} and conclude this case.

\item \case{\PIfI}: then $P_i = \ITE{M}{M'}{P_i^\top}{P_i^\bot}$ and $Q_i = \ITE{N}{N'}{Q_i^\top}{Q_i^\bot}$ for some $Q_i^\top$, $Q_i^\bot$. This case is similar to the \PIfS case: the incompatibility of the types of
$M$, $N$ and $M'$, $N'$ ensures that the processes can only follow the else branch.

$P_i$ reduces to $P_i'$ which is either $P_i^\top$ via the If-Then rule, or $P_i^\bot$ via the If-Else rule.
In addition 
\[\Pi_i =
\inferrule{%
 \inferrule*{\Pi^\bot}{\teqP{\Gamma}{\Delta}{\EPO}{\EQO}{P_i^\bot}{Q_i^\bot}{C_i'}}\\
 \inferrule*{\Pi}{\teqTc{\Gamma}{\Delta}{\EPO}{\EQO}{M}{N}{T*T'}{c}}\\
 \inferrule*{\Pi'}{\teqTc{\Gamma}{\Delta}{\EPO}{\EQO}{M'}{N'}{\LRTn{l}{\oneorinf}{m}{l'}{\oneorinf}{n}}{c'}}}
 {\teqP{\Gamma}{\Delta}{\EPO}{\EQO}{P_i}{Q_i}{C_i = C_i'}}
\]

We have $\alpha = \silentAction$ in any case.

By hypothesis, $\sigma_P$, $\sigma_Q$ are ground and $\wtc{\Delta}{\E_P}{\E_Q}{\sigma_P}{\sigma_Q}{\Gamma}{c_\sigma}$.
Hence, by Lemma~\ref{lem-proof:subst-typing}, using $\Pi$, there exists $c''$ such that
$\teqTc{\novar{\Gamma}}{\Delta}{\EPO}{\EQO}{M\sigma_P}{N\sigma_Q}{T*T'}{c''}$.
By Lemma~\ref{lem-proof:pair-types}, this implies that $M\sigma_P$ and $N\sigma_Q$ both are pairs.
%
Similarly we can show that $\teqTc{\novar{\Gamma}}{\Delta}{\EPO}{\EQO}{M'\sigma_P}{N'\sigma_Q}{\LRTn{l}{\oneorinf}{m}{l'}{\oneorinf}{n}}{c'''}$ for some $c'''$.
By Lemma~\ref{lem-proof:lr-ground}, this implies that $M'\sigma_P = m$ and $N'\sigma_Q = n$.
Thus neither of these two terms are pairs.

Therefore
$M\sigma_P \neq M'\sigma_P$ and $N\sigma_Q \neq N'\sigma_Q$.
The end of the proof for this case is then the same as for the \PIfS case.

\item \case{\PIfP}: then $P_i = \ITE{M}{t}{P_i^\top}{P_i^\bot}$ and $Q_i = \ITE{N}{t}{Q_i^\top}{Q_i^\bot}$ for some $Q_i^\top$, $Q_i^\bot$, some messages $M$, $N$, and some $t \in C \cup \K \cup \N$.
$P_i$ reduces to $P_i'$ which is either $P_i^\top$ via the If-Then rule, or $P_i^\bot$ via the If-Else rule.
In addition 
\[\Pi_i =
\inferrule{%
 \inferrule*{\Pi^\top}{\teqP{\Gamma}{\Delta}{\EPO}{\EQO}{P_i^\top}{Q_i^\top}{C_i^\top}}\\
 \inferrule*{\Pi^\bot}{\teqP{\Gamma}{\Delta}{\EPO}{\EQO}{P_i^\bot}{Q_i^\bot}{C_i^\bot}}\\\\
 \inferrule*{\Pi}{\teqTc{\Gamma}{\Delta}{\EPO}{\EQO}{M}{N}{\L}{c}}\\
 \inferrule*{\Pi'}{\teqTc{\Gamma}{\Delta}{\EPO}{\EQO}{t}{t}{\L}{c'}}\\
 t \in \CST \cup \K \cup \N}
 {\teqP{\Gamma}{\Delta}{\EPO}{\EQO}{P_i}{Q_i}{C_i = C_i^\top \cup C_i^\bot}}
\]

We have in any case $\alpha = \silentAction$.


By hypothesis, $\sigma_P$, $\sigma_Q$ are ground and $\wtc{\Delta}{\E_P}{\E_Q}{\sigma_P}{\sigma_Q}{\Gamma}{c_\sigma}$.
Hence, by Lemma~\ref{lem-proof:subst-typing}, using $\Pi$, there exists $c'' \subseteq \inst{c}{\sigma_P}{\sigma_Q} \cup c_\sigma$ such that
$\teqTc{\novar{\Gamma}}{\Delta}{\EPO}{\EQO}{M\sigma_P}{N\sigma_Q}{\L}{c''}$.

We then show that $M\sigma_P = t$ if and only if $N\sigma_Q = t$ (note that since $t$ is ground, $t = t\sigma_P = t\sigma_Q$).
If $M\sigma_P = t$, then $\teqTc{\novar{\Gamma}}{\Delta}{\EPO}{\EQO}{t}{N\sigma_Q}{\L}{c''}$.
In all possible cases for $t$, \ie $t\in\K$, $t\in\N$, and $t\in\CST$, Lemma~\ref{lem-proof:type-key-nonce} implies
that $N\sigma_Q = t$. This proves the first direction of the equivalence, the other direction is similar. 

Therefore, if rule If-Then is applied to $P_i$ then it can also be applied to reduce $Q_i$ into $Q_i^\top$,
and if the rule applied to $P_i$ is If-Else then it can also be applied to reduce $Q_i$ into $Q_i^\bot$.
This proves point \ref{item:eqt-red}.
We prove here the If-Then case. The If-Else case is similar.

We choose $\Gamma' = \Gamma$.
We have $\sigma_P' = \sigma_P$ and $\sigma_Q' = \sigma_Q$.

Since the substitutions and typing environments do not change in this reduction, point \ref{item:eqt-subst}
trivially holds.

Moreover, by hypothesis, 
\[\inst{({\UnionCart}_{j\neq i} C_j) \UnionCart (C_i^\top \cup C_i^\bot) \UnionAll c_\phi}{\sigma_P}{\sigma_Q} \UnionAll c_\sigma\]
is consistent. 
Since, using
$C_i' = C_i^\top$ and $C_i = (C_i^\top \cup C_i^\bot)$, we have
\[\inst{({\UnionCart}_{j\neq i} C_j) \UnionCart C_i' \UnionAll c_\phi}{\sigma_P}{\sigma_Q} \UnionAll c_\sigma \subseteq 
\inst{({\UnionCart}_{j\neq i} C_j) \UnionCart (C_i^\top \cup C_i^\bot) \UnionAll c_\phi}{\sigma_P}{\sigma_Q} \UnionAll c_\sigma,
\]
we have by Lemma~\ref{lem-proof:cons-subset}
that $\inst{({\UnionCart}_{j\neq i} C_j) \UnionCart C_i' \UnionAll c_\phi}{\sigma_P}{\sigma_Q} \UnionAll c_\sigma$ is consistent. 
This fact proves point \ref{item:eqt-proc} and concludes this case.
%
%
%
\item \case{\PIfLRinf}: then $P_i = \ITE{M_1}{M_2}{P_i^\top}{P_i^\bot}$ and $Q_i = \ITE{N_1}{N_2}{Q_i^\top}{Q_i^\bot}$ for some $Q_i^\top$, $Q_i^\bot$.
$P_i$ reduces to $P_i'$ which is either $P_i^\top$ via the If-Then rule, or $P_i^\bot$ via the If-Else rule.
In addition 
\[\Pi_i =
\inferrule{%
 \inferrule*{\Pi}{\teqTc{\Gamma}{\Delta}{\EPO}{\EQO}{M_1}{N_1}{\LRTn{l}{\infty}{m}{l'}{\infty}{n}}{c_1}}\\
 \inferrule*{\Pi'}{\teqTc{\Gamma}{\Delta}{\EPO}{\EQO}{M_2}{N_2}{\LRTn{l}{\infty}{m}{l'}{\infty}{n}}{c_2}}\\
 \inferrule*{\Pi^\top}{\teqP{\Gamma}{\Delta}{\EPO}{\EQO}{P_i^\top}{Q_i^\top}{C_i^\top}}\\
 \inferrule*{\Pi^\bot}{\teqP{\Gamma}{\Delta}{\EPO}{\EQO}{P_i^\bot}{Q_i^\bot}{C_i^\bot}}}
 {\teqP{\Gamma}{\Delta}{\EPO}{\EQO}{P_i}{Q_i}{C_i = C_i^\top \cup C_i^\bot}}
\]

We have $\alpha = \silentAction$ in any case.

By hypothesis, $\sigma_P$, $\sigma_Q$ are ground and $\wtc{\Delta}{\E_P}{\E_Q}{\sigma_P}{\sigma_Q}{\Gamma}{c_\sigma}$.
Hence, by Lemma~\ref{lem-proof:subst-typing}, using $\Pi$, there exists $c''$ such that
$\teqTc{\novar{\Gamma}}{\Delta}{\EPO}{\EQO}{M_1\sigma_P}{N_1\sigma_Q}{\LRTn{l}{\infty}{m}{l'}{\infty}{m}}{c''}$.
Therefore by Lemma~\ref{lem-proof:lr-ground},
$M_1\sigma_P = m$ and $N_1\sigma_Q = n$.
Similarly we can show that $M_2\sigma_P = m$ and $N_2\sigma_Q = n$. 

Hence $M_1' = M_2'$ and $N_1' = N_2'$.

Thus the reduction rule applied to $P_i$ is If-Then and $P_i' = P_i^\top$.
On the other hand, rule If-Then can also be applied to reduce $Q_i$ into $Q_i' = Q_i^\top$.
This proves point \ref{item:eqt-red}.

Note that we still need to type the other branch, even though it is not used here, as when replicating the process this 
test may fail if $M_1$, $N_1$ and $M_2$, $N_2$
are nonces from different sessions.

We choose $\Gamma' = \Gamma$.
We have $\sigma_P' = \sigma_P$ and $\sigma_Q' = \sigma_Q$.

Since the substitutions and environments do not change in this reduction, point \ref{item:eqt-subst}
trivially holds.

Moreover, $\Pi''$ and the fact that, with $C_i' = C_i^\top$, 
\begin{align*}
\inst{({\UnionCart}_{j\neq i} C_j) \UnionCart C_i' \UnionAll c_\phi}{\sigma_P}{\sigma_Q} \UnionAll c_\sigma &\subseteq 
\inst{({\UnionCart}_{j\neq i} C_j) \UnionCart (C_i^\top\cup C_i^\bot) \UnionAll c_\phi}{\sigma_P}{\sigma_Q} \UnionAll c_\sigma \\
& =
\inst{({\UnionCart}_{j} C_j) \UnionAll c_\phi}{\sigma_P}{\sigma_Q} \UnionAll c_\sigma
\end{align*}
prove point \ref{item:eqt-proc} and conclude this case.

\item \case{PIfLR'*}: then $P_i = \ITE{M_1}{M_2}{P_i^\top}{P_i^\bot}$ and $Q_i = \ITE{N_1}{N_2}{Q_i^\top}{Q_i^\bot}$ for some $Q_i^\top$, $Q_i^\bot$.
$P_i$ reduces to $P_i'$ which is either $P_i^\top$ via the If-Then rule, or $P_i^\bot$ via the If-Else rule.
In addition 
\[\Pi_i =
\inferrule{%
 \inferrule*{\Pi}{\teqTc{\Gamma}{\Delta}{\EPO}{\EQO}{M_1}{N_1}{\LRTn{l}{\oneorinf}{m}{l'}{\oneorinf}{n}}{c_1}}\\
 \inferrule*{\Pi'}{\teqTc{\Gamma}{\Delta}{\EPO}{\EQO}{M_2}{N_2}{\LRTn{l''}{\oneorinf}{m'}{l'''}{\oneorinf}{n'}}{c_2}}\\
 \noncetypelab{l}{\oneorinf}{m} \neq \noncetypelab{l''}{\oneorinf}{m'}\\
 \noncetypelab{l'}{\oneorinf}{n} \neq \noncetypelab{l'''}{\oneorinf}{n'}\\
 \inferrule*{\Pi''}{\teqP{\Gamma}{\Delta}{\EPO}{\EQO}{P_i^\bot}{Q_i^\bot}{C_i'}}}
 {\teqP{\Gamma}{\Delta}{\EPO}{\EQO}{P_i}{Q_i}{C_i = C_i'}}
\]

We have $\alpha = \silentAction$ in any case.

By hypothesis, $\sigma_P$, $\sigma_Q$ are ground and $\wtc{\Delta}{\E_P}{\E_Q}{\sigma_P}{\sigma_Q}{\Gamma}{c_\sigma}$.
Hence, by Lemma~\ref{lem-proof:subst-typing}, using $\Pi$, there exists $c''$ such that
$\teqTc{\novar{\Gamma}}{\Delta}{\EPO}{\EQO}{M_1\sigma_P}{N_1\sigma_Q}{\LRTn{l}{\oneorinf}{m}{l'}{\oneorinf}{n}}{c''}$.
Therefore by Lemma~\ref{lem-proof:lr-ground},
$M_1\sigma_P = m$ and $N_1\sigma_Q = n$.
Similarly, using Lemma~\ref{lem-proof:lr-ground}, we can show that $M_2\sigma_P = m'$ and $N_2\sigma_Q = n'$. 

Moreover, since $\noncetypelab{l}{\oneorinf}{m} \neq \noncetypelab{l''}{\oneorinf}{m'}$, we know that $m\neq m'$
(by well-formedness of the processes), and similarly $n \neq n'$.

Hence, $M_1\sigma_P\neq M_2\sigma_P$ and $N_1\sigma_Q\neq N_2\sigma_Q$.

Thus the reduction rule applied to $P_i$ is If-Else and $P_i' = P_i^\bot$.
On the other hand, rule If-Else can also be applied to reduce $Q_i$ into $Q_i' = Q_i^\bot$.
This proves point \ref{item:eqt-red}.

We choose $\Gamma' = \Gamma$.
We have $\sigma_P' = \sigma_P$ and $\sigma_Q' = \sigma_Q$.

Since the substitutions and environments do not change in this reduction, point \ref{item:eqt-subst}
trivially holds.

Moreover, $\Pi''$ and the fact that 
\[\inst{({\UnionCart}_{j\neq i} C_j) \UnionCart C_i' \UnionAll c_\phi}{\sigma_P}{\sigma_Q} \UnionAll c_\sigma = 
\inst{({\UnionCart}_{j} C_j) \UnionAll c_\phi}{\sigma_P}{\sigma_Q} \UnionAll c_\sigma\]
prove point \ref{item:eqt-proc} and conclude this case.

\end{itemize}

\end{proof}

\begin{theorem}[Typing implies trace inclusion]
\label{lem-proof:red-steq}
For all processes $P$, $Q$,
for all 
$\phi_P$, $\phi_Q$, $\sigma_P$, $\sigma_Q$, 
for all multisets of processes $\PP$, $\QQ$ for all constraints $C$, for all sequence $s$ of actions,
for all $\Gamma$ containing only keys,
\[\teqP{\Gamma}{\Delta}{\EPO}{\EQO}{P}{Q}{C},\]
and if $C$ is consistent, 
then
\[P\incTrace Q\]
that is, if
\[(\emptyset,\{P\},\emptyset,\emptyset) \redWord{s} (\E_P,\PP,\phi_P,\sigma_P),\]

then there exists a sequence $s'$ of actions, a multiset $\QQ$, a set of names $\E_Q$, a frame $\phi_Q$, a substitution $\sigma_Q$, such that 
\begin{itemize}
\item $s \eqSilent s'$
\item $(\emptyset,\{Q\},\emptyset,\emptyset) \redWord{s'} (\E_Q,\QQ,\phi_Q,\sigma_Q)$,
\item $\NEWN{\E_P}. \phi_P\sigma_P$ and $\NEWN{\E_Q}.\phi_Q\sigma_Q$ are statically equivalent.
\end{itemize}
\end{theorem}

\begin{proof}
We successively apply Lemma~\ref{lem-proof:invariant} to each of the reduction steps in the reduction 
\[(\emptyset,\{P\},\emptyset,\emptyset) \redWord{s} (\E_P,\PP,\phi_P,\sigma_P).\]

The lemma can indeed be applied successively. At each reduction step of $P$ we obtain a sequence of reduction steps for $Q$ with the same actions, and the conclusions the lemma provides imply the conditions needed for its next application.

It is clear, for the first application, that all the hypotheses of this lemma are satisfied.

In the end, we know that there exist $\Gamma'$, some constraint sets $C_i$, some $c_\phi$, $c_\sigma$, and a reduction 
\[(\emptyset,\{Q\},\emptyset,\emptyset) \redWord{s'} (\EGG,\QQ,\phi_Q,\sigma_Q)\]
with $s \eqSilent s'$, such that (among other conclusions)
\begin{itemize}
\item $\E_P = \EGG$,
\item $\sigma_P$, $\sigma_Q$ are ground and $\wtc{\Delta}{\EPO}{\EQO}{\sigma_P}{\sigma_Q}{\Gamma'}{c_\sigma}$, 
\item $\teqTc{\Gamma'}{\Delta}{\EPO}{\EQO}{\phi_P}{\phi_Q}{\L}{c_\phi}$,
\item $\dom{\phi_P} = \dom{\phi_Q}$,
\item $\forall i, \teqP{\Gamma'}{\Delta}{\EPO}{\EQO}{P_i}{Q_i}{C_i}$,
\item for all $i\neq j$, the sets of bound variables in $P_i$ and $P_j$ (resp. $Q_i$ and $Q_j$) are disjoint, and similarly for the bound names;
\item $\inst{({\UnionCart}_i C_i) \UnionAll c_\phi}{\sigma_P}{\sigma_Q} \UnionAll c_\sigma$ is consistent in $\Delta$, $\EPO$, $\EQO$.
\end{itemize}

To prove the claim, it is then sufficient to show that $\NEWN{\EGG}.\phi_P\sigma_P$ and $\NEWN{\EGG}.\phi_Q\sigma_Q$ are statically equivalent.

We have $\teqTc{\Gamma'}{\Delta}{\EPO}{\EQO}{\phi_P}{\phi_Q}{\L}{c_\phi}$ and 
$\wtc{\Delta}{\EPO}{\EQO}{\sigma_P}{\sigma_Q}{\Gamma'}{c_\sigma}$.
Hence, by Lemma~\ref{lem-proof:subst-typing}, there exists $c \subseteq \inst{c_\phi}{\sigma_P}{\sigma_Q} \cup c_\sigma$
such that $\teqTc{\novar{\Gamma'}}{\Delta}{\EPO}{\EQO}{\phi_P\sigma_P}{\phi_Q\sigma_Q}{\L}{c}$.

We will now show that $(c, \novar{\Gamma'})$ is consistent. 
Since $c \subseteq \inst{c_\phi}{\sigma_P}{\sigma_Q} \cup c_\sigma$,
by Lemma~\ref{lem-proof:cons-subset}, it suffices to show that
$(\inst{c_\phi}{\sigma_P}{\sigma_Q} \cup c_\sigma, \novar{\Gamma'})$ is consistent.

We have $\wtc{\Delta}{\EPO}{\EQO}{\sigma_P}{\sigma_Q}{\Gamma'}{c_\sigma}$.
By Lemma~\ref{lem-proof:ground-subst-branch}, there is a $\Gamma''' \in \branch{\Gamma'}$ such that $\wtc{\Delta}{\EPO}{\EQO}{\sigma_P}{\sigma_Q}{\Gamma'''}{c_\sigma}$.

By Lemma~\ref{lem-proof:env-const-branch}, there exists for all $i$ some $(c_i,\Gamma''_i)\in C_i$ such that $\Gamma''' \subseteq \Gamma''_i$.
The disjointness condition on the bound variables implies by Lemma~\ref{lem-proof:env-constraints-bound}
that for all $i$, $j$, $\Gamma''_i$ and $\Gamma''_j$ are compatible.
Thus ${\UnionCart}_i C_i$ contains $(c',\Gamma'') \eqdef (\bigcup_i c_i, \bigcup_i \Gamma''_i)$.
We have $\Gamma''' \subseteq \Gamma''$.
Therefore $\inst{({\UnionCart}_i C_i) \UnionAll c_\phi}{\sigma_P}{\sigma_Q} \UnionAll c_\sigma$, which is consistent, 
contains
$(\inst{c' \cup c_\phi}{\sigma_P}{\sigma_Q} \cup c_\sigma,\Gamma'')$.
Hence, by Lemma~\ref{lem-proof:cons-subset}, $(\inst{c_\phi}{\sigma_P}{\sigma_Q} \cup c_\sigma, \Gamma'')$ is consistent. 

Therefore, $(c, \Gamma'')$ is consistent. Since $c$ is ground, it follows from the definition of consistency
that $(c,\novar{\Gamma''})$ is also consistent.
Moreover, $\Gamma''' \subseteq \Gamma''$, and $\Gamma'''$ is a branch of $\Gamma'$.
It is then clear that $\novar{\Gamma'}\subseteq \Gamma''$.
Hence, by Lemma~\ref{lem-proof:typing-contextinclusion}, 
since $\teqTc{\novar{\Gamma'}}{\Delta}{\EPO}{\EQO}{\phi_P\sigma_P}{\phi_Q\sigma_Q}{\L}{c}$,
we have $\teqTc{\Gamma''}{\Delta}{\EPO}{\EQO}{\phi_P\sigma_P}{\phi_Q\sigma_Q}{\L}{c}$.

\vspace{1em}
Hence, we have $\teqTc{\Gamma''}{\Delta}{\EPO}{\EQO}{\phi_P\sigma_P}{\phi_Q\sigma_Q}{\L}{c}$ with $(c, \novar{\Gamma''})$
consistent. 

Moreover, $\phi_P\sigma_P$ and $\phi_Q\sigma_Q$ are ground (by well-formedness of the processes).

Therefore, by Lemma~\ref{lem-proof:l-cons-stat-eq}, the frames $\NEWN{\E_{\Gamma''}}.\phi_P\sigma_P$ and $\NEWN{\E_{\Gamma''}}.\phi_Q\sigma_Q$ are statically equivalent.

By definition of the reduction relation, and by well-formedness of the processes,
since 
\[(\emptyset,\{P\},\emptyset,\emptyset) \redWord{s} (\EGG,\PP,\phi_P,\sigma_P)\]
and
\[(\emptyset,\{Q\},\emptyset,\emptyset) \redWord{s'} (\EGG,\QQ,\phi_Q,\sigma_Q)\]
it is clear that $\names{\phi_P\sigma_P} \subseteq \EGG$ and $\names{\phi_Q\sigma_Q} \subseteq \EGG$.

Thus, the only names that are relevant to the frames are $\EGG$.

Hence, $\NEWN{\EGG}.\phi_P\sigma_P$ and $\NEWN{\EGG}.\phi_Q\sigma_Q$ are statically equivalent.
\end{proof}

This theorem corresponds to Theorem~\ref{thm:typing-sound}.

\begin{theorem}[Typing implies trace equivalence]
\label{thm-proof:typing-sound}
For all $\Gamma$ containing only keys, for all 
$P$ and $Q$, if 
\[\teqP{\Gamma}{\Delta}{\EPO}{\EQO}{P}{Q}{C}\] 
and $C$ is consistent, 
then
\[P \equivTrace Q.\]
\end{theorem}

\begin{proof}
Theorem~\ref{lem-proof:red-steq} proves that under these assumptions, $P \incTrace Q$.
This is sufficient to prove the theorem. Indeed, it is clear from the typing rules for processes and terms that
\[\teqP{\Gamma}{\Delta}{\EPO}{\EQO}{P}{Q}{C} \Leftrightarrow
 \teqP{\Gamma'}{\Delta}{\EQO}{\EPO}{Q}{P}{C'}\]
where $C'$ is the constraint obtained from $C$ by swapping the left and right hand sides of all of its elements,
and $\Gamma'$ is the environment obtained from $\Gamma$ by swapping the left and right types in all refinement types.
Clearly from the definition of consistency, $C$ is consistent 
if and only if $C'$ is.
Therefore, by symmetry, proving that the assumptions imply $P \incTrace Q$ also proves that they imply $Q \incTrace P$, and
thus $P \equivTrace Q$.
\end{proof}

%% file: proofs/typerepl.tex


\subsection{Typing replicated processes}

In this subsection, we prove the soundness result for replicated processes.

In this subsection, as well as the following ones, without loss of generality
we assume, for each infinite nonce type $\noncetypelab{l}{\infty}{m}$ appearing in the processes 
we consider, that $\N$ contains an infinite number of fresh names which we will denote by $ \{m_i\;|\; i\in\mathbb{N}\}$;
such that the $m_i$ do not appear in the processes or environments considered.
We will denote by $\NN$ the set of unindexed names and by $\NI$ the set of indexed names.
We similarly assume that for all the variables $x$ appearing in the processes,
the set $\X$ of all variables also contains variables $\{x_i\;|\;i\in\mathbb{N}\}$.
We denote $\XX$ the set of unindexed variables, and $\XI$ the set of indexed variables.

\begin{definition}[Renaming of a process]
For all process $P$,
for all $i\in\mathbb{N}$, for all environment $\Gamma$,
we define $\instProca{P}{i}$,
the renaming of $P$ for session $i$ with respect to $\Gamma$,
as the process obtained from $P$ by:
\begin{itemize}
\item for each nonce $n$ declared in $P$ by $\NEW{n}{\noncetypelab{l}{\infty}{n}}$, and each nonce $n$ such
that $\Gamma(n)=\noncetypelab{l}{\infty}{n}$ for some $l$,
replacing every occurrence of $n$ with $n_i$, and the declaration $\NEW{n}{\noncetypelab{l}{\infty}{n}}$ with
$\NEW{n_i}{\noncetypelab{l}{1}{n_i}}$;
\item replacing every occurence of a variable $x$ with $x_i$.
\end{itemize}
\end{definition}

\begin{lemma}[Typing terms with replicated names]
\label{lem-proof:typing-terms-stars}
For all $\Gamma$, 
$M$, $N$, $T$ and $c$, if 
\[\teqTc{\Gamma}{\Delta}{\E}{\E'}{M}{N}{T}{c}\] 
then for all $i, n\in\mathbb{N}$ such that $1 \leq i \leq n$,
for all $\Gamma'\in\branch{\instG{\Gamma}{i}{n}}$,

\[\teqTc{\Gamma'}{\instD{\Delta}{n}}{\instE{\E}{n}}{\instE{\E'}{n}}{\instTerma{M}{i}}{\instTerma{N}{i}}{\instTyp{T}{n}}{\instConsta{c}{i}}\]%

%
%
%
\end{lemma}
\begin{proof}
Let $\Gamma$, 
$M$, $N$, $T$, $c$ be such as assumed in the statement of the lemma.
Let $i, n\in\mathbb{N}$ such that $1\leq i \leq n$.
Let $\Gamma' \in\branch{\instG{\Gamma}{i}{n}}$.

We prove this property by induction on the proof $\Pi$ of \[\teqTc{\Gamma}{\Delta}{\E}{\E'}{M}{N}{T}{c}.\]

There are several possible cases for the last rule applied in $\Pi$.
\newcommand{\Mii}{\instTerma{M}{i}}
\newcommand{\Nii}{\instTerma{N}{i}}
\newcommand{\Gammain}{\instG{\Gamma}{i}{n}}
\newcommand{\ci}{\instConsta{c}{i}}
\newcommand{\Tn}{\instTyp{T}{n}}

\begin{itemize}

\item \case{\TNonce:}
then $M=m$ and $N=p$ for some $m, p\in\N$, $T = l$ for some $l\in\{\S, \H\}$, and

\[\Pi=
\inferrule
  {\Gamma(m) = \noncetypelab{l}{\oneorinf}{m} \\ \Gamma(p) = \noncetypelab{l}{\oneorinf}{p}}
  {\teqTc{\Gamma}{\Delta}{\E}{\E'}{m}{p}{l}{\emptyset}}.
\]

It is clear from the definition of $\Gammain$ that $\Gammain(\instTerma{m}{i}) = \noncetypelab{l}{1}{\instTerma{m}{i}}$,
and that $\Gammain(\instTerma{p}{i}) = \noncetypelab{l}{1}{\instTerma{p}{i}}$.
Hence, $\Gamma'(\instTerma{m}{i}) = \noncetypelab{l}{1}{\instTerma{m}{i}}$ and $\Gamma'(\instTerma{p}{i}) = \noncetypelab{l}{1}{\instTerma{p}{i}}$.
Then, by rule \TNonce, we have $\teqTc{\Gamma'}{\Deltan}{\En}{\EEn}{\Mii}{\Nii}{l}{\emptyset}$ and the claim holds.

\item \case{\TNonceL, \TCst, \TKey, \TPubkey, \TVkey, \THash, \THigh, \TLRone:}
Similarly to the \TNonce case, the claim follows directly from the definition of $\Gammain$, 
$\Mii$, $\Nii$, $\Tn$ and $\ci$ in these cases.

\item \case{\TEncH:}
then $T = \L$ and there exist $T'$, $k$, $c'$ such that

\[\Pi=
\inferrule
  {\inferrule*{\Pi'}{\teqTc{\Gamma}{\Delta}{\E}{\E'}{M}{N}{\encT{T'}{k}}{c'}}\\ \Gamma(k) = \skey{\S}{T'}}
  {\teqTc{\Gamma}{\Delta}{\E}{\E'}{M}{N}{\L}{c = c' \cup \{M\eqC N\}}}.
\]

By applying the induction hypothesis to $\Pi'$,
since $\instTyp{\encT{T'}{k}}{n} = \encT{\instTyp{T'}{n}}{k}$,
there exists a proof $\Pi''$ of
$\teqTc{\Gamma'}{\Deltan}{\En}{\EEn}{\Mii}{\Nii}{\encT{\instTyp{T'}{n}}{k}}{\instConsta{c'}{i}}$

In addition $\Gammain(k) = \skey{\S}{\instTyp{T'}{n}}$ by definition of $\Gammain$.
Hence $\Gamma'(k)=\skey{\S}{\instTyp{T'}{n}}$.

Therefore by rule \TEncH, we have
\[\teqTc{\Gamma'}{\Deltan}{\En}{\EEn}{\Mii}{\Nii}{\L}{\instConsta{c'}{i}\cup \{\Mii \eqC \Nii\} = \ci}\]

\item \case{\TPair, \TEnc, \TEncL, \TAenc, \TAencH, \TAencL, \TSignH, \TSignL, \THashL, \TOr:}
Similarly to the \TEncH case, the claim is proved directly by applying the induction hypothesis to the type
judgement appearing in the conditions of the last rule in these cases.

\item \case{\TVar:}
then $M=N=x$ for some $x\in\X$, and

\[\Pi=
\inferrule
  {\Gamma(x) = T}
  {\teqTc{\Gamma}{\Delta}{\E}{\E'}{x}{x}{T}{\emptyset}}.
\]

We have $\Mii = \Nii = x_i$.

Since $\Gamma'\in\branch{\Gammain}$, we have $\Gamma'(x_i)\in\branch{\Tn}$.

Hence by rule \TVar, $\teqTc{\Gamma'}{\Deltan}{\En}{\EEn}{x_i}{x_i}{\Gamma'(x_i)}{\emptyset}$.
Therefore, by rule \TOr, we have 
\[\teqTc{\Gamma'}{\Deltan}{\En}{\EEn}{x_i}{x_i}{\Tn}{\emptyset}\]
which proves the claim.

\item \case{\TLRp (the \TLRLp case is similar):}
then there exist $m, p, l$ 
such that $T = l$, and
\[\Pi=
\inferrule{
  \inferrule*{\Pi'}{\teqTc{\Gamma}{\Delta}{\E}{\E'}{M}{N}{\LRTn{l}{a}{m}{l}{a}{p}}{c}}}
  {\teqTc{\Gamma}{\Delta}{\E}{\E'}{M}{N}{l}{c}}.
\]
Let us distinguish the case where $a$ is $1$ from the case where $a$ is $\infty$.

\case{If $a$ is 1:} by applying the induction hypothesis to $\Pi'$,
since $\instTyp{\LRTn{l}{a}{m}{l}{a}{p}}{n} = \LRTn{l}{1}{m}{l}{1}{p}$, we have
\[\teqTc{\Gamma'}{\Deltan}{\En}{\EEn}{\Mii}{\Nii}{\LRTn{l}{a}{m}{l}{a}{p}}{\ci}.\]
Thus by rule \TLRp, we have \[\teqTc{\Gammain}{\Deltan}{\En}{\EEn}{\Mii}{\Nii}{l}{\ci}.\]

\case{If $a$ is $\infty$:}
by applying the induction hypothesis to $\Pi'$,
since $\instTyp{\LRTn{l}{a}{m}{l}{a}{p}}{n} = \bigvee_{1\leq j\leq n} \LRTn{l}{1}{m_j}{l}{1}{p_j}$, we have
\[\teqTc{\Gamma'}{\Deltan}{\En}{\EEn}{\Mii}{\Nii}{\bigvee_{1\leq j\leq n} \LRTn{l}{1}{m_j}{l}{1}{p_j}}{\ci}.\]

Thus, by Lemma~\ref{lem-proof:type-terms-branches}, there exists $j\in\llbracket 1, n\rrbracket$ and a proof $\Pi''$ of
\[\teqTc{\Gamma'}{\Deltan}{\En}{\EEn}{\Mii}{\Nii}{\LRTn{l}{1}{m_j}{l}{1}{p_j}}{\ci}.\]
Thus, by rule \TLRp,
\[\teqTc{\Gamma'}{\Deltan}{\En}{\EEn}{\Mii}{\Nii}{l}{\ci},\]
which proves the claim.

\item \case{\TLRVar:} this case is similar to the \TLRp case, but only the case where $a$ is $1$ is possible.

\item \case{\TSub:}
then there exists $T'\subtyp T$ such that
\[\Pi=
\inferrule{
  \inferrule*{\Pi'}{\teqTc{\Gamma}{\Delta}{\E}{\E'}{M}{N}{T'}{c}}\\ T'\subtyp T}
  {\teqTc{\Gamma}{\Delta}{\E}{\E'}{M}{N}{T}{c}}.
\]

By applying the induction hypothesis to $\Pi'$, we have
\[\teqTc{\Gamma'}{\Deltan}{\En}{\EEn}{\Mii}{\Nii}{\instTyp{T'}{n}}{\ci}.\]

Since it is clear by induction on the subtyping rules that $T'\subtyp T$ implies that $\instTyp{T'}{n}\subtyp \Tn$,
rule \TSub can be applied and proves the claim.

\item \case{\TLRinf:}
then $M=m$, $N=p$, $c=\emptyset$, and $T = \LRTn{l}{\infty}{m}{l'}{\infty}{p}$ for some $m,p\in\N$, 
$c = \emptyset$, and

\[\Pi=
\inferrule{
   \Gamma(m) = \noncetypelab{l}{\infty}{m} \\ \Gamma(p) = \noncetypelab{l'}{\infty}{p}}
  {\teqTc{\Gamma}{\Delta}{\E}{\E'}{m}{p}{\LRTn{l}{\infty}{m}{l'}{\infty}{p}}{\emptyset}}.
\]

We have by definition $\Mii = m_i$ and $\Nii=p_i$,
and $\Gammain(m_i) = \noncetypelab{l}{1}{m_i}$, and $\Gammain(p_i) = \noncetypelab{l'}{1}{p_i}.$
Thus $\Gamma'(m_i) = \noncetypelab{l}{1}{m_i}$, and $\Gamma'(p_i) = \noncetypelab{l'}{1}{p_i}.$
Hence by rule \TLRone, we have $\teqTc{\Gamma'}{\Delta}{\E}{\E'}{\Mii}{\Nii}{\LRTn{l}{1}{m_i}{l'}{1}{p_i}}{\emptyset}$.

In addition, $\instTyp{\LRTn{l}{\infty}{m}{l'}{\infty}{p}}{n} = \bigvee_{1\leq j\leq n} \LRTn{l}{1}{m_j}{l'}{1}{p_j}$.
Therefore, by applying rule \TOr,
we have
\[\teqTc{\Gamma'}{\Deltan}{\En}{\EEn}{\Mii}{\Nii}{\instTyp{\LRTn{l}{\infty}{m}{l'}{\infty}{p}}{n}}{\emptyset}\]
which proves the claim.

\end{itemize}

\end{proof}

\begin{lemma}[Typing destructors with replicated names]
\label{lem-proof:typing-dest-stars}
For all $\Gamma$, 
$d$, $x$, $T$, if
\[\tDestnew{\Gamma}{d(x)}{T}\] 
then for all $i, n\in\mathbb{N}$ such that $1 \leq i \leq n$,

\[\tDestnew{\instG{\Gamma}{i}{n}}{d(x_i)}{\instTyp{T}{n}}\] 
\end{lemma}
\begin{proof}
Immediate by examining the typing rules for destructors.
\end{proof}

\begin{lemma}[Branches and expansion]
\label{lem-proof:inst-branch}
\begin{itemize}
\item For all $T$, 
\[\bigcup_{T'\in\branch{T}} \branch{\instTyp{T'}{n}} = \branch{\instTyp{T}{n}}\]

\item For all $\Gamma$ 
for all $i, n\in\mathbb{N}$,
\[\bigcup_{\Gamma'\in\branch{\Gamma}} \branch{\instG{\Gamma'}{i}{n}} = \branch{\instG{\Gamma}{i}{n}}\]
\end{itemize}
\end{lemma}

\begin{proof}
The first point is proved by induction on $T$.

If $T = T' \orT T''$ for some $T'$, $T''$, then
\[
\begin{array}{rcl}
\branch{\instTyp{T}{n}} &=& \branch{\instTyp{T'}{n}} \cup \branch{\instTyp{T''}{n}}\\
    &=& (\bigcup_{T'''\in \branch{T'}} \branch{\instTyp{T'''}{n}}) \cup (\bigcup_{T'''\in \branch{T''}} \branch{\instTyp{T'''}{n}})
\end{array}
\]
by the induction hypothesis. Since $\branch{T} = \branch{T'}\cup\branch{T''}$, this proves the claim.

Otherwise, $\branch{T} = \{T\}$ and the claim trivially holds.

\vspace{1em}

The second point directly follows from the first point, using the definition of $\instG{\Gamma}{i}{n}$.
\end{proof}

\begin{lemma}[Typing processes in all branches]
\label{lem-proof:process-all-branches}
For all $P$, $Q$, $\Gamma$, 
${\{C_{\Gamma'}\}}_{\Gamma'\in\branch{\Gamma}}$,
if
\[\forall \Gamma'\in\branch{\Gamma}.\quad
\teqP{\Gamma'}{\Delta}{\E}{\E'}{P}{Q}{C_{\Gamma'}}\]
then
\[\teqP{\Gamma}{\Delta}{\E}{\E'}{P}{Q}{\bigcup_{\Gamma'\in\branch{\Gamma}}C_{\Gamma'}}.\]

Consequently if for some $C$, $C_{\Gamma'}\subseteq C$ for all $\Gamma'$, then there exists $C' \subseteq C$ such that
\[\teqP{\Gamma}{\Delta}{\E}{\E'}{P}{Q}{C'}.\]
\end{lemma}
\begin{proof}
The first point is easily proved by successive applications of rule \POr.

The second point is a direct consequence of the first point.
\end{proof}

\begin{lemma}[Expansion and union]
\label{lem-proof:inst-union}
\begin{itemize}
\item For all $C$, $C'$, 
such that $\forall (c, \Gamma)\in C\cup C'.\; \branch{\Gamma}=\{\Gamma\}$, \ie such that $\Gamma$ does not contain union types, and $\names{c}\subseteq\dom{\Gamma}\cup\FN$, and $\Gamma$ only nonce types with names from $\NN$ (\ie unindexed names), we have
\[\instCst{C\UnionCart C'}{i}{n} = \instCst{C}{i}{n}\UnionCart\instCst{C'}{i}{n}\]

\item For all $C$, $c$, $\Gamma$, such that $\names{c}\subseteq \dom{\Gamma}$ and $\forall (c, \Gamma')\in C. \;\novar{\Gamma}\subseteq \Gamma'$,
we have
\[\instCst{C\UnionAll c}{i}{n} = \instCst{C}{i}{n}\UnionAll\instConsta{c}{i}\]
\end{itemize}
\end{lemma}
\begin{proof}
The first point follows from the definition of $\instCst{\cdot}{i}{n}$ and $\UnionCart$.
Indeed, if $C$, $C'$ are as assumed in the claim, we have:
\[
\begin{array}{rcl}
\instCst{C\UnionCart C'}{i}{n} & = & \{(\instConsta{c}{i},\Gamma')|\exists\Gamma.\;(c,\Gamma)\in C\UnionCart C' \;\wedge\;
  \Gamma'\in\branch{\instG{\Gamma}{i}{n}}\}\\
    & = & \{(\instConst{c_1\cup c_2}{i}{\Gamma_1\cup\Gamma_2},\Gamma')|\exists\Gamma_1 \Gamma_2.\;
      (c_1,\Gamma_1)\in C \;\wedge\;(c_2,\Gamma_2)\in C' \;\wedge\; \Gamma_1, \Gamma_2 \text{ are compatible}\;\wedge\;\\
  && \quad\Gamma'\in\branch{\instG{\Gamma_1\cup\Gamma_2}{i}{n}}\}\\
    & = &\{(\instConst{c_1}{i}{\Gamma_1}\cup \instConst{c_2}{i}{\Gamma_2},\Gamma')|\exists\Gamma_1 \Gamma_2.\;
      (c_1,\Gamma_1)\in C \;\wedge\;(c_2,\Gamma_2)\in C' \;\wedge\; \Gamma_1, \Gamma_2 \text{ are compatible}\;\wedge\;\\
  &&\quad\Gamma'\in\branch{\instG{\Gamma_1}{i}{n}\cup\instG{\Gamma_2}{i}{n}}\}\\
    & = &\{(\instConst{c_1}{i}{\Gamma_1}\cup \instConst{c_2}{i}{\Gamma_2},\Gamma\cup\Gamma')|\exists\Gamma_1 \Gamma_2.\;
      (c_1,\Gamma_1)\in C \;\wedge\;(c_2,\Gamma_2)\in C' \;\wedge\; \Gamma_1, \Gamma_2 \text{ are compatible}\;\wedge\;\\
  &&\quad\Gamma\in\branch{\instG{\Gamma_1}{i}{n}}\;\wedge\;\Gamma'\in\branch{\instG{\Gamma_2}{i}{n}}\;\wedge\;
  \Gamma, \Gamma'\text{ are compatible}\}
 \end{array}
\]
The last step is proved by directly showing both inclusions.

On the other hand we have:
\[
\begin{array}{rcl}
\instCst{C}{i}{n}\UnionCart \instCst{C'}{i}{n} & = &
 \{(c \cup c', \Gamma \cup \Gamma') | (c, \Gamma) \in \instCst{C}{i}{n} \; \wedge \; (c',\Gamma') \in \instCst{C'}{i}{n} \;\wedge\; \Gamma, \Gamma' \text{ are compatible}\}\\
    & = & \{(\instConst{c_1}{i}{\Gamma_1}\cup \instConst{c_2}{i}{\Gamma_2},\Gamma\cup\Gamma')|\exists\Gamma_1 \Gamma_2.\;
      (c_1,\Gamma_1)\in C \;\wedge\;(c_2,\Gamma_2)\in C' \;\wedge\;\\
  &&\quad \Gamma\in\branch{\instG{\Gamma_1}{i}{n}}\;\wedge\; \Gamma'\in\branch{\instG{\Gamma_2}{i}{n}}\;\wedge\;
  \Gamma, \Gamma' \text{ are compatible}\}\\
    & = & \{(\instConst{c_1}{i}{\Gamma_1}\cup \instConst{c_2}{i}{\Gamma_2},\Gamma\cup\Gamma')|\exists\Gamma_1 \Gamma_2.\;
      (c_1,\Gamma_1)\in C \;\wedge\;(c_2,\Gamma_2)\in C' \;\wedge\;\\
  &&\quad \Gamma\in\branch{\instG{\Gamma_1}{i}{n}}\;\wedge\; \Gamma'\in\branch{\instG{\Gamma_2}{i}{n}}\;\wedge\;
  \Gamma, \Gamma' \text{ are compatible}\;\wedge\;\Gamma_1, \Gamma_2 \text{ are compatible}\}
\end{array}
\]

This last step comes from the fact that if $(c_1,\Gamma_1)\in C$ and $(c_2,\Gamma_2)\in C'$,
then by assumption $\Gamma_1$ and $\Gamma_2$ do not contain union types.
This implies that if
$\Gamma\in\branch{\instG{\Gamma_1}{i}{n}}$ and $\Gamma'\in\branch{\instG{\Gamma_2}{i}{n}}$ are compatible, then
$\Gamma_1$ and $\Gamma_2$ are compatible.
Indeed, let $x\in\dom{\Gamma_1}\cap\dom{\Gamma_2}$.
Hence $x_i\in\dom{\Gamma}\cap\dom{\Gamma'}$, and since they are compatible, $\Gamma(x_i) = \Gamma'(x_i)$.
That is to say that there exists $T \in \branch{\instTyp{\Gamma_1(x)}{n}}\cap\branch{\instTyp{\Gamma_2(x)}{n}}$.

If $\Gamma_1(x) = \LRTn{l}{\infty}{m}{l'}{\infty}{p}$ (for some $m, p, l, l'$), then 
$\instTyp{\Gamma_1(x)}{n} = \bigvee_{1\leq j \leq n} \LRTn{l}{1}{m_j}{l'}{1}{p_j}$, and thus there exists $j\in\llbracket 1, n\rrbracket$
such that $T = \LRTn{l}{1}{m_j}{l'}{1}{p_j}$.
Hence,  $\LRTn{l}{1}{m_j}{l'}{1}{p_j}\in\branch{\instTyp{\Gamma_2(x)}{n}}$. Because of the definition of $\instTyp{\cdot}{n}$,
and since $\Gamma_2(x)$ is not a union type (by assumption),
this implies that $\Gamma_2(x) = \LRTn{l}{\infty}{m}{l'}{\infty}{p}$, and therefore $\Gamma_1(x) = \Gamma_2(x)$.

If $\Gamma_1(x)$ is not of the form $\LRTn{l}{\infty}{m}{l'}{\infty}{p}$ (for some $m, p, l, l'$), then
neither is $\Gamma_2(x)$ (by contraposition, following the same reasoning as in the previous case).
$\Gamma_1(x)$ and $\Gamma_2(x)$ are not of the form $T'\orT T''$ either, by assumption.
Therefore, neither $\instTyp{\Gamma_1(x)}{n}$ nor $\instTyp{\Gamma_2(x)}{n}$ are union types (from the definition of
$\instTyp{\cdot}{n}$).
This implies that $T = \instTyp{\Gamma_1(x)}{n} = \instTyp{\Gamma_2(x)}{n}$, which implies $\Gamma_1(x) = \Gamma_2(x)$.

In both cases $\Gamma_1(x) = \Gamma_2(x)$, and $\Gamma_1$, $\Gamma_2$ are therefore compatible.

Hence $\instCst{C}{i}{n}\UnionCart \instCst{C'}{i}{n} = \instCst{C\UnionCart C'}{i}{n}$, which proves the claim.

\vspace{1em}

The second point directly follows from the definition of $\instCst{\cdot}{i}{n}$ and $\UnionAll$.
Indeed, for all $C$, $c$, $\Gamma$ satisfying the assumptions, we have:
\[
\begin{array}{rcl}
\instCst{C\UnionAll c}{i}{n} & = & \{(\instConst{c'}{i}{\Gamma'},\Gamma'')|\exists\Gamma'.\;(c',\Gamma')\in C\UnionAll c \;\wedge\;
  \Gamma''\in\branch{\instG{\Gamma'}{i}{n}}\}\\
    & = & \{(\instConst{c''\cup c}{i}{\Gamma'},\Gamma'')|\exists\Gamma'.\;
      (c'',\Gamma')\in C \;\wedge\; \Gamma''\in\branch{\instG{\Gamma'}{i}{n}}\}\\
    & = &\{(\instConst{c''}{i}{\Gamma'}\cup \instConst{c}{i}{\Gamma'},\Gamma'')|\exists\Gamma'.\;
      (c'',\Gamma')\in C \;\wedge\; \Gamma''\in\branch{\instG{\Gamma'}{i}{n}}\}\\
    & = &\{(\instConst{c''}{i}{\Gamma'}\cup \instConst{c}{i}{\Gamma},\Gamma'')|\exists\Gamma'.\;
      (c'',\Gamma')\in C \;\wedge\; \Gamma''\in\branch{\instG{\Gamma'}{i}{n}}\}\\
    &   &\quad\text{(since $\Gamma$, $\Gamma'$ give the same types to names and keys)}\\
    & = &\{(\instConst{c'}{i}{\Gamma'},\Gamma'')|\exists\Gamma'.\;
      (c',\Gamma')\in C \;\wedge\;\Gamma''\in\branch{\instG{\Gamma'}{i}{n}}\} \UnionAll \instConst{c}{i}{\Gamma}\\
    & = & \instCst{C}{i}{n} \UnionAll \instConst{c}{i}{\Gamma}
 \end{array}
\]
\end{proof}

\begin{theorem}[Typing processes with expanded types]
\label{lem-proof:typing-process-star}
For all $\Gamma$, 
$P$, $Q$ and $C$, if
\[\teqP{\Gamma}{\Delta}{\E}{\E'}{P}{Q}{C}\] 
then for all $i, n\in\mathbb{N}$ such that $1 \leq i \leq n$, there exists $C'\subseteq \instCst{C}{i}{n}$
such that

\[\teqP{\instG{\Gamma}{i}{n}}{\instD{\Delta}{n}}{\instE{\E}{n}}{\instE{\E'}{n}}{\instProca{P}{i}}{\instProca{Q}{i}}{C'}\] 
\end{theorem}

\begin{proof}
We prove this theorem by induction on the derivation $\Pi$ of $\teqP{\Gamma}{\Delta}{\E}{\E'}{P}{Q}{C}$.
We distinguish several cases for the last rule applied in this derivation.

\newcommand{\Pii}{\instProca{P}{i}}
\newcommand{\Qii}{\instProca{Q}{i}}
\newcommand{\Gammain}{\instG{\Gamma}{i}{n}}

\begin{itemize}
\item \case{\PZero:}
then $P = Q = \Pii = \Qii = \ZERO$, and $C = \{(\emptyset,\Gamma)\}$. 
Hence
\begin{align*}
\Cin 
	&= \{(\emptyset, \Gamma') \;|\; \Gamma'\in\branch{\instG{\Gamma}{i}{n}}\}
\end{align*}

Thus, by applying rule \POr as many times as necessary to split $\instG{\Gamma}{i}{n}$ into all of its branches,
followed by rule \PZero, we have $\teqP{\Gammain}{\Deltan}{\En}{\EEn}{\Pii}{\Qii}{\Cin}$.

\item \case{\POut:}
then $P = \OUT{M}.P'$, $Q=\OUT{N}.Q'$ for some messages $M$, $N$ and some processes $P'$, $Q'$,
and
\[\Pi=
\inferrule
  {\inferrule*{\Pi'}{\teqP{\Gamma}{\Delta}{\E}{\E'}{P'}{Q'}{C'}} \\
   \inferrule*{\Pi''}{\teqTc{\Gamma}{\Delta}{\E}{\E'}{M}{N}{\L}{c}}}
  {\teqP{\Gamma}{\Delta}{\E}{\E'}{P}{Q}{C = C' \UnionAll c}}.
\]

By applying the induction hypothesis to $\Pi'$,
there exists $C''\subseteq \instCst{C'}{i}{n}$ and
a proof $\Pi'''$ of $\teqP{\Gammain}{\Deltan}{\En}{\EEn}{\instProca{P'}{i}}{\instProca{Q'}{i}}{C''}$.

Hence, by Lemma~\ref{lem-proof:type-processes-branches},
for all $\Gamma'\in\branch{\Gammain}$, there exist $C_{\Gamma'}\subseteq C''$ and a proof $\Pi_{\Gamma'}$ of
$\teqP{\Gamma'}{\Deltan}{\En}{\EEn}{\instProca{P'}{i}}{\instProca{Q'}{i}}{C_{\Gamma'}}$.

Moreover, by Lemma~\ref{lem-proof:typing-terms-stars}, for all $\Gamma'\in\branch{\Gammain}$,
there exists a proof $\Pi'_{\Gamma'}$ of
$\teqTc{\Gamma'}{\Deltan}{\En}{\EEn}{\instTerma{M}{i}}{\instTerma{N}{i}}{\L}{\instConsta{c}{i}}$.

In addition, $\Pii = \instProca{\OUT{M}.P'}{i} = \OUT{\instTerma{M}{i}}.\instProca{P'}{i}$.
Similarly, $\Qii = \OUT{\instTerma{N}{i}}.\instProca{Q'}{i}$.
Therefore, using $\Pi_{\Gamma'}$, $\Pi'_{\Gamma'}$
and rule \POut, we have for all $\Gamma'\in\branch{\Gammain}$ that
$\teqP{\Gamma'}{\Deltan}{\En}{\EEn}{\Pii}{\Qii}{C_{\Gamma'}\UnionAll\instConsta{c}{i}\subseteq
C''\UnionAll\instConsta{c}{i}}$.

Thus by Lemma~\ref{lem-proof:process-all-branches}, there exists
$C_1 \subseteq C''\UnionAll\instConsta{c}{i}$ such that
\[\teqP{\Gammain}{\Deltan}{\En}{\EEn}{\Pii}{\Qii}{C_1}.\]

Finally, $\Cin = \instCst{C'\UnionAll c}{i}{n} = \instCst{C'}{i}{n}\UnionAll\instConsta{c}{i}$ (by Lemma~\ref{lem-proof:inst-union}, whose conditions are satisfied, by Lemma~\ref{lem-proof:env-const-branch}).
Hence 
$C''\UnionAll\instConsta{c}{i} \subseteq \Cin$, which proves the claim.

\item \case{\PIn:}
then $P = \IN{x}.P'$, $Q=\IN{x}.Q'$ for some variable $x$ 
and some processes $P'$, $Q'$,
and
\[\Pi=
\inferrule
  {
   \inferrule*{\Pi'}{\teqP{\Gamma, x:\L}{\Delta}{\E}{\E'}{P'}{Q'}{C}}}
  {\teqP{\Gamma}{\Delta}{\E}{\E'}{P}{Q}{C}}.
\]

Since
\[\instG{\Gamma, x:\L}{i}{n} = \Gammain, x_i:\instTyp{\L}{n} = \Gammain, x_i:\L,\]
by applying the induction hypothesis to $\Pi'$,
there exists $C'\subseteq \instCst{C}{i}{n}$ and a proof $\Pi''$ of $\teqP{\Gammain, x_i:\L}{\Deltan}{\En}{\EEn}{\instProc{P'}{i}{\Gamma,x:\L}}{\instProc{Q'}{i}{\Gamma,x:\L}}{C'}$.


In addition, $\Pii = \instProca{\IN{x}.P'}{i} = \IN{x_i}.\instProca{P'}{i} = \IN{x_i}.\instProc{P'}{i}{\Gamma,x:\L}$.
Similarly, $\Qii = \IN{x_i}.\instProc{Q'}{i}{\Gamma,x:\L}$.

Therefore, using $\Pi''$ and rule \PIn, we have
$\teqP{\Gammain}{\Deltan}{\En}{\EEn}{\Pii}{\Qii}{C' \subseteq \Cin}$.

\item \case{\PNew:}
then $P = \NEW{m}{\noncetypelab{l}{a}{m}}.P'$, $Q = \NEW{m}{\noncetypelab{l}{a}{m}}. Q'$ for some $m, l, a$ and some processes $P', Q'$,
and
\[\Pi=
\inferrule
  {
  \inferrule*{\Pi'}{\teqP{\Gamma, m:\noncetypelab{l}{a}{m}}{\Delta}{\E}{\E'}{P'}{Q'}{C}}}
  {\teqP{\Gamma}{\Delta}{\E}{\E'}{P}{Q}{C}}.
\]

\begin{itemize}
\item \case{If $a = 1$:}

 Since
 \[\instG{\Gamma, m:\noncetypelab{l}{1}{m}}{i}{n} = \Gammain, m:\noncetypelab{l}{1}{m},\]
by applying the induction hypothesis to $\Pi'$,
there exists $C' \subseteq \Cin$ and a proof $\Pi''$ of $\teqP{\Gammain, m:\noncetypelab{l}{1}{m}}{\Deltan}{\En}{\EEn}{\instProc{P'}{i}{\Gamma, m:\noncetypelab{l}{1}{m}}}{\instProc{Q'}{i}{\Gamma, m:\noncetypelab{l}{1}{m}}}{C'}$.

In addition $\Pii = \instProca{\NEWN{m}.P'}{i} = \NEW{m}{\noncetypelab{l}{1}{m}}.\instProca{P'}{i} = \NEW{m}{\noncetypelab{l}{1}{m}}.\instProc{P'}{i}{\Gamma, m:\noncetypelab{l}{1}{m}}$; and similarly for $Q$.
Therefore, using $\Pi''$ and rule \PNew, we have
$\teqP{\Gammain}{\Deltan}{\En}{\EEn}{\Pii}{\Qii}{C'\subseteq \Cin}$.

\item \case{If $a=\infty$:}
 Since
 \[\instG{\Gamma, m:\noncetypelab{l}{\infty}{m}}{i}{n} = \Gammain, m_i:\noncetypelab{l}{1}{m_i},\]
by applying the induction hypothesis to $\Pi'$,
there exists $C' \subseteq \Cin$ and a proof $\Pi''$ of $\teqP{\Gammain, m_i:\noncetypelab{l}{1}{m_i}}{\Deltan}{\En}{\EEn}{\instProc{P'}{i}{\Gamma, m:\noncetypelab{l}{\infty}{m}}}{\instProc{Q'}{i}{\Gamma, m:\noncetypelab{l}{\infty}{m}}}{C'}$.

In addition $\Pii = \instProca{\NEWN{m}.P'}{i} = \NEW{m_i}{\noncetypelab{l}{1}{m_i}}.\instProca{P'[m_i/m]}{i} = \NEW{m_i}{\noncetypelab{l}{1}{m_i}}.\instProc{P'}{i}{\Gamma, m:\noncetypelab{l}{\infty}{m}}$; and similarly for $Q$.
Therefore, using $\Pi''$ and rule \PNew, we have
$\teqP{\Gammain}{\Deltan}{\En}{\EEn}{\Pii}{\Qii}{C'\subseteq \Cin}$.

\end{itemize}


\item \case{\PPar:}
then $P = P'\PAR P''$, $Q=Q'\PAR Q''$ for some processes $P'$, $Q'$, $P''$, $Q''$,
and
\[\Pi=
\inferrule
  {\inferrule*{\Pi'}{\teqP{\Gamma}{\Delta}{\E}{\E'}{P'}{Q'}{C'}}\\
   \inferrule*{\Pi''}{\teqP{\Gamma}{\Delta}{\E}{\E'}{P''}{Q''}{C''}}}
  {\teqP{\Gamma}{\Delta}{\E}{\E'}{P}{Q}{C = C' \UnionCart C''}}.
\]

By applying the induction hypothesis to $\Pi'$,
there exists $C'''\subseteq \instCst{C'}{i}{n}$ and a proof $\Pi'''$ of $\teqP{\Gammain}{\Deltan}{\En}{\EEn}{\instProca{P'}{i}}{\instProca{Q'}{i}}{C'''}$.
Similarly, by applying the induction hypothesis to $\Pi''$,
there exists $C''''\subseteq \instCst{C''}{i}{n}$ and a proof $\Pi''''$ of $\teqP{\Gammain}{\Deltan}{\En}{\EEn}{\instProca{P''}{i}}{\instProca{Q''}{i}}{C''''}$.

In addition, $\Pii = \instProca{P'\PAR P''}{i} = \instProca{P'}{i} \PAR \instProca{P''}{i}$.
Similarly, $\Qii = \instProca{Q'\PAR Q''}{i} = \instProca{Q'}{i} \PAR \instProca{Q''}{i}$.
Finally, $\Cin = \instCst{C'\UnionCart C''}{i}{n} = \instCst{C'}{i}{n}\UnionCart\instCst{C''}{i}{n}$, by Lemma~\ref{lem-proof:inst-union} (using Lemma~\ref{lem-proof:env-constr-union} to ensure the condition that the environments do not contain union types).

Therefore, using $\Pi'''$, $\Pi''''$ and rule \PPar, we have
$\teqP{\Gammain}{\Deltan}{\En}{\EEn}{\Pii}{\Qii}{C'''\UnionCart C''''\subseteq \Cin}$.

\item \case{\POr:}
then $\Gamma = \Gamma', x:T\orT T'$ for some $\Gamma'$, some $x\in\X$ and some types $T$, $T'$, and
\[\Pi=
\inferrule
  {\inferrule*{\Pi_T}{\teqP{\Gamma', x:T}{\Delta}{\E}{\E'}{P}{Q}{C'}}\\
   \inferrule*{\Pi_{T'}}{\teqP{\Gamma', x:T'}{\Delta}{\E}{\E'}{P}{Q}{C''}}}
  {\teqP{\Gamma}{\Delta}{\E}{\E'}{P}{Q}{C = C'\cup C''}}.
\]

By applying the induction hypothesis to $\Pi_T$,
there exist $C_1\subseteq \instCst{C'}{i}{n}$ and
a proof $\Pi_1$ of $\teqP{\instG{\Gamma'}{i}{n}, x_i:\instTyp{T}{n}}{\Deltan}{\En}{\EEn}{\instProc{P}{i}{\Gamma',x:T}}{\instProc{Q}{i}{\Gamma',x:T}}{C_1}$.
Similarly with $\Pi_{T'}$,
there exist $C_2\subseteq \instCst{C''}{i}{n}$ and
a proof $\Pi_2$ of $\teqP{\instG{\Gamma'}{i}{n}, x_i:\instTyp{T'}{n}}{\Deltan}{\En}{\EEn}{\instProc{P}{i}{\Gamma',x:T'}}{\instProc{Q}{i}{\Gamma',x:T'}}{C_2}$.

In addition $\Pii = \instProc{P}{i}{\Gamma',x:T} = \instProc{P}{i}{\Gamma',x:T'}$, and similarly for $Q$.

Thus by rule \POr, we have
\[\teqP{\instG{\Gamma'}{i}{n}, x_i:\instTyp{T}{n} \orT \instTyp{T'}{n}}{\Deltan}{\En}{\EEn}{\Pii}{\Qii}{C_1\cup C_2 \subseteq \instCst{C'}{i}{n} \cup \instCst{C''}{i}{n} = \Cin}.\]

Since $\Gammain = \instG{\Gamma', x:T\orT T'}{i}{n} = \instG{\Gamma'}{i}{n}, x_i:\instTyp{T}{n} \orT \instTyp{T'}{n}$,
this proves the claim in this case.

\item \case{\PLet:}
then $P = \LET{x}{d(y)}{P'}{P''}$, $Q=\LET{x}{d(y)}{Q'}{Q''}$ for some variable $x$ and some processes $P'$, $Q'$,
$P''$, $Q''$, and
\[\Pi=
\inferrule
  {
   \inferrule*{\Pi_d}{\tDestnew{\Gamma}{d(y)}{T}}\\
   \inferrule*{\Pi'}{\teqP{\Gamma, x:T}{\Delta}{\E}{\E'}{P'}{Q'}{C'}}\\
   \inferrule*{\Pi''}{\teqP{\Gamma}{\Delta}{\E}{\E'}{P''}{Q''}{C''}}}
  {\teqP{\Gamma}{\Delta}{\E}{\E'}{P}{Q}{C=C'\cup C''}}.
\]

Since
\[\instG{\Gamma, x:T}{i}{n} = \Gammain, x_i:\instTyp{T}{n},\]
by applying the induction hypothesis to $\Pi'$,
there exist $C'''\subseteq \instCst{C'}{i}{n}$ and a proof $\Pi'''$ of $\teqP{\Gammain, x_i:\instTyp{T}{n}}{\Deltan}{\En}{\EEn}{\instProc{P'}{i}{\Gamma,x:T}}{\instProc{Q'}{i}{\Gamma,x:T}}{C'''}$.
Similarly,
there exist $C''''\subseteq \instCst{C''}{i}{n}$ and a proof $\Pi''''$ of $\teqP{\Gammain}{\Deltan}{\En}{\EEn}{\instProca{P''}{i}}{\instProca{Q''}{i}}{C''''}$.

By Lemma~\ref{lem-proof:typing-dest-stars} applied to $\Pi_d$, we also have
\[\tDestnew{\Gammain}{d(y_i)}{\instTyp{T}{n}}\]


In addition, $\Pii = \instProca{\LET{x}{d(y)}{P'}{P''}}{i} = \LET{x_i}{d(y_i)}{\instProca{P'}{i}}{\instProca{P''}{i}} = \LET{x_i}{d(y_i)}{\instProc{P'}{i}{\Gamma,x:T}}{\instProca{P''}{i}}$.
Similarly, $\Qii = \LET{x_i}{d(y_i)}{\instProc{Q'}{i}{\Gamma,x:T}}{\instProca{Q''}{i}}$.

Therefore, using $\Pi''$ and rule \PLet, we have
$\teqP{\Gammain}{\Deltan}{\En}{\EEn}{\Pii}{\Qii}{C'''\cup C'''' \subseteq \Cin}$.

\item \case{\PLetLR:}
then $P = \LET{x}{d(y)}{P'}{P''}$, $Q=\LET{x}{d(y)}{Q'}{Q''}$ for some variable $x\in\XR$ and some processes $P'$, $Q'$,
$P''$, $Q''$, and
\[\Pi=
\inferrule
  {
   \Gamma(y) = \LRTn{l}{\oneorinf}{m}{l'}{\oneorinf}{p}\\
   \inferrule*{\Pi'}{\teqP{\Gamma}{\Delta}{\E}{\E'}{P''}{Q''}{C}}}
  {\teqP{\Gamma}{\Delta}{\E}{\E'}{P}{Q}{C}}
\]
for some $m$, $p$. 

By applying the induction hypothesis to $\Pi'$,
there exists $C'\subseteq \Cin$ and a proof $\Pi''$ of $\teqP{\Gammain}{\Deltan}{\En}{\EEn}{\instProca{P''}{i}}{\instProca{Q''}{i}}{C'}$.

We have $\Pii = \instProca{\LET{x}{d(y)}{P'}{P''}}{i} = \LET{x_i}{d(y_i)}{\instProca{P'}{i}}{\instProca{P''}{i}}$.
Similarly, $\Qii = \LET{x_i}{d(y_i)}{\instProca{Q'}{i}}{\instProca{Q''}{i}}$.

We distinguish two cases, depending on whether the types in the refinement $\LRTn{l}{\oneorinf}{m}{l'}{\oneorinf}{p}$
are finite nonce types or infinite nonce types, $\ie$ whether $\oneorinf$ is $1$ or $\infty$.

\begin{itemize}
\item \case{If $\oneorinf$ is $1$:}
Then by definition of $\Gammain$, we have $\Gammain(y_i) = \LRTn{l}{1}{m}{l'}{1}{p}$.

Therefore, using $\Pi''$ and rule \PLetLR, we have
$\teqP{\Gammain}{\Deltan}{\En}{\EEn}{\Pii}{\Qii}{C' \subseteq \Cin}$.

\item \case{If $\oneorinf$ is $\infty$:}
Then by definition of $\Gammain$, we have $\Gammain(y_i) = \bigvee_{1\leq j \leq n} \LRTn{l}{1}{m_j}{l'}{1}{p_j}$.
Let $\Gamma'\in\branch{\Gammain}$.
By definition, there exists $j\in\llbracket 1, n \rrbracket$, such that 
$\Gamma'(y_i) = \LRTn{l}{1}{m_j}{l'}{1}{p_j}$.

Using $\Pi''$ and Lemma~\ref{lem-proof:type-processes-branches}, there exist $C'_{\Gamma'} \subseteq \Cin$
and a derivation $\Pi''_{\Gamma'}$ of
$\teqP{\Gamma'}{\Deltan}{\En}{\EEn}{\instProca{P''}{i}}{\instProca{Q''}{i}}{C'_{\Gamma'}}$.

Therefore, using rule \PLetLR, we have, for all $\Gamma'\in\branch{\Gammain}$,
$\teqP{\Gamma'}{\Deltan}{\En}{\EEn}{\Pii}{\Qii}{C'_{\Gamma'} \subseteq \Cin}$.
Thus, by Lemma~\ref{lem-proof:process-all-branches}, we have
\[\teqP{\Gammain}{\Deltan}{\En}{\EEn}{\Pii}{\Qii}{C''},\]
where $C'' \subseteq \Cin$,
which proves the claim in this case.
\end{itemize}

\item \case{\PIfL:}
then $P = \ITE{M}{M'}{P'}{P''}$, $Q=\ITE{N}{N'}{Q'}{Q''}$ for some messages $M$, $N$, $M'$, $N'$, and some processes $P'$, $Q'$, $P''$, $Q''$,
and
\[\Pi=
\inferrule
  {\inferrule*{\Pi'}{\teqP{\Gamma}{\Delta}{\E}{\E'}{P'}{Q'}{C'}} \\
   \inferrule*{\Pi''}{\teqP{\Gamma}{\Delta}{\E}{\E'}{P''}{Q''}{C''}}\\
   \inferrule*{\Pi_1}{\teqTc{\Gamma}{\Delta}{\E}{\E'}{M}{N}{\L}{c}}\\
   \inferrule*{\Pi_2}{\teqTc{\Gamma}{\Delta}{\E}{\E'}{M'}{N'}{\L}{c'}}}
  {\teqP{\Gamma}{\Delta}{\E}{\E'}{P}{Q}{C = (C'\cup C'') \UnionAll (c\cup c')}}.
\]

By applying the induction hypothesis to $\Pi'$,
there exists $C'''\subseteq \instCst{C'}{i}{n}$ and
a proof $\Pi'''$ of $\teqP{\Gammain}{\Deltan}{\En}{\EEn}{\instProca{P'}{i}}{\instProca{Q'}{i}}{C'''}$.
Similarly,
there exists $C''''\subseteq \instCst{C''}{i}{n}$ and
a proof $\Pi''''$ of $\teqP{\Gammain}{\Deltan}{\En}{\EEn}{\instProca{P''}{i}}{\instProca{Q''}{i}}{C''''}$.

Hence, by Lemma~\ref{lem-proof:type-processes-branches},
for all $\Gamma'\in\branch{\Gammain}$, there exist $C_{\Gamma'}\subseteq C'''$ and a proof $\Pi_{1,\Gamma'}$ of
$\teqP{\Gamma'}{\Deltan}{\En}{\EEn}{\instProca{P'}{i}}{\instProca{Q'}{i}}{C_{\Gamma'}}$;
as well as $C'_{\Gamma'}\subseteq C''''$ and a proof $\Pi_{2,\Gamma'}$ of
$\teqP{\Gamma'}{\Deltan}{\En}{\EEn}{\instProca{P''}{i}}{\instProca{Q''}{i}}{C'_{\Gamma'}}$.

Moreover, by Lemma~\ref{lem-proof:typing-terms-stars}, for all $\Gamma'\in\branch{\Gammain}$,
there exists a proof $\Pi'_{1,\Gamma'}$ of
$\teqTc{\Gamma'}{\Deltan}{\En}{\EEn}{\instTerma{M}{i}}{\instTerma{N}{i}}{\L}{\instConsta{c}{i}}$.
Similarly,
there exists a proof $\Pi'_{2,\Gamma'}$ of
$\teqTc{\Gamma'}{\Deltan}{\En}{\EEn}{\instTerma{M'}{i}}{\instTerma{N'}{i}}{\L}{\instConsta{c'}{i}}$.

In addition, $\Pii = \instProca{\ITE{M}{M'}{P'}{P''}}{i} = \ITE{\instTerma{M}{i}}{\instTerma{M'}{i}}{\instProca{P'}{i}}{\instProca{P''}{i}}$.
Similarly, $\Qii = \ITE{\instTerma{N}{i}}{\instTerma{N'}{i}}{\instProca{Q'}{i}}{\instProca{Q''}{i}}$.

Therefore, using $\Pi_{1,\Gamma'}$, $\Pi_{2,\Gamma'}$, $\Pi'_{1,\Gamma'}$, $\Pi'_{2,\Gamma'}$ and rule \PIfL, we have 
\[\teqP{\Gamma'}{\Deltan}{\En}{\EEn}{\Pii}{\Qii}{(C_{\Gamma'}\cup C'_{\Gamma'})\UnionAll(\instConsta{c}{i}\cup\instConsta{c'}{i})\subseteq (C'''\cup C'''')\UnionAll(\instConsta{c}{i}\cup\instConsta{c'}{i})}.\]

Thus by Lemma~\ref{lem-proof:process-all-branches}, there exists
$C_1 \subseteq (C'''\cup C'''')\UnionAll(\instConsta{c}{i}\cup\instConsta{c'}{i})$ such that
\[\teqP{\Gammain}{\Deltan}{\En}{\EEn}{\Pii}{\Qii}{C_1}.\]

Finally, $\Cin = \instCst{(C'\cup C'')\UnionAll (c\cup c')}{i}{n} =
(\instCst{C'}{i}{n}\cup \instCst{C''}{i}{n})\UnionAll(\instConsta{c}{i}\cup \instConsta{c'}{i})$ (by Lemma~\ref{lem-proof:inst-union}, whose conditions are satisfied, by Lemma~\ref{lem-proof:env-const-branch}).
Hence $(C'''\cup C'''')\UnionAll(\instConsta{c}{i}\cup \instConsta{c'}{i}) \subseteq \Cin$, which proves the claim.

\item \case{\PIfP:}
then $P = \ITE{M}{t}{P'}{P''}$, $Q=\ITE{N}{t}{Q'}{Q''}$ for some messages $M$, $N$, some $t\in\K\cup\N\cup\CST$, and some processes $P'$, $Q'$, $P''$, $Q''$,
and
\[\Pi=
\inferrule
  {\inferrule*{\Pi'}{\teqP{\Gamma}{\Delta}{\E}{\E'}{P'}{Q'}{C'}} \\
   \inferrule*{\Pi''}{\teqP{\Gamma}{\Delta}{\E}{\E'}{P''}{Q''}{C''}}\\
   \inferrule*{\Pi_1}{\teqTc{\Gamma}{\Delta}{\E}{\E'}{M}{N}{\L}{c}}\\
   \inferrule*{\Pi_2}{\teqTc{\Gamma}{\Delta}{\E}{\E'}{t}{t}{\L}{c'}}}
  {\teqP{\Gamma}{\Delta}{\E}{\E'}{P}{Q}{C = C' \cup C''}}.
\]

By applying the induction hypothesis to $\Pi'$,
there exist $C'''\subseteq \instCst{C'}{i}{n}$ and
a proof $\Pi'''$ of $\teqP{\Gammain}{\Deltan}{\En}{\EEn}{\instProca{P'}{i}}{\instProca{Q'}{i}}{C'''}$.
Similarly,
there exist $C''''\subseteq \instCst{C''}{i}{n}$ and
a proof $\Pi''''$ of $\teqP{\Gammain}{\Deltan}{\En}{\EEn}{\instProca{P''}{i}}{\instProca{Q''}{i}}{C''''}$.

Hence, by Lemma~\ref{lem-proof:type-processes-branches},
for all $\Gamma'\in\branch{\Gammain}$, there exist $C_{\Gamma'}\subseteq C'''$ and a proof $\Pi_{1,\Gamma'}$ of
$\teqP{\Gamma'}{\Deltan}{\En}{\EEn}{\instProca{P'}{i}}{\instProca{Q'}{i}}{C_{\Gamma'}}$;
as well as $C'_{\Gamma'}\subseteq C''''$ and a proof $\Pi_{2,\Gamma'}$ of
$\teqP{\Gamma'}{\Deltan}{\En}{\EEn}{\instProca{P''}{i}}{\instProca{Q''}{i}}{C'_{\Gamma'}}$.

Moreover, by Lemma~\ref{lem-proof:typing-terms-stars}, for all $\Gamma'\in\branch{\Gammain}$,
there exists a proof $\Pi'_{1,\Gamma'}$ of
$\teqTc{\Gamma'}{\Deltan}{\En}{\EEn}{\instTerma{M}{i}}{\instTerma{N}{i}}{\L}{\instConsta{c}{i}}$.
Similarly,
there exists a proof $\Pi'_{2,\Gamma'}$ of
$\teqTc{\Gamma'}{\Deltan}{\En}{\EEn}{\instTerma{t}{i}}{\instTerma{t}{i}}{\L}{\instConsta{c'}{i}}$.

Since $t\in\K\cup\N\cup\CST$, we also have $\instTerma{t}{i}\in\K\cup\N\cup\CST$.

In addition, $\Pii = \instProca{\ITE{M}{M'}{P'}{P''}}{i} = \ITE{\instTerma{M}{i}}{\instTerma{M'}{i}}{\instProca{P'}{i}}{\instProca{P''}{i}}$.
Similarly, $\Qii = \ITE{\instTerma{N}{i}}{\instTerma{N'}{i}}{\instProca{Q'}{i}}{\instProca{Q''}{i}}$.

Therefore, using $\Pi_{1,\Gamma'}$, $\Pi_{2,\Gamma'}$, $\Pi'_{1,\Gamma'}$, $\Pi'_{2,\Gamma'}$ and rule \PIfP, we have for all $\Gamma'\in\branch{\Gammain}$
\[\teqP{\Gamma'}{\Deltan}{\En}{\EEn}{\Pii}{\Qii}{C_{\Gamma'}\cup C'_{\Gamma'} \subseteq C'''\cup C''''}.\]

Thus by Lemma~\ref{lem-proof:process-all-branches}, there exists
$C_1 \subseteq (C'''\cup C'''')$ such that
\[\teqP{\Gammain}{\Deltan}{\En}{\EEn}{\Pii}{\Qii}{C_1}.\]

Finally, $\Cin = \instCst{C'\cup C''}{i}{n} = \instCst{C'}{i}{n}\cup \instCst{C''}{i}{n}$ (by Lemma~\ref{lem-proof:inst-union}).
Hence $(C'''\cup C'''') \subseteq \Cin$, which proves the claim.

\item \case{\PIfLR:}
then $P = \ITE{M_1}{M_2}{P_\top}{P_\bot}$, $Q=\ITE{N_1}{N_2}{Q_\top}{Q_\bot}$ for some messages $M_1$, $N_1$, $M_2$, $N_2$, and some processes $P_\top$, $Q_\top$, $P_\bot$, $Q_\bot$,
and there exist $m$, $p$, $m'$, $p'$ such that
\[\Pi=
\inferrule
  {\inferrule*{\Pi_1}{\teqTc{\Gamma}{\Delta}{\E}{\E'}{M_1}{N_1}{\LRTn{l}{1}{m}{l'}{1}{p}}{c}}\\
   \inferrule*{\Pi_2}{\teqTc{\Gamma}{\Delta}{\E}{\E'}{M_2}{N_2}{\LRTn{l''}{1}{m'}{l'''}{1}{p'}}{c'}}\\
    b = (\noncetypelab{l}{1}{m} \overset{?}{=} \noncetypelab{l''}{1}{m'}) \\ 
    b' = (\noncetypelab{l''}{1}{p} \overset{?}{=} \noncetypelab{l'''}{1}{p'}) \\
   \inferrule*{\Pi'}{\teqP{\Gamma}{\Delta}{\E}{\E'}{P_b}{Q_{b'}}{C}}}
  {\teqP{\Gamma}{\Delta}{\E}{\E'}{P}{Q}{C}}.
\]

By applying the induction hypothesis to $\Pi'$,
there exist $C'\subseteq \Cin$ and
a proof $\Pi''$ of $\teqP{\Gammain}{\Deltan}{\En}{\EEn}{\instProca{P_b}{i}}{\instProca{Q_{b'}}{i}}{C'}$.

Hence, by Lemma~\ref{lem-proof:type-processes-branches},
for all $\Gamma'\in\branch{\Gammain}$, there exist $C_{\Gamma'}\subseteq C'$, and a proof
$\Pi_{\Gamma'}$ of
$\teqP{\Gamma'}{\Deltan}{\En}{\EEn}{\instProca{P_b}{i}}{\instProca{Q_{b'}}{i}}{C_{\Gamma'}}$.

Moreover, by Lemma~\ref{lem-proof:typing-terms-stars} applied to $\Pi_1$, for all $\Gamma'\in\branch{\Gammain}$,
there exists a proof $\Pi'_{1,\Gamma'}$ of
$\teqTc{\Gamma'}{\Deltan}{\En}{\EEn}{\instTerma{M_1}{i}}{\instTerma{N_1}{i}}{\LRTn{l}{1}{m}{l'}{1}{p}}{\instConsta{c}{i}}$.
Similarly,
there exists a proof $\Pi'_{2,\Gamma'}$ of
$\teqTc{\Gamma'}{\Deltan}{\En}{\EEn}{\instTerma{M_2}{i}}{\instTerma{N_2}{i}}{\LRTn{l''}{1}{m'}{l'''}{1}{p'}}{\instConsta{c'}{i}}$.

In addition, $\Pii = \instProca{\ITE{M_1}{M_2}{P_\top}{P_\bot}}{i} = \ITE{\instTerma{M_1}{i}}{\instTerma{M_2}{i}}{\instProca{P_\top}{i}}{\instProca{P_\bot}{i}}$.
Similarly, $\Qii = \ITE{\instTerma{N_1}{i}}{\instTerma{N_2}{i}}{\instProca{Q_\top}{i}}{\instProca{Q_\bot}{i}}$.

Therefore, using $\Pi_{\Gamma'}$, $\Pi'_{1,\Gamma'}$, $\Pi'_{2,\Gamma'}$ and rule \PIfLR, we have 
for all $\Gamma'\in\branch{\Gammain}$
\[\teqP{\Gamma'}{\Deltan}{\En}{\EEn}{\Pii}{\Qii}{C_{\Gamma'}\subseteq C'\subseteq \Cin}.\]

Thus by Lemma~\ref{lem-proof:process-all-branches}, there exists
$C_1 \subseteq \Cin$ such that
\[\teqP{\Gammain}{\Deltan}{\En}{\EEn}{\Pii}{\Qii}{C_1},\]
which proves the claim.

\item \case{\PIfS:}
then $P = \ITE{M}{M'}{P'}{P''}$, $Q=\ITE{N}{N'}{Q'}{Q''}$ for some messages $M$, $N$, $M'$, $N'$, and some processes $P'$, $Q'$, $P''$, $Q''$,
and
\[\Pi=
\inferrule
  {\inferrule*{\Pi'}{\teqP{\Gamma}{\Delta}{\E}{\E'}{P''}{Q''}{C}}\\
   \inferrule*{\Pi_1}{\teqTc{\Gamma}{\Delta}{\E}{\E'}{M}{N}{\L}{c}}\\
   \inferrule*{\Pi_2}{\teqTc{\Gamma}{\Delta}{\E}{\E'}{M'}{N'}{\S}{c'}}}
  {\teqP{\Gamma}{\Delta}{\E}{\E'}{P}{Q}{C}}.
\]

By applying the induction hypothesis to $\Pi'$,
there exists $C'\subseteq \Cin$ and
a proof $\Pi''$ of $\teqP{\Gammain}{\Deltan}{\En}{\EEn}{\instProca{P''}{i}}{\instProca{Q''}{i}}{C'}$.

Hence, by Lemma~\ref{lem-proof:type-processes-branches},
for all $\Gamma'\in\branch{\Gammain}$, there exist $C_{\Gamma'}\subseteq C'$ and a proof $\Pi_{\Gamma'}$ of
$\teqP{\Gamma'}{\Deltan}{\En}{\EEn}{\instProca{P'}{i}}{\instProca{Q'}{i}}{C_{\Gamma'}}$.

Moreover, by Lemma~\ref{lem-proof:typing-terms-stars}, for all $\Gamma'\in\branch{\Gammain}$,
there exists a proof $\Pi'_{1,\Gamma'}$ of
$\teqTc{\Gamma'}{\Deltan}{\En}{\EEn}{\instTerma{M}{i}}{\instTerma{N}{i}}{\L}{\instConsta{c}{i}}$.
Similarly,
there exists a proof $\Pi'_{2,\Gamma'}$ of
$\teqTc{\Gamma'}{\Deltan}{\En}{\EEn}{\instTerma{M'}{i}}{\instTerma{N'}{i}}{\S}{\instConsta{c'}{i}}$.

In addition, $\Pii = \instProca{\ITE{M}{M'}{P'}{P''}}{i} = \ITE{\instTerma{M}{i}}{\instTerma{M'}{i}}{\instProca{P'}{i}}{\instProca{P''}{i}}$.
Similarly, $\Qii = \ITE{\instTerma{N}{i}}{\instTerma{N'}{i}}{\instProca{Q'}{i}}{\instProca{Q''}{i}}$.

Therefore, using $\Pi_{\Gamma'}$, $\Pi'_{1,\Gamma'}$, $\Pi'_{2,\Gamma'}$ and rule \PIfS, we have for all
$\Gamma'\in\branch{\Gammain}$
\[\teqP{\Gamma'}{\Deltan}{\En}{\EEn}{\Pii}{\Qii}{C_{\Gamma'}\subseteq C' \subseteq \Cin}.\]

Thus by Lemma~\ref{lem-proof:process-all-branches}, there exists
$C_1 \subseteq \Cin$ such that
\[\teqP{\Gammain}{\Deltan}{\En}{\EEn}{\Pii}{\Qii}{C_1},\]
which proves the claim.

\item \case{\PIfI:}
then $P = \ITE{M}{M'}{P'}{P''}$, $Q=\ITE{N}{N'}{Q'}{Q''}$ for some messages $M$, $N$, $M'$, $N'$, and some processes $P'$, $Q'$, $P''$, $Q''$, and there exist types $T$, $T'$, and names $m$, $p$,
such that
\[\Pi=
\inferrule
  {\inferrule*{\Pi'}{\teqP{\Gamma}{\Delta}{\E}{\E'}{P''}{Q''}{C}}\\
   \inferrule*{\Pi_1}{\teqTc{\Gamma}{\Delta}{\E}{\E'}{M}{N}{T*T'}{c}}\\
   \inferrule*{\Pi_2}{\teqTc{\Gamma}{\Delta}{\E}{\E'}{M'}{N'}{\LRTn{l}{\oneorinf}{m}{l'}{\oneorinf}{p}}{c'}}}
  {\teqP{\Gamma}{\Delta}{\E}{\E'}{P}{Q}{C}}.
\]

By applying the induction hypothesis to $\Pi'$,
there exist $C'\subseteq \Cin$ and
a proof $\Pi''$ of $\teqP{\Gammain}{\Deltan}{\En}{\EEn}{\instProca{P''}{i}}{\instProca{Q''}{i}}{C'}$.

Let $\Gamma'\in\branch{\Gammain}$.
By applying Lemma~\ref{lem-proof:type-processes-branches} to $\Pi''$,
there exists $C_{\Gamma'}\subseteq C'$, such that there exists a proof
$\Pi_{\Gamma'}$ of
$\teqP{\Gamma'}{\Deltan}{\En}{\EEn}{\instProca{P''}{i}}{\instProca{Q''}{i}}{C_{\Gamma'}}$.

Moreover, by Lemma~\ref{lem-proof:typing-terms-stars} applied to $\Pi_1$,
there exists a proof $\Pi'_{1,\Gamma'}$ of
$\teqTc{\Gamma'}{\Deltan}{\En}{\EEn}{\instTerma{M}{i}}{\instTerma{N}{i}}{\instTyp{T}{n}*\instTyp{T'}{n}}{\instConsta{c}{i}}$.
Similarly,
there exists a proof $\Pi'_{2,\Gamma'}$ of
$\teqTc{\Gamma'}{\Deltan}{\En}{\EEn}{\instTerma{M'}{i}}{\instTerma{N'}{i}}{\instTyp{\LRTn{l}{\oneorinf}{m}{l'}{\oneorinf}{p}}{n})}{\instConsta{c'}{i}}$.

In addition, $\Pii = \instProca{\ITE{M}{M'}{P'}{P''}}{i} = \ITE{\instTerma{M}{i}}{\instTerma{M'}{i}}{\instProca{P'}{i}}{\instProca{P''}{i}}$.
Similarly, $\Qii = \ITE{\instTerma{N}{i}}{\instTerma{N'}{i}}{\instProca{Q'}{i}}{\instProca{Q''}{i}}$.

We distinguish two cases.

\begin{itemize}
\item \case{If $\oneorinf$ is $1$:}
Then $\instTyp{\LRTn{l}{1}{m}{l'}{1}{p}}{n} = \LRTn{l}{1}{m}{l'}{1}{p}$,
and using $\Pi_{\Gamma'}$, $\Pi'_{1,\Gamma'}$, $\Pi'_{2,\Gamma'}$ and rule \PIfI, we have for all $\Gamma'\in\branch{\Gammain}$
\[\teqP{\Gamma'}{\Deltan}{\En}{\EEn}{\Pii}{\Qii}{C_{\Gamma'}\subseteq C'\subseteq \Cin}.\]
Thus by Lemma~\ref{lem-proof:process-all-branches}, there exists
$C_1 \subseteq \Cin$ such that
\[\teqP{\Gammain}{\Deltan}{\En}{\EEn}{\Pii}{\Qii}{C_1},\]
which proves the claim in this case.

\item \case{If $\oneorinf$ is $\infty$:}
Moreover,
by applying Lemma~\ref{lem-proof:type-terms-branches} to $\Pi'_{2,\Gamma'}$,
there exists a type $T''\in\branch{\instTyp{\LRTn{l}{\infty}{m}{l'}{\infty}{p}}{n}}$, such that there exists a proof $\Pi''_{2,\Gamma'}$
of
$\teqTc{\Gamma'}{\Deltan}{\En}{\EEn}{\instTerma{M'}{i}}{\instTerma{N'}{i}}{T''}{\instConsta{c'}{i}}$.

By definition, $\instTyp{\LRTn{l}{\infty}{m}{l'}{\infty}{p}}{n} = \bigvee_{1\leq j \leq n} \LRTn{l}{1}{m_j}{l'}{1}{p_j}$.
Therefore, by definition of branches, there exists $j$ such that $T'' = \LRTn{l}{1}{m_j}{l'}{1}{p_j}$.

Hence, using $\Pi_{\Gamma'}$, $\Pi'_{1,\Gamma'}$, $\Pi''_{2,\Gamma'}$,
by applying rule \PIfI, we have 
for all $\Gamma'\in\branch{\Gammain}$ that
\[\teqP{\Gamma'}{\Deltan}{\En}{\EEn}{\Pii}{\Qii}{C_{\Gamma'} \subseteq C'\subseteq \Cin}.\]

Thus by Lemma~\ref{lem-proof:process-all-branches}, there exists
$C_1 \subseteq \Cin$ such that
\[\teqP{\Gammain}{\Deltan}{\En}{\EEn}{\Pii}{\Qii}{C_1}\]
which proves the claim in this case.
\end{itemize}

\item \case{\PIfLRinf:}
then $P = \ITE{M_1}{M_2}{P_\top}{P_\bot}$, $Q=\ITE{N_1}{N_2}{Q_\top}{Q_\bot}$ for some messages $M_1$, $N_1$, $M_2$, $N_2$, and some processes $P_\top$, $Q_\top$, $P_\bot$, $Q_\bot$,
and there exist $m, p, l, l'$ such that

\[\Pi=
\inferrule
  {\inferrule*{\Pi_1}{\teqTc{\Gamma}{\Delta}{\E}{\E'}{M_1}{N_1}{\LRTn{l}{\infty}{m}{l'}{\infty}{p}}{\emptyset}}\\
   \inferrule*{\Pi_2}{\teqTc{\Gamma}{\Delta}{\E}{\E'}{M_2}{N_2}{\LRTn{l}{\infty}{m}{l'}{\infty}{p'}}{\emptyset}}\\
   \inferrule*{\Pi_\top}{\teqP{\Gamma}{\Delta}{\E}{\E'}{P_\top}{Q_\top}{C_1}}\\
   \inferrule*{\Pi_\bot}{\teqP{\Gamma}{\Delta}{\E}{\E'}{P_\bot}{Q_\bot}{C_2}}}
  {\teqP{\Gamma}{\Delta}{\E}{\E'}{P}{Q}{C = C_1\cup C_2}}.
\]

By applying the induction hypothesis to $\Pi_\top$,
there exist $C'\subseteq \instCst{C_1}{i}{n} \subseteq \Cin$ and
a proof $\Pi'$ of $\teqP{\Gammain}{\Deltan}{\En}{\EEn}{\instProca{P_\top}{i}}{\instProca{Q_\top}{i}}{C'}$.
Similarly with $\Pi_\bot$,
there exist $C''\subseteq \instCst{C_2}{i}{n} \subseteq \Cin$ and
a proof $\Pi''$ of $\teqP{\Gammain}{\Deltan}{\En}{\EEn}{\instProca{P_\bot}{i}}{\instProca{Q_\bot}{i}}{C''}$.

Hence, by Lemma~\ref{lem-proof:type-processes-branches},
for all $\Gamma'\in\branch{\Gammain}$, there exist $C_{1,\Gamma'}\subseteq C' (\subseteq \Cin)$,
$C_{2,\Gamma'}\subseteq C'' (\subseteq \Cin)$,
and proofs $\Pi_{\top, \Gamma'}$ and $\Pi_{\bot,\Gamma'}$ of
$\teqP{\Gamma'}{\Deltan}{\En}{\EEn}{\instProca{P_\top}{i}}{\instProca{Q_\top}{i}}{C_{1,\Gamma'}}$
and
$\teqP{\Gamma'}{\Deltan}{\En}{\EEn}{\instProca{P_\bot}{i}}{\instProca{Q_\bot}{i}}{C_{2,\Gamma'}}$.

Moreover, by Lemma~\ref{lem-proof:typing-terms-stars} applied to $\Pi_1$, for all $\Gamma'\in\branch{\Gammain}$,
there exists a proof $\Pi'_{1,\Gamma'}$ of
$\teqTc{\Gamma'}{\Deltan}{\En}{\EEn}{\instTerma{M_1}{i}}{\instTerma{N_1}{i}}{\bigvee_{1\leq j \leq n} \LRTn{l}{1}{m_j}{l'}{1}{p_j}}{\instConsta{c}{i}}$.
Similarly,
there exists a proof $\Pi'_{2,\Gamma'}$ of
$\teqTc{\Gamma'}{\Deltan}{\En}{\EEn}{\instTerma{M_2}{i}}{\instTerma{N_2}{i}}{\bigvee_{1\leq j \leq n} \LRTn{l}{1}{m_j}{l'}{1}{p_j}}{\instConsta{c'}{i}}$.

Let $\Gamma'\in\branch{\Gammain}$.
By Lemma~\ref{lem-proof:type-terms-branches}, there exists $T\in\branch{\bigvee_{1\leq j \leq n} \LRTn{l}{1}{m_j}{l'}{1}{p_j}}$
such that there exists a proof $\Pi''_{1, \Gamma'}$ of 
$\teqTc{\Gamma'}{\Deltan}{\En}{\EEn}{\instTerma{M_1}{i}}{\instTerma{N_1}{i}}{T}{\instConsta{c}{i}}$.

Similarly, there exists $T'\in\branch{\bigvee_{1\leq j \leq n} \LRTn{l}{1}{m_j}{l'}{1}{p_j}}$
such that there exists a proof $\Pi''_{2, \Gamma'}$ of 
$\teqTc{\Gamma'}{\Deltan}{\En}{\EEn}{\instTerma{M_2}{i}}{\instTerma{N_2}{i}}{T'}{\instConsta{c}{i}}$.

By definition of branches, there exist $j, j'$ such that
$T = \LRTn{l}{1}{m_j}{l'}{1}{p_j}$ and $T' = \LRTn{l}{1}{m_{j'}}{l'}{1}{p_{j'}}$.

In addition, $\Pii = \instProca{\ITE{M_1}{M_2}{P_\top}{P_\bot}}{i} = \ITE{\instTerma{M_1}{i}}{\instTerma{M_2}{i}}{\instProca{P_\top}{i}}{\instProca{P_\bot}{i}}$.
Similarly, $\Qii = \ITE{\instTerma{N_1}{i}}{\instTerma{N_2}{i}}{\instProca{Q_\top}{i}}{\instProca{Q_\bot}{i}}$.

Therefore, using $\Pi_{\top,\Gamma'}$, $\Pi_{\bot,\Gamma'}$, $\Pi''_{1,\Gamma'}$, $\Pi''_{2,\Gamma'}$ and rule \PIfLR,
either $j = j'$ and we have
\[\teqP{\Gamma'}{\Deltan}{\En}{\EEn}{\Pii}{\Qii}{C_{1,\Gamma'} (\subseteq \Cin)}\]
or $j \neq j'$ and we have
\[\teqP{\Gamma'}{\Deltan}{\En}{\EEn}{\Pii}{\Qii}{C_{2,\Gamma'} (\subseteq \Cin)}.\]

This holds for any $\Gamma'\in\branch{\Gammain}$.

Thus by Lemma~\ref{lem-proof:process-all-branches}, there exists
$C' \subseteq \Cin$ such that
\[\teqP{\Gammain}{\Deltan}{\En}{\EEn}{\Pii}{\Qii}{C'}\]
which proves the claim in this case.

\item \case{\PIfLRp:}
then $P = \ITE{M_1}{M_2}{P_\top}{P_\bot}$, $Q=\ITE{N_1}{N_2}{Q_\top}{Q_\bot}$ for some messages $M_1$, $N_1$, $M_2$, $N_2$, and some processes $P_\top$, $Q_\top$, $P_\bot$, $Q_\bot$,
and there exist names $m,p,m', p'$ such that

\[\Pi=
\inferrule
  {\inferrule*{\Pi_1}{\teqTc{\Gamma}{\Delta}{\E}{\E'}{M_1}{N_1}{\LRTn{l}{\oneorinf}{m}{l'}{\oneorinf}{p}}{\emptyset}}\\
   \inferrule*{\Pi_2}{\teqTc{\Gamma}{\Delta}{\E}{\E'}{M_2}{N_2}{\LRTn{l''}{\oneorinf'}{m'}{l'''}{\oneorinf'}{p'}}{\emptyset}}\\
   \inferrule*{\Pi'}{\teqP{\Gamma}{\Delta}{\E}{\E'}{P_\bot}{Q_\bot}{C}}}
  {\teqP{\Gamma}{\Delta}{\E}{\E'}{P}{Q}{C}}.
\]

By applying the induction hypothesis to $\Pi'$,
there exist $C'\subseteq \Cin$ and
a proof $\Pi''$ of $\teqP{\Gammain}{\Deltan}{\En}{\EEn}{\instProca{P_\bot}{i}}{\instProca{Q_\bot}{i}}{C'}$.

Hence, by Lemma~\ref{lem-proof:type-processes-branches},
for all $\Gamma'\in\branch{\Gammain}$, there exist $C_{\Gamma'}\subseteq C' (\subseteq \Cin)$,
and a proofs$\Pi_{\Gamma'}$ of
$\teqP{\Gamma'}{\Deltan}{\En}{\EEn}{\instProca{P_\bot}{i}}{\instProca{Q_\bot}{i}}{C_{\Gamma'}}$.

Moreover, by Lemma~\ref{lem-proof:typing-terms-stars} applied to $\Pi_1$, for all $\Gamma'\in\branch{\Gammain}$,
there exists a proof $\Pi'_{1,\Gamma'}$ of
$\teqTc{\Gamma'}{\Deltan}{\En}{\EEn}{\instTerma{M_1}{i}}{\instTerma{N_1}{i}}{\instTyp{\LRTn{l}{\oneorinf}{m}{l'}{\oneorinf}{p}}{n}}{\instConsta{c}{i}}$.
Similarly,
there exists a proof $\Pi'_{2,\Gamma'}$ of
$\teqTc{\Gamma'}{\Deltan}{\En}{\EEn}{\instTerma{M_2}{i}}{\instTerma{N_2}{i}}{\instTyp{\LRTn{l''}{\oneorinf'}{m'}{l'''}{\oneorinf'}{p'}}{n}}{\instConsta{c'}{i}}$.

We distinguish several cases, depending on $\oneorinf$ and $\oneorinf'$.
\begin{itemize}
\item \case{if $\oneorinf$ and $\oneorinf'$ are both $1$:}
Then this rule is a particular case of rule \PIfLR, and the result is proved in a similar way.

\item \case{if $\oneorinf$ is $1$ and $\oneorinf'$ is $\infty$:}
Then $\instTyp{\LRTn{l}{\oneorinf}{m}{l'}{\oneorinf}{p}}{n} = \instTyp{\LRTn{l}{1}{m}{l'}{1}{p}}{n}$,
and $\instTyp{\LRTn{l''}{\oneorinf}{m'}{l'''}{\oneorinf}{p'}}{n} = \bigvee_{1\leq j \leq n}
\LRTn{l''}{1}{m'_j}{l'''}{1}{p_j}$.

Let $\Gamma'\in\branch{\Gammain}$.

By Lemma~\ref{lem-proof:type-terms-branches}, using $\Pi'_{2,\Gamma'}$,
there exists $j \in\llbracket 1, n\rrbracket$ such that
there exists a proof $\Pi''_{2,\Gamma'}$ of
$\teqTc{\Gamma'}{\Deltan}{\En}{\EEn}{\instTerma{M_2}{i}}{\instTerma{N_2}{i}}{\LRTn{l''}{1}{m'_j}{l'''}{1}{p'_j}}{\instConsta{c'}{i}}$.

In addition, $\Pii = \instProca{\ITE{M_1}{M_2}{P_\top}{P_\bot}}{i} = \ITE{\instTerma{M_1}{i}}{\instTerma{M_2}{i}}{\instProca{P_\top}{i}}{\instProca{P_\bot}{i}}$.
Similarly, $\Qii = \ITE{\instTerma{N_1}{i}}{\instTerma{N_2}{i}}{\instProca{Q_\top}{i}}{\instProca{Q_\bot}{i}}$.

For any $j\in\llbracket1, n\rrbracket$, $\noncetypelab{l}{1}{m}\neq \noncetypelab{l''}{1}{m'_j}$;
and $\noncetypelab{l'}{1}{p}\neq \noncetypelab{l'''}{1}{p'_j}$.

Therefore, using $\Pi_{\Gamma'}$, $\Pi'_{1,\Gamma'}$, $\Pi''_{2,\Gamma'}$ and rule \PIfLR,
we have
\[\teqP{\Gamma'}{\Deltan}{\En}{\EEn}{\Pii}{\Qii}{C_{\Gamma'} (\subseteq \Cin)}.\]

This holds for any $\Gamma'\in\branch{\Gammain}$.

Thus by Lemma~\ref{lem-proof:process-all-branches}, there exists
$C' \subseteq \Cin$ such that
\[\teqP{\Gammain}{\Deltan}{\En}{\EEn}{\Pii}{\Qii}{C'}\]
which proves the claim in this case.

\item \case{if $\oneorinf$ is $\infty$ and $\oneorinf'$ is $1$:}
This case is similar to the symmetric one.

\item \case{if $\oneorinf$ and $\oneorinf'$ both are $\infty$:}
This case is similar to the case where $\oneorinf$ is $1$ and $\oneorinf'$ is $\infty$.
\end{itemize}
\end{itemize}
\end{proof}


\begin{theorem}[Typing $n$ sessions]
\label{thm-proof:typing-n-sessions}
For all $\Gamma$, 
$P$, $Q$ and $C$, such that 
\[\teqP{\Gamma}{\Delta}{\E}{\E'}{P}{Q}{C}\] 
then for all $n\in\mathbb{N}$, there exists $C'\subseteq \UnionCart_{1\leq i \leq n}\instCst{C}{i}{n}$
such that

\[\teqP{\instGG{\Gamma}{n}}{\instD{\Delta}{n}}{\instE{\E}{n}}{\instE{\E'}{n}}
{\instProca{P}{1}\PAR \dots \PAR \instProca{P}{n}}{\instProca{Q}{1}\PAR\dots\PAR\instProca{Q}{n}}{C'}\] 
where $\instGG{\Gamma}{n}$ is defined as $\bigcup_{1\leq i\leq n} \instG{\Gamma}{i}{n}$.
\end{theorem}
\begin{proof}
\newcommand{\Pii}[1]{\instProca{P}{#1}}
\newcommand{\Qii}[1]{\instProca{Q}{#1}}
\newcommand{\Gammain}{\instG{\Gamma}{i}{n}}
\newcommand{\Gamman}{\instGG{\Gamma}{n}}

Let us assume $\Gamma$, 
$P$, $Q$ and $C$
are such that 
\[\teqP{\Gamma}{\Delta}{\E}{\E'}{P}{Q}{C}.\]
Let $n\in\mathbb{N}$.

Note that the union $\bigcup_{1\leq i\leq n} \instG{\Gamma}{i}{n}$ is well-defined, as for $i\neq j$,
$\dom{\instG{\Gamma}{i}{n}} \cap \dom{\instG{\Gamma}{j}{n}} \subseteq \K\cup\N$,
and the types associated to keys and nonces are the same in each $\instG{\Gamma}{i}{n}$.

\medskip

The property follows from Theorem~\ref{lem-proof:typing-process-star}.
Indeed, this theorem guarantees that for all $i\in\llbracket 1, n\rrbracket$, there exists
$C_i\subseteq \Cin$ such that
\[\teqP{\Gammain}{\Deltan}{\En}{\EEn}{\Pii{i}}{\Qii{i}}{C_i}.\] 

By construction, all variables in $\dom{\Gammain}$ are indexed with $i$, and as we mentioned earlier,
for all $i, j$, $\Gammain$ and $\instG{\Gamma}{j}{n}$ have the same values on their common domain.

Hence we have $\Gamman = \Gammain \uplus (\bigcup_{j\neq i} \onlyvar{(\instG{\Gamma}{j}{n})})$.
Therefore, by Lemma~\ref{lem-proof:typing-contextinclusion}, we have for all $i\in\llbracket 1, n\rrbracket$
\[\teqP{\Gamman}{\Deltan}{\En}{\EEn}{\Pii{i}}{\Qii{i}}{C'_i}\] 
where
\[C'_i = \{(c,\Gamma'\cup\Gamma'')|(c, \Gamma')\in C_i \;\wedge\;\Gamma''\in\branch{\bigcup_{j\neq i} \onlyvar{(\instG{\Gamma}{j}{n})}}\}\]

Thus, by applying rule \PPar $n-1$ times, we have
\[\teqP{\Gamman}{\Deltan}{\En}{\EEn}
{\Pii{1}\PAR \dots \PAR \Pii{n}}{\Qii{1}\PAR\dots\Qii{n}}{\UnionCart_{1\leq i \leq n}C'_i}.\] 

It only remains to be proved that $\UnionCart_{1\leq i \leq n}C'_i \subseteq \UnionCart_{1\leq i \leq n}\instCst{C}{i}{n}$.
Since for all $i\in\llbracket 1, n\rrbracket$ we have $C_i\subseteq \Cin$, by Lemma~\ref{lem-proof:cons-subset}
we know that 
$\UnionCart_{1\leq i \leq n}C_i \subseteq \UnionCart_{1\leq i \leq n}\instCst{C}{i}{n}$.

Hence it suffices to show that $\UnionCart_{1\leq i \leq n}C'_i \subseteq \UnionCart_{1\leq i \leq n}C_i$.

Let $(c, \Gamma')\in\UnionCart_{1\leq i \leq n}C'_i$.
By definition there exist $(c_1,\Gamma_1)\in C'_1,\dots,(c_n,\Gamma_n)\in C'_n$ such that
$c = \bigcup_{1\leq i \leq n} c_i$, $\Gamma' = \bigcup_{1\leq i \leq n} \Gamma_i$, and for all $i\neq j$, $\Gamma_i$ and $\Gamma_j$ are compatible.

For all $i$, $(c_i,\Gamma_i)\in C'_i$. Thus by definition of $C'_i$ there exist
$\Gamma'_i$ and $\Gamma''_i$ such that $(c_i, \Gamma'_i)\in C_i$, $\Gamma''_i\in\branch{\bigcup_{j\neq i} \onlyvar{(\instG{\Gamma}{j}{n})}}$, and $\Gamma_i = \Gamma'_i \cup \Gamma''_i$.
Since for all $i\neq j$, $\Gamma_i$ and $\Gamma_j$ are compatible, we know that $\Gamma'_i$ and $\Gamma'_j$ also are, as well as
$\Gamma'_i$ and $\Gamma''_j$.

Hence, 
$\Gamma' = \bigcup_{1\leq i \leq n} \Gamma_i = \bigcup_{1\leq i \leq n} (\Gamma'_i\cup\Gamma''_i)
= (\bigcup_{1\leq i \leq n} \Gamma'_i) \cup (\bigcup_{1\leq i \leq n} \Gamma''_i)$.

Moreover, $(\bigcup_{1\leq i \leq n} \Gamma''_i) = (\bigcup_{1\leq i \leq n} \Gamma'_i)$.
Indeed, they have the same domain, \ie $\{x_i\;|\; x\in\dom{\Gamma}\wedge 1\leq i \leq n\}$,
and are compatible since for all $i\neq j$, $\Gamma'_i$ and $\Gamma''_j$ are compatible.

Thus $\Gamma' = (\bigcup_{1\leq i \leq n} \Gamma'_i)$, and since the $\Gamma'_i$ are all pairwise compatible,
and for all $i$, $(c_i,\Gamma'_i)\in C_i$, we have
$(c, \Gamma')\in\UnionCart_{1\leq i \leq n}C_i$.

This proves that $\UnionCart_{1\leq i \leq n}C'_i \subseteq \UnionCart_{1\leq i \leq n}C_i$, which concludes the proof.
\end{proof}

\bigskip

This next theorem corresponds to Theorem~\ref{thm:typing-sound-replicated}:

\begin{theorem}
\label{thm-proof:typing-sound-replicated}
\newcommand{\Gamman}{\instGG{\Gamma}{n}}
Consider $P$, $Q$, $P'$ ,$Q'$, $C$, $C'$,
such that $P$, $Q$ and $P'$, $Q'$ do not share any variable.
Consider $\Gamma$, containing only keys and nonces with types of the form $\noncetypelab{l}{1}{n}$.

Assume that $P$ and $Q$ only bind nonces with infinite nonce types,
\ie using $\NEWnew{m}{\noncetypelab{l}{\infty}{m}}$ for some label $l$;
while $P'$ and $Q'$ only bind nonces with finite types, \ie using $\NEWnew{m}{\noncetypelab{l}{1}{m}}$.

Let us abbreviate by $\NEWN{\nm}$ the sequence of declarations of each nonce $m\in\dom{\Gamma}$.
If
\begin{itemize}
\item $\teqPnew{\Gamma}{P}{Q}{C}$,
\item $\teqPnew{\Gamma}{P'}{Q'}{C'}$,
\item $C'\UnionCart(\UCn)$ is consistent for all $n$,
\end{itemize}
then \hfill$\NEWN{\nm}. \;((!P)\PAR P') \equivTrace
\NEWN{\nm}. \;((!Q)\PAR Q')$.\hfill~
\end{theorem}
\begin{proof}
\newcommand{\Gammain}{\instG{\Gamma}{i}{n}}
\newcommand{\Gamman}{\instGG{\Gamma}{n}}

Note that since $\Gamma$ only contains keys and nonces with finite types,
for all $i$, $\instProc{P}{i}{\Gamma} = \instProc{P}{i}{\emptyset}$ is just $P$ where all variables and some names have been $\alpha$-renamed,
and similarly for $Q$. 
Since $P'$, $Q'$ only contain nonces with finite types, $\instProc{P'}{1}{\Gamma}$ and $\instProc{Q'}{1}{\Gamma}$
are $P'$, $Q'$ where all variables have been $\alpha$-renamed.

By Theorem~\ref{thm-proof:typing-n-sessions},
we know that for all $i, n$,
\[\teqP{\instGG{\Gamma}{n}}{\instD{\Delta}{n}}{\instE{\E}{n}}{\instE{\E'}{n}}
{\instProca{P}{1}\PAR \dots \PAR \instProca{P}{n}}{\instProca{Q}{1}\PAR\dots\PAR\instProca{Q}{n}}{C''}\] 
where $\instGG{\Gamma}{n}=\bigcup_{1\leq i\leq n} \instG{\Gamma}{i}{n}$,
and $C''\subseteq \UnionCart_{1\leq i \leq n}\instCst{C}{i}{n}$.

By Theorem~\ref{lem-proof:typing-process-star}, there also exists $C'''\subseteq \instCst{C'}{1}{n}$, such that
\[\teqP{\instG{\Gamma}{1}{n}}{\instD{\Delta}{n }}{\instE{\E}{n}}{\instE{\E'}{n}}
{\instProca{P'}{1}}{\instProca{Q'}{1}}{C'''}.\] 

Therefore, by Lemma~\ref{lem-proof:typing-contextinclusion}, we have
\[\teqP{\Gamman}{\Deltan}{\En}{\EEn}{\instProca{P'}{1}}{\instProca{Q'}{1}}{C''''}\] 
where $C''''$ is $C'''$ where all the environments have been extended with $\bigcup_{1\leq i \leq n}
\onlyvar{(\Gammain)}$
(note that this environment still only contains nonces and keys).

Therefore, by rules \PPar and \PNew,
\[\teqP{\Gamma'}{\instD{\Delta}{n}}{\instE{\E}{n}}{\instE{\E'}{n}}
{\NEWN{\nm}. \;(\instProca{P}{1}\PAR \dots \PAR \instProca{P}{n}) \PAR \instProca{P'}{1}}{\NEWN{\nm}. \;(\instProca{Q}{1}\PAR\dots\PAR\instProca{Q}{n})\PAR \instProca{Q'}{1}}{C''\UnionCart C''''}\]
where $\Gamma'$ is the restriction of $\Gamman$ to keys.

If $\instConst{C'}{1}{n}\UnionCart(\UCn)$ is consistent,
similarly to the reasoning in the proof of Theorem~\ref{thm-proof:typing-n-sessions},
$C''\UnionCart C''''$ also is.

Then, by Theorem~\ref{thm-proof:typing-sound},
\[\NEWN{\nm}. \;(\instProca{P}{1}\PAR \dots \PAR \instProca{P}{n}) \PAR \instProca{P'}{1} \equivTrace \NEWN{\nm}. \;(\instProca{Q}{1}\PAR\dots\PAR\instProca{Q}{n})\PAR \instProca{Q'}{1}\]
which implies (since $\instProca{P'}{1}$ is just a renaming of the variables in $P'$) that
\[\NEWN{\nm}. \;(\instProca{P}{1}\PAR \dots \PAR \instProca{P}{n}) \PAR P' \equivTrace \NEWN{\nm}. \;(\instProca{Q}{1}\PAR\dots\PAR\instProca{Q}{n})\PAR Q'\]
Since $\instProc{P}{i}{\Gamma}$ and $\instProc{Q}{i}{\Gamma}$ are just $\alpha$-renamings of $P$, $Q$,
this implies that for all $n$,
\[\NEWN{\nm}. \;(P_1\PAR \dots \PAR P_n) \PAR P' \equivTrace \NEWN{\nm}. \;(Q_1\PAR\dots\PAR Q_n)\PAR Q'\]
where $P_1 = \dots = P_n = P$, and $Q_1 = \dots = Q_n = Q$.
Therefore
\[\NEWN{\nm}. \;((! P) \PAR P') \equivTrace \NEWN{\nm}. \;((!Q)\PAR Q').\]
\end{proof}

%% file: proofs/consistency2.tex
\subsection{Checking consistency}

In this subsection, we first recall the $\checkconststar$ procedure from Section~\ref{sec:consistency-check}, in more detail, and prove its correctness in the non-replicated case.

\bigskip
\bigskip

For a constraint $c$ and an environment $\Gamma$, let
\[\stepI_\Gamma(c) := (\inst{c}{\sigma_F}{\sigma_F'}, \Gamma'),\]
where
\[F = \{x\in\dom{\Gamma} \;|\;\exists m,n,l,l'.\;\Gamma(x) = \LRTnewnew{\noncetypelab{l}{1}{m}}{\noncetypelab{l'}{1}{n}}\},
\]
$\sigma_F,\sigma_F'$ are the substitutions defined by
\begin{itemize}
\item $\dom{\sigma_F} = \dom{\sigma_F'} = F$
\item $\forall x \in F.\;\forall m,n,l,l'. \LRTnewnew{\noncetypelab{l}{1}{m}}{\noncetypelab{l'}{1}{n}} \Rightarrow \sigma_F(x) = m\;\wedge\;\sigma_F'(x) = n,$
\end{itemize}
and $\Gamma'$ is the environment obtained by
extending the restriction of $\Gamma$ to $\dom{\Gamma}\backslash F$ with $\Gamma'(n) = \noncetypelab{l}{1}{n}$
for all nonce $n$ such that $\noncetypelab{l}{1}{n}$ occurs in $\Gamma$.
This is well defined, since by assumption on the well-formedness of the processes and by definition of the processes,
a name $n$ is always associated with the same label.

\bigskip
\bigskip

Let $\splitrd$ be the reduction relation defined on couples of sets of constraints by (all variables are universally quantified)
\[
\begin{array}{r@{\quad}c@{\quad}l}
(\{\PAIR{M}{N} \eqC \PAIR{M'}{N'}\} \cup c, c') & \splitrd & (\{M\eqC M', N\eqC N'\} \cup c, c')\\
[1em]
(\{\ENC{M}{k} \eqC \ENC{M'}{k}\} \cup c, c') & \splitrd & (\{M\eqC M'\} \cup c, c')\\
\multicolumn{3}{r}{\text{if $\Gamma(k) = \skey{\L}{T}$ for some $T$}} \\
[1em]
(\{\AENC{M}{\PUBK{k}} \eqC \AENC{M'}{\PUBK{k}}\} \cup c, c') & \splitrd & (\{M\eqC M'\} \cup c, c')\\
\multicolumn{3}{r}{\text{if $\Gamma(k) = \skey{\L}{T}$ for some $T$}} \\
[1em]
(\{\SIGN{M}{k} \eqC \SIGN{M'}{k}\} \cup c, c') & \splitrd & (\{M\eqC M'\} \cup c, c')\\
\multicolumn{3}{r}{\text{if $\Gamma(k) = \skey{\L}{T}$ for some $T$}} \\
[1em]
(\{\SIGN{M}{k} \eqC \SIGN{M'}{k}\} \cup c, c') & \splitrd & (\{M\eqC M'\} \cup c, \{\SIGN{M}{k} \eqC \SIGN{M'}{k}\} \cup c')\\
\multicolumn{3}{r}{\text{if $\Gamma(k) = \skey{\S}{T}$ for some $T$}} \\
\end{array}
\]

Let then $\stepII_\Gamma(c) = c_1 \cup c_2$ where $(c_1, c_2)$ is the normal form of $(\c, \emptyset)$ for $\splitrd$.
This definition is equivalent to the one described in Section~\ref{sec:consistency-check},
but more practical for the proofs.

\bigskip
\bigskip

We define the condition $\stepIII_\Gamma(c)$ as: check that
$c$ only contains elements of the form $M \eqC N$ where $M$ and $N$ are both
\begin{itemize}
\item a key $k\in\K$ such that $\exists T. \Gamma(k)=\skey{\L}{T}$;
\item nonces $m,n\in\N$ such that $\Gamma(n) = \noncetypelab{\L}{\oneorinf}{n}\;\wedge\;\Gamma(m)=\noncetypelab{\L}{\oneorinf}{n}$,
\item or public keys, verification keys, or constants;
\item or $\ENC{M'}{k}$, $\ENC{N'}{k}$ such that $\exists T. \Gamma(k)=\skey{\L}{T}$;
\item or either hashes $\HASH{M'}$, $\HASH{N'}$ or encryptions $\AENC{M'}{\PUBK{k}}$, $\AENC{N'}{\PUBK{k}}$ with a honest key $k$, \ie such that $\exists T. \Gamma(k)=\skey{\S}{T}$; 
such that $M'$ and $N'$ contain directly under pairs a nonce $n$ such that $\Gamma(n) = \noncetypelab{\S}{a}{n}$ or a secret key $k$ such that $\exists T. \Gamma(k)=\skey{\S}{T}$;
\item or signatures $\SIGN{M'}{k}$, $\SIGN{N'}{k}$ with honest keys, such that $\exists T. \Gamma(k)=\skey{\S}{T}$;
\end{itemize}
$\stepIII_\Gamma(c)$ returns \texttt{true} if this check succeeds and \texttt{false} otherwise.

\bigskip
\bigskip

We then proceed to $\stepIV$. We define condition $\stepIV_\Gamma(c)$ as follows.
We consider all $M \eqC M' \in c$ and $N \eqC N'\in c$, such that $M$, $N$ are unifiable with a most general unifier $\mu$,
and such that
\[\forall x\in\dom{\mu}.\; \forall l, l', m, n.\; 
(\Gamma(x)=\LRTnewnew{ \noncetypelab{l}{\infty}{m}}{ \noncetypelab{l'}{\infty}{p}}) \Rightarrow (x\mu\in\X\;\vee\;\exists i.\; x\mu=m_i).
\]

We then define the substitution $\theta$, over all variables $x\in\dom{\mu}$ such that
$\Gamma(x)=\LRTnewnew{ \noncetypelab{l}{\infty}{m}}{ \noncetypelab{l'}{\infty}{p}}$
by

\[
\forall x\in\dom{\mu}.\; \forall l, l', m, p, i.\; 
\quad(\Gamma(x)=\LRTnewnew{ \noncetypelab{l}{\infty}{m}}{ \noncetypelab{l'}{\infty}{p}} \;\wedge\;\mu(x)=m_i)\Rightarrow\theta(x)=p_i
\]
and $\theta(x)=x$ otherwise.

Let then $\alpha$ be the restriction of $\mu$ to
$\{x\in\dom{\mu}\;|\; \Gamma(x)=\L \;\wedge\; \mu(x)\in\N\}$.

We then check that $M'\alpha\theta = N'\alpha\theta$.

\medskip 

Similarly, we check that the symmetric condition, when $M'$ and $N'$ are unifiable, holds for all
$M \eqC M' \in c$ and $N \eqC N'\in c$.

\medskip

If all these checks succeed, $\stepIV_\Gamma(c)$ returns $\mathtt{true}$.

\bigskip
\bigskip

Finally, $\checkconststar(C)$ is computed by considering all $(c,\Gamma)\in C$.
We let $(\c,\Gammao) = \stepI_\Gamma(c)$, and $\cc = \stepII_{\Gammao}(\c)$.
We then check that $\stepIII_{\Gammao}(\cc) = \mathtt{true}$ and $\stepIV_{\Gammao}(\cc) = \mathtt{true}$.
If this check succeeds for all $(c,\Gamma)\in C$, we say that $\checkconststar(C) = \mathtt{true}$.

\bigskip
\bigskip
\bigskip

%
%
%
%
%
%
%
%
%
%
%
%
%
%

\bigskip\bigskip

Note that we only consider constraints obtained by typing, and therefore such that
there exists $c_\phi$ such that $\teqTc{\Gamma}{\Delta}{\E}{\E'}{\phiL{c}}{\phiR{c}}{\L}{c_\phi}$.

Indeed, it is clear by induction on the typing rules for terms that:
\[\forall \Gamma, 
M, N, T, c.\quad
\teqTcGDE{M}{N}{T}{c} \Longrightarrow (\forall u\eqC v\in c.\quad \exists c'.\quad \teqTcGDE{u}{v}{\L}{c'}).
\]
From this result, and using Lemmas~\ref{lem-proof:term-branch} and~\ref{lem-proof:env-const-branch}, it follows clearly by induction on the typing rules for processes that
\[
\forall \Gamma, 
P, Q, C.\quad 
\teqPGDE{P}{Q}{C} \Longrightarrow
(\forall (c,\Gamma') \in C.\quad \forall u\eqC v\in c.\quad \exists c'.\quad
\teqTc{\Gamma'}{\Delta}{\E}{\E'}{u}{v}{\L}{c'}).
\]

\bigskip\bigskip
Let us now prove 
that the procedure is correct for constraints without infinite nonce types, \ie constraint sets $C$ such that
\[\forall (c,\Gamma)\in C.\;\forall l, l', m, n.\;
 \Gamma(x) \neq \LRTn{l}{\infty}{m}{l'}{\infty}{n}.\]

We fix such a constraint set $C$ (obtained by typing).

Let $(c, \Gamma)\in C$.
Let $(\c,\Gammao) = \stepI_\Gamma(c)$, and $\cc = \stepII_{\Gammao}(\c)$.
Let us assume that $\stepIII_{\Gammao}(\cc) = \mathtt{true}$ and $\stepIV_{\Gammao}(\cc) = \mathtt{true}$.

\bigskip

\begin{lemma}
\label{lem-proof:co-c-const}
If $\c$ is consistent in $\Gammao$, 
then $c$ is consistent in $\Gamma$.
\end{lemma}
\begin{proof}
Let $c'$ be a set of constraints and $\Gamma'$ be a typing environment such that
$c' \subseteq c$, $\Gamma' \subseteq \Gamma$, $\novar{\Gamma'} = \novar{\Gamma}$ and $\var{c'}\subseteq \dom{\Gamma'}$.
Let $\sigma$, $\sigma'$ be two substitutions such that $\wtc{\Delta}{\E}{\E'}{\sigma}{\sigma'}{\Gamma'}{c_\sigma}$
(for some set of constraints $c_\sigma$).

To prove the claim, we need to show that the frames $\NEWN{\E_\Gamma}.(\phiEE \cup \phiL{\inst{c'}{\sigma}{\sigma'}})$ and $\NEWN{\E_\Gamma}.(\phiEE \cup \phiR{\inst{c'}{\sigma}{\sigma'}})$ are statically equivalent.
Let $D$ denote $\dom{\onlyvar{\Gamma'}} (= \dom{\sigma} = \dom{\sigma'})$.

For all $x\in F\cap D$, by definition of $F$, there exist 
$m,n,l,l'$ such that $\Gamma(x) = \LRTnewnew{\noncetypelab{l}{1}{m}}{\noncetypelab{l'}{1}{n}}$.
Thus, by well-typedness of $\sigma$, $\sigma'$, there exists $c_x$ such that $\teqTc{\Gamma}{\Delta}{\E}{\E'}{\sigma(x)}{\sigma'(x)}{\LRTnewnew{\noncetypelab{l}{1}{m}}{\noncetypelab{l'}{1}{n}}}{c_x}$.
Hence, by Lemma~\ref{lem-proof:lr-ground}, since $\sigma$, $\sigma'$ are ground, we have $\sigma(x) = m$ and $\sigma'(x) = n$.
Therefore, $\sigma|_{D\cap F} = \sigma_F|_D$ and $\sigma'|_{D\cap F} = \sigma'_F|_D$.

Let $c''$ be the set $\inst{c'}{\sigma|_{D\cap F}}{\sigma'|_{D\cap F}}$.
By Lemma~\ref{lem-proof:cons-subset}, we have $\inst{c'}{\sigma}{\sigma'} = \inst{c''}{\sigma|_{D\backslash F}}{\sigma'|_{D\backslash F}}$.
We also have $c'' = \inst{c'}{\sigma_F|_D}{\sigma'_F|_D}$, which is equal to $\inst{c'}{\sigma_F}{\sigma'_F}$ since
$\var{c'} \subseteq D$.
Hence $c'' \subseteq \c$.

Let $\Gamma'' = \Gamma'|_{\dom{\Gamma'}\backslash F}$. We have $\Gamma'' \subseteq \Gammao$.

Moreover, since $\wtc{\Delta}{\E}{\E'}{\sigma}{\sigma'}{\Gamma'}{c_\sigma}$,
it is clear from the definition of well-typedness for substitutions that we also have $\wtc{\Delta}{\E}{\E'}{\sigma|_{D\backslash F}}{\sigma'|_{D\backslash F}}{\Gamma''}{c'_\sigma}$ for some $c'_\sigma$.
Finally, $\var{c''} \subseteq \var{c'}\backslash F$ by definition of instantiation, thus $\var{c''}\subseteq \dom{\Gamma''}$.

We have established that $c''\subseteq \c$, $\Gamma''\subseteq \Gammao$, $\var{c''}\subseteq \dom{\Gamma''}$,
and $\wtc{\Delta}{\E}{\E'}{\sigma|_{D\backslash F}}{\sigma'|_{D\backslash F}}{\Gamma''}{c'_\sigma}$.
Therefore, by definition of the consistency of $\c$ in $\Gammao$, 
the frames
$\NEWN{\E_{\Gammao}}.(\phiEEO \cup \phiL{\inst{c''}{\sigma|_{D\backslash F}}{\sigma'|_{D\backslash F}}})$
and $\NEWN{\E_{\Gammao}}.(\phiEEO \cup \phiR{\inst{c''}{\sigma|_{D\backslash F}}{\sigma'|_{D\backslash F}}})$
are statically equivalent.

Since $\phiEEO \subseteq \phiEE$, that is to say that
$\NEWN{\E_\Gamma}.(\phiEE \cup \phiL{\inst{c'}{\sigma}{\sigma'}})$ and $\NEWN{\E_\Gamma}.(\phiEE \cup \phiR{\inst{c'}{\sigma}{\sigma'}})$
are statically equivalent.
This proves the consistency of $c$ in $\Gamma$.
\end{proof}

\begin{lemma}
\label{lem-proof:split-recipe}
If $(c_1, c_2) \splitrd^* (c_1', c_2')$ then for all $x\in \dom{\phiL{c_1'\cup c_2'}}$ there exists a recipe $R$ such that
\begin{itemize}
\item $\var{R}\subseteq \dom{\phiEE\cup\phiL{c_1\cup c_2}}$
\item $\phiL{c_1'\cup c_2'}(x) = \eval{R(\phiEE\cup\phiL{c_1\cup c_2})}$
\item $\phiR{c_1'\cup c_2'}(x) = \eval{R(\phiEE\cup\phiR{c_1\cup c_2})}$.
\end{itemize}

Conversely,
if $(c_1, c_2) \splitrd^* (c_1', c_2')$ then for all $x\in \dom{\phiL{c_1\cup c_2}}$ there exists a recipe $R$ without
destructors, \ie in which $\DECNA$, $\ADECNA$, $\CHECKNA$, $\FSTNA$, $\SNDNA$ 
do not appear,
such that
\begin{itemize}
\item $\var{R}\subseteq \dom{\phiEE\cup\phiL{c_1'\cup c_2'}}$
\item $\phiL{c_1\cup c_2}(x) = R(\phiEE\cup\phiL{c_1'\cup c_2'})$
\item $\phiR{c_1\cup c_2}(x) = R(\phiEE\cup\phiR{c_1'\cup c_2'})$.
\end{itemize}
\end{lemma}
\begin{proof}
For both directions, 
it suffices to prove that the claim holds for all $c_1$, $c_2$, $c_1'$, $c_2'$ such that $(c_1, c_2) \splitrd (c_1', c_2')$.
Indeed, in that case we prove the result for $\splitrd^*$ by composing all the recipes.
The proof for one reduction step is 
clear by examining the cases for the reduction $\splitrd$.

\end{proof}

\begin{lemma}
\label{lem-proof:cc-co-const}
If $\cc$ is consistent in $\Gammao$, 
then $\c$ is consistent in $\Gammao$.
\end{lemma}
\begin{proof}
This follows directly from Lemma~\ref{lem-proof:split-recipe}.

Let $c'$ be a set of constraints and $\Gamma'$ be a typing environment such that
$c' \subseteq \c$, $\Gamma' \subseteq \Gammao$, $\novar{\Gamma'} = \novar{\Gammao}$ and $\var{c'}\subseteq \dom{\Gamma'}$.
Let $\sigma$, $\sigma'$ be two substitutions such that $\wtc{\Delta}{\E}{\E'}{\sigma}{\sigma'}{\Gamma'}{c_\sigma}$
(for some set of constraints $c_\sigma$).

To prove the claim, we need to show that the frames $\NEWN{\EGO}.(\phiEEO \cup \phiL{\inst{c'}{\sigma}{\sigma'}})$ and $\NEWN{\EGO}.(\phiEEO \cup \phiR{\inst{c'}{\sigma}{\sigma'}})$ are statically equivalent.

Since $\cc$ is consistent in $\Gammao$, we know that
the frames $\NEWN{\EGO}.(\phiEEO \cup \phiL{\inst{\cc}{\sigma}{\sigma'}})$ and $\NEWN{\EGO}.(\phiEEO \cup \phiR{\inst{\cc}{\sigma}{\sigma'}})$ are statically equivalent.

By Lemma~\ref{lem-proof:split-recipe}, the frames $\NEWN{\EGO}.(\phiEEO \cup \phiL{\inst{c'}{\sigma}{\sigma'}})$ and $\NEWN{\EGO}.(\phiEEO \cup \phiR{\inst{c'}{\sigma}{\sigma'}})$ can be written as a recipe on
the frames
$\NEWN{\EGO}.(\phiEEO \cup \phiL{\inst{\cc}{\sigma}{\sigma'}})$ and $\NEWN{\EGO}.(\phiEEO \cup \phiR{\inst{\cc}{\sigma}{\sigma'}})$.

Therefore, they are also statically equivalent, which proves the claim.
\end{proof}

%
%

\begin{lemma}
\label{lem-proof:co-low}
There exists $c_\phi$ such that $\teqTc{\Gammao}{\Delta}{\E}{\E'}{\phiEE \cup \phiL{\c}}{\phiEE\cup \phiR{\c}}{\L}{c_\phi}$.
\end{lemma}
\begin{proof}
As explained previously, 
there exists $c'_\phi$ such that $\teqTc{\Gamma}{\Delta}{\E}{\E'}{\phiL{c}}{\phiR{c}}{\L}{c'_\phi}$.

Moreover, we have by definition $\Gamma = \Gamma_F \uplus \Gammao'$, where $F$ is defined as in $\stepI$
and $\Gamma_F$ is the restriction of $\Gamma$ to $F$, and for some $\Gammao'\subseteq\Gammao$.
In addition $\wtc{\Delta}{\E}{\E'}{\sigma_F}{\sigma'_F}{\Gamma_F}{c'}$ for some $c'$.
By definition of $F$, and since the refinement types in $\Gamma$ only contain ground terms by assumption,
we also know that $\Gammao$ does not contain refinement types.
Hence, by Lemma~\ref{lem-proof:subst-typing}, and Lemma~\ref{lem-proof:typing-contextinclusion},
there exists $c_\phi$
such that $\teqTc{\Gammao}{\Delta}{\E}{\E'}{\phiL{c}\sigma_F}{\phiR{c}\sigma'_F}{\L}{c_\phi}$.
Since $\c = \inst{c}{\sigma_F}{\sigma'_F}$, this proves that
$\teqTc{\Gammao}{\Delta}{\E}{\E'}{\phiL{\c}}{\phiR{\c}}{\L}{c_\phi}$.
Besides, it is clear from the definition of $\phiEE$ and rules \TCst, \TNonceL, \TKey, \TPubkey, \TVkey that
$\teqTc{\Gammao}{\Delta}{\E}{\E'}{\phiEE}{\phiEE}{\L}{\emptyset}$.

These two results prove the lemma.
\end{proof}

\begin{lemma}
\label{lem-proof:cc-low}
There exists $c_\phi$ such that
$\teqTc{\Gammao}{\Delta}{\E}{\E'}{\phiEE \cup \phiL{\cc}}{\phiEE\cup \phiR{\cc}}{\L}{c_\phi}$.
\end{lemma}
\begin{proof}
By Lemma~\ref{lem-proof:co-low}, 
there exists $c'_\phi$ such that $\teqTc{\Gammao}{\Delta}{\E}{\E'}{\phiEEO\cup\phiL{\c}}{\phiEEO\cup\phiR{\c}}{\L}{c'_\phi}$.

Moreover, by definition, there exist $c_1$, $c_2$ such that $(\c, \emptyset) \splitrd^* (c_1, c_2)$
and $\cc = c_1 \cup c_2$.
Hence, we know by Lemma~\ref{lem-proof:split-recipe} that for all $x\in \dom{\phiL{\cc}} (=\dom{\phiR{\cc}})$ there exists a 
recipe $R$ such that
$\names{R} = \emptyset$,
$\phiL{\cc}(x) = \eval{R(\phiEEO\cup\phiL{\c})}$ and 
$\phiR{\cc}(x) = \eval{R(\phiEEO\cup\phiR{\c})}$.
Thus, by Lemma~\ref{lem-proof:l-type-recipe}, there exists $c_\phi$ such that
$\teqTc{\Gammao}{\Delta}{\E}{\E'}{\phiL{\cc}}{\phiR{\cc}}{\L}{c_\phi}$.
Besides, it is clear from the definition of $\phiEEO$ and rules \TCst, \TNonceL, \TKey, \TPubkey, \TVkey that
$\teqTc{\Gammao}{\Delta}{\E}{\E'}{\phiEEO}{\phiEEO}{\L}{\emptyset}$.

These two results prove the lemma.
\end{proof}

\bigskip\bigskip

We now assume that $\cc$ satisfies the condition $\stepIII_{\Gammao}(\cc)$.

Note that we write $\splitrd$ for $\splitr{\Gammao}$ as these relations are equal.

\begin{lemma}
\label{lem-proof:split-sign-recipe}
If $(\c, \emptyset) \splitrd^* (c_1, c_2)$ and $\SIGN{M}{k} \eqC \SIGN{N}{k'} \in c_2$ then
there exists a recipe $R$ without
destructors, \ie in which $\DECNA$, $\ADECNA$, $\CHECKNA$, $\FSTNA$, $\SNDNA$ 
do not appear,
such that
\begin{itemize}
\item $\var{R}\subseteq \dom{\phiEEO\cup\phiL{c_1\cup c_2}}$
\item $M = R(\phiEEO\cup\phiL{c_1\cup c_2})$
\item $N = R(\phiEEO\cup\phiR{c_1\cup c_2})$.
\end{itemize}
\end{lemma}
\begin{proof}
We prove this property by induction on the length of the reduction.
It trivially holds if no reduction step is performed since in that case $c_2 = \emptyset$.
Otherwise there exist $c_1'$, $c_2'$ such that
$(\c, \emptyset) \splitrd^* (c_1', c_2') \splitrd (c_1, c_2)$.

If $(c_1', c_2') \splitrd (c_1, c_2)$ is any case except the honest signature case,
we have $c_2 = c_2'$.
Thus if $\SIGN{M}{k} \eqC \SIGN{N}{k'} \in c_2$, then by the induction hypothesis
there exists $R'$ without destructors such that
\begin{itemize}
\item $\var{R'}\subseteq \dom{\phiEEO\cup\phiL{c_1'\cup c_2'}}$
\item $M = R'(\phiEEO\cup\phiL{c_1'\cup c_2'})$
\item $N = R'(\phiEEO\cup\phiR{c_1'\cup c_2'})$.
\end{itemize}
We then prove the claim by applying (the second part of) Lemma~\ref{lem-proof:split-recipe} 
and composing the recipes.

If $(c_1', c_2') \splitrd (c_1, c_2)$ corresponds to the honest signature case,
we have $c_1' = c_1'' \cup \{\SIGN{M'}{k''} \eqC \SIGN{N'}{k''}\}$, $c_1 = c_1'' \cup \{M' \eqC N'\}$, and
$c_2 = c_2' \cup \{\SIGN{M'}{k''} \eqC \SIGN{N'}{k''}\}$
for some $c_1''$, $M'$, $N'$, $k''$, $T$ such that $\Gamma(k') = \skey{\S}{T}$.
If $(\SIGN{M'}{k''}, \SIGN{N'}{k''})\neq (\SIGN{M}{k}, \SIGN{N}{k'})$, then the same proof as in the previous case shows the claim.
Otherwise, $M \eqC N \in c_1$, and therefore the claim trivially holds.
\end{proof}

\begin{lemma}
\label{lem-proof:cc-sign-recipe}
If $\SIGN{M}{k} \eqC \SIGN{N}{k'} \in \cc$ then
there exists a recipe $R$ without
destructors, \ie in which $\DECNA$, $\ADECNA$, $\CHECKNA$, $\FSTNA$, $\SNDNA$ 
do not appear,
such that
\begin{itemize}
\item $\names{R} = \emptyset$,
\item $\var{R}\subseteq \dom{\phiEEO\cup\phiL{\cc}}$
\item $M = R(\phiEEO\cup\phiL{\cc})$
\item $N = R(\phiEEO\cup\phiR{\cc})$.
\end{itemize}
\end{lemma}
\begin{proof}
This property directly follows from Lemma~\ref{lem-proof:split-sign-recipe}, applied to $(\c, \emptyset) \splitrd^* (c_1, c_2)$
such that $\cc = c_1\cup c_2$.
Indeed, since $(c_1, c_2)$ is a normal form for $\splitrd$, if $\SIGN{M}{k} \eqC \SIGN{N}{k'} \in \cc = c_1\cup c_2$,
then $\SIGN{M}{k} \eqC \SIGN{N}{k'} \in c_2$, as if it was an element of $c_1$ then another reduction step would be possible.
\end{proof}

\begin{lemma}
\label{lem-proof:cc-recipe-dest}
For all recipe $R$ such that $\var{R}\subseteq \dom{\phiEEO\cup\phiL{\cc}}$, 
and $\eval{R(\phiEEO\cup\phiL{\cc})} \neq \bot$ or $\eval{R(\phiEEO\cup\phiR{\cc})} \neq \bot$,
there exists a recipe $R'$ without destructors, \ie in which $\DECNA$, $\ADECNA$, $\CHECKNA$, $\FSTNA$, $\SNDNA$, 
do not appear,
such that
\begin{itemize}
\item $\var{R'}\subseteq \dom{\phiEEO\cup\phiL{\cc}}$,
\item $\eval{R(\phiEEO\cup\phiL{\cc})} = R'(\phiEEO\cup\phiL{\cc})$,
\item $\eval{R(\phiEEO\cup\phiR{\cc})} = R'(\phiEEO\cup\phiR{\cc})$.
\end{itemize}
\end{lemma}
\begin{proof}

Let us first note that
\[\eval{R(\phiEEO\cup\phiL{\cc})} \neq \bot\text{ or } \eval{R(\phiEEO\cup\phiR{\cc})} \neq \bot\]
is equivalent to 
\[\eval{R(\phiEEO\cup\phiL{\cc})} \neq \bot \text{ and }\eval{R(\phiEEO\cup\phiR{\cc})} \neq \bot.\]
This follows from Lemmas~\ref{lem-proof:cc-low} and~\ref{lem-proof:l-type-recipe}.

We prove the property by induction on $R$.

\begin{itemize}
\item \case{If $R = n \in \N$ or $R = x\in\AX$ or $R = a\in\CST\cup\FN$} then the claim holds with $R' = R$.
\item \case{If the head symbol of $R$ is a constructor},\ie if there exist $R_1$, $R_2$ such that
$R = \PUBK{R_1}$ or $R = \VK{R_1}$ or $R = \ENC{R_1}{R_2}$ or $R = \AENC{R_1}{R_2}$ or $R = \SIGN{R_1}{R_2}$ or $R = \PAIR{R_1}{R_2}$ or $R = \HASH{R_1}$, we may apply the induction hypothesis to
$R_1$ (and $R_2$ when it is present). All these case are similar, we write the proof generically for $R = f(R_1,R_2)$.
By the induction hypothesis, there exist $R_1'$, $R_2'$ such that
\begin{itemize}
\item $\var{R_1'}\cup\var{R_2'}\subseteq \dom{\phiEEO\cup\phiL{\cc}}$,
\item for all $i\in\{1, 2\}$, $\eval{R_i(\phiEEO\cup\phiL{\cc})} = R_i'(\phiEEO\cup\phiL{\cc})$,
\item for all $i\in\{1, 2\}$, $\eval{R_i(\phiEEO\cup\phiR{\cc})} = R_i'(\phiEEO\cup\phiR{\cc})$.
\end{itemize}

Let $R'$ be $f(R_1',R_2')$. The first two points imply that $R'$ satisfies the conditions on variables.
Since $\eval{R(\phiEEO\cup\phiL{\cc})} = f(\eval{R_1(\phiEEO\cup\phiL{\cc})}, \eval{R_2(\phiEEO\cup\phiL{\cc})})$,
the third point implies that $\eval{R(\phiEEO\cup\phiL{\cc})} = R'(\phiEEO\cup\phiL{\cc})$.
Similarly, $\eval{R(\phiEEO\cup\phiR{\cc})} = R'(\phiEEO\cup\phiR{\cc})$, and the claim holds.

\item \case{If $R = \DEC{S}{K}$ for some recipes $S$, $K$}, then since 
$\eval{R(\phiEEO\cup\phiL{\cc})} \neq \bot$, we have $\eval{K(\phiEEO\cup\phiL{\cc})} = k$ for some $k\in K$,
and $\eval{S(\phiEEO\cup\phiL{\cc})} = \ENC{M}{k}$, where $M = \eval{R(\phiEEO\cup\phiL{\cc})}$.

Similarly, there exists $k'\in\K$ such that
$\eval{K(\phiEEO\cup\phiR{\cc})} = k'$ and
$\eval{S(\phiEEO\cup\phiR{\cc})} = \ENC{N}{k'}$, where $N = \eval{R(\phiEEO\cup\phiR{\cc})}$.

In addition, by Lemma~\ref{lem-proof:cc-low}, there exists $c'$ such that
$\teqTc{\Gammao}{\Delta}{\E}{\E'}{\phiEEO \cup \phiL{\cc}}{\phiEEO\cup \phiR{\cc}}{\L}{c'}$.
Thus by Lemma~\ref{lem-proof:l-type-recipe}, there exists $c''$ such that
$\teqTc{\Gammao}{\Delta}{\E}{\E'}{\eval{K(\phiEEO\cup\phiL{\cc})}}{\eval{K(\phiEEO\cup\phiR{\cc})}}{\L}{c''}$, which is to say
$\teqTc{\Gammao}{\Delta}{\E}{\E'}{k}{k'}{\L}{c''}$.
Hence by Lemma~\ref{lem-proof:type-key-nonce}, $k = k'$ and $\Gamma(k) = \skey{\L}{T}$ for some type $T$.

Since $\eval{S(\phiEEO\cup\phiL{\cc})} = \ENC{M}{k} \neq \bot$, 
by the induction hypothesis, there exists $S'$ such that 
$\var{S'}\subseteq \var{S}$,
$S'(\phiEEO\cup\phiL{\cc}) = \eval{S(\phiEEO\cup\phiL{\cc})} = \ENC{M}{k}$, and
$S'(\phiEEO\cup\phiR{\cc}) = \eval{S(\phiEEO\cup\phiR{\cc})} = \ENC{N}{k}$.

It is then clear that
either $S' = x$ for some variable $x\in\AX$, or $S' = \ENC{S''}{K'}$ for some $S''$, $K'$.
The first case is impossible, since we have already shown that $\Gamma(k) = \skey{\L}{T}$,
and since by $\stepIII_{\Gammao}(\cc)$, 
$\cc$ only contains messages encrypted with keys $k''$ such that 
$\Gamma(k'') = \skey{\S}{T'}$ for some $T'$.

Hence there exist $S''$, $K'$ such that $S' = \ENC{S''}{K'}$.
Since $S'(\phiEEO\cup\phiL{\cc}) = \ENC{M}{k}$, we have $S''(\phiEEO\cup\phiL{\cc}) = M$.
Hence $\eval{R(\phiEEO\cup\phiL{\cc})} = M = S''(\phiEEO\cup\phiL{\cc})$, and similarly for $\phiR{\cc}$.
Moreover, $S''$ being a subterm of $S'$ it also satisfies the conditions on the domains, and thus
the property holds with $R' = S''$.

\item \case{If $R = \ADEC{S}{K}$ for some recipes $S$, $K$}: this case is similar to the symmetric case.

\item \case{If $R = \CHECK{S}{K}$ for some recipes $S$, $K$}: then since 
$\eval{R(\phiEEO\cup\phiL{\cc})} \neq \bot$, we have $\eval{K(\phiEEO\cup\phiL{\cc})} = \VK{k}$ for some $k\in K$,
and $\eval{S(\phiEEO\cup\phiL{\cc})} = \SIGN{M}{k}$, where $M = \eval{R(\phiEEO\cup\phiL{\cc})}$.

Similarly, there exists $k'\in\K$ such that
$\eval{K(\phiEEO\cup\phiR{\cc})} = \VK{k'}$ and
$\eval{S(\phiEEO\cup\phiR{\cc})} = \SIGN{N}{k'}$, where $N = \eval{R(\phiEEO\cup\phiR{\cc})}$.

Since $\eval{S(\phiEEO\cup\phiL{\cc})} = \SIGN{M}{k} \neq \bot$, 
by the induction hypothesis, there exists $S'$ such that 
$\var{S'}\subseteq \var{S}$,
$S'(\phiEEO\cup\phiL{\cc}) = \eval{S(\phiEEO\cup\phiL{\cc})} = \SIGN{M}{k}$, and
$S'(\phiEEO\cup\phiR{\cc}) = \eval{S(\phiEEO\cup\phiR{\cc})} = \SIGN{N}{k}$.

Since $S'(\phiEEO\cup\phiL{\cc}) = \SIGN{M}{k}$, it is clear from the definition of $\evalNA$
that either $S' = x$ for some $x\in\AX$, or $S' = \SIGN{S''}{K'}$ for some $S''$, $K'$.

In the first case, we therefore have $\SIGN{M}{k} \eqC \SIGN{N}{k'} \in \cc$, and Lemma~\ref{lem-proof:cc-sign-recipe} directly proves the claim.

In the second case, there exist $S''$, $K'$ such that $S' = \SIGN{S''}{K'}$.
Since $S'(\phiEEO\cup\phiL{\cc}) = \SIGN{M}{k}$, we have $S''(\phiEEO\cup\phiL{\cc}) = M$.
Hence $\eval{R(\phiEEO\cup\phiL{\cc})} = M = S''(\phiEEO\cup\phiL{\cc})$, and similarly for $\phiR{\cc}$.
Moreover, $S''$ being a subterm of $S'$ it also satisfies the conditions on the domains, and thus
the property holds with $R' = S''$.

\item \case{If $R = \FST{S}$ for some recipe $S$} then since 
$\eval{R(\phiEEO\cup\phiL{\cc})} \neq \bot$, we have 
$\eval{S(\phiEEO\cup\phiL{\cc})} = \PAIR{M_1}{M_2}$, where $M_1 = \eval{R(\phiEEO\cup\phiL{\cc})}$, and $M_2$ is a message.

Similarly,
$\eval{S(\phiEEO\cup\phiR{\cc})} = \PAIR{N_1}{N_2}$, where $N_1 = \eval{R(\phiEEO\cup\phiR{\cc})}$, and $N_2$ is a message.

Since $\eval{S(\phiEEO\cup\phiL{\cc})} = \PAIR{M_1}{M_2} \neq \bot$, 
by the induction hypothesis, there exists $S'$ such that 
$\var{S'}\subseteq \var{S}$,
$S'(\phiEEO\cup\phiL{\cc}) = \eval{S(\phiEEO\cup\phiL{\cc})} = \PAIR{M}{k}$, and
$S'(\phiEEO\cup\phiR{\cc}) = \eval{S(\phiEEO\cup\phiR{\cc})} = \PAIR{N}{k}$.

Since $S'(\phiEEO\cup\phiL{\cc}) = \PAIR{M_1}{M_2}$, it is clear from the definition of $\evalNA$
that either $S' = x$ for some $x\in\AX$, or $S' = \PAIR{S_1}{S_2}$ for some $S_1$, $S_2$.

The first case is impossible, since by $\stepIII_{\Gammao}(\cc)$, 
$\cc$ does not contain pairs.

In the second case, there exist $S_1$, $S_2$ such that $S' = \PAIR{S_1}{S_2}$.
Since $S'(\phiEEO\cup\phiL{\cc}) = \PAIR{M_1}{M_2}$, we have $S_1(\phiEEO\cup\phiL{\cc}) = M_1$.
Hence $\eval{R(\phiEEO\cup\phiL{\cc})} = M_1 = S_1(\phiEEO\cup\phiL{\cc})$, and similarly for $\phiR{\cc}$.
Moreover, $S_1$ being a subterm of $S'$ it also satisfies the conditions on the domains, and thus
the property holds with $R' = S_1$.

\item \case{If $R = \SND{S}$ for some $S$}: this case is similar to the
$\FSTNA$ case.

\end{itemize}
\end{proof}

\begin{lemma}
\label{lem-proof:subst-red}
For all term $t$ and substitution $\sigma$ containing only messages, if $\eval{t}\neq\bot$, then $\eval{(t\sigma)} = (\eval{t})\sigma$.
\end{lemma}
\begin{proof}
This property is easily proved by induction on $t$.
In the base case where $t$ is a variable $x$, by definition of $\evalNA$, since $\sigma(x)$ is a messages, $\eval{\sigma(x)} = \sigma(x)$ and the claim holds.
In the other base cases where $t$ is a name, key or constant the claim trivially holds.
We prove the case where $t$ starts with a constructor other than $\ENCNA$, $\AENCNA$, $\SIGNNA$ generically for
$t = f(t_1,t_2)$. We then have $\eval{t_1\sigma}\neq \bot$ and $\eval{t_2\sigma}\neq \bot$, and
$\eval{t\sigma}=f(\eval{t_1\sigma},\eval{t_2\sigma})$, which, by the induction hypothesis, is equal to
$f(\eval{t_1}\sigma,\eval{t_2}\sigma)$, \ie to $\eval{f(t_1,t_2)}\sigma$.
The case where $f$ is $\ENCNA$, $\AENCNA$ or $\SIGNNA$ is similar, but we in addition know that $\eval{t_2}$ is a key.

Finally if $t$ starts with a destructor, $t = d(t_1,t_2)$, we know that $\eval{t_1}$ starts with the corresponding constructor $f$: $\eval{t_1} = f(t_3,t_4)$. In the case of encryptions and signatures we know in addition that $\eval{t_4}$
and $\eval{t_2}$ are the same key (resp. public key/verification key).
We then have $\eval{t} = \eval{t_3}$, and $\eval{t\sigma} = \eval{t_3\sigma}$ (or $t_4$ in the case of the second projection $\SND$). Hence by the induction hypothesis, $\eval{t\sigma} = \eval{t_3}\sigma = \eval{t}\sigma$ and the claim holds.

\end{proof}

\begin{lemma}
\label{lem-proof:red-inst}
For all $\sigma$, $\sigma'$, 
for all recipe $R$ such that $\var{R}\subseteq \dom{\phiEEO\cup\phiL{\cc}}$, 
if $\eval{R(\phiEEO\cup\phiL{\inst{\cc}{\sigma}{\sigma'}})} \neq \bot$ then
$\eval{R(\phiEEO\cup\phiL{\cc})} \neq \bot$;
and similarly
if $\eval{R(\phiEEO\cup\phiR{\inst{\cc}{\sigma}{\sigma'}})} \neq \bot$ then
$\eval{R(\phiEEO\cup\phiR{\cc})} \neq \bot$.
\end{lemma}
\begin{proof}
We prove the property for $\phiL{\cc}$, as the proof for $\phiR{\cc}$ is similar.

We prove this by induction on $R$.

\begin{itemize}
\item \case{if $R = n \in \N$ or $R = x\in\AX$ or $R = a\in\CST\cup\FN$}: then
$\eval{R(\phiEEO\cup\phiL{\cc})} \neq \bot$ is trivial.

\item \case{If the head symbol of $R$ is $\PAIRNA$, 
$\HASHNA$}:
all these cases are similar, we detail here the pair case.
We have $R = \PAIR{R_1}{R_2}$. By assumption, we have
$\eval{R_1(\phiEEO\cup\phiL{\inst{\cc}{\sigma}{\sigma'}})} \neq \bot$ and
$\eval{R_2(\phiEEO\cup\phiL{\inst{\cc}{\sigma}{\sigma'}})} \neq \bot$.
By the induction hypothesis, we thus have
$\eval{R_1(\phiEEO\cup\phiL{\cc})} \neq \bot$ and $\eval{R_2(\phiEEO\cup\phiL{\cc})} \neq \bot$, and therefore
$\eval{R(\phiEEO\cup\phiL{\cc})} \neq \bot$.

\item \case{if the head symbol of $R$ is $\ENCNA$, $\AENCNA$, $\SIGNNA$, $\PUBKNA$, $\VKNA$}:
all these cases are similar, we detail the proof for the asymmetric encryption case.
We have $R = \AENC{S}{K}$ for some recipes $S$, $K$.
Since $\eval{R(\phiEEO\cup\phiR{\inst{\cc}{\sigma}{\sigma'}})} \neq \bot$,
we have $\eval{S(\phiEEO\cup\phiR{\inst{\cc}{\sigma}{\sigma'}})} \neq \bot$,
and $\eval{K(\phiEEO\cup\phiR{\inst{\cc}{\sigma}{\sigma'}})} = \PUBK{k}$ for some $k\in\K$. 

By the induction hypothesis, we thus have
$\eval{S(\phiEEO\cup\phiL{\cc})} \neq \bot$ and $\eval{K(\phiEEO\cup\phiL{\cc})} \neq \bot$.
Hence, by Lemma~\ref{lem-proof:cc-recipe-dest}, there exists a recipe $K'$ without destructors, such that
$\eval{K(\phiEEO\cup\phiL{\cc})} = K'(\phiEEO\cup\phiL{\cc})$.
Hence, using Lemma~\ref{lem-proof:subst-red}, we have
\[\PUBK{k} = \eval{K(\phiEEO\cup\phiR{\inst{\cc}{\sigma}{\sigma'}})} =
\eval{K((\phiEEO\cup\phiR{\cc})\sigma)} = (\eval{K(\phiEEO\cup\phiL{\cc})})\sigma
= K'(\phiEEO\cup\phiL{\cc})\sigma\]

Since $K'(\phiEEO\cup\phiL{\cc})\sigma = \PUBK{k}$, there exists a variable $x$ such that
\begin{itemize}
\item either $K' = x$, and then since $\phiL{\cc}$ does not contain variables by $\stepIII_{\Gammao}(\cc)$, we have
$(\phiEEO\cup\phiL{\cc})(x) = \PUBK{k}$.
If $x\in\dom{\phiL{\cc}}$, since by Lemma~\ref{lem-proof:cc-low}, there exists $c_x$ such that
$\teqTc{\Gammao}{\Delta}{\E}{\E'}{\phiL{\cc}(x)}{\phiR{\cc}(x)}{\L}{c_x}$, we know by Lemma~\ref{lem-proof:type-key-nonce}
that $\phiR{\cc}(x) = \PUBK{k}$.

\item or $K' = \PUBK{x}$, and then 
since $\phiL{\cc}$ does not contain variables by $\stepIII_{\Gammao}(\cc)$, we have
$(\phiEEO\cup\phiL{\cc})(x) = k$.
If $x\in\dom{\phiL{\cc}}$, since by Lemma~\ref{lem-proof:cc-low}, there exists $c_x$ such that
$\teqTc{\Gammao}{\Delta}{\E}{\E'}{\phiL{\cc}(x)}{\phiR{\cc}(x)}{\L}{c_x}$, we know by Lemma~\ref{lem-proof:type-key-nonce}
that $\phiR{\cc}(x) = k$.
\end{itemize}

In any case $K'(\phiEEO\cup\phiL{\cc})\sigma = K'(\phiEEO\cup\phiL{\cc})$, and therefore
\[\eval{K(\phiEEO\cup\phiR{\cc})} = K'(\phiEEO\cup\phiL{\cc}) = \PUBK{k}\]
which, together with the fact that $\eval{S(\phiEEO\cup\phiL{\cc})} \neq \bot$, implies that
$\eval{R(\phiEEO\cup\phiL{\cc})} \neq \bot$.

\item \case{if the head symbol of $R$ is $\DECNA$, $\ADECNA$, $\CHECKNA$}:
all these cases are similar, we detail the proof for the signature verification case.
We have $R = \CHECK{S}{K}$ for some recipes $S$, $K$.
Since $\eval{R(\phiEEO\cup\phiR{\inst{\cc}{\sigma}{\sigma'}})} \neq \bot$,
we know that $\eval{S(\phiEEO\cup\phiR{\inst{\cc}{\sigma}{\sigma'}})} = \SIGN{M}{k}$ for some message $M$ and some $k\in\K$,
and $\eval{K(\phiEEO\cup\phiR{\inst{\cc}{\sigma}{\sigma'}})} = \VK{k}$.

By the induction hypothesis, we thus have
$\eval{S(\phiEEO\cup\phiL{\cc})} \neq \bot$ and $\eval{K(\phiEEO\cup\phiL{\cc})} \neq \bot$.
Hence, by Lemma~\ref{lem-proof:cc-recipe-dest}, there exist recipes $S'$, $K'$ without destructors, such that
$\eval{S(\phiEEO\cup\phiL{\cc})} = S'(\phiEEO\cup\phiL{\cc})$ and
$\eval{K(\phiEEO\cup\phiL{\cc})} = K'(\phiEEO\cup\phiL{\cc})$.

Hence, using Lemma~\ref{lem-proof:subst-red}, we have
\[\SIGN{M}{k} = \eval{S(\phiEEO\cup\phiR{\inst{\cc}{\sigma}{\sigma'}})} =
\eval{S((\phiEEO\cup\phiR{\cc})\sigma)} = (\eval{S(\phiEEO\cup\phiL{\cc})})\sigma
= S'(\phiEEO\cup\phiL{\cc})\sigma\]

Since $S'(\phiEEO\cup\phiL{\cc})\sigma = \SIGN{M}{k}$, and considering that $\stepIII_{\Gammao}(\cc)$ holds
there exists a variable $x$ such that
either $S' = x$ and $\phiL{\cc}(x) = \SIGN{M'}{k}$ for some message $M'$;
or $S' = \SIGN{S''}{x}$ (for some recipe $S''$) and $\phiEEO(x) = k$ (similarly to the previous case).

In any case $S'(\phiEEO\cup\phiL{\cc})$ is a signature by $k$.
That is to say $\eval{S(\phiEEO\cup\phiL{\cc})}$ is a signature by $k$.

Similarly to the asymmetric encryption case, we can also show that $\eval{K(\phiEEO\cup\phiL{\cc})} = \VK{k}$.

Therefore, $\eval{R(\phiEEO\cup\phiL{\cc})} = \eval{(\CHECK{S}{K})(\phiEEO\cup\phiL{\cc})} \neq \bot$.

\item \case{if the head symbol of $R$ is $\FSTNA$, $\SNDNA$:}
all these cases are similar, we detail the proof for the $\FSTNA$ case.
We have $R = \FST{S}$ for some recipe $S$.
Since $\eval{R(\phiEEO\cup\phiR{\inst{\cc}{\sigma}{\sigma'}})} \neq \bot$,
we know that $\eval{S(\phiEEO\cup\phiR{\inst{\cc}{\sigma}{\sigma'}})} = \PAIR{M}{N}$ for some messages $M$, $N$.

By the induction hypothesis, we thus have
$\eval{S(\phiEEO\cup\phiL{\cc})} \neq \bot$.
Hence, by Lemma~\ref{lem-proof:cc-recipe-dest}, there exists a recipe $S'$ without destructors, such that
$\eval{S(\phiEEO\cup\phiL{\cc})} = S'(\phiEEO\cup\phiL{\cc})$.

Hence, using Lemma~\ref{lem-proof:subst-red}, we have
\[\PAIR{M}{N} = \eval{S(\phiEEO\cup\phiR{\inst{\cc}{\sigma}{\sigma'}})} =
\eval{S((\phiEEO\cup\phiR{\cc})\sigma)} = (\eval{S(\phiEEO\cup\phiL{\cc})})\sigma
= S'(\phiEEO\cup\phiL{\cc})\sigma\]

Since $S'(\phiEEO\cup\phiL{\cc})\sigma = \PAIR{M}{N}$, and cconsidering that $\stepIII_{\Gammao}(\cc)$ holds
there exist some recipes $S''$ and $S'''$ such that $S' = \PAIR{S''}{S'''}$.
That is to say $\eval{S(\phiEEO\cup\phiL{\cc})}$ is a pair.

Therefore, $\eval{R(\phiEEO\cup\phiL{\cc})} = \eval{(\FST{S})(\phiEEO\cup\phiL{\cc})} \neq \bot$.
\end{itemize}
\end{proof}

\begin{lemma}
\label{lem-proof:cc-inst-recipe-dest}
For all $\sigma$, $\sigma'$, 
for all recipe $R$ such that $\var{R}\subseteq \dom{\phiEEO\cup\phiL{\cc}}$, 
$\eval{R(\phiEEO\cup\phiL{\inst{\cc}{\sigma}{\sigma'}})} \neq \bot$ and
$\eval{R(\phiEEO\cup\phiR{\inst{\cc}{\sigma}{\sigma'}})} \neq \bot$,
there exists a recipe $R'$ without destructors, \ie in which $\DECNA$, $\ADECNA$, $\CHECKNA$, $\FSTNA$, $\SNDNA$,
do not appear,
such that
\begin{itemize}
\item $\var{R'}\subseteq \dom{\phiEEO\cup\phiL{\cc}}$,
\item $\eval{R(\phiEEO\cup\phiL{\inst{\cc}{\sigma}{\sigma'}})} = R'(\phiEEO\cup\phiL{\inst{\cc}{\sigma}{\sigma'}})$,
\item $\eval{R(\phiEEO\cup\phiR{\inst{\cc}{\sigma}{\sigma'}})} = R'(\phiEEO\cup\phiR{\inst{\cc}{\sigma}{\sigma'}})$.
\end{itemize}
\end{lemma}
\begin{proof}

By Lemma~\ref{lem-proof:red-inst},
$\eval{R(\phiEEO\cup\phiL{\cc})} \neq \bot$ and
$\eval{R(\phiEEO\cup\phiR{\cc})} \neq \bot$.

Hence we may apply Lemma~\ref{lem-proof:cc-recipe-dest}, and there exists a recipe $R'$ without destructors, such that
\begin{itemize}
\item $\var{R'}\subseteq \dom{\phiEEO\cup\phiL{\cc}}$,
\item $\eval{R(\phiEEO\cup\phiL{\cc})} = R'(\phiEEO\cup\phiL{\cc})$,
\item $\eval{R(\phiEEO\cup\phiR{\cc})} = R'(\phiEEO\cup\phiR{\cc})$.
\end{itemize}

By Lemma~\ref{lem-proof:subst-red}, we have 
\[\eval{(R(\phiEEO\cup\phiL{\inst{\cc}{\sigma}{\sigma'}}))} = \eval{(R(\phiEEO\cup\phiL{\cc})\sigma)}
= (\eval{R(\phiEEO\cup\phiL{\cc})})\sigma\]
Hence
\[\eval{(R(\phiEEO\cup\phiL{\inst{\cc}{\sigma}{\sigma'}}))} = R'(\phiEEO\cup\phiL{\cc})\sigma
= R'(\phiEEO\cup\phiL{\inst{\cc}{\sigma}{\sigma'}}).\]

Similarly we can show that 
\[\eval{(R(\phiEEO\cup\phiR{\inst{\cc}{\sigma}{\sigma'}}))} = R'(\phiEEO\cup\phiR{\inst{\cc}{\sigma}{\sigma'}}).\]

This proves the claim.
\end{proof}

\begin{lemma}
\label{lem-proof:recipe-variable}
For all $\sigma$, $\sigma'$, for all recipe $R$ such that $\var{R}\subseteq \dom{\phiEEO\cup\phiL{\cc}}$, 
for all $x\in \dom{\phiEEO\cup\phiL{\cc}}$,
if $R(\phiEEO\cup\phiL{\inst{\cc}{\sigma}{\sigma'}} = (\phiEEO\cup\phiL{\inst{\cc}{\sigma}{\sigma'}})(x)$ then
$R$ is a variable $y\in\dom{\phiEEO\cup\phiL{\cc}}$, or $R\in\CST\cup\FN$.

Similarly,
if $R(\phiEEO\cup\phiR{\inst{\cc}{\sigma}{\sigma'}} = (\phiEEO\cup\phiR{\inst{\cc}{\sigma}{\sigma'}})(x)$ then
$R$ is a variable $y\in\dom{\phiEEO\cup\phiR{\cc}}$ or $R\in\CST\cup\FN$.
\end{lemma}
\begin{proof}
We only detail the proof for $\phiL{\inst{\cc}{\sigma}{\sigma'}}$, as the proof for $\phiR{\cc}$ is similar.

We distinguish several cases for $R$.

\begin{itemize}

\item \case{If $R = a\in\CST\cup\FN$}: the claim clearly holds.

\item \case{If $R = x\in\AX$} then the claim trivially holds.

\item \case{If $R = \ENC{S}{K}$ or $\SIGN{S}{K}$ for some recipes $S$, $K$}:
these two cases are similar, we only detail the encryption case.
$(\phiEEO\cup\phiL{\inst{\cc}{\sigma}{\sigma'}})(x)$ is then an encrypted message, which, because of the form of $\cc$,
implies that $K(\phiEEO\cup\phiL{\inst{\cc}{\sigma}{\sigma'}}) = k$ for some $k\in\K$ such that $\Gamma(k) = \skey{\S}{T}$ for some $T$.
This is only possible if there exists a variable $z$ such that $K = z$
and $(\phiEEO\cup\phiL{\inst{\cc}{\sigma}{\sigma'}})(z) = k$, which, by $\stepIII_{\Gammao}(\cc)$ and the definition of $\phiEEO$, implies that $\Gamma(k) = \skey{\L}{T'}$ for some $T'$, which is contradictory.

\item \case{If $R = \AENC{S}{K}$ or $\HASH{S}$ for some recipes $S$, $K$}:
these two cases are similar, we only detail the encryption case.
$(\phiEEO\cup\phiL{\inst{\cc}{\sigma}{\sigma'}})(x)$ is then an asymmetrically encrypted message, which, because of the form
of $\cc$, implies that
$S(\phiEEO\cup\phiL{\inst{\cc}{\sigma}{\sigma'}})$ contains directly under pairs a nonce $n$ such that $\Gammao(n) = \noncetypelab{\S}{a}{n}$, or a key $k\in\K$ such that $\Gammao(k) = \skey{\S}{T}$ for some $T$.
This is only possible if there exists a recipe $S'$ such that $S(\phiEEO\cup\phiL{\inst{\cc}{\sigma}{\sigma'}}) = n$ (resp. $k$).

Since $R$ can only contain names from $\FN$, this implies that there exists a variable $z$ such that
$(\phiEEO\cup\phiL{\inst{\cc}{\sigma}{\sigma'}})(z) = n$ (resp. $k$), which, by $\stepIII_{\Gammao}(\cc)$, and the definition of $\phiEEO$, implies that $\Gammao(n) = \noncetypelab{\L}{a}{n}$ (resp. $\Gammao(k) = \skey{\L}{T'}$ for some $T'$), which is contradictory.

\item Finally, the head symbol of $R$ cannot be $\PAIRNA$, 
$\DECNA$, $\ADECNA$, $\CHECKNA$, $\FSTNA$, $\SNDNA$, 
because of the form of $\cc$ ($\stepIII_{\Gammao}(\cc)$).
\end{itemize}

\end{proof}

\bigskip\bigskip
We now assume that $\stepIV_{\Gammao}(\cc)$ holds.
Note that since $\Gamma$, and hence $\Gammao$, do not contain refinements or nonces with infinite nonce types, the
$\stepIV_{\Gammao}(\cc)$ condition is simpler than in the general case.
Indeed the condition on the most general unifier $\mu$ is always trivially satisfied,
and the substitution $\theta$ is the identity.

\begin{lemma}
\label{lem-proof:inst-stateq}
For all $\sigma$, $\sigma'$ such that $\exists \Gamma'\subseteq \Gammao. \;\exists c_\sigma.\;\wtc{\Delta}{\E}{\E'}{\sigma}{\sigma'}{\Gamma'}{c_\sigma}$,
the frames $\NEWN{\EGO}.\;\phiEEO \cup \phiL{\inst{\cc}{\sigma}{\sigma'}}$ and $\NEWN{\EGO}.\;\phiEEO \cup \phiR{\inst{\cc}{\sigma}{\sigma'}}$
are statically equivalent.
\end{lemma}
\begin{proof}
Let $R$, $S$ be two recipes such that
$\var{R}\cup\var{S} \subseteq \dom{\phiEEO\cup\phiL{\inst{\cc}{\sigma}{\sigma'}}} = \dom{\phiEEO\cup\phiR{\inst{\cc}{\sigma}{\sigma'}}}$.

Let us show that
\[\eval{(R(\phiEEO\cup\phiL{\inst{\cc}{\sigma}{\sigma'}}))} = \eval{(S(\phiEEO\cup\phiL{\inst{\cc}{\sigma}{\sigma'}}))}
\quad\Longleftrightarrow\quad
\eval{(R(\phiEEO\cup\phiR{\inst{\cc}{\sigma}{\sigma'}}))} = \eval{(S(\phiEEO\cup\phiR{\inst{\cc}{\sigma}{\sigma'}}))}\]

We only detail the proof for the $(\Rightarrow)$, as the other direction is similar.
We then assume that $\eval{(R(\phiEEO\cup\phiL{\inst{\cc}{\sigma}{\sigma'}}))} = \eval{(S(\phiEEO\cup\phiL{\inst{\cc}{\sigma}{\sigma'}}))}$.

Let us first note that
\[\eval{(R(\phiEEO\cup\phiL{\cc}))} \neq \bot \quad \Longleftrightarrow\quad \eval{(R(\phiEEO\cup\phiR{\cc}))} \neq \bot.\]
This follows from Lemmas~\ref{lem-proof:cc-low} and~\ref{lem-proof:l-type-recipe}.

Since $\wtc{\Delta}{\E}{\E'}{\sigma}{\sigma'}{\Gamma'}{c_\sigma}$, as by definition we only type messages,
$\sigma$ and $\sigma'$ only contain messages.
Hence, by Lemmas~\ref{lem-proof:red-inst} and~\ref{lem-proof:subst-red},
\[\eval{R(\phiEEO\cup\phiL{\inst{\cc}{\sigma}{\sigma'}})} \neq \bot \quad\Longleftrightarrow\quad \eval{R(\phiEEO\cup\phiL{\cc})} \neq \bot\]
and 
\[\eval{R(\phiEEO\cup\phiR{\inst{\cc}{\sigma}{\sigma'}})} \neq \bot \quad\Longleftrightarrow\quad \eval{R(\phiEEO\cup\phiR{\cc})} \neq \bot.\]
Hence, by chaining all these equivalences, we have
\[\eval{R(\phiEEO\cup\phiL{\inst{\cc}{\sigma}{\sigma'}})} \neq \bot \quad\Longleftrightarrow\quad \eval{R(\phiEEO\cup\phiR{\inst{\cc}{\sigma}{\sigma'}})} \neq \bot.\]

Similarly, we can show that
%
\[\eval{S(\phiEEO\cup\phiL{\inst{\cc}{\sigma}{\sigma'}})} \neq \bot \quad\Longleftrightarrow\quad \eval{S(\phiEEO\cup\phiR{\inst{\cc}{\sigma}{\sigma'}})} \neq \bot.\]

Therefore,
if $\eval{R(\phiEEO\cup\phiL{\inst{\cc}{\sigma}{\sigma'}})} = \bot$, \ie $\eval{S(\phiEEO\cup\phiL{\inst{\cc}{\sigma}{\sigma'}})} = \bot$, then
$\eval{R(\phiEEO\cup\phiR{\inst{\cc}{\sigma}{\sigma'}})} = \eval{S(\phiEEO\cup\phiR{\inst{\cc}{\sigma}{\sigma'}})} = \bot$,
and the claim holds.

Let us now assume that $\eval{R(\phiEEO\cup\phiL{\inst{\cc}{\sigma}{\sigma'}})} \neq \bot$,
then $\eval{R(\phiEEO\cup\phiR{\inst{\cc}{\sigma}{\sigma'}})} \neq \bot$ and $\eval{S(\phiEEO\cup\phiR{\inst{\cc}{\sigma}{\sigma'}})} \neq \bot$.

By Lemma~\ref{lem-proof:cc-inst-recipe-dest}, there exist $R'$, $S'$ without destructors such that
\begin{itemize}
\item $\var{R'}\subseteq \dom{\phiEEO\cup\phiL{\cc}}$ and $\var{S'}\subseteq \dom{\phiEEO\cup\phiL{\cc}}$,
\item $\eval{R(\phiEEO\cup\phiL{\inst{\cc}{\sigma}{\sigma'}})} = R'(\phiEEO\cup\phiL{\inst{\cc}{\sigma}{\sigma'}})$,
\item $\eval{R(\phiEEO\cup\phiR{\inst{\cc}{\sigma}{\sigma'}})} = R'(\phiEEO\cup\phiR{\inst{\cc}{\sigma}{\sigma'}})$.
\item $\eval{S(\phiEEO\cup\phiL{\inst{\cc}{\sigma}{\sigma'}})} = S'(\phiEEO\cup\phiL{\inst{\cc}{\sigma}{\sigma'}})$,
\item $\eval{S(\phiEEO\cup\phiR{\inst{\cc}{\sigma}{\sigma'}})} = S'(\phiEEO\cup\phiR{\inst{\cc}{\sigma}{\sigma'}})$.
\end{itemize}

Since $\eval{(R(\phiEEO\cup\phiL{\inst{\cc}{\sigma}{\sigma'}}))} = \eval{(S(\phiEEO\cup\phiL{\inst{\cc}{\sigma}{\sigma'}}))}$,
we have $R'(\phiEEO\cup\phiL{\inst{\cc}{\sigma}{\sigma'}}) = S'(\phiEEO\cup\phiL{\inst{\cc}{\sigma}{\sigma'}})$.

We show that $R'(\phiEEO\cup\phiR{\inst{\cc}{\sigma}{\sigma'}}) = S'(\phiEEO\cup\phiR{\inst{\cc}{\sigma}{\sigma'}})$.
by induction on the recipes $R'$, $S'$.
Since
$R'(\phiEEO\cup\phiL{\inst{\cc}{\sigma}{\sigma'}}) = S'(\phiEEO\cup\phiL{\inst{\cc}{\sigma}{\sigma'}})$,
we can distinguish four cases for $R'$ and $S'$.

\begin{itemize}
\item \case{If they have the same head symbol},
either this symbol is a nonce or constant and the claim is trivial, or it is a variable, and we handle this case later,
or it is a destructor or constructor.
We write the proof for this last case generically for $R' = f(R'')$ and $S' = f(S'')$.
We have necessarily
$R''(\phiEEO\cup\phiL{\inst{\cc}{\sigma}{\sigma'}}) = S''(\phiEEO\cup\phiL{\inst{\cc}{\sigma}{\sigma'}})$.
It then follows by applying the induction hypothesis to $R''$ and $S''$ that
$R''(\phiEEO\cup\phiR{\inst{\cc}{\sigma}{\sigma'}}) = S''(\phiEEO\cup\phiR{\inst{\cc}{\sigma}{\sigma'}})$.
The claim follows by applying $f$ on both sides of this equalities.

\item \case{If $R'$ is a variable and not $S'$:} then by Lemma~\ref{lem-proof:recipe-variable}, $S'\in\CST\cup\FN$.
Let us denote $R' = x$.
By Lemma~\ref{lem-proof:cc-low}, there exists $c_x$ such that
$\teqTc{\Gammao}{\Delta}{\E}{\E'}{(\phiEE\cup\phiL{\cc})(x)}{(\phiEE\cup \phiR{\cc})(x)}{\L}{c_x}$.
Since $(\phiEE\cup\phiL{\cc})(x) = R'(\phiEE\cup\phiL{\cc}) = S'(\phiEE\cup\phiL{\cc}) \in \CST\cup\FN$,
by Lemma~\ref{lem-proof:type-key-nonce}, we have $(\phiEE\cup\phiR{\cc})(x) = S'(\phiEE\cup\phiL{\cc})$,
\ie $R'(\phiEE\cup\phiR{\cc}) = S'(\phiEE\cup\phiR{\cc})$.

\item \case{If $S'$ is a variable and not $R'$:} this case is similar to the previous one.

\item \case{If $R'$, $S'$ are two variables $x$ and $y$,} we have
$(\phiEEO\cup\phiL{\inst{\cc}{\sigma}{\sigma'}})(x) = (\phiEEO\cup\phiL{\inst{\cc}{\sigma}{\sigma'}})(y)$.
We can then prove $(\phiEEO\cup\phiR{\inst{\cc}{\sigma}{\sigma'}})(x) = (\phiEEO\cup\phiR{\inst{\cc}{\sigma}{\sigma'}})(y)$.
Indeed:
\begin{itemize}
\item if $x,y\in\dom{\phiEEO}$, this follows from the definition of $\phiEEO$.
\item if $x\in\dom{\phiEEO}$ and $y\in\dom{\phiL{\inst{\cc}{\sigma}{\sigma'}}}$:
then by definition of $\phiEEO$, 
$R'(\phiEEO\cup\phiL{\inst{\cc}{\sigma}{\sigma'}}) = \phiEEO(x)$ is a nonce, key, public key, or verification key.
Hence $\phiL{\inst{\cc}{\sigma}{\sigma'}}(y)$ is also a nonce, key, public key or verification key.
By $\stepIII_{\Gammao}(\cc)$, $\phiL{\inst{\cc}{\sigma}{\sigma'}}(y) = \phiL{\cc}(y)$.
This implies that $\phiL{\cc}(y)$ is also a nonce, key, public key or verification key.
By Lemma~\ref{lem-proof:cc-low}, there exists $c_y$ such that
$\teqTc{\Gammao}{\Delta}{\E}{\E'}{\phiL{\cc}(y)}{\phiR{\cc}(y)}{\L}{c_y}$,
and hence by Lemma~\ref{lem-proof:type-key-nonce}, $\phiR{\cc}(y) = \phiL{\cc}(y)$.
That is to say
\[R'(\phiEEO\cup\phiR{\inst{\cc}{\sigma}{\sigma'}} = (\phiEEO\cup\phiR{\inst{\cc}{\sigma}{\sigma'}}(x) =
\phiEEO(x) = (\phiEEO\cup\phiL{\inst{\cc}{\sigma}{\sigma'}})(y) = \phiL{\cc}(y) = \phiR{\cc}(y) = 
S'(\phiEEO\cup\phiR{\inst{\cc}{\sigma}{\sigma'}}.\]

\item if $x, y \in \dom{\phiL{\inst{\cc}{\sigma}{\sigma'}}}$:
then there exist $M\eqC M'\in \cc$, $N\eqC N'\in\cc$ such that
$\phiL{\inst{\cc}{\sigma}{\sigma'}}(x)=M\sigma$,
$\phiR{\inst{\cc}{\sigma}{\sigma'}}(x)=M'\sigma'$,
$\phiL{\inst{\cc}{\sigma}{\sigma'}}(y)=N\sigma$,
$\phiR{\inst{\cc}{\sigma}{\sigma'}}(y)=N'\sigma'$.
Since $M\sigma = N\sigma$, $M$, $N$ are unifiable,
let $\mu$ be their most general unifier. There exists $\theta$ such that $\sigma = \mu\theta$.

Let then $\alpha$ be the restriction of $\mu$ to
$\{x\in\var{M}\cup\var{N}\;|\; \Gammao(x)=\L \;\wedge\; \mu(x)\in\N \text{ is a nonce}\}$.

By $\stepIV_{\Gammao}(\cc)$, we have $M'\alpha = N'\alpha$.

\medskip

Let $x\in\dom{\alpha}$. Then $\mu(x) = n$ for some $n\in\N$.
Thus $\sigma(x)=\mu(x)\theta = n$.

Since, by assumption, $\exists \Gamma'\subseteq \Gammao. \;\exists c_\sigma.\;\wtc{\Delta}{\E}{\E'}{\sigma}{\sigma'}{\Gamma'}{c_\sigma}$, there exists some $c_x$ such that
$\teqTc{\emptyset}{\Delta}{\E}{\E'}{n}{\sigma'(x)}{\L}{c_x}$.

Hence by Lemma~\ref{lem-proof:type-key-nonce}, $\sigma'(x)= n = \mu(x) = \alpha(x)$.

Since this holds for any $x\in\dom{\alpha}$, we have $\alpha\sigma'=\sigma'$.

\medskip
Therefore, since $M'\alpha = N'\alpha$, we have $M'\alpha\sigma' = N'\alpha\sigma'$, \ie $M'\sigma' = N'\sigma'$.
This proves the claim.
\end{itemize}
\end{itemize}

Finally, since
$R'(\phiEEO\cup\phiL{\inst{\cc}{\sigma}{\sigma'}}) = S'(\phiEEO\cup\phiL{\inst{\cc}{\sigma}{\sigma'}})$,
we have
$\eval{R(\phiEEO\cup\phiR{\inst{\cc}{\sigma}{\sigma'}})} = \eval{S(\phiEEO\cup\phiR{\inst{\cc}{\sigma}{\sigma'}})}$.

This proves the property.
\end{proof}

\begin{lemma}
$c$ is consistent in $\Gamma$. 
\end{lemma}
\begin{proof}
By Lemmas~\ref{lem-proof:co-c-const} and~\ref{lem-proof:cc-co-const}, it suffices to show that 
$\cc$ is consistent in $\Gammao$. 

Let $c'\subseteq \cc$, $\Gamma'\subseteq\Gammao$ be such that $\var{c'}\subseteq\dom{\Gamma'}$.
Let $\sigma$, $\sigma'$ be such that $\wtc{\Delta}{\E}{\E'}{\sigma}{\sigma'}{\Gamma'}{c_\sigma}$ for some $c_\sigma$.


By Lemma~\ref{lem-proof:inst-stateq},
the frames $\NEWN{\EGO}.\;\phiEEO \cup \phiL{\inst{\cc}{\sigma}{\sigma'}}$ and $\NEWN{\EGO}.\;\phiEEO \cup \phiR{\inst{\cc}{\sigma}{\sigma'}}$
are statically equivalent.

Since the frames $\NEWN{\EGO}.\;\phiEEO \cup \phiL{\inst{c'}{\sigma}{\sigma'}}$ and $\NEWN{\EGO}.\;\phiEEO \cup \phiR{\inst{c'}{\sigma}{\sigma'}}$
are subsets of these frames, they also are statically equivalent.

Therefore $\cc$ is consistent in $\Gammao$. 
\end{proof}

This next theorem corresponds to Theorem~\ref{thm:soundness-proc-consistency}.

\begin{theorem}[Soundness of the procedure]
\label{thm-proof:soundness-proc-consistency}
Let $C$ be a constraint set without infinite nonce types, \ie
\[\forall (c,\Gamma)\in C.\;\forall l, l', m, n. 
 \Gamma(x) \neq \LRTn{l}{\infty}{m}{l'}{\infty}{n}.\]
If $\checkconst(C)$ succeeds,
then $C$ is consistent. 
\end{theorem}
\begin{proof}
The previous lemmas directly imply that for all $(c,\Gamma)\in C$, $c$ is consistent in $\Gamma$.
This proves the theorem.
\end{proof}

%% file: proofs/consistency-repl.tex
\subsection{Consistency for replicated processes}

In this subsection, we prove the results regarding the procedure when checking consistency in the replicated case.

In this subsection, we only consider constraints obtained by typing processes (with the same key types).
Notably, by the well-formedness assumptions on the processes, this means that a nonce $n$ is always associated with the same nonce type.

This theorem corresponds to Theorem~\ref{thm:consistency-two-to-n}.

\begin{theorem}
\label{thm-proof:consistency-two-to-n}
Let 
$C$ and $C'$ be two constraint sets such that
\[\forall (c,\Gamma)\in C.\; \forall (c',\Gamma')\in C'.\; \dom{\onlyvar{\Gamma}}\cap\dom{\onlyvar{\Gamma'}} = \emptyset.\]

For all $n\in\mathbb{N}$, 
if $\checkconststar(\instCst{C}{1}{n} \UnionCart \instCst{C}{2}{n} \UnionCart \instCst{C'}{1}{n}) = \mathtt{true}$,
then $\checkconststar(\UnionCart_{1\leq i \leq n}\instCst{C}{i}{n})\UnionCart\instCst{C'}{1}{n}) = \mathtt{true}$.
\end{theorem}
\begin{proof}
\newcommand{\Gammain}{\instG{\Gamma}{i}{n}}
\newcommand{\Cpi}{\instCst{C'}{1}{n}}
\newcommand{\Cudp}{\Ci\UnionCart\Cii\UnionCart\Cpi}
\newcommand{\UCnp}{(\UCn)\UnionCart\Cpi}
\newcommand{\bc}{\overline{c}}
\newcommand{\bci}[1]{\overline{c_{#1}}}
\newcommand{\cci}[1]{\widetilde{c_{#1}}}
\newcommand{\cij}[3]{c^{#2,#3}_#1}
\newcommand{\bcij}[3]{\overline{c}^{#2,#3}_#1}
\newcommand{\ccij}[3]{\widetilde{c}^{#2,#3}_#1}
\newcommand{\Gammaij}[3]{\Gamma^{#2,#3}_#1}
\newcommand{\Gammaoij}[3]{\Gammao^{#2,#3}_#1}
\newcommand{\bcp}[1]{\overline{c'}^{#1}}
\newcommand{\ccp}[1]{\widetilde{c'}^{#1}}
\newcommand{\Gammap}[1]{{\Gamma'}^{#1}}
\newcommand{\Gammaop}[1]{\overline{\Gamma'}^{#1}}

Let $n\in\mathbb{N}$.
Let $C$ be such that $\forall (c,\Gamma)\in C.\; \forall (c',\Gamma')\in C'.\; \dom{\Gamma}\cap\dom{\Gamma'} = \emptyset$.

Let $(c,\Gamma)\in\UCnp$.
By definition of $\UnionCart$, there exists $(c',\Gamma')\in\Cpi$,
and for all $i\in\unn$, there exists $(c_i, \Gamma_i)\in\Cin$, such that
\begin{itemize}
\item $c = (\cup_{1\leq i \leq n} c_i)\cup c'$;
\item $\Gamma = (\uplus_{1\leq i \leq n} \Gamma_i)\uplus \Gamma'$.
\end{itemize}

\medskip

Let $i\in\unn$.
Since $(c_i, \Gamma_i)\in\Cin$, by definition of $\Cin$
there exists $(c'_i, \Gamma'_i)\in C$ such that
\begin{itemize}
\item $c_i = \instConst{c'_i}{i}{\Gamma_i}$;
\item $\Gamma_i \in \branch{\instG{\Gamma'_i}{i}{n}}$.
\end{itemize}
Note that this implies $\dom{\onlyvar{\Gamma_i}}$ only contains variables indexed by $i$,
and, from the assumption
that $\var{c'_i}\subseteq\dom{\onlyvar{\Gamma'_i}}$, that $\var{c_i}\subseteq\dom{\onlyvar{\Gamma_i}}$.


\medskip

For all $i\in\unn$, let $\delta^1_i$ denote the function on terms
which consists in exchanging all occurrences of the indices $i$ and 1,
\ie replacing any occurrence of $m_i$ (for all nonce $m$ with an infinite nonce type) with $m_1$,
any occurrence of $m_1$ with $m_i$,
any occurrence of $x_i$ with $x_1$ (for all variable $x$),
and any occurrence of $x_1$ with $x_i$ (also for all variable $x$).

We also (abusing notations) apply this function to constraints, types, typing environments and constraint sets.
In the case of types we use it to denote the replacement of nonces appearing in the refinements.
In the case of typing  environments it denotes the replacement of nonces appearing in the types, and of nonces and variables in the domain of the environment, \ie $(\delta^1_i(\Gamma))(x_1) = \delta^1_i (\Gamma(x_i))$.
In the case of constraint sets it denotes the application of the function to each constraint and environment in the constraint set.


Similarly, we denote $\delta^2_i$ the function exchanging indices $i$ and 2.

For all $h\in\unn$ and all $i\neq j\in\unn$, such that $i\neq 2$ and $j\neq 1$,
let \[(\cij{h}{i}{j}, \Gammaij{h}{i}{j}) = (c_h, \Gamma_h)\delta^1_i\delta^2_j.\]

%

Similarly, for all $h\in\unn$ and all $i\in\unn$,
let
\[(\cij{h}{i}{i}, \Gammaij{h}{i}{i}) = (c_h, \Gamma_h)\delta^1_i.\]

Finally, for all $i\in\unn$, let $\Gammap{i}$ be the typing environment such that
$\dom{\Gammap{i}}=\dom{\Gamma'}$ and $\forall x\in\dom{\Gammap{i}}.\; \Gammap{i}(x)=\Gamma'(x)\delta^1_i$.

\medskip

Since $(c_i, \Gamma_i)\in\instCst{C}{i}{n}$,
we can show that $(\cij{i}{i}{j}, \Gammaij{i}{i}{j})\in\instCst{C}{1}{n}$.
Indeed, recall that 
there exists $(c'_i, \Gamma'_i)\in C$ such that
$c_i = \instConst{c'_i}{i}{\Gamma_i}$ and $\Gamma_i \in \branch{\instG{\Gamma'_i}{i}{n}}$.
$c_i$ only contains variables and names indexed by $i$, hence it is clear that 
$\cij{i}{i}{j} = \instConst{c'_i}{1}{\Gamma_i}$.
Moreover, since $\Gamma_i \in \branch{\instG{\Gamma'_i}{i}{n}}$,
it is clear that
$\Gammaij{i}{i}{j}\in \branch{\instG{\Gamma'_i}{i}{n}\delta^1_i\delta^2_j}$.
By definition, indexed nonces or variables appear in $\instG{\Gamma'_i}{i}{n}$ only in its domain, and as parts
of union types of the form $\LRT{m_1}{p_1}\orT\ldots\LRT{m_n}{p_n}$.
This union type is left unchanged by $\delta^1_i\delta^2_j$: since $i\neq j$,
$i\neq 2$, and $j\neq 1$, $\delta^1_i\delta^2_j$ is indeed only performing a permutation of the indices.
Hence, $\instG{\Gamma'_i}{i}{n}\delta^1_i\delta^2_j = \instG{\Gamma'_i}{1}{n}$.
Thus $\Gammaij{i}{i}{j}\in \branch{\instG{\Gamma'_i}{1}{n}}$.
Therefore, $(\cij{i}{i}{j}, \Gammaij{i}{i}{j})\in\instCst{C}{1}{n}$.


Note that $\dom{\Gammaij{i}{i}{j}}$ only contains variables indexed by $1$;
and that, since $\var{c_i}\subseteq\dom{\Gamma_i}$, we have $\var{\cij{i}{i}{j}}\subseteq\dom{\Gammaij{i}{i}{j}}$.

Similarly, if $j\neq i$, $i \neq 2$ and $j\neq 1$, $(\cij{j}{i}{j}, \Gammaij{j}{i}{j})\in\instCst{C}{2}{n}$.
Note that $\dom{\Gammaij{j}{i}{j}}$ only contains variables indexed by 2;
and that $\var{\cij{j}{i}{j}}\subseteq\dom{\Gammaij{j}{i}{j}}$.

Similarly, we also have $(c', \Gammap{i})\in \instCst{C'}{1}{n}$.

\medskip

Note that for all $(c'',\Gamma'')\in\instCst{C}{1}{n}$, and all $(c''',\Gamma''')\in\instCst{C'}{1}{n}$,
and since by assumption:
\[\forall (c,\Gamma)\in C.\; \forall (c',\Gamma')\in C'.\; \dom{\Gamma}\cap\dom{\Gamma'} = \emptyset,\]
we know that $\Gamma''$ and $\Gamma'''$ are compatible.
In particular this applies to all the $\Gammaij{i}{i}{j}$ and $\Gamma'$ (as well as $\Gammaij{i}{i}{j}$ and $\Gammap{i}$).

Moreover, for all $(c'',\Gamma'')\in\instCst{C}{2}{n}$, and all $(c''',\Gamma''')\in\instCst{C'}{1}{n}$,
since $\dom{\Gamma'''}\subseteq\{x_1\;|\; x\in\X\}$, and $\dom{\Gamma''}\subseteq\{x_2\;|\; x\in\X\}$,
$\Gamma''$ and $\Gamma'''$ are also compatible.
This in particular applies to $\Gammaij{j}{i}{j}$ for all $i\neq j\in\unn$ and $\Gamma'$
(as well as $\Gammaij{j}{i}{j}$ and $\Gammap{i}$).

\medskip

If $C$ is empty, then so is $\UCn$ and the claim clearly holds.
Let us now assume that $C$ is not empty. Hence for all $i\in\unn$, $\Cin$ is not empty.

The procedure for $c, \Gamma$ 
is as follows:
\begin{enumerate}
\item We compute $\stepI_\Gamma(c)$.
Following the notations used in the procedure, we have
\[F = \{x\in\dom{\Gamma} \;|\;\exists m,n,l,l'.\;\Gamma(x) = \LRTnewnew{\noncetypelab{l}{1}{m}}{\noncetypelab{l'}{1}{n}}\},
\]
and we write
$(\bc,\Gammao)\eqdef\stepI_\Gamma(c)$.

For all $i\in\unn$, let $(\bci{i}, \Gammao_i) \eqdef \stepI_{\Gamma_i}(c_i)$.
Let also $(\bc', \Gammao')\eqdef \stepI_{\Gamma'}(c')$.
We have $\bc=(\cup_{1\leq i\leq n} \bci{i})\cup \bc'$, and $\Gammao = (\cup_{1\leq i \leq n} \Gammao_i)\cup\Gammao'$.

\medskip

For all $h, i, j\in\unn$, such that either $i\neq j$ and $i\neq 2$ and $j\neq 1$, or $i=j$,
let also $(\bcij{h}{i}{j}, \Gammaoij{h}{i}{j})\eqdef\stepI_{\Gammaij{h}{i}{j}}(\cij{h}{i}{j})$.
Similarly, for all $i\in\unn$, let also $(\bcp{i},\Gammaop{i})\eqdef\stepI_{\Gammap{i}}(c')$.
Since, for $i\neq j$, $(\cij{h}{i}{j}, \Gammaij{h}{i}{j}) = (c_h,\Gamma_h)\delta^1_i\delta^2_j$,
it can easily be shown (by induction on the terms) that
$(\bcij{h}{i}{j}, \Gammaoij{h}{i}{j}) = (\bci{h}, \Gammao_h)\delta^1_i\delta^2_j$.
Similarly, 
$(\bcij{h}{i}{i}, \Gammaoij{h}{i}{i}) = (\bci{h}, \Gammao_h)\delta^1_i$.
Finally, we similarly also have 
$(\bcp{i},\Gammaop{i})=(\bc',\Gammao')\sigma^1_p$.

\bigskip
\item\label{item:splitting} We compute $\cc \eqdef \stepII_{\Gammao}(\bc)$.

Note that, by the assumption that $C$ and $C'$ are obtained by typing processes, all the environments
$\Gamma$, $\Gamma_i$, $\Gamma'$, $\Gammaij{h}{i}{j}$, $\Gammaij{h}{i}{i}$, $\Gamma'_i$ contain the same keys,
associated with the same labels.
The same is thus true for $\Gammao$, $\Gammao_i$, $\Gammao'$, $\Gammaoij{h}{i}{j}$, $\Gammaoij{h}{i}{i}$, $\Gammao'_i$.
Hence, $\stepII_{\Gammao}$, $\stepII_{\Gammao_i}$, $\stepII_{\Gammao'}$, $\stepII_{\Gammaoij{h}{i}{j}}$, $\stepII_{\Gammaoij{h}{i}{i}}$, $\stepII_{\Gammao'_i}$
all denote the same function.

For all $i\in\unn$, let $\cci{i}\eqdef\stepII_{\Gammao}(\bci{i})$.
Similarly, let $\cc'=\stepII_{\Gammao}(\bc')$.
It is clear that $\cc = (\cup_{1\leq i \leq n} \cci{i})\cup\cc'$.

Similarly, for all $h, i, j\in\unn$, such that either $i\neq j$ and $i\neq 2$ and $j\neq 1$, or $i=j$,
let $\ccij{h}{i}{j}\eqdef\stepII_{\Gammao}(\bcij{h}{i}{j})$.
Let also $\ccp{i}\eqdef\stepII_{\Gammao}(\bcp{i})$.

It can easily be seen 
that
for all $h\in\unn$, all $i\neq j\in\unn$ such that $i\neq 2$ and $j\neq 1$, since $\bcij{h}{i}{j} = \bci{h}\delta^1_i\delta^2_j$,
we have $\ccij{h}{i}{j} = \cci{h}\delta^1_i\delta^2_j$.
Similarly, $\ccij{h}{i}{i} = \cci{h}\delta^1_i$.
Finally, we similarly also have $\ccp{i} = \cc'\delta^1_i$.

\bigskip
\item We check that $\stepIII_{\Gammao}(\cc)$ holds, \ie that each $M\eqC N\in\cc$ has the correct form (with respect to the definition of $\stepIII$).

If $M\eqC N\in \cc$, either $M\eqC N\in \cc'$, or there exists $i\in\unn$ such that $M\eqC N\in\cci{i}$.

\begin{itemize}
\item In the first case, $M \eqC N\in\cc'$.
By assumption,
$\instCst{C}{1}{n}$ and $\instCst{C}{2}{n}$ are not empty.
Hence there exist $(c'', \Gamma'')\in\instCst{C}{1}{n}$ and $(c''', \Gamma''')\in\instCst{C}{2}{n}$.
Thus, $(c''\cup c'''\cup c', \Gamma''\cup \Gamma''' \cup \Gamma') \in \Cudp$
(as noted previously, $\Gamma''$, $\Gamma'''$, and $\Gamma'$ are compatible).
Hence, by assumption, $\checkconst(\{(c''\cup c'''\cup c', \Gamma''\cup \Gamma''' \cup \Gamma')\})$ succeeds.

If $\cc''=\stepII_{\Gammao}(fst(\stepI_{\Gamma''}(c'')))$, and $\cc'''=\stepII_{\Gammao}(fst(\stepI_{\Gamma'''}(c''')))$,
then $\cc''\cup \cc'''\cup \cc' = \stepII_{\Gammao}(fst(\stepI_{\Gamma'' \cup \Gamma''' \cup \Gamma'}(c''\cup c'''\cup c')))$.
Therefore, $\stepIII_{\Gammao}(\cc''\cup \cc'''\cup \cc') = \mathtt{true}$.

In particular, $M \eqC N\in\cc'$ has the correct form.

\item In the second case, $M\eqC N\in \cci{i}$ for some $i\in\unn$.

Let $M'= M\delta^1_i$ and $N'=N\delta^1_i$.
Since $\ccij{i}{i}{i} = \cci{i}\delta^1_i$, we have $M'\eqC N'\in\ccij{i}{i}{i}$.

By assumption,
$\instCst{C}{2}{n}$ is not empty, hence there exists $(c'', \Gamma'')\in\instCst{C}{2}{n}$.
Thus, $(\cij{i}{i}{i}\cup c''\cup c', \Gammaij{i}{i}{i} \cup \Gamma'' \cup \Gamma') \in \Cudp$
(as noted previously, $\Gammaij{i}{i}{i}$, $\Gamma''$, and $\Gamma'$ are compatible).
Hence, by assumption, $\checkconst(\{(\cij{i}{i}{i}\cup c''\cup c', \Gammaij{i}{i}{i} \cup \Gamma'' \cup \Gamma')\})$ succeeds.
If $\cci{}''=\stepII_{\Gammao}(fst(\stepI_{\Gamma''}(c'')))$,
then $\ccij{i}{i}{i}\cup \cci{}''\cup \cc'=\stepII_{\Gammao}(fst(\stepI_{\Gamma''_{i,1} \cup \Gamma'' \cup \Gamma'}(\ccij{i}{i}{i}\cup c''\cup c'))$.
Therefore, $\stepIII_{\Gammao}(\ccij{i}{i}{i}\cup \cci{}''\cup \cc')$ holds.

In particular, $M'\eqC N'\in\ccij{i}{i}{i}$, has the correct form.
It follows by examining all the cases 
and using the fact that for all $m_i$, $m_j$, if $m_i$ is associated with the type $\noncetypelab{l}{a}{m_i}$ and $m_j$ with
$\noncetypelab{l'}{a}{m_j}$ then $l=l'$,
that $M\eqC N$ also has the correct form.
\end{itemize}

Therefore, $\stepIII_{\Gammao}(\cc)$ holds.

\bigskip
\item Finally, we check that $\stepIV_{\Gammao}(\cc)$ holds.
Let $M_1\eqC N_1 \in \cc$ and $M_2\eqC N_2 \in \cc$.
Let us prove the property in the case where $M_1$ and $M_2$ are unifiable with a most general unifier $\mu$.
The case where $N_1$ and $N_2$ are unifiable is similar.

Let then $\alpha$ be the restriction of $\mu$ to
$\{x\in\var{M_1}\cup\var{M_2}\;|\; \Gammao(x)=\L \;\wedge\; \mu(x)\in\N\}$.

We have prove that $N_1\alpha = N_2\alpha$.

\medskip

Since we already have $\cc = (\cup_{1\leq i \leq n} \cci{i})\cup\cc'$, we know that:
\begin{itemize}
\item either there exists $i\in\unn$ such that $M_1\eqC N_1\in\cci{i}$;
\item or $M_1\eqC N_1\in\cc'$;
\end{itemize}
and
\begin{itemize}
\item either there exists $j\in\unn$ such that $M_2\eqC N_2\in\cci{j}$ ;
\item or $M_2\eqC N_2\in\cc'$.
\end{itemize}

Let us first prove the case where there exist $i, j\in\unn$ such that $M_1\eqC N_1\in\cci{i}$ and $M_2\eqC N_2\in\cci{j}$.
We distinguish two cases.
\begin{itemize}
\item\case{if $i\neq j$:}
The property to prove is symmetric between $M_1\eqC N_1 \in \cc$ and $M_2\eqC N_2 \in \cc$.
Hence without loss of generality we may assume that $i\neq 2$ and $j\neq 1$.
Indeed, if we assume that the property can be proved in that case,
then in the case where $i = 2$ or $j = 1$, we may exchange the two constraints.
The property holds for $M_2\eqC N_2 \in \cci{j}$ and $M_1\eqC N_1 \in \cci{i}$:
as $i\neq j$, and $i=2$ or $j=1$, we know that $i\neq 1$ and $j\neq 2$.
Then by symmetry it also holds for $M_1\eqC N_1\in\cci{i}$ and $M_2\eqC N_2\in\cci{j}$.

Let us hence assume that $i\neq 2$ and $j\neq 1$.

Let then $M'_1= M_1\delta^1_i\delta^2_j$, $N'_1=N_1\delta^1_i\delta^2_j$, $M'_2= M_2\delta^1_i\delta^2_j$,
$N'_2=N_2\delta^1_i\delta^2_j$.

Since $\ccij{i}{i}{j} = \cci{i}\delta^1_i\delta^2_j$, we have $M'_1\eqC N'_1\in\ccij{i}{i}{j}$.
Similarly, $M'_2\eqC N'_2\in\ccij{j}{i}{j}$.

Since $M_1$ and $M_2$ are unifiable,
then so are $M'_1$ and $M'_2$, with a most general unifier $\mu'$ which satisfies
$\mu(x)=t \Leftrightarrow \mu'(x\delta^1_i\delta^2_j) = t\delta^1_i\delta^2_j$.

Let then $\alpha'$ be the restriction of $\mu'$ to
$\{x\in\var{M_1'}\cup\var{M_2'}\;|\; (\Gammaoij{i}{i}{j} \cup \Gammaoij{j}{i}{j})(x)=\L \;\wedge\; \mu'(x)\in\N \text{ is a nonce}\}$.

Similarly $\alpha'$ is such that $\forall x\in\dom{\alpha'}.\forall n.\; \alpha(x) = n \Leftrightarrow \alpha'(x\delta^1_i\delta^2_j) = n\delta^1_i\delta^2_j$, \ie
$\delta^1_i\delta^2_j\alpha'\delta^1_i\delta^2_j = \alpha$.

By assumption, $\checkconst(\{(\cij{i}{i}{j}\cup \cij{j}{i}{j}\cup c', \Gammaij{i}{i}{j} \cup \Gammaij{j}{i}{j} \cup \Gamma')\})$ succeeds
since $(\cij{i}{i}{j}\cup \cij{j}{i}{j}\cup c', \Gammaij{i}{i}{j} \cup \Gammaij{j}{i}{j} \cup \Gamma')\in \Cudp$.

Thus, $\stepIV_{\Gammaoij{i}{i}{j} \cup \Gammaoij{j}{i}{j} \cup \Gammao'}(\ccij{i}{i}{j}\cup\ccij{j}{i}{j}\cup\cc')$ holds,
and since $\{M'_1\eqC N'_1, M'_2\eqC N'_2\}\subseteq\ccij{i}{i}{j}\cup\ccij{j}{i}{j}\cup\cc'$, we know that
since $M'_1$, $M'_2$ are unifiable, $N'_1\alpha' = N'_2\alpha'$.

Thus $N'_1\alpha'\delta^1_i\delta^2_j = N'_2\alpha'\delta^1_i\delta^2_j$, \ie, since $i\neq j$, and $\delta^1_i\delta^2_j\alpha'\delta^1_i\delta^2_j = \alpha$, $N_1\alpha = N_2\alpha$.
Therefore the claim holds in this case.

\item\case{if $i = j$}
then let $M'_1= M_1\delta^1_i$, $N'_1=N_1\delta^1_i$, $M'_2= M_2\delta^1_i$, $N'_2=N_2\delta^1_i$.
Since $\ccij{i}{i}{i} = \cci{i}\delta^1_i$, we have $M'_1\eqC N'_1\in\ccij{i}{i}{i}$.
Similarly, $M'_2\eqC N'_2\in\ccij{i}{i}{i}$.

Since $M_1$ and $M_2$ are unifiable,
then so are $M'_1$ and $M'_2$, with a most general unifier $\mu'$ which satisfies 
$\mu(x)=t \Leftrightarrow \mu'(x\delta^1_i) = t\delta^1_i$.

Let then $\alpha'$ be the restriction of $\mu'$ to
$\{x\in\var{M_1'}\cup\var{M_2'}\;|\; \Gammaoij{i}{i}{i} (x)=\L \;\wedge\; \mu'(x)\in\N \text{ is a nonce}\}$.

Similarly $\alpha'$ is such that $\forall x\in\dom{\alpha'}.\forall n.\; \alpha(x) = n \Leftrightarrow \alpha'(x\delta^1_i) = n\delta^1_i$, \ie
$\delta^1_i\alpha'\delta^1_i = \alpha$.

By assumption,
$\instCst{C}{2}{n}$ is not empty, hence there exists $(c'', \Gamma'')\in\instCst{C}{2}{n}$.
Thus, $(\cij{i}{i}{i}\cup c'' \cup c', \Gammaij{i}{i}{i} \cup \Gamma'' \cup \Gamma') \in \Cudp$.
Hence, by assumption, $\checkconst(\{(\cij{i}{i}{i}\cup c'' \cup c', \Gammaij{i}{i}{i} \cup \Gamma'' \cup \Gamma')\})$ succeeds.
If $\cci{}''=\stepII_{\Gammao}(fst(\stepI_{\Gamma''}(c'')))$,
then $\ccij{i}{i}{i}\cup \cci{}'' \cup \cc'=\stepII_{\Gammao}(fst(\stepI_{\Gammaij{i}{i}{i} \cup \Gamma'' \cup \Gamma'}(\cij{i}{i}{i}\cup c''\cup c')))$.

Thus, $\stepIV_{\Gammaoij{i}{i}{i} \cup \Gammao'' \cup \Gammao'}(\ccij{i}{i}{i}\cup \cci{}'' \cup \cc')$ holds, and 
since $\{M'_1\eqC N'_1, M'_2\eqC N'_2\}\subseteq\ccij{i}{i}{i}\cup\cci{}''\cup \cc'$, we know that
$N'_1\alpha' = N'_2\alpha'$.

Thus $N'_1\alpha'\delta^1_i = N'_2\alpha'\delta^1_i$, \ie, since $\delta^1_i\alpha'\delta^1_i$,
$N_1\alpha = N_2\alpha$.
Therefore the claim holds in this case.
\end{itemize}

\medskip

Let us now prove the case where there exists $i\in\unn$ such that $M_1\eqC N_1\in\cci{i}$, and $M_2\eqC N_2\in\cc'$.
The symmetric case, where $M_1\eqC N_1\in\cc'$ and there exists $j\in\unn$ such that $M_2\eqC N_2\in\cci{j}$, is similar.

Let then $M'_1= M_1\delta^1_i$, $N'_1=N_1\delta^1_i$, $M'_2= M_2\delta^1_i$, $N'_2=N_2\delta^1_i$.
Since $\ccij{i}{i}{i} = \cci{i}\delta^1_i$, we have $M'_1\eqC N'_1\in\ccij{i}{i}{i}$.
Similarly, $M'_2\eqC N'_2\in\ccp{i}$.

Since $M_1$ and $M_2$ are unifiable,
then so are $M'_1$ and $M'_2$, with a most general unifier $\mu'$ which satisfies 
$\mu(x)=t \Leftrightarrow \mu'(x\delta^1_i) = t\delta^1_i$.

Let then $\alpha'$ be the restriction of $\mu'$ to
$\{x\in\var{M_1'}\cup\var{M_2'}\;|\; (\Gammaoij{i}{i}{i}\cup\Gammap{i})(x)=\L \;\wedge\; \mu'(x)\in\N \text{ is a nonce}\}$.

Similarly $\alpha'$ is such that $\forall x\in\dom{\alpha'}.\forall n.\; \alpha(x) = n \Leftrightarrow \alpha'(x\delta^1_i) = n\delta^1_i$, \ie
$\delta^1_i\alpha'\delta^1_i = \alpha$.

By assumption,
$\instCst{C}{2}{n}$ is not empty, hence there exists $(c'', \Gamma'')\in\instCst{C}{2}{n}$.
Moreover, as noted previously, $(c',\Gammap{i})\in\instCst{C'}{1}{n}$.
Thus, $(\cij{i}{i}{i}\cup c'' \cup c', \Gammaij{i}{i}{i} \cup \Gamma'' \cup \Gammap{i}) \in \Cudp$.

Hence, by assumption,
$\checkconst(\{(\cij{i}{i}{i}\cup c'' \cup c', \Gammaij{i}{i}{i} \cup \Gamma'' \cup \Gammap{i})\})$ succeeds.
If $\cci{}''=\stepII_{\Gammao}(fst(\stepI_{\Gamma''}(c'')))$,
then $\ccij{i}{i}{i}\cup \cci{}'' \cup \ccp{i}=\stepII_{\Gammao}(fst(\stepI_{\Gammaij{i}{i}{i} \cup \Gamma'' \cup \Gammap{i}}(\cij{i}{i}{i}\cup c''\cup c')))$.

Thus, $\stepIV_{\Gammaoij{i}{i}{i} \cup \Gammao'' \cup \Gammaop{i}}(\ccij{i}{i}{i}\cup \cci{}'' \cup \ccp{i})$ holds,
and since $\{M'_1\eqC N'_1, M'_2\eqC N'_2\}\subseteq\ccij{i}{i}{i}\cup\cci{}''\cup \ccp{i}$, we know that
$N'_1\alpha' = N'_2\alpha'$.

Thus $N'_1\alpha'\delta^1_i = N'_2\alpha'\delta^1_i$, \ie, since $\delta^1_i\alpha'\delta^1_i=\alpha$,
$N_1\alpha = N_2\alpha$.

\bigskip

Finally, only the case where $M_1\eqC N_1\in \cc'$ and $M_2\eqC N_2\in\cc'$ remains.
By assumption,
$\instCst{C}{1}{n}$ and $\instCst{C}{2}{n}$ are not empty, hence there exist
$(c'', \Gamma'')\in\instCst{C}{1}{n}$ and $(c''', \Gamma''')\in\instCst{C}{2}{n}$.
Thus, $(c''\cup c''' \cup c', \Gamma'' \cup \Gamma''' \cup \Gamma') \in \Cudp$.

Hence, by assumption, $\checkconst(\{(c''\cup c''' \cup c', \Gamma'' \cup \Gamma''' \cup \Gamma')\})$ succeeds.
If $\cc''=\stepII_{\Gammao}(fst(\stepI_{\Gamma''}(c'')))$ and $\cc'''=\stepII_{\Gammao}(fst(\stepI_{\Gamma'''}(c''')))$,
then $\cc''\cup \cc''' \cup \cc'=\stepII_{\Gammao}(fst(\stepI_{\Gamma'' \cup \Gamma''' \cup \Gamma'}(c''\cup c'''\cup c')))$.

Thus, $\stepIV_{\Gammao'' \cup \Gammao''' \cup \Gammao'}(\cc''\cup \cc''' \cup \cc')$ holds, and since $\{M_1\eqC N_1, M_2\eqC N_2\}\subseteq\cc''\cup\cc'''\cup \cc'$, we know that
$N_1\alpha =  N_2\alpha$.

Therefore the claim holds in this case, which concludes the proof that $\stepIV_{\Gammao}(\cc)$ holds.

\end{enumerate}

Therefore, for every $(c,\Gamma)\in(\UCn)\UnionCart \instCst{C'}{1}{n}$, $\checkconst(\{(c,\Gamma)\})$ succeeds,
which proves the claim.
\end{proof}

\bigskip

\bigskip
\bigskip

This next lemma is a more general version of Theorem~\ref{lem:consistency-stars-sound}.

\begin{lemma}
\label{lem-proof:consistency-stars-sound}
For all 
$(c,\Gamma)$ such that
$\var{c} \subseteq \dom{\Gamma}$ which only contains variables indexed by 1 or 2, and all names in $c$ have finite nonce types,
if $\checkconststar(\{(c,\Gamma)\})$ succeeds,
then for all $\Gamma''\in\branch{\Gamma'}$,
where $\Gamma' = \Gamma[\bigvee_{1\leq i\leq n} \LRTn{l}{1}{m_i}{l'}{1}{p_i} \;/\; \LRTn{l}{\infty}{m}{l'}{\infty}{p}]_{m,p\in\N}$,
$\checkconst(\{(c,\Gamma'')\})$ succeeds.
\end{lemma}
\begin{proof}

\newcommand{\bc}{\overline{c}}
\newcommand{\bci}[1]{\overline{c_{#1}}}
\newcommand{\cci}[1]{\widetilde{c_{#1}}}

Let $n\in\mathbb{N}$.

Let $(c,\Gamma)$ be as assumed in the statement of the lemma.

Let us assume that $\checkconststar(\{(c,\Gamma)\})$ succeeds.
Let $\Gamma' = \Gamma[\bigvee_{1\leq i\leq n} \LRTn{l}{1}{m_i}{l'}{1}{p_i} \;/\; \LRTn{l}{\infty}{m}{l'}{\infty}{p}]_{m,p\in\N}$,
and let $\Gamma''\in\branch{\Gamma'}$.

\medskip

The procedure $\checkconst(\{(c, \Gamma'')\})$ is as follows:
\begin{enumerate}
%
%
%
%

\bigskip
\item\label{item:starinstantiation} We compute $(\c,\Gammao'')=\stepI_{\Gamma''}(c)$.
Following the notations in the procedure, we denote
\[F = \{x\in\dom{\Gamma''} \;|\;\exists m,p,l,l'.\;\Gamma''(x) = \LRTnewnew{\noncetypelab{l}{1}{m}}{\noncetypelab{l'}{1}{p}}\}.
\]
Let $F' = \{x\in\dom{\Gamma} \;|\;\exists m,p,l,l'.\;\Gamma(x) = \LRTnewnew{\noncetypelab{l}{1}{m}}{\noncetypelab{l'}{1}{p}}\}$;
and $F'' = \{x\in\dom{\Gamma} \;|\;\exists m,p,l,l'.\;\Gamma(x) = \LRTnewnew{\noncetypelab{l}{\infty}{m}}{\noncetypelab{l'}{\infty}{p}}\}$.

It is easily seen from the definition of $\Gamma'$ that $F = F'\uplus F''$.

By definition of $\stepI_{\Gamma''}(c)$, $\Gammao''$ contains $\Gamma''|_{\dom{\Gamma''}\backslash F}$.

Let $(\c',\Gammao) = \stepI_{\Gamma}(c)$.

It is clear from the definitions of $\Gamma'$ and $\Gamma''$
that for all $x\in F''$, there exists $i\in\unn$ and $m, p, l, l'$ such that $\Gamma(x) = \LRTnewnew{\noncetypelab{l}{\infty}{m}}{\noncetypelab{l'}{\infty}{p}}$
and $\Gamma''(x)=\LRTnewnew{\noncetypelab{l}{1}{m_i}}{\noncetypelab{l'}{1}{p_i}}$.
Let $\sigma_\Lleft$ and $\sigma_\Rright$ be the substitutions defined by
\[\dom{\sigma_\Lleft}=\dom{\sigma_\Rright} = F''\]
and
\[\forall x\in F''.\forall m,p \in \N.\;\forall l,l'.\; \forall i\in\unn.\; \Gamma''(x)=\LRTnewnew{\noncetypelab{l}{1}{m_i}}{\noncetypelab{l'}{1}{p_i}} \Rightarrow
(\sigma_\Lleft(x)=m_i \;\wedge\;\sigma_\Rright(x)=p_i).\]

It is clear from the definition of $\c$ and $\c'$
that $\c = \inst{\c'}{\sigma_\Lleft}{\sigma_\Rright}$.

\bigskip
\item\label{item:starsplitting} We compute $\cc\eqdef\stepII_{\Gammao}(\c)$.

Similarly, let $\cc'\eqdef\stepII_{\Gammao}(\c')$.

It can easily be seen by induction on the reduction $(\c,\emptyset)\splitr{\Gammao}^*(c_1,c_2)$
(using the fact that $\splitr{\Gammao} = \splitrd$)
that $\cc = \inst{\cc'}{\sigma_\Lleft}{\sigma_\Rright}$.

\bigskip
\item We check that $\stepIII_{\Gammao''}(\cc)$ holds.

Let $u\eqC v\in\cc$.
Since $\cc = \inst{\cc'}{\sigma_\Lleft}{\sigma_\Rright}$,
there exists $u'\eqC v' \in\cc'$ such that $u=u'\sigma_\Lleft$ and $v=v'\sigma_\Rright$.

Since $\checkconststar(\{(c,\Gamma)\})=\mathtt{true}$,
we know that $u'$ and $v'$ have the required form.
Note that by definition of $\Gammao''$,
the keys which are low in $\Gammao''$, \ie the keys $k\in\K$ such that there exist $T$ such that $\Gammao''(k)=\skey{\L}{T}$,
are exactly the keys which are low in $\Gammao$.

It clearly follows, by examining all cases for $u'$ and $v'$,
that $u'\sigma_\Lleft$ and $v'\sigma_\Rright$, \ie $u$ and $v$,
also have the required form.

Therefore, $\stepIII_{\Gammao''}(\cc)$ holds.

\bigskip
\item Finally, we check the condition $\stepIV_{\Gammao''}(\cc)$.

Let $M_1\eqC N_1 \in \cc$ and $M_2\eqC N_2 \in \cc$.
Since $\cc = \inst{\cc'}{\sigma_\Lleft}{\sigma_\Rright}$,
there exist $M_1'\eqC N_1'\in \cc'$ and $M_2'\eqC N_2'\in \cc'$ such that
$M_1=M_1'\sigma_\Lleft$, $N_1=N_1'\sigma_\Rright$,
$M_2=M_2'\sigma_\Lleft$, and $N_2=N_2'\sigma_\Rright$.

Let us prove the first direction of the equivalence, \ie the case where $M_1$, $M_2$ are unifiable.
The proof for the case where $N_1$, $N_2$ are unifiable is similar.

If $M_1$, $M_2$ are unifiable, let $\mu$ be their most general unifier.
We have $M_1\mu = M_2\mu$, \ie $(M_1'\sigma_\Lleft)\mu = (M_2'\sigma_\Lleft)\mu$.

Let $\tau$ denote the substitution $\sigma_\Lleft\mu$.
Since $M_1'\tau = M_2'\tau$, $M_1'$ and $M_2'$ are unifiable.
Let $\mu'$ be their most general unifier.
There exists $\theta$ such that $\tau = \mu'\theta$.

Let also $\alpha$ be the restriction of $\mu$ to
$\{x\in\var{M_1}\cup\var{M_2}\;|\; \Gammao''(x)=\L \;\wedge\; \mu(x)\in\N\}$.

Note that $\Gammao''(x) = \L \Leftrightarrow \Gammao(x)=\L$.

We have to prove that $N_1\alpha = N_2\alpha$.

\medskip

Let $x\in\var{M_1'}\cup\var{M_2'}$ such that there exist $m, p, l, l'$ such that $\Gammao(x)=\LRTn{l}{\infty}{m}{l'}{\infty}{p}$, \ie $x\in F''$.
By definition of $\sigma_\Lleft$ (point~\ref{item:starinstantiation}),
there exists $i\in\unn$ such that $x\sigma_\Lleft = m_i$ (and $x\sigma_\Rright=p_i$).
Hence, we have
\[(x\mu')\theta = x\tau = x\sigma_\Lleft\mu = (x\sigma_\Lleft)\mu = m_i\mu = m_i.\]
Thus, $x\mu'$ can only be either a variable $y$ such that $y\theta=m_i$,
or the nonce $m_i$.

Therefore, $\mu'$ satisfies the conditions on the most general unifier expressed in $\stepIVstar_{\Gammao}(\cc')$.

\medskip

Let $x\in\var{M_1}\cup\var{M_2}$ such that $\Gammao''(x)=\L$ and $\mu(x)\in\N$.
We have $(x\mu')\theta = x\tau = x\sigma_\Lleft\mu = (x\sigma_\Lleft)\mu = x\mu = \mu(x)\in\N$.
Thus, $x\mu'$ can only be either a variable $y$ (such that $y\theta=\mu(x)$),
or the nonce $\mu(x)$.

Conversely, let $x\in\var{M_1'}\cup\var{M_2'}$ such that $\Gammao(x)=\L$ and $\mu'(x)\in\N$.
We have $x\mu = (x\sigma_\Lleft)\mu = x\tau = (x\mu')\theta = \mu'(x)$.

\medskip

Let then $\theta'$ be the substitution with domain 
$\{x\in\var{M_1'}\cup\var{M_2'}\;|\; \exists m, p, l, l'.\exists i\in\llbracket 1, n\rrbracket.\;\Gammao(x)=\LRTn{l}{\infty}{m}{l'}{\infty}{p} \;\wedge\; \mu'(x)=m_i\}$
such that $\forall x\in\dom{\theta'}.\;\theta'(x)=p_i$ if $\mu'(x)=m_i$ and $\Gammao(x)=\LRTn{l}{\infty}{m}{l'}{\infty}{p}$.

Let also $\alpha'$ be the restriction of $\mu'$ to
$\{x\in\var{M_1'}\cup\var{M_2'}\;|\; \Gammao(x)=\L \;\wedge\; \mu'(x)\in\N\}$.

Since $\checkconststar(\{(c,\Gamma)\}) = \mathtt{true}$, we know that
$\stepIVstar_{\Gammao}(\cc')$ holds.
Since $M_1'\eqC N_1'\in\cc'$, and $M_2'\eqC N_2'\in\cc'$,
this implies that $N_1'\alpha'\theta' = N_2'\alpha'\theta'$.

\medskip

As we have just shown, for all $x\in\dom{\theta'}$, there exists $i\in\unn$ such that $x\sigma_\Lleft = m_i$ and
$x\sigma_\Rright=p_i$, and $\mu'(x)$ is either $m_i$ or a variable.
By definition of $\dom{\theta'}$, only the case where $\mu'(x)=m_i$ is actually possible, and we have $\theta'(x)=p_i$.

Thus, $\forall x\in\dom{\theta'}.\;\sigma_\Rright(x)=\theta'(x)$.

It then is clear from the definitions of the domains of $\theta'$ and $\sigma_\Rright$ that there exists $\tau'$ such that
$\sigma_\Rright = \theta'\tau'$.

\medskip

Thus, since we have shown that $N_1'\alpha'\theta' = N_2'\alpha'\theta'$, we have $(N_1'\alpha'\theta')\tau' = (N_2'\alpha'\theta')\tau'$, that is to say $N_1'\alpha'\sigma_\Rright = N_2'\alpha'\sigma_\Rright$,
\ie, since $\alpha'$ and $\sigma_\Rright$ have disjoint domains, and are both ground, $N_1\alpha' = N_2\alpha'$.

\medskip

Moreover, we have shown that 
for all $x\in\var{M_1'}\cup\var{M_2'}$ such that $\Gammao(x)=\L$ and $\mu'(x)\in\N$, $\mu(x) = \mu'(x)$.
That is to say that
for all $x\in\dom{\alpha'}$, $\mu(x) = \alpha'(x)$.

In addition, it is clear from the definition of $\sigma_\Lleft$ that
\[\{x\in\var{M_1'}\cup\var{M_2'}\;|\; \Gammao(x)=\L\} = \{x\in\var{M_1}\cup\var{M_2}\;|\; \Gammao''(x)=\L\}.\]

Hence
\begin{align*}
\dom{\alpha} &= \{x\in\var{M_1}\cup\var{M_2}\;|\; \Gammao''(x)=\L \;\wedge\;\mu(x)\in\N\}\\
&= \{x\in\var{M_1'}\cup\var{M_2'}\;|\; \Gammao(x)=\L\;\wedge\;\mu(x)\in\N\}\\
&\supseteq \{x\in\var{M_1'}\cup\var{M_2'}\;|\; \Gammao(x)=\L\;\wedge\;\mu(x)\in\N\;\wedge\;\mu'(x)\in\N\}\\
&= \{x\in\var{M_1'}\cup\var{M_2'}\;|\; \Gammao(x)=\L\;\wedge\;\mu'(x)\in\N\}\\
&=\dom{\alpha'}.
\end{align*}

Therefore, $\forall x\in\dom{\alpha'}.\; x\in\dom{\alpha}\;\wedge\;\alpha'(x)=\alpha(x)$.
Thus there exists $\alpha''$ such that $\alpha = \alpha'\alpha''$.

\medskip
Since we already have $N_1\alpha' = N_2\alpha'$, this implies that $N_1\alpha = N_2\alpha$, 
which concludes the proof that $\stepIV_{\Gammao''}(\cc)$ holds.
Hence, $\checkconst(\{(c,\Gamma'')\}) = \mathtt{true}$.
\end{enumerate}
\end{proof}

\bigskip
\bigskip
We can now prove the following theorem:

\begin{theorem}
Let $C$, and $C'$ be two constraint sets without any common variable.
\[\checkconststar(\instCstr{C}{1}\UnionCart\instCstr{C}{2}\UnionCart\instCstr{C'}{1}) = \mathtt{true}\;\Rightarrow
\forall n. \;\instCst{C'}{1}{n}\UnionCart(\UCn)\text{ is consistent.}
\]
\end{theorem}
\begin{proof}
Assume $\checkconststar(\instCstr{C}{1}\UnionCart\instCstr{C}{2}\UnionCart\instCstr{C'}{1}) = \mathtt{true}$.
Let $n>0$. Let us show that $\instCst{C'}{1}{n}\UnionCart(\UCn)$ is consistent.

By Theorem~\ref{thm-proof:soundness-proc-consistency}, it suffices to show that
$\checkconst(\instCst{C'}{1}{n}\UnionCart(\UCn)) = \mathtt{true}$.

By Theorem~\ref{thm-proof:consistency-two-to-n}, it suffices to show that 
$\checkconst(\Ci\UnionCart \Cii\UnionCart\instCst{C'}{1}{n}) = \mathtt{true}$.

By assumption, we know that
$\checkconststar(\instCstr{C}{1}\UnionCart\instCstr{C}{2}\UnionCart\instCstr{C'}{1}) = \mathtt{true}$.

That is to say, for each $(c_1,\Gamma_1)\in C$, $(c_2,\Gamma_2)\in C$, $(c_3,\Gamma_3)\in C'$,
if $c' = \instConst{c_1}{1}{\Gamma_1} \cup \instConst{c_2}{2}{\Gamma_2} \cup \instConst{c_3}{1}{\Gamma_3}$,
and $\Gamma' = \instGr{\Gamma_1}{1}\cup\instGr{\Gamma_2}{2}\cup\instGr{\Gamma_3}{1}$,
$\checkconststar(\{(c',\Gamma')\}) = \mathtt{true}$.

Thus, by Lemma~\ref{lem-proof:consistency-stars-sound},
for all $(c_1,\Gamma_1)\in C$, $(c_2,\Gamma_2)\in C$, $(c_3,\Gamma_3)\in C'$,
if $c' = \instConst{c_1}{1}{\Gamma_1} \cup \instConst{c_2}{2}{\Gamma_2} \cup \instConst{c_3}{1}{\Gamma_3}$,
and $\Gamma' = \instG{\Gamma_1}{1}{n}\cup\instG{\Gamma_2}{2}{n}\cup\instG{\Gamma_3}{1}{n}$,
$\checkconststar(\{(c',\Gamma')\}) = \mathtt{true}$.

That is to say, $\checkconststar(\instCst{C}{1}{n}\UnionCart\instCst{C}{2}{n}\UnionCart\instConst{C'}{1}{n}) = \mathtt{true}$,
which concludes the proof.

\end{proof}

%% file: main-ccs.bbl

\begin{thebibliography}{00}


\ifx \showCODEN    \undefined \def \showCODEN     #1{\unskip}     \fi
\ifx \showDOI      \undefined \def \showDOI       #1{#1}\fi
\ifx \showISBNx    \undefined \def \showISBNx     #1{\unskip}     \fi
\ifx \showISBNxiii \undefined \def \showISBNxiii  #1{\unskip}     \fi
\ifx \showISSN     \undefined \def \showISSN      #1{\unskip}     \fi
\ifx \showLCCN     \undefined \def \showLCCN      #1{\unskip}     \fi
\ifx \shownote     \undefined \def \shownote      #1{#1}          \fi
\ifx \showarticletitle \undefined \def \showarticletitle #1{#1}   \fi
\ifx \showURL      \undefined \def \showURL       {\relax}        \fi
\providecommand\bibfield[2]{#2}
\providecommand\bibinfo[2]{#2}
\providecommand\natexlab[1]{#1}
\providecommand\showeprint[2][]{arXiv:#2}

\bibitem[\protect\citeauthoryear{??}{our}{[n. d.]}]%
        {oursite}
 \bibinfo{year}{[n. d.]}\natexlab{}.
\newblock
\showURL{%
\url{https://sites.google.com/site/typesystemeq/}}


\bibitem[\protect\citeauthoryear{Abadi}{Abadi}{2000}]%
        {Abadi2000}
\bibfield{author}{\bibinfo{person}{Martín Abadi}.}
  \bibinfo{year}{2000}\natexlab{}.
\newblock \showarticletitle{Security Protocols and their Properties}.
\newblock In \bibinfo{booktitle}{{\em Foundations of Secure Computation}}.
  \bibinfo{series}{NATO Science Series}, Vol.~\bibinfo{volume}{for the 20th
  International Summer School on Foundations of Secure Computation held in
  Marktoberdorf Germany}. \bibinfo{publisher}{IOS Press},
  \bibinfo{pages}{39--60}.
\newblock


\bibitem[\protect\citeauthoryear{Abadi and Fournet}{Abadi and Fournet}{2001}]%
        {AbadiFournet2001}
\bibfield{author}{\bibinfo{person}{Mart\'{\i}n Abadi} {and}
  \bibinfo{person}{C{\'e}dric Fournet}.} \bibinfo{year}{2001}\natexlab{}.
\newblock \showarticletitle{Mobile Values, New Names, and Secure
  Communication}. In \bibinfo{booktitle}{{\em 28th ACM SIGPLAN-SIGACT Symposium
  on Principles of Programming Languages (POPL'01)}}. \bibinfo{publisher}{ACM},
  \bibinfo{pages}{104--115}.
\newblock


\bibitem[\protect\citeauthoryear{Abadi and Fournet}{Abadi and Fournet}{2004}]%
        {private-auth}
\bibfield{author}{\bibinfo{person}{Mart\'in Abadi} {and}
  \bibinfo{person}{C\'edric Fournet}.} \bibinfo{year}{2004}\natexlab{}.
\newblock \showarticletitle{Private authentication}.
\newblock \bibinfo{journal}{{\em Theoretical Computer Science\/}}
  \bibinfo{volume}{322}, \bibinfo{number}{3} (\bibinfo{year}{2004}),
  \bibinfo{pages}{427 -- 476}.
\newblock


\bibitem[\protect\citeauthoryear{Abadi and Rogaway}{Abadi and Rogaway}{2000}]%
        {AbadiRogaway}
\bibfield{author}{\bibinfo{person}{M. Abadi} {and} \bibinfo{person}{P.
  Rogaway}.} \bibinfo{year}{2000}\natexlab{}.
\newblock \showarticletitle{Reconciling two views of cryptography}. In
  \bibinfo{booktitle}{{\em {I}nternational {C}onference on {T}heoretical
  {C}omputer {S}cience {(IFIP TCS2000)}}}. \bibinfo{pages}{3--22}.
\newblock


\bibitem[\protect\citeauthoryear{Adida}{Adida}{2008}]%
        {helios}
\bibfield{author}{\bibinfo{person}{Ben Adida}.}
  \bibinfo{year}{2008}\natexlab{}.
\newblock \showarticletitle{Helios: web-based open-audit voting}. In
  \bibinfo{booktitle}{{\em 17th conference on Security symposium}} {\em
  (\bibinfo{series}{SS'08})}. \bibinfo{pages}{335--348}.
\newblock


\bibitem[\protect\citeauthoryear{Arapinis, Chothia, Ritter, and Ryan}{Arapinis
  et~al\mbox{.}}{2009}]%
        {MyrtoRFID09}
\bibfield{author}{\bibinfo{person}{M. Arapinis}, \bibinfo{person}{T. Chothia},
  \bibinfo{person}{E. Ritter}, {and} \bibinfo{person}{M. Ryan}.}
  \bibinfo{year}{2009}\natexlab{}.
\newblock \showarticletitle{Untraceability in the applied pi calculus}. In
  \bibinfo{booktitle}{{\em 1st International Workshop on RFID Security and
  Cryptography}}.
\newblock


\bibitem[\protect\citeauthoryear{Arapinis, Chothia, Ritter, and Ryan}{Arapinis
  et~al\mbox{.}}{2010}]%
        {passport}
\bibfield{author}{\bibinfo{person}{Myrto Arapinis}, \bibinfo{person}{Tom
  Chothia}, \bibinfo{person}{Eike Ritter}, {and} \bibinfo{person}{Mark Ryan}.}
  \bibinfo{year}{2010}\natexlab{}.
\newblock \showarticletitle{{A}nalysing unlinkability and anonymity using the
  applied pi calculus}. In \bibinfo{booktitle}{{\em 2nd {IEEE} {C}omputer
  {S}ecurity {F}oundations {S}ymposium ({CSF}'10)}}. \bibinfo{publisher}{{IEEE}
  Computer Society Press}.
\newblock


\bibitem[\protect\citeauthoryear{Arapinis, Cortier, and Kremer}{Arapinis
  et~al\mbox{.}}{2016}]%
        {Myrto3Voters}
\bibfield{author}{\bibinfo{person}{Myrto Arapinis},
  \bibinfo{person}{V\'eronique Cortier}, {and} \bibinfo{person}{Steve Kremer}.}
  \bibinfo{year}{2016}\natexlab{}.
\newblock \showarticletitle{When are three voters enough for privacy
  properties?}. In \bibinfo{booktitle}{{\em 21st {E}uropean {S}ymposium on
  {R}esearch in {C}omputer {S}ecurity (ESORICS'16)}} {\em
  (\bibinfo{series}{Lecture Notes in Computer Science})}.
  \bibinfo{publisher}{Springer}, \bibinfo{address}{Heraklion, Crete},
  \bibinfo{pages}{241--260}.
\newblock


\bibitem[\protect\citeauthoryear{Armando, Basin, Boichut, Chevalier, Compagna,
  Cuellar, Hankes~Drielsma, H\'eam, Kouchnarenko, Mantovani, M\"odersheim, von
  Oheimb, Rusinowitch, Santiago, Turuani, Vigan\`o, and Vigneron}{Armando
  et~al\mbox{.}}{2005}]%
        {avispa}
\bibfield{author}{\bibinfo{person}{A. Armando}, \bibinfo{person}{D. Basin},
  \bibinfo{person}{Y. Boichut}, \bibinfo{person}{Y. Chevalier},
  \bibinfo{person}{L. Compagna}, \bibinfo{person}{J. Cuellar},
  \bibinfo{person}{P. Hankes~Drielsma}, \bibinfo{person}{P.-C. H\'eam},
  \bibinfo{person}{O. Kouchnarenko}, \bibinfo{person}{J. Mantovani},
  \bibinfo{person}{S. M\"odersheim}, \bibinfo{person}{D. von Oheimb},
  \bibinfo{person}{M. Rusinowitch}, \bibinfo{person}{J. Santiago},
  \bibinfo{person}{M. Turuani}, \bibinfo{person}{L. Vigan\`o}, {and}
  \bibinfo{person}{L. Vigneron}.} \bibinfo{year}{2005}\natexlab{}.
\newblock \showarticletitle{{The AVISPA Tool for the automated validation of
  internet security protocols and applications}}. In \bibinfo{booktitle}{{\em
  {17th International Conference on Computer Aided Verification, CAV'2005}}}
  {\em (\bibinfo{series}{Lecture Notes in Computer Science})},
  Vol.~\bibinfo{volume}{3576}. \bibinfo{publisher}{Springer},
  \bibinfo{address}{Edinburgh, Scotland}, \bibinfo{pages}{281--285}.
\newblock


\bibitem[\protect\citeauthoryear{Backes, Catalin, and Maffei}{Backes
  et~al\mbox{.}}{2014}]%
        {Backes:2014:UIR}
\bibfield{author}{\bibinfo{person}{Michael Backes}, \bibinfo{person}{Hritcu
  Catalin}, {and} \bibinfo{person}{Matteo Maffei}.}
  \bibinfo{year}{2014}\natexlab{}.
\newblock \showarticletitle{Union, Intersection and Refinement Types and
  Reasoning About Type Disjointness for Secure Protocol Implementations}.
\newblock \bibinfo{journal}{{\em Journal of Computer Security\/}}
  \bibinfo{volume}{22}, \bibinfo{number}{2} (\bibinfo{date}{March}
  \bibinfo{year}{2014}), \bibinfo{pages}{301--353}.
\newblock
\showISSN{0926-227X}
\showURL{%
\url{http://dl.acm.org/citation.cfm?id=2595841.2595845}}


\bibitem[\protect\citeauthoryear{Backes, Hritcu, and Maffei}{Backes
  et~al\mbox{.}}{2008}]%
        {Backes:2008:AVR}
\bibfield{author}{\bibinfo{person}{Michael Backes}, \bibinfo{person}{Catalin
  Hritcu}, {and} \bibinfo{person}{Matteo Maffei}.}
  \bibinfo{year}{2008}\natexlab{}.
\newblock \showarticletitle{Automated Verification of Remote Electronic Voting
  Protocols in the Applied Pi-Calculus}. In \bibinfo{booktitle}{{\em
  Proceedings of the 2008 21st IEEE Computer Security Foundations Symposium}}
  {\em (\bibinfo{series}{CSF '08})}. \bibinfo{publisher}{IEEE Computer
  Society}, \bibinfo{address}{Washington, DC, USA}, \bibinfo{pages}{195--209}.
\newblock
\showISBNx{978-0-7695-3182-3}
\showDOI{%
\url{https://doi.org/10.1109/CSF.2008.26}}


\bibitem[\protect\citeauthoryear{Baelde, Delaune, and Hirschi}{Baelde
  et~al\mbox{.}}{2015}]%
        {APTE-por}
\bibfield{author}{\bibinfo{person}{David Baelde},
  \bibinfo{person}{St{\'e}phanie Delaune}, {and} \bibinfo{person}{Lucca
  Hirschi}.} \bibinfo{year}{2015}\natexlab{}.
\newblock \showarticletitle{Partial Order Reduction for Security Protocols}. In
  \bibinfo{booktitle}{{\em 26th {I}nternational {C}onference on {C}oncurrency
  {T}heory ({CONCUR}'15)}} {\em (\bibinfo{series}{LIPIcs})},
  Vol.~\bibinfo{volume}{42}. \bibinfo{publisher}{Leibniz-Zentrum f{\"u}r
  Informatik}, \bibinfo{pages}{497--510}.
\newblock


\bibitem[\protect\citeauthoryear{Barthe, Fournet, Gr{\'{e}}goire, Strub, Swamy,
  and B{\'{e}}guelin}{Barthe et~al\mbox{.}}{2014}]%
        {BartheFGSSB14}
\bibfield{author}{\bibinfo{person}{Gilles Barthe},
  \bibinfo{person}{C{\'{e}}dric Fournet}, \bibinfo{person}{Benjamin
  Gr{\'{e}}goire}, \bibinfo{person}{Pierre{-}Yves Strub},
  \bibinfo{person}{Nikhil Swamy}, {and} \bibinfo{person}{Santiago~Zanella
  B{\'{e}}guelin}.} \bibinfo{year}{2014}\natexlab{}.
\newblock \showarticletitle{Probabilistic relational verification for
  cryptographic implementations}. In \bibinfo{booktitle}{{\em 41st Annual {ACM}
  {SIGPLAN-SIGACT} Symposium on Principles of Programming Languages ({POPL}
  '14)}}. \bibinfo{publisher}{ACM}, \bibinfo{pages}{193--206}.
\newblock


\bibitem[\protect\citeauthoryear{Barthe, Gr{\'{e}}goire, and
  {Zanella-B{\'{e}}guelin}}{Barthe et~al\mbox{.}}{2009}]%
        {BartheGB09}
\bibfield{author}{\bibinfo{person}{Gilles Barthe}, \bibinfo{person}{Benjamin
  Gr{\'{e}}goire}, {and} \bibinfo{person}{Santiago {Zanella-B{\'{e}}guelin}}.}
  \bibinfo{year}{2009}\natexlab{}.
\newblock \showarticletitle{Formal certification of code-based cryptographic
  proofs}. In \bibinfo{booktitle}{{\em 36th {ACM} {SIGPLAN-SIGACT} Symposium on
  Principles of Programming Languages, {POPL} 2009, Savannah, GA, USA, January
  21-23, 2009}}. \bibinfo{publisher}{{ACM}}, \bibinfo{pages}{90--101}.
\newblock


\bibitem[\protect\citeauthoryear{Basin, Dreier, and Sasse}{Basin
  et~al\mbox{.}}{2015}]%
        {tamarin-equiv}
\bibfield{author}{\bibinfo{person}{David Basin}, \bibinfo{person}{Jannik
  Dreier}, {and} \bibinfo{person}{Ralf Sasse}.}
  \bibinfo{year}{2015}\natexlab{}.
\newblock \showarticletitle{{Automated Symbolic Proofs of Observational
  Equivalence}}. In \bibinfo{booktitle}{{\em {22nd ACM SIGSAC Conference on
  Computer and Communications Security (ACM CCS 2015)}}}. {ACM},
  \bibinfo{pages}{1144--1155}.
\newblock


\bibitem[\protect\citeauthoryear{Bengtson, Bhargavan, Fournet, Gordon, and
  Maffeis}{Bengtson et~al\mbox{.}}{2011}]%
        {Bengtson:2011}
\bibfield{author}{\bibinfo{person}{Jesper Bengtson},
  \bibinfo{person}{Karthikeyan Bhargavan}, \bibinfo{person}{C{\'e}dric
  Fournet}, \bibinfo{person}{Andrew~D. Gordon}, {and} \bibinfo{person}{Sergio
  Maffeis}.} \bibinfo{year}{2011}\natexlab{}.
\newblock \showarticletitle{Refinement Types for Secure Implementations}.
\newblock \bibinfo{journal}{{\em ACM Transactions on Programming Languages and
  Systems\/}} \bibinfo{volume}{33}, \bibinfo{number}{2} (\bibinfo{year}{2011}),
  \bibinfo{pages}{8:1--8:45}.
\newblock


\bibitem[\protect\citeauthoryear{Benton}{Benton}{2004}]%
        {Benton04}
\bibfield{author}{\bibinfo{person}{Nick Benton}.}
  \bibinfo{year}{2004}\natexlab{}.
\newblock \showarticletitle{Simple Relational Correctness Proofs for Static
  Analyses and Program Transformations}. In \bibinfo{booktitle}{{\em 31st ACM
  SIGPLAN-SIGACT Symposium on Principles of Programming Languages}} {\em
  (\bibinfo{series}{POPL '04})}. \bibinfo{publisher}{ACM},
  \bibinfo{address}{New York, NY, USA}, \bibinfo{pages}{14--25}.
\newblock
\showISBNx{1-58113-729-X}


\bibitem[\protect\citeauthoryear{Blanchet}{Blanchet}{2001}]%
        {proverif}
\bibfield{author}{\bibinfo{person}{Bruno Blanchet}.}
  \bibinfo{year}{2001}\natexlab{}.
\newblock \showarticletitle{An {E}fficient {C}ryptographic {P}rotocol
  {V}erifier {B}ased on {P}rolog {R}ules}. In \bibinfo{booktitle}{{\em 14th
  IEEE Computer Security Foundations Workshop (CSFW-14)}}.
  \bibinfo{publisher}{IEEE Computer Society}, \bibinfo{address}{Cape Breton,
  Nova Scotia, Canada}, \bibinfo{pages}{82--96}.
\newblock


\bibitem[\protect\citeauthoryear{Blanchet}{Blanchet}{2016}]%
        {BlanchetFnTPS16}
\bibfield{author}{\bibinfo{person}{Bruno Blanchet}.}
  \bibinfo{year}{2016}\natexlab{}.
\newblock \showarticletitle{Modeling and Verifying Security Protocols with the
  Applied Pi Calculus and {P}ro{V}erif}.
\newblock \bibinfo{journal}{{\em Foundations and Trends in Privacy and
  Security\/}} \bibinfo{volume}{1}, \bibinfo{number}{1--2}
  (\bibinfo{year}{2016}), \bibinfo{pages}{1--135}.
\newblock


\bibitem[\protect\citeauthoryear{Blanchet, Abadi, and Fournet}{Blanchet
  et~al\mbox{.}}{2008}]%
        {proverif-equiv}
\bibfield{author}{\bibinfo{person}{Bruno Blanchet},
  \bibinfo{person}{Mart{\'\i}n Abadi}, {and} \bibinfo{person}{C{\'e}dric
  Fournet}.} \bibinfo{year}{2008}\natexlab{}.
\newblock \showarticletitle{Automated Verification of Selected Equivalences for
  Security Protocols}.
\newblock \bibinfo{journal}{{\em Journal of Logic and Algebraic Programming\/}}
  \bibinfo{volume}{75}, \bibinfo{number}{1} (\bibinfo{date}{Feb.--March}
  \bibinfo{year}{2008}), \bibinfo{pages}{3--51}.
\newblock


\bibitem[\protect\citeauthoryear{Bugliesi, Calzavara, Eigner, and
  Maffei}{Bugliesi et~al\mbox{.}}{2015}]%
        {Bugliesi:2015:ART}
\bibfield{author}{\bibinfo{person}{Michele Bugliesi}, \bibinfo{person}{Stefano
  Calzavara}, \bibinfo{person}{Fabienne Eigner}, {and} \bibinfo{person}{Matteo
  Maffei}.} \bibinfo{year}{2015}\natexlab{}.
\newblock \showarticletitle{Affine Refinement Types for Secure Distributed
  Programming}.
\newblock \bibinfo{journal}{{\em ACM Transactions on Programming Languages and
  Systems\/}} \bibinfo{volume}{37}, \bibinfo{number}{4}, Article
  \bibinfo{articleno}{11} (\bibinfo{date}{Aug.} \bibinfo{year}{2015}),
  \bibinfo{numpages}{66}~pages.
\newblock
\showISSN{0164-0925}
\showDOI{%
\url{https://doi.org/10.1145/2743018}}


\bibitem[\protect\citeauthoryear{Chadha, Ciob{\^a}c{\u{a}}, and Kremer}{Chadha
  et~al\mbox{.}}{2012}]%
        {akiss}
\bibfield{author}{\bibinfo{person}{Rohit Chadha}, \bibinfo{person}{{{S}}tefan
  Ciob{\^a}c{\u{a}}}, {and} \bibinfo{person}{Steve Kremer}.}
  \bibinfo{year}{2012}\natexlab{}.
\newblock \showarticletitle{Automated verification of equivalence properties of
  cryptographic protocols}. In \bibinfo{booktitle}{{\em {P}rogramming
  {L}anguages and {S}ystems~---21th {E}uropean {S}ymposium on {P}rogramming
  ({ESOP}'12)}} {\em (\bibinfo{series}{Lecture Notes in Computer Science})},
  Vol.~\bibinfo{volume}{7211}. \bibinfo{publisher}{Springer},
  \bibinfo{address}{Tallinn, Estonia}, \bibinfo{pages}{108--127}.
\newblock


\bibitem[\protect\citeauthoryear{Cheval}{Cheval}{2014}]%
        {APTE}
\bibfield{author}{\bibinfo{person}{Vincent Cheval}.}
  \bibinfo{year}{2014}\natexlab{}.
\newblock \showarticletitle{APTE: an Algorithm for Proving Trace Equivalence}.
  In \bibinfo{booktitle}{{\em 20th {I}nternational {C}onference on {T}ools and
  {A}lgorithms for the {C}onstruction and {A}nalysis of {S}ystems
  ({TACAS}'14)}} {\em (\bibinfo{series}{Lecture Notes in Computer Science})},
  Vol.~\bibinfo{volume}{8413}. \bibinfo{address}{Grenoble, France},
  \bibinfo{pages}{587--592}.
\newblock


\bibitem[\protect\citeauthoryear{Cheval, Cortier, and Plet}{Cheval
  et~al\mbox{.}}{2013}]%
        {Length-cav2013}
\bibfield{author}{\bibinfo{person}{Vincent Cheval},
  \bibinfo{person}{V\'eronique Cortier}, {and} \bibinfo{person}{Antoine Plet}.}
  \bibinfo{year}{2013}\natexlab{}.
\newblock \showarticletitle{Lengths may break privacy -- or how to check for
  equivalences with length}. In \bibinfo{booktitle}{{\em 25th International
  Conference on Computer Aided Verification (CAV'13)}} {\em
  (\bibinfo{series}{Lecture Notes in Computer Science})},
  Vol.~\bibinfo{volume}{8043}. \bibinfo{publisher}{Springer},
  \bibinfo{address}{St Petersburg, Russia}, \bibinfo{pages}{708--723}.
\newblock


\bibitem[\protect\citeauthoryear{Clark and Jacob}{Clark and Jacob}{1997}]%
        {clark1997}
\bibfield{author}{\bibinfo{person}{John Clark} {and} \bibinfo{person}{Jeremy
  Jacob}.} \bibinfo{year}{1997}\natexlab{}.
\newblock \bibinfo{title}{A survey of authentication protocol literature:
  Version 1.0}.
\newblock   (\bibinfo{year}{1997}).
\newblock


\bibitem[\protect\citeauthoryear{Comon{-}Lundh and Cortier}{Comon{-}Lundh and
  Cortier}{2008}]%
        {CLC-CCS2008}
\bibfield{author}{\bibinfo{person}{Hubert Comon{-}Lundh} {and}
  \bibinfo{person}{V{\'e}ronique Cortier}.} \bibinfo{year}{2008}\natexlab{}.
\newblock \showarticletitle{Computational Soundness of Observational
  Equivalence}. In \bibinfo{booktitle}{{\em 15th {ACM} {C}onference on
  {C}omputer and {C}ommunications {S}ecurity ({CCS}'08)}}.
  \bibinfo{publisher}{ACM Press}, \bibinfo{address}{Alexandria, Virginia, USA},
  \bibinfo{pages}{109--118}.
\newblock


\bibitem[\protect\citeauthoryear{Cortier, Delaune, and Dallon}{Cortier
  et~al\mbox{.}}{2017}]%
        {SAT-equiv}
\bibfield{author}{\bibinfo{person}{V\'eronique Cortier},
  \bibinfo{person}{St\'ephanie Delaune}, {and} \bibinfo{person}{Antoine
  Dallon}.} \bibinfo{year}{2017}\natexlab{}.
\newblock \showarticletitle{SAT-Equiv: an efficient tool for equivalence
  properties}. In \bibinfo{booktitle}{{\em 30th {IEEE} {C}omputer {S}ecurity
  {F}oundations {S}ymposium ({CSF}'17)}}. \bibinfo{publisher}{{IEEE} Computer
  Society Press}.
\newblock


\bibitem[\protect\citeauthoryear{Cortier, Eigner, Kremer, Maffei, and
  Wiedling}{Cortier et~al\mbox{.}}{2015}]%
        {Cortier:2015:TVE}
\bibfield{author}{\bibinfo{person}{V{\'e}ronique Cortier},
  \bibinfo{person}{Fabienne Eigner}, \bibinfo{person}{Steve Kremer},
  \bibinfo{person}{Matteo Maffei}, {and} \bibinfo{person}{Cyrille Wiedling}.}
  \bibinfo{year}{2015}\natexlab{}.
\newblock \showarticletitle{Type-Based Verification of Electronic Voting
  Protocols}. In \bibinfo{booktitle}{{\em Proceedings of the 4th International
  Conference on Principles of Security and Trust - Volume 9036}}.
  \bibinfo{publisher}{Springer-Verlag New York, Inc.}, \bibinfo{address}{New
  York, NY, USA}, \bibinfo{pages}{303--323}.
\newblock
\showISBNx{978-3-662-46665-0}
\showDOI{%
\url{https://doi.org/10.1007/978-3-662-46666-7_16}}


\bibitem[\protect\citeauthoryear{Cortier, Filipiak, Gharout, and
  Traore}{Cortier et~al\mbox{.}}{[n. d.]}]%
        {CFGT17}
\bibfield{author}{\bibinfo{person}{Veronique Cortier}, \bibinfo{person}{Alicia
  Filipiak}, \bibinfo{person}{Said Gharout}, {and} \bibinfo{person}{Jacques
  Traore}.} \bibinfo{year}{[n. d.]}\natexlab{}.
\newblock \showarticletitle{Designing and proving an EMV-compliant payment
  protocol for mobile devices}. In \bibinfo{booktitle}{{\em 2nd IEEE European
  Symposium on Security and Privacy (EuroS\&P'16)}}. \bibinfo{publisher}{IEEE
  Computer Society}.
\newblock


\bibitem[\protect\citeauthoryear{Cortier, Rusinowitch, and
  Z{\u{a}}linescu}{Cortier et~al\mbox{.}}{2006}]%
        {Cortier2006}
\bibfield{author}{\bibinfo{person}{V{\'e}ronique Cortier},
  \bibinfo{person}{Micha{\"e}l Rusinowitch}, {and} \bibinfo{person}{Eugen
  Z{\u{a}}linescu}.} \bibinfo{year}{2006}\natexlab{}.
\newblock \bibinfo{booktitle}{{\em Relating Two Standard Notions of Secrecy}}.
\newblock \bibinfo{publisher}{Springer Berlin Heidelberg},
  \bibinfo{pages}{303--318}.
\newblock


\bibitem[\protect\citeauthoryear{Cortier and Smyth}{Cortier and Smyth}{2011}]%
        {Helios-CSF11}
\bibfield{author}{\bibinfo{person}{V\'eronique Cortier} {and}
  \bibinfo{person}{Ben Smyth}.} \bibinfo{year}{2011}\natexlab{}.
\newblock \showarticletitle{Attacking and fixing Helios: An analysis of ballot
  secrecy}. In \bibinfo{booktitle}{{\em 24th {IEEE} {C}omputer {S}ecurity
  {F}oundations {S}ymposium ({CSF}'11)}}. \bibinfo{publisher}{{IEEE} Computer
  Society Press}.
\newblock


\bibitem[\protect\citeauthoryear{Cremers}{Cremers}{2008}]%
        {scyther}
\bibfield{author}{\bibinfo{person}{C.J.F. Cremers}.}
  \bibinfo{year}{2008}\natexlab{}.
\newblock \showarticletitle{The {S}cyther {T}ool: Verification, Falsification,
  and Analysis of Security Protocols}. In \bibinfo{booktitle}{{\em Computer
  Aided Verification, 20th International Conference, CAV 2008, Princeton, USA}}
  {\em (\bibinfo{series}{Lecture Notes in Computer Science})},
  Vol.~\bibinfo{volume}{5123/2008}. \bibinfo{publisher}{Springer},
  \bibinfo{pages}{414--418}.
\newblock


\bibitem[\protect\citeauthoryear{Dawson and Tiu}{Dawson and Tiu}{2010}]%
        {SPEC}
\bibfield{author}{\bibinfo{person}{Jeremy Dawson} {and} \bibinfo{person}{Alwen
  Tiu}.} \bibinfo{year}{2010}\natexlab{}.
\newblock \showarticletitle{Automating open bisimulation checking for the
  spi-calculus}. In \bibinfo{booktitle}{{\em IEEE Computer Security Foundations
  Symposium (CSF 2010)}}.
\newblock


\bibitem[\protect\citeauthoryear{Delaune, Kremer, and Ryan}{Delaune
  et~al\mbox{.}}{2009}]%
        {DKR-jcs08}
\bibfield{author}{\bibinfo{person}{St{\'e}phanie Delaune},
  \bibinfo{person}{Steve Kremer}, {and} \bibinfo{person}{Mark~D. Ryan}.}
  \bibinfo{year}{2009}\natexlab{}.
\newblock \showarticletitle{Verifying Privacy-type Properties of Electronic
  Voting Protocols}.
\newblock \bibinfo{journal}{{\em Journal of Computer Security\/}}
  \bibinfo{volume}{17}, \bibinfo{number}{4} (\bibinfo{year}{2009}),
  \bibinfo{pages}{435--487}.
\newblock


\bibitem[\protect\citeauthoryear{Focardi and Maffei}{Focardi and
  Maffei}{2011}]%
        {FM:2011}
\bibfield{author}{\bibinfo{person}{Riccardo Focardi} {and}
  \bibinfo{person}{Matteo Maffei}.} \bibinfo{year}{2011}\natexlab{}.
\newblock \showarticletitle{Types for Security Protocols}.
\newblock In \bibinfo{booktitle}{{\em Formal Models and Techniques for
  Analyzing Security Protocols}}. \bibinfo{series}{Cryptology and Information
  Security Series}, Vol.~\bibinfo{volume}{5}. \bibinfo{publisher}{IOS Press},
  Chapter~7, \bibinfo{pages}{143--181}.
\newblock


\bibitem[\protect\citeauthoryear{Lowe}{Lowe}{1996}]%
        {lowe96breaking}
\bibfield{author}{\bibinfo{person}{Gavin Lowe}.}
  \bibinfo{year}{1996}\natexlab{}.
\newblock \showarticletitle{Breaking and fixing the {N}eedham-{S}chroeder
  public-key protocol using {FDR}}. In \bibinfo{booktitle}{{\em Tools and
  Algorithms for the Construction and Analysis of Systems ({TACAS'96})}} {\em
  (\bibinfo{series}{LNCS})}, Vol.~\bibinfo{volume}{1055}.
  \bibinfo{publisher}{Springer-Verlag}, \bibinfo{pages}{147--166}.
\newblock


\bibitem[\protect\citeauthoryear{Maffei, Pecina, and Reinert}{Maffei
  et~al\mbox{.}}{2013}]%
        {Maffei:2013:SPD}
\bibfield{author}{\bibinfo{person}{Matteo Maffei}, \bibinfo{person}{Kim
  Pecina}, {and} \bibinfo{person}{Manuel Reinert}.}
  \bibinfo{year}{2013}\natexlab{}.
\newblock \showarticletitle{Security and Privacy by Declarative Design}. In
  \bibinfo{booktitle}{{\em Proceedings of the 2013 IEEE 26th Computer Security
  Foundations Symposium}} {\em (\bibinfo{series}{CSF '13})}.
  \bibinfo{publisher}{IEEE Computer Society}, \bibinfo{address}{Washington, DC,
  USA}, \bibinfo{pages}{81--96}.
\newblock
\showISBNx{978-0-7695-5031-2}
\showDOI{%
\url{https://doi.org/10.1109/CSF.2013.13}}


\bibitem[\protect\citeauthoryear{M.Boreale, Nicola, and Pugliese}{M.Boreale
  et~al\mbox{.}}{2002}]%
        {trace-equivalence}
\bibfield{author}{\bibinfo{person}{M.Boreale}, \bibinfo{person}{R.~D. Nicola},
  {and} \bibinfo{person}{R. Pugliese}.} \bibinfo{year}{2002}\natexlab{}.
\newblock \showarticletitle{Proof techniques for cryptographic processes}.
\newblock \bibinfo{journal}{{\it SIAM J. Comput.}} \bibinfo{volume}{31},
  \bibinfo{number}{3} (\bibinfo{year}{2002}), \bibinfo{pages}{947--986}.
\newblock


\bibitem[\protect\citeauthoryear{Meier, Schmidt, Cremers, and Basin}{Meier
  et~al\mbox{.}}{2013}]%
        {tamarin}
\bibfield{author}{\bibinfo{person}{Simon Meier}, \bibinfo{person}{Benedikt
  Schmidt}, \bibinfo{person}{Cas Cremers}, {and} \bibinfo{person}{David
  Basin}.} \bibinfo{year}{2013}\natexlab{}.
\newblock \showarticletitle{{The TAMARIN Prover for the Symbolic Analysis of
  Security Protocols}}. In \bibinfo{booktitle}{{\em Computer Aided
  Verification, 25th International Conference, CAV 2013, Princeton, USA}} {\em
  (\bibinfo{series}{Lecture Notes in Computer Science})},
  Vol.~\bibinfo{volume}{8044}. \bibinfo{publisher}{Springer},
  \bibinfo{pages}{696--701}.
\newblock


\bibitem[\protect\citeauthoryear{Roenne}{Roenne}{2016}]%
        {Roenne}
\bibfield{author}{\bibinfo{person}{Peter Roenne}.}
  \bibinfo{year}{2016}\natexlab{}.
\newblock \bibinfo{howpublished}{Private communication}.
  (\bibinfo{year}{2016}).
\newblock


\bibitem[\protect\citeauthoryear{Santiago, Escobar, Meadows, and
  Meseguer}{Santiago et~al\mbox{.}}{2014}]%
        {maude-equiv}
\bibfield{author}{\bibinfo{person}{Sonia Santiago}, \bibinfo{person}{Santiago
  Escobar}, \bibinfo{person}{Catherine~A. Meadows}, {and}
  \bibinfo{person}{Jos\'e Meseguer}.} \bibinfo{year}{2014}\natexlab{}.
\newblock \showarticletitle{{A Formal Definition of Protocol
  Indistinguishability and Its Verification Using Maude-NPA}}. In
  \bibinfo{booktitle}{{\em STM 2014}} {\em (\bibinfo{series}{LNCS})}.
  \bibinfo{pages}{162--177}.
\newblock


\bibitem[\protect\citeauthoryear{Schmidt, Meier, Cremers, and Basin}{Schmidt
  et~al\mbox{.}}{2012}]%
        {SchmidtMCB12}
\bibfield{author}{\bibinfo{person}{Benedikt Schmidt}, \bibinfo{person}{Simon
  Meier}, \bibinfo{person}{Cas J.~F. Cremers}, {and} \bibinfo{person}{David~A.
  Basin}.} \bibinfo{year}{2012}\natexlab{}.
\newblock \showarticletitle{Automated Analysis of Diffie-Hellman Protocols and
  Advanced Security Properties.}. In \bibinfo{booktitle}{{\em 24th IEEE
  Computer Security Foundations Symposium (CSF'12)}}. \bibinfo{publisher}{IEEE
  Computer Society}, \bibinfo{pages}{78--94}.
\newblock


\bibitem[\protect\citeauthoryear{Yang}{Yang}{2007}]%
        {Yang07}
\bibfield{author}{\bibinfo{person}{Hongseok Yang}.}
  \bibinfo{year}{2007}\natexlab{}.
\newblock \showarticletitle{Relational separation logic}.
\newblock \bibinfo{journal}{{\em Theor. Comput. Sci.\/}} \bibinfo{volume}{375},
  \bibinfo{number}{1-3} (\bibinfo{year}{2007}), \bibinfo{pages}{308--334}.
\newblock


\end{thebibliography}
